\documentclass[A4,10pt]{article}
\usepackage[margin=1in,dvips]{geometry}
\usepackage{graphicx,cancel}
\usepackage{graphics}
\usepackage{color, soul}
\usepackage{amsmath, amsthm, slashed}
\usepackage{amssymb}
\usepackage{marvosym}
\usepackage{rotating}
\usepackage{mathrsfs}
\usepackage{upgreek}
\usepackage{epsfig}
\usepackage{verbatim}
\usepackage[affil-it]{authblk}
\usepackage{rotating}
\usepackage{graphicx}
\usepackage{accents}
\usepackage{harmony, slashed}
\usepackage{tikz}

\usepackage{enumerate}

\newtheorem{definition}{Definition}[section]
\newtheorem{theorem}{Theorem}[section]
 
\newtheorem*{conjecture*}{Conjecture}
\newtheorem{corollary}{Corollary}[section]
\newtheorem*{theorem*}{Theorem}
\newtheorem*{corollary*}{Corollary}
\newtheorem{proposition}{Proposition}[subsection]
\newtheorem{lemma}{Lemma}[subsection]
\newtheorem{remark}{Remark}[section]

\newcommand{\Lb}{\underline{L}}

\DeclareMathAlphabet\mathbfcal{OMS}{cmsy}{b}{n}

\title{Boundedness and decay for the Teukolsky equation\\ on Kerr spacetimes I: the case $|a|\ll M$}

\author[1,2]{Mihalis Dafermos}
\author[3]{Gustav Holzegel}
\author[2]{Igor Rodnianski}

\affil[1]{\small University of Cambridge, Department of Pure Mathematics and Mathematical
Statistics, Wilberforce~Road,~Cambridge~CB3~0WA,~United~Kingdom\vskip.2pc \ }

\affil[2]{\small Princeton University, Department of Mathematics, Fine~Hall,~Washington~Road,~Princeton,~NJ~08544,~United~States~of~America\vskip.2pc \ }

\affil[3]{\small Imperial College London,
Department of Mathematics,
South~Kensington~Campus,~London~SW7~2AZ,~United~Kingdom\vskip.2pc \ }

\begin{document}

\maketitle

\begin{abstract}
We prove boundedness and polynomial decay statements 
for  solutions of the spin $\pm2$ Teukolsky equation on
a Kerr exterior background with parameters satisfying $|a|\ll M$. 
The bounds are obtained by introducing generalisations of the higher order
quantities $P$ and $\underline{P}$ used in our previous work on the linear stability of Schwarzschild.  The existence of these quantities in the Schwarzschild case is related to the transformation theory
of Chandrasekhar.
In a followup paper, we shall extend this result to the general
sub-extremal range of parameters $|a|<M$. 
As in the  Schwarzschild case, 
these bounds provide  the first step in proving the full linear stability
of the Kerr metric to gravitational perturbations.
\end{abstract}

\tableofcontents

\section{Introduction}
The   stability of the celebrated Schwarzschild~\cite{schwarzschild1916} 
and Kerr metrics~\cite{Kerr}  remains
one of the most important open problems of classical general relativity and has generated a large
number of studies over the years since the pioneering paper of Regge--Wheeler~\cite{Regge}. See~\cite{Mihalisnotes, dafrodlargea} and the introduction of~\cite{holzstabofschw} for recent surveys of the problem. 

The ultimate question is  that of \emph{nonlinear} stability,
that is to say, the dynamic stability of the Kerr family $(\mathcal{M},g_{a,M})$
(including the Schwarzschild case $a=0$), without symmetry assumptions,
as solutions to the Einstein vacuum equations
\begin{equation}
\label{vaceqhere}
{\rm Ric}[g]=0,
\end{equation}
in analogy to the nonlinear stability of Minkowski space, first proven in the monumental~\cite{CK}.
A necessary step to understand nonlinear stability is of course proving suitable versions of
\emph{linear stability}, i.e.~boundedness and decay statements for the
linearisation of $(\ref{vaceqhere})$ around 
the Schwarzschild and Kerr solutions.  This requires first imposing a gauge in which
the equations $(\ref{vaceqhere})$ become well-posed.
A complete study of the linear stability of Schwarzschild in a double null gauge has
been obtained in our recent~\cite{holzstabofschw}. A key step in~\cite{holzstabofschw} 
was proving boundedness and decay for the so-called \emph{Teukolsky equation}, 
to be discussed below in {\bf Section~\ref{Teukinintrosec}},
which can be thought to suitably control 
the ``gauge invariant'' part of the perturbations. See  already
equation~$(\ref{Teukphysicintro})$.
These decay results were 
then used in~\cite{holzstabofschw} to recover appropriate estimates
for the full linearisation of $(\ref{vaceqhere})$.

The purpose of the present paper is to extend the boundedness and decay results of~\cite{holzstabofschw} concerning the  
Teukolsky equation $(\ref{Teukphysicintro})$ from the
Schwarzschild $a=0$ case to the very slowly rotating Kerr case, corresponding to parameters
$|a|\ll M$.  We give a rough
statement of the main result
already in {\bf Section~\ref{herethemainresultintro}} below.

In part II of this series, we shall obtain an analogue of our main theorem for the case 
of general
subextremal Kerr parameters  $|a|<M$. The extremal case $|a|=M$ is
exceptional; see {\bf Section~\ref{otherrelatedsec}} for remarks on this and other related problems.
In a separate paper, following our previous work 
on Schwarzschild~\cite{holzstabofschw}, we will use the above result to show the full linear stability
of the Kerr solution in an appropriate gauge.

We end this introduction in {\bf Section~\ref{outlinesec}} 
with an outline of the paper.

\subsection{The Teukolsky equation for general spin}
\label{Teukinintrosec}
The original approach to linear stability in the Schwarzschild case 
centred on so-called metric perturbations, leading to the decoupled equations
of Regge--Wheeler~\cite{Regge} and Zerilli~\cite{Zerilli!}. The Regge--Wheeler equation
will in fact appear below as formula $(\ref{RWeqintro})$. This approach does not, however, 
appear to easily generalise
to Kerr.
Thus, it was a fundamental  advance when Teukolsky~\cite{teukolsky1973}
 identified two gauge invariant
quantities which decouple
from the full linearisation of $(\ref{vaceqhere})$ in the general Kerr case.
The quantities, corresponding to the extremal curvature components in the Newman--Penrose formalism~\cite{newmanpenrose}, 
can each be expressed by complex scalars $\upalpha^{[\pm2]}$  which
satisfy  a wave equation, now known as the Teukolsky equation:
\begin{eqnarray}
\label{Teukphysicintro}
\nonumber
\Box_g \upalpha^{[s]} +\frac{2s}{\rho^2}(r-M)\partial_r \upalpha^{[s]} +\frac{2s}{\rho^2}
\left(\frac{a(r-M)}{\Delta} +i\frac{\cos \theta}{\sin^2\theta}\right) \partial_\phi \upalpha^{[s]}
+\frac{2s}{\rho^2}\left(\frac{M(r^2-a^2)}{\Delta}-r-ia\cos\theta\right)\partial_t \upalpha^{[s]}\\
+\frac{1}{\rho^2}(s-s^2\cot^2\theta) \upalpha^{[s]}=0,
\end{eqnarray}
with $s=+ 2$ and $-2$ respectively. The scalars are more properly thought of as
spin $\pm2$ weighted quantities.
This generalised an analogous
property in the Schwarzschild case identified by Bardeen and Press~\cite{bardeen1973}. 
 These quantities govern the ``gauge invariant'' part of the perturbations in the sense that
an admissible solution of the linearised Einstein equations whose corresponding $\upalpha^{[\pm2]}$ both vanish
must be a combination of a linearised Kerr solution and a pure
gauge solution \cite{WaldKerr}. 

Note that equation $(\ref{Teukphysicintro})$ can be considered for arbitrary values of
$s\in \frac12\mathbb Z$. For $s=0$, $(\ref{Teukphysicintro})$ reduces to the covariant wave equation $\Box_g\psi=0$, 
while for $s=\pm 1$, $(\ref{Teukphysicintro})$ arises as an equation
satisfied by the extreme components of the Maxwell equations in a null 
frame~\cite{Chandrasekhar}.

\subsubsection{Separability and the mode stability of Whiting and Shlapentokh-Rothman}

An additional remarkable property of the  
Teukolsky 
equation $(\ref{Teukphysicintro})$ is that it can be 
formally separated, in analogy
with  Carter's separation~\cite{carter1968hamilton} of the wave equation (i.e.~the case
of $s=0$).
The  separation of the $\theta$-dependence is
surprising in the case $a\ne 0$ for all $s$ because    
the Kerr metric only admits $\partial_\phi$ and $\partial_t$ as Killing fields. 
It turns out that considering the ansatz
\begin{equation}
\label{separinintro}
\upalpha^{[s]} (r) e^{-i\omega t} S_{m\ell}^{[s]}(a\omega, \cos\theta)   e^{im\phi}
\end{equation}
where $S_{m\ell}^{[s]}(\nu, \cos\theta)$ denote spin-weighted oblate spheroidal harmonics, 
one can
derive from $(\ref{Teukphysicintro})$
an ordinary differential equation for $\upalpha$, which in rescaled form (see (\ref{rescaled})) can be
written as
\begin{equation}
\label{eqforinintro}
u'' +  V^{[s]}(\omega, m, \ell, r) u = 0 
\end{equation}
where for $s\ne 0$, the potential $V^{[s]}$ is complex valued. (Here $'$ denotes
differentiation with respect to $r^*$. See Section~\ref{kerrmetricsection}.)
See already~$(\ref{basic})$.
The separation $(\ref{separinintro})$ 
was subsequently understood to be related to the presence of an additional
 Killing tensor~\cite{kalnins}.

Of course, the problem of decomposing general, 
initially finite-energy solutions of $(\ref{Teukphysicintro})$ as
appropriate superpositions of
$(\ref{separinintro})$ is intimately tied with the validity of boundedness
and decay results, in view of the necessity of taking the Fourier transform in time.
A preliminary question that can be addressed already solely at the level of $(\ref{eqforinintro})$
is that of ``mode stability''.
Mode stability is the statement that there
are no initially finite-energy  solutions of the form $(\ref{separinintro})$
with ${\rm Im}(\omega)>0$. This reduces to showing the non-existence of solutions of
$(\ref{eqforinintro})$ with ${\rm Im}(\omega)>0$ and exponentially decaying boundary
conditions both as $r^*\to \infty$ and $r^*\to-\infty$.

In the case $a=0$, $s=0$, then mode stability can be immediately inferred
by applying the physical space energy estimate associated to the Killing
vector field $\partial_t$
to a solution of the form $(\ref{separinintro})$.
The question is highly nontrivial for $a\ne 0$, already in the
case $s=0$, in view  of the phenomenon of \emph{superradiance},
connected to the presence of the so-called \emph{ergoregion} where
$\partial_t$ is spacelike.  
For $s=\pm 2$, the question is non-trivial
even in the case $a=0$, as there does not exist an obvious conserved energy current.
(In separated form $(\ref{eqforinintro})$, this is related to the fact that the potential $V^{[s]}$
is now complex valued.)  
In a remarkable paper, Whiting~\cite{Whiting} nonetheless succeeded in proving 
mode stability 
for $(\ref{Teukphysicintro})$ for all $s$ in the general
subextremal range of parameters $|a|<M$ 
by cleverly exploiting certain algebraic transformations of the ode~$(\ref{eqforinintro})$.

Mode stability has been extended to exclude ``resonances'' on the real axis,
i.e.~solutions $u$ of $(\ref{eqforinintro})$ with  $\omega\in \mathbb R$ with appropriate boundary
conditions, by 
Shlapentokh-Rothman~\cite{SRT} in the case $s=0$,
who had the insight that the transformations applied in~\cite{Whiting} 
could be extended to the real axis
using the theory of oscillatory integrals.
Together with a continuity argument in $a$, \cite{SRT} can be used
to reprove the original~\cite{Whiting}, and this leads to certain simplifications.
The argument generalises to $s=\pm2$. See also~\cite{doi:10.1063/1.4991656} 
where the techniques
of~\cite{SRT} are combined with an alternative complex analytic treatment.

We emphasise that mode stability is a remarkable property tied to 
the specific form of the equation $(\ref{Teukphysicintro})$ and to the specific Kerr background, even
for $s=0$. Indeed, mode stability fails for $a\ne 0$  when an arbitrarily small Klein--Gordon mass is added, as was first suggested by~\cite{zouros1979instabilities, detweiler1980klein} and proven recently in~\cite{yakovinsta}. 
Even more surprisingly, mode stability fails when a well-chosen positive compactly supported potential is added to $(\ref{Teukphysicintro})$, or when the Kerr metric is itself sufficiently deformed, keeping
however all its symmetries and separation properties, in 
a spatially compact region which can be taken arbitrarily far from the ergoregion~\cite{moschidis2017superradiant}. 

\subsubsection{Previous work on boundedness and decay}
\label{prevworksec}
The quantitative study of the Cauchy problem for $(\ref{Teukphysicintro})$ with $s=0$,
beyond  statements for fixed modes, has become an active field in recent years.
The study for higher spin is still less developed beyond the Schwarzschild case.
We review some relevant previous work below.

\paragraph{The case $s=0$, $|a|<M$.}
An early  result~\cite{KayWald} 
obtained boundedness
for solutions to the Cauchy problem for
the scalar wave equation on Schwarzschild (i.e.~the case $s=0$ and $a=0$
of $(\ref{Teukphysicintro})$) with regular, localised initial data. Even this involved non-trivial considerations on the
event horizon, which can now be understood in a more robust way using
the red-shift energy identity~\cite{DafRod2, Mihalisnotes}.
Following intense activity in the last decade
(e.g.~\cite{MR1972492, DRPrice, Sterbenz, DafRod2, DafRodKerr, Mihalisnotes, DafRodsmalla, Toha2, AndBlue}) 
there are now complete boundedness
and decay results   for $(\ref{Teukphysicintro})$ with $s=0$ in the full subextremal range of Kerr parameters $|a|<M$~\cite{partiii}.

The main difficulties in passing from $a=0$ to $a\ne 0$ arise from superradiance, mentioned
already in the context of mode stability, and
the fact that trapped null geodesics no longer approach a unique value of $r$ in physical space.
The latter is relevant because integrated local energy decay estimates, an important
step in the proof of quantitative decay, must necessarily degenerate at trapping.\footnote{In 
the non-trapping case, such estimates are non-degenerate and can be
derived by classical virial identities whose use  goes 
back to~\cite{morawetz1968time}.}
One way of dealing with the latter difficulty is employing the separation $(\ref{separinintro})$ as
a method of frequency localising integrated local energy decay estimates. 
See~\cite{Mihalisnotes, DafRodsmalla}.
Once such an estimate is obtained, 
the difficulty of superradiance can easily be overcome in the $|a|\ll M$
case as the error terms in the ergoregion are small and can be  absorbed.
For alternative approaches, see~\cite{Toha2, AndBlue}.

The $|a|<M$ case appears a priori to be much more complicated. 
It turns out, however, that the Schwarzschild-like  structure  of trapping
survives, when appropriately
viewed in phase space. Moreover, in the high frequency regime,
one can quantify superradiance with the help of the fact that, quite fortuitously,  superradiant frequencies
happen \underline{not} to be trapped. See~\cite{partiii}. 
These good high frequency properties, 
together with Shlapentokh-Rothman's real mode stability~\cite{SRT}
and a continuity argument in $a$, allow one to extend the exact
same boundedness and integrated local energy decay
results originally obtained on Schwarzschild to the
whole sub-extremal range $|a|<M$ of Kerr parameters. Suitable polynomial decay
then follows from an application of the method of $r^p$ weighted energy
estimates~\cite{DafRodnew, Moschnewmeth}. See~\cite{partiii}. 
For comments on the extremal case $|a|=M$, see Section~\ref{extremality}.

\paragraph{The case $s=\pm2$, $a=0$.}
As we remarked already above, the Teukolsky equation with $s=\pm 2$, $a=0$ has been studied in our previous~\cite{holzstabofschw}
as part of our complete study of the linear stability of Schwarzschild.

The main difficulty of the $s=\pm2$ case as opposed to the case $s=0$, is that, as discussed
already in the context of mode stability, there
does not exist an obvious analogue of the conserved energy associated to the Killing 
field $\partial_t$. Thus,  proving even just boundedness for $a=0$ is non-trivial, even just far away
from the event horizon.
The key to understanding $(\ref{Teukphysicintro})$ for $s=\pm2$, $a=0$ in~\cite{holzstabofschw}
was associating quantities
$P^{[\pm 2]}$ to  $\upalpha^{[\pm2]}$ satisfying $(\ref{Teukphysicintro})$.
These are physical space versions of  transformations first considered
by Chandrasekhar~\cite{Chandrasekhar} 
and are  defined by the expressions
\begin{align}
\label{Pdefshere}
 P^{[+2]} &= -\frac{1}{2(r-2M)} \underline{L} \left( \frac{r^3}{r-2M} \underline{L} \left(\frac{(r-2M)^2}{r} \upalpha^{[+2]}\right)\right) \,  ,  \\
 \label{Pdefsheretwo}
  P^{[-2]} &=-\frac{1}{2(r-2M)} {L} \left( \frac{r^3}{r-2M} {L} \left( r^{-3} \upalpha^{[-2]}\right)\right) \, .
\end{align}
Here $L= \partial_t-\partial_{r*}$, $\underline{L}=\partial_t+\partial_{r*}$ are a null frame,
where $r^*$ is the Regge--Wheeler coordinate.
The quantities $\Psi^{[+2]} = r^3 P^{[+2]} $ and $\Psi^{[-2]}=r^3 P^{[-2]} $ can be shown to
satisfy  the Regge--Wheeler equation\footnote{See Sections \ref{relwiththesystemsec} and \ref{relagainwithDHR} for the precise relation between the tensorial Regge--Wheeler equation defined in \cite{holzstabofschw} and equation (\ref{RWeqintro}). Note in particular that $\Psi^{[+2]}$ and $\overline{\Psi^{[-2]}}$ satisfy the same equation, which explains the appearance of a single Regge--Wheeler equation in \cite{holzstabofschw}. We also note that both the operators $\mathring{\slashed\triangle}^{[\pm 2]}$ and $\mathring{\slashed\triangle}^{[\pm 2]}   \pm 2$ have non-negative eigenvalues. See Sections \ref{spinweigtsec} and \ref{sec:swest}.}
\begin{equation}
\label{RWeqintro}
L\underline{L} \Psi^{[\pm 2]} + \frac{\Omega^2}{r^2}
\left( \mathring{\slashed\triangle}^{[\pm 2]}  \Psi^{[\pm 2]}  \pm 2\right) + \Omega^2\left(\frac{4}{r^2} - \frac{6M}{r^3}\right) \Psi^{[\pm 2]} = 0 \, ,
\end{equation}
where $\mathring{\slashed\triangle}^{[+2]}$ denotes the spin-2-weighted Laplacian on the
unit sphere
\begin{equation} \label{defspinlaplacianintro}
\mathring{\slashed\triangle}^{[\pm 2]}=  -\frac{1}{\sin \theta} \frac{\partial}{\partial \theta} \left(\sin \theta \frac{\partial}{\partial \theta}\right) - \frac{1}{\sin^2 \theta} \partial_\phi^2 + 2 \left(\pm2\right) i\frac{ \cos \theta}{\sin^2 \theta} \partial_\phi + 4 \cot^2 \theta \mp 2 \, .
\end{equation}
Remarkably, $(\ref{RWeqintro})$ is precisely the
same equation which appeared as one of the equations
governing the ``metric
perturbations'' approach discussed at the beginning of Section~\ref{Teukinintrosec}! 

Unlike $(\ref{Teukphysicintro})$ with $s=\pm2$, 
the above equation $(\ref{RWeqintro})$ can  be 
 estimated on Schwarzschild just as for the
wave equation $s=0$, since $(\ref{RWeqintro})$ 
is indeed endowed with the usual structure
of energy estimates. 
In particular, both boundedness and
integrated local energy decay can be obtained. 
(For analysis of $(\ref{RWeqintro})$, see the 
previous~\cite{BlueSoffer, Holzegelspin2} as well
as the self-contained treatment in~\cite{holzstabofschw}.) 
Estimates for $\upalpha^{[\pm2]}$ could then be recovered
directly by integrating $(\ref{Pdefshere})$ as transport equations from initial data.  
Such integration on its own 
would lead, however, to ``loss of derivatives'' in the resulting estimates for $\upalpha^{[\pm2]}$.  
The Teukolsky equation itself $(\ref{Teukphysicintro})$ can be viewed as a further
elliptic relation which allows one to gain back these derivatives, leading finally
to boundedness results without loss of derivative, as well as 
integrated local energy decay and pointwise decay.

We remark  that, beyond $(\ref{Teukphysicintro})$, in the context of the full proof of linear stability in~\cite{holzstabofschw}, further transport equations
and elliptic equations
could then be used to appropriately estimate the remaining gauge dependent quantities.

\paragraph{Other spins.} 
We note that the scheme of~\cite{holzstabofschw}
has recently been applied also to the $s=\pm 1$ case by
Pasqualotto~\cite{pasqualotto2016spin}. This gives an alternative proof of boundedness and 
polynomial decay
for the Maxwell equations on Schwarzschild, proven originally by Blue~\cite{Blue}.
See also~\cite{MR3373052}.
Decay for Maxwell in the case $|a|\ll M$ was obtained in~\cite{andersson2013uniform}.
For a direct treatment of $(\ref{Teukphysicintro})$ for $s=\pm 1$ 
in the case $|a|\ll  M$, generalising some of the results of~\cite{pasqualotto2016spin},
there is the recent~\cite{Ma:2017yui}.
 For the cases $s=\pm1/2$
and $s=\pm3/2$ see~\cite{Smoller2012}. 
See~\cite{finster2003} for the related massive Dirac equation not covered by $(\ref{Teukphysicintro})$.
We note also the papers~\cite{MR3607063, Finster:2016tky}.

\subsection{The main result and first comments on the proof}
\label{herethemainresultintro}

The aim of the present paper is to extend the analysis of $(\ref{Teukphysicintro})$
for $s=\pm 2$ from the Schwarzschild
$a=0$ case considered in~\cite{holzstabofschw}
to the very slowly rotating Kerr case with parameters $|a|\ll M$. 
A rough version of our main result is
the following:

\begin{theorem*}[Rough version]
Let $|a|\ll M$.  Solutions $\upalpha^{[\pm2]}$ 
to the spin $s=\pm2$ Teukolsky equation $(\ref{Teukphysicintro})$ on Kerr exterior 
spacetimes $(\mathcal{M}, g_{a,M})$
arising from regular localised initial data on a Cauchy hypersurface $\Sigma_0$
remain uniformly bounded and satisfy an $r^p$-weighted  energy  
hierarchy and polynomial decay.
\end{theorem*} 

The precise statements embodying the above will be given as 
Theorem~\ref{finalstatetheor}.

The proof of our Theorem 
combines the use of the quantities $P^{[\pm2]}$ 
introduced in our previous~\cite{holzstabofschw} 
with
a simplified version of the
framework introduced in~\cite{DafRodsmalla, partiii} for frequency localised energy estimates,
which as discussed in Section~\ref{prevworksec}
 are useful to capture the obstruction to decay associated with   trapped null geodesics.
(In the special case of axisymmetric solutions, this frequency localisation
can be avoided and our proof can be expressed entirely in physical space. 
See already  Section~\ref{axi-intro}.)

The crucial  observation which allows this technique to work is the following:
In the scheme introduced in~\cite{holzstabofschw},
it is not in fact absolutely necessary 
that the quantities $P^{[\pm2]}$ each satisfy a  completely decoupled equation
 $(\ref{RWeqintro})$. It would
 have been permissible if the equation $(\ref{RWeqintro})$  for $P^{[\pm2]}$ was 
 somehow still coupled to $\upalpha^{[\pm2]}$ on the right hand side,
 provided that this coupling was at a suitable ``lower order'', in the sense
 that these lower order terms can 
indeed be recovered
by the transport (and elliptic equations) which were used in~\cite{holzstabofschw}
to estimate $\upalpha^{[\pm2]}$.

It turns out, remarkably, that when analogues of the quantities
$P^{[\pm2]}$  are defined for Kerr,
even though the exact decoupling from $\upalpha^{[\pm2]}$, respectively,  breaks down, the resulting equations indeed
only couple to $\upalpha^{[\pm2]}$
in the ``weak'' sense described above.

We explain below this structure in more detail, and how it is implemented
in our proof (where we will in fact be able to circumvent use of elliptic estimates).

\subsubsection{The generalisation of $P^{[\pm2]}$ to Kerr}
Our physical-space
definition for $P^{[+2]}$, generalising $(\ref{Pdefshere})$, is given as
\begin{equation}
\label{therelationhere}
 P^{[+2]} = -\frac{(r^2+a^2)^{1/2}}{2\Delta} \underline{L} \left( \frac{(r^2+a^2)^2}{\Delta} \underline{L}\left(\Delta^2 \left(r^2+a^2\right)^{-\frac{3}{2}} \upalpha^{[+2]}\right)\right) \,  . 
\end{equation}
A similar formula holds for $P^{[-2]}$. See already Section~\ref{physicalspacedefsec}.
A  computation reveals that
the rescaled $\Psi^{[+2]}=(r^2+a^2)^{\frac32}P^{[+2]}$ satisfies an equation of the form
\begin{equation}
\label{Reggewheelertype}
\mathfrak{R}^{[+2]} \Psi^{[+2]} =
c_1(r) \partial_\phi (\underline{L}\upalpha^{[+2]}) +c_2(r)
\underline{L} \upalpha^{[+2]}+c_3(r)\partial_\phi \upalpha^{[+2]}
+c_4(r)\upalpha^{[+2]},
\end{equation}
where $\mathfrak{R}^{[+2]}$ is a second order operator defined on Kerr generalising the 
Regge--Wheeler operator appearing on the left hand side of $(\ref{RWeqintro})$, which 
has good divergence properties and thus admits energy currents. 
Consistent with the total decoupling in the Schwarzschild case, 
the coefficient functions $c_i(r)$ above are all $O(|a|)$.
Provided that $\upalpha^{[+2]}$ 
can indeed be viewed as being of two degrees lower in differentiability
than $\Psi^{[+2]}$, then the right hand side is ``zero'th order'' in $\Psi^{[+2]}$.
Let us note, however,
that if we use only the transport relation $(\ref{therelationhere})$, 
then the right hand side of $(\ref{Reggewheelertype})$ 
can only be viewed as  ``first order'' in $\Psi^{[+2]}$,
as integration of $(\ref{therelationhere})$ does not improve
differentiability. Thus, to exploit fully this structure, one must also invoke
in general elliptic relations connecting $\alpha^{[+2]}$ and $\Psi^{[+2]}$ that
can be derived by revisiting equation $(\ref{Teukphysicintro})$ itself.
As we shall see below, it turns out, however, that we shall be able to avoid
invoking this by exploiting more carefully the special structure and the non-degeneration
of the derivative $\partial_{r^*}\Psi^{[\pm2]}$.
We describe 
in Sections~\ref{estim.away}--\ref{odeanalysisintro} how these
terms can be controlled.

We emphasise already that the above structure of
the terms appearing on the right hand side of $(\ref{Reggewheelertype})$ 
is surprising. Upon perturbing $(\ref{RWeqintro})$
one would expect higher order terms in $\Psi^{[\pm2]}$ to appear which cannot be
incorporated in the definition of
$\mathfrak{R}^{[+2]}$ so as to preserve its good divergence properties. 
We note already that in the axisymmetric case, the right hand side
of $(\ref{Reggewheelertype})$ is of even lower order, as the $\partial_\phi$ derivatives
vanish.
The deeper reason why these terms cancel is  
not at all clear. See also the remarks in Section~\ref{finalremhere} below.

\subsubsection{Estimates away from trapping}
\label{estim.away}

Away from trapping, 
it suffices to treat the right hand side
of $(\ref{Reggewheelertype})$ as if it were at the level of a general
``first order'' perturbation in $\Psi^{[+2]}$.

To see this, let us note first that 
suitably away from $r=3M$, 
the $f$ and $y$-multiplier estimate of~\cite{holzstabofschw} leads
in the Schwarzschild case to a coercive spacetime integral containing
\emph{all} first derivatives of $\Psi^{[+2]}$ (with suitable weights towards
the horizon and infinity). This coercivity property away from trapping
is manifestly preserved
to  perturbations to Kerr for $|a|<a_0\ll M$ sufficiently small.  We may
add also a small multiple of the $r^\eta$-current for an $\eta>0$ to generate extra useful
weights near infinity.
Moreover, we may add a suitable multiple of the energy
estimate associated to a vector field
$\partial_t + \chi \upomega_+ \partial_\phi$ which connects
the Hawking vector field on the horizon with the stationary
vector field $\partial_t$. This ensures positive boundary terms on suitable
spacelike and null boundaries,
at the expense of generating an $O(|a|)$ bulk term supported where $\chi'=0$,
which is chosen to be away from trapping. Thus, this bulk term 
can again be absorbed by the coercive terms of the $f$ and $y$-multipliers.

On the other hand, commutation of equation 
$(\ref{therelationhere})$
by the Killing fields $\partial_t$ and $\partial_\phi$ 
allows one to estimate all terms involving
$\upalpha$ and $\underline{L}\upalpha$ and their derivatives
appearing on the right
hand side of $(\ref{Reggewheelertype})$ 
from the spacetime estimate for $\Psi^{[+2]}$
by appropriate transport estimates. (Here, we  note that we
must make use of the extra 
$r^\eta$ weight, just as in~\cite{holzstabofschw}.)
Thus, were it not for trapping,
one could easily absorb
the error terms on the right hand side of $(\ref{Reggewheelertype})$.

\subsubsection{Frequency localised analysis of the coupled system near trapping}
\label{odeanalysisintro}

In view of the above discussion, the terms on the right hand side of
$(\ref{Reggewheelertype})$ are most dangerous near trapping.
Let us take a more careful look at the structure of 
$(\ref{therelationhere})$--$(\ref{Reggewheelertype})$ 
using our frequency analysis.

At the level of the formally separated solutions $(\ref{separinintro})$,
the operator $\underline{L}$ takes the form
\begin{equation}
\label{takes.the.form.sep}
-\underline{L}=\frac{d}{dr^*} +i\omega-\frac{iam}{r^2+a^2},
\end{equation}
where $r^*$ is a Regge--Wheeler type coordinate,
the relation $(\ref{therelationhere})$ reads 
\begin{equation}
\label{inadditionto}
\Psi^{[+2]} = -\frac 12 w^{-1}\Lb \left(w^{-1}\Lb \left(w\cdot u^{[+2]}\right)\right) 
\end{equation}
where 
\begin{equation}
\label{wdefinit}
w:= \frac{\Delta}{(r^2+a^2)^2}
\end{equation}
and the ``Regge--Wheeler'' type equation $(\ref{Reggewheelertype})$ 
takes the form
\begin{equation}
\label{ReggeWheelerlikeintro}
\frac{d^2}{(dr^*)^2}{\Psi^{[+2]}} +(\omega^2 -\mathcal{V} )\Psi^{[+2]} 
   =a\left(\mathfrak{c}_1(r) i m +\mathfrak{c}_2(r) \frac{a}{r} \right) \underline{L} (u^{[+2]} w)
  + a^2 w\left(\mathfrak{c}_3(r) \frac{1}{r} a i m  + \mathfrak{c}_4(r) \right) (u^{[+2]} w).
\end{equation}
Here $\mathcal{V}$ is a real potential depending smoothly on $a$
which reduces to the separated version of the
Regge--Wheeler potential for $a=0$ and the $\mathfrak{c}_i$ are bounded functions. Cf.~(\ref{Reggewheelertype}) and see Appendix \ref{appendix:RW}.

At the separated level, using a frequency localised version of
the current $f$ of~\cite{holzstabofschw}, chosen to vanish at the (frequency-dependent) 
maximum 
of the potential $\mathcal{V}$ as in
of~\cite{DafRodsmalla}, 
together with a frequency localised $y$-current
and the frequency-localised energy estimate (multiplication by
$\omega \Psi$)
one can prove the ODE analogue of
a degenerating integrated local energy decay 
for $\Psi^{[\pm2]}$,  with a right hand side involving the right hand side of
$(\ref{ReggeWheelerlikeintro})$.
Considerations are different in the ``trapped frequency range''
\begin{equation}
\label{trapped.rang}
1\ll \omega^2\sim \lambda_{m \ell}^{[s]} + s , 
\end{equation}
and the non-trapped frequencies.
(Here $\lambda_{m \ell}$ are the eigenvalues of the spin-weighted Laplacian (\ref{defspinlaplacianintro}) reducing to $\ell(\ell+1)-s^2 \geq 2$ in the case $a=0$.)

 In the trapped frequency range $(\ref{trapped.rang})$, 
 the above multiplier gives an estimate which can schematically be written as:
\begin{align}
\label{onlyforthis}
\int_{r \sim 3M} |\Psi^{[+2]}|^2 + |\partial_{r^\star}\Psi^{[+2]}|^2 dr^* \lesssim
   \text{terms controllable by physical space estimates (cf.~Section~\ref{estim.away})} \\
  + |a|\int_{r \sim 3M}( \omega \Psi^{[+2]} + \partial_{r^\star}\Psi^{[+2]})\Big\{ (a i m +1) \underline{L} (u^{[+2]} w)
  + a^2 \left( a i m  + 1 \right) (u^{[+2]} w)\Big\}dr^* \, .
\nonumber
\end{align}
This should be thought of as a degenerate integrated local energy decay bound
for $\Psi^{[+2]}$.
Considering the right hand side of $(\ref{onlyforthis})$, we
note that naive integration of $(\ref{inadditionto})$ as a transport equation 
is not sufficient to control the integral on the right hand side by the left hand side.
This is not surprising: In constrast to the considerations away
from trapping of Section~\ref{estim.away}, in general now only terms
which can be truly thought of as ``zero'th order''  in $\Psi^{[+2]}$ 
can manifestly be absorbed by  the left hand side of $(\ref{onlyforthis})$, in view 
of the absence of an $\omega^2|\Psi^{[+2]}|^2$
and $\Lambda |\Psi^{[+2]}|^2$ coercive term.

One way to try to realise the right hand side of $(\ref{onlyforthis})$ as 
``zero'th order'' in $\Psi^{[+2]}$ would be to 
invoke,
in addition to the transport $(\ref{inadditionto})$, also the elliptic estimates
of~\cite{holzstabofschw}. It turns out, however, that exploiting the presence
of the good first order term $|\partial_r\Psi^{[+2]}|^2$ on the left
hand side of $(\ref{onlyforthis})$, one can
argue in a more elementary manner:
Indeed, by commuting $(\ref{inadditionto})$  with $\partial_{r*}$ and exploiting
the relation $(\ref{takes.the.form.sep})$, one can indeed rewrite the right hand
side so as to absorb it into the left hand side.

Let us note finally
that for ``non-trapped'' frequencies (i.e.~outside the frequency range~$(\ref{trapped.rang})$), 
one can arrange the 
frequency localised multiplier
so that terms $m^2|\Psi^{[+2]}|^2$ and $\omega^2|\Psi^{[+2]}|^2$  appear 
on the left hand side of $(\ref{onlyforthis})$ without degeneration. 
One can then easily absorb the right hand side just as in Section~\ref{estim.away}
treating it essentially as one would a general ``first order'' term.

\subsubsection{Technical comments}
Let us discuss briefly the technical implementation
of the above argument.

As in~\cite{DafRodsmalla},
by using the smallness of  the Kerr parameter $a$, 
the fixed  
frequency analysis of Section~\ref{odeanalysisintro},
restricted entirely to real frequencies $\omega\in \mathbb R$, 
can indeed be implemented to general solutions $\upalpha^{[\pm2]}$ of the 
Cauchy problem for
$(\ref{Teukphysicintro})$, despite the fact that we do not know
a priori that solutions are square integrable in time.
This requires, however, applying cutoffs to $\upalpha$ 
in order to justify the Fourier
transform, and thus one must estimate inhomogeneous versions of 
$(\ref{Teukphysicintro})$ and thus also inhomogeneous versions
of the resulting ODE
$(\ref{ReggeWheelerlikeintro})$.  These inhomogeneous terms must themselves
be bound by the final estimates.

As opposed to the cutoffs of~\cite{DafRodsmalla, partiii}, we here will only
cut off the solution in a region $r^*\in[2A_1^*,2A_2^*]$ near trapping. Thus,
the resulting inhomogeneous terms will be supported in a fixed region of
finite $r^*$. Moreover, the fixed frequency ODE estimates of 
Section~\ref{odeanalysisintro}
will only
be applied in the region $r^*\in [A_1^*,A_2^*]$. They will be combined
with physical space estimates of Section~\ref{estim.away}.
These estimates are now coupled however via boundary terms
on $r=A_1$ and $r=A_2$.  The fixed frequency multipliers applied
to $\Psi^{[+2]}$ are chosen so as to be frequency independent
near $A_1$ and $A_2$ and coincide precisely with those used
in the physical space estimates in the \emph{away} region. As a result,
after summation over frequencies, the boundary terms in the mutliplier
currents exactly cancel. There are also boundary terms associated with the transport 
equations, but these can be absorbed using the smallness of $a$.

The above argument leads to a 
degenerate energy boundedness and integrated local energy
decay  for both $\Psi^{[\pm2]}$ and $\upalpha^{[\pm2]}$.
This preliminary decay bound will be stated  as Theorem~\ref{degenerateboundednessandILED}.
From Theorem~\ref{degenerateboundednessandILED},
we can easily improve our estimates at the event horizon, using
the red-shift technique of~\cite{DafRod2}, and 
then we can easily infer polynomial decay
using the weighted $r^p$ method of~\cite{DafRodnew}---all directly in
physical space.

\subsubsection{The axisymmetric case}
\label{axi-intro}

We have already remarked that in the axisymmetric case $\partial_\phi\upalpha^{[\pm2]}=0$,
the right hand side of $(\ref{Reggewheelertype})$ is of lower order.
An  even more important simplification is that trapped null geodesics all asymptote
to a single value of $r=r_{\rm trap}$ which is near $3M$, independent
of frequency. As a result, there is no need
for frequency-localised analysis and
the whole argument can be expressed entirely in physical space. 
This is convenient for non-linear applications.
We shall explain
how this simplified argument can be explicitly read off from our paper in Section~\ref{axi-note}.

\subsubsection{Final remarks}
\label{finalremhere}

Given the analogue of~\cite{SRT} for $s=\pm 2$, 
the argument
can in principle be applied for the whole subextremal range $|a|<M$
following the continuity argument of~\cite{partiii}, 
but in the present paper we shall only consider the case $|a|\ll M$,
where the lower order terms also have a useful smallness factor bounded by $a$,
and the relevant multiplier currents
can thus be constructed as perturbations of Schwarzschild.
The general case will be considered in part II of this series, 
following the more general constructions of~\cite{partiii}.

There are other generalisations of $P^{[\pm2]}$  to Kerr
which have been considered previously in the literature, see~\cite{Chandrasekhar165, doi:10.1143/PTP.67.1788} and the recent review~\cite{Glampedakis:2017rar}. In contrast to our situation,
the quantities of~\cite{Chandrasekhar165, doi:10.1143/PTP.67.1788} do indeed satisfy decoupled equations, though the transformations must now
be defined in phase space, and the transformed potentials are somewhat 
non-standard in their frequency dependence. 
It would be interesting to find an alternative argument using these
transformations.  We hope to emphasise with our method, however,
that exact decoupling is not absolutely
necessary for  quantities to be useful.

\subsection{Other related results}
\label{otherrelatedsec}
We collect  other related recent results concerning the stability of
black holes. The literature has already become vast so the  list below is in no way
exhaustive. See also the surveys~\cite{Mihalisnotes, dafrodlargea}.

\subsubsection{Metric perturbations}
An alternative approach to 
linear stability in the Schwarzschild case would go through the theory of so-called metric perturbations.
See for instance~\cite{johnson, hung2017linear} for estimates on the additional 
Zerilli equation which must be understood in that approach.  We note the 
paper~\cite{dotti2016black}.

\subsubsection{Canonical energy}
As discussed above, one of the difficulties in understanding linearised
gravity is the lack of an obvious coercive energy quantity for the full system,
even in the $a=0$ case. 
The Lagrangian structure of the Einstein equations 
$(\ref{vaceqhere})$ does give rise however to a notion of canonical energy,
albeit somewhat non-standard in view of diffeomorphism invariance,
and this can indeed be used to infer certain weak stability statements.
For some recent
results which 
have been obtained using this approach, see~\cite{hollands2013stability, prabhu2015black}
and the related~\cite{Holzegelfluxes}.

\subsubsection{Precise power-law asymptotics}
Though one expects that 
the decay bounds obtained here are in principle 
sufficient for non-linear applications, it is of considerable
interest for a wide range of problems
to obtain sharp asymptotics of solutions, of the type first suggested
by~\cite{Price}. For  upper bounds on decay compatible with some of the asymptotics of~\cite{Price},
 see~\cite{Donninger2012, tohaneanu, METCALFE201753}.  Lower bounds
were first obtained in~\cite{LukOhpub}.
 The most satisfying results are the sharp asymptotics recently obtained
by~\cite{Angelopoulos:2016wcv, Angelopoulos:2016moe} for the $s=0$, $a=0$ case.
Such results in particular have applications
to the interior
structure of black holes (see~\cite{LukOhpub}).

\subsubsection{Extremality and the Aretakis instability}
\label{extremality}
Whereas some stability results for $s=0$
carry over to the extremal case $|a|=M$, it turns out that,
already in axisymmetry~\cite{Aretakis},
the \emph{transversal} derivatives along the horizon grow 
polynomially~\cite{Aretakis, Aretakis2}. This phenomenon is now known
as the \emph{Aretakis instability}.
The Aretakis instability has been shown to hold also in the
case $s=\pm2$  by~\cite{lucietti2012gravitational}.
Understanding the non-axisymmetric case is completely open;
see~\cite{andersson2000superradiance} 
for some of the additional new phenomena that arise.

\subsubsection{Nonlinear model problems and stability under symmetry}
Though nonlinear stability of both Schwarzschild and Kerr is still open, 
various model problems have been considered which address some of
the specific technical difficulties expected to occur.

Issues connected to the handling
of decay
for quadratic nonlinearities in derivatives 
are addressed in the models considered in~\cite{MR3082240, lindblad2016global}. 
The Maxwell--Born--Infeld equations on Schwarzschild were
recently considered in~\cite{Pasqualotto:2017rkh}. This latter system, of
independent interest
in the context of high energy physics, can be thought to
capture at the same time aspects of both the quasilinear difficulties
as well as the tensorial difficulties (at the level of $s=\pm1$) inherent in $(\ref{vaceqhere})$.

Turning to stability under symmetry, the literature is  vast. 
For the Einstein--scalar field system under spherical symmetry, see~\cite{maththeory, DRPrice}. 
For the  vacuum equations~$(\ref{vaceqhere})$,~\cite{holzbiax} provides
the first result on the non-linear stability of the
Schwarzschild solution in symmetry,
considering  biaxial symmetry in $4+1$-dimensions. 
This again reduces to a $1+1$ problem. Beyond $1+1$, 
some aspects of the vacuum stability problem in 
axisymmetry are captured in a wave-map
model problem whose study was initiated by~\cite{Ionescu2015}.
Very recently, Klainerman--Szeftel \cite{SKqasp} have announced a proof of the non-linear
stability of Schwarzschild in the polarised, axisymmetric case.

\subsubsection{Analogues with $\Lambda \ne 0$}
There are analogues of the questions addressed here when
the Schwarzschild and Kerr solutions are replaced with 
the Schwarzschild--(anti) de Sitter metrics and Kerr--(anti) de Sitter metrics,
which are solutions of $(\ref{vaceqhere})$ 
when a cosmological term $\Lambda g_{\mu\nu}$
is added to the right hand side. These solutions are discussed in~\cite{kerrdesitrefcarter}.

In the de Sitter case $\Lambda>0$, the 
analogous problem is to understand the stability
of the spatially compact region bounded by the event and so-called cosmological horizons.
Following various linear results~\cite{Haefner, Dafermos:2007jd, vasy, dyatlov2011quasi, hintz2015resonance} the full non-linear
stability of this region has been obtained 
in remarkable work of Hintz--Vasy~\cite{hintz2016global}.
This de Sitter case is characterized by exponential 
decay, so many of the usual difficulties
of the asymptotically flat case are not present.   
The stability of the ``cosmological region'' beyond the event horizon has been considered in~\cite{schlue2016decay}.

The case of $\Lambda<0$ has been of considerable interest in the context
of high energy physics. Already, pure AdS spacetime fails to be globally
hyperbolic. In general, asymptotically AdS spacetimes have a timelike boundary
at infinity where boundary conditions must be prescribed to obtained well-posed problems. 

For reflective boundary conditions, the analogue of equation $(\ref{Teukphysicintro})$ on 
pure AdS space
admits infinitely many periodic solutions. 
In view of this lack of decay  in the reflective case, it is natural to conjecture
instability at the non-linear level~\cite{eguchiha}, once 
backreaction is taken into account.\footnote{In contrast,
good quantitative decay rates for solutions the Bianchi equations on pure AdS
with dissipative boundary conditions have been proven in~\cite{Holzegel:2015swa},
suggesting nonlinear stability.}
This nonlinear instability has indeed been seen in the seminal numerical study~\cite{bizon2011weakly},
which moreover sheds light on the relevance of resonant frequencies for calculating
a time-scale for growth.  Very recently, the full nonlinear instability of pure AdS space
 has  been proven in the simplest model for which the problem can  be
studied~\cite{moschidis2017proof}, exploiting an alternative physical-space mechanism.

In the case of Kerr--AdS, one has
logarithmic decay \cite{holzegel2013decay}---but in general no faster~\cite{HolzSmulevici}---for the analogue of $(\ref{Teukphysicintro})$ with $s=0$, on account of the fact
that trapped null geodesics, in contrast with the situation described in Section~\ref{prevworksec},
 are now stable.
Again, these results may  suggest instability at the non-linear level, 
as this slow rate of decay is in itself 
insufficient to control backreaction.

\subsubsection{Scattering theory}
A related problem to that of proving boundedness and decay
is developing a scattering theory for $(\ref{Teukphysicintro})$. 
Fixed frequency scattering theory for $(\ref{Teukphysicintro})$ is 
discussed in~\cite{Chandrasekhar}.
It was in fact the equality of the reflexion and transmission coefficients
between the Teukolsky, Regge--Wheeler and Zerilli equations that first suggested
the existence of  Chandrasekhar's transformations~\cite{Chandrasekhar}.
A definitive physical space scattering theory was developed
in the Schwarzschild case in~\cite{Dimock1, Dimock2}  for $s=0$, see also~\cite{nicolas}, 
and was recently extended to the Kerr
case in~\cite{Dafermos:2014jwa} for the full sub-extremal range of parameters $|a|<M$. 

Turning to the fully non-linear theory of $(\ref{vaceqhere})$,
a scattering construction of dynamic vacuum spacetimes settling
down to Kerr was given in~\cite{vacuumscatter}. 
The free scattering data allowed in the latter were very restricted, however,
as the radiation tail was required to decay exponentially in retarded time, 
and thus the spacetimes produced are measure
zero in the set of small perturbations of Kerr relevant for the stability problem.

For scattering for the Maxwell equations, see~\cite{MR1069953}.
For results in the $\Lambda>0$ case, see~\cite{georgescu, Mokdad:2017aua} and references therein.

\subsubsection{Stability and instability of the Kerr black hole interior}
The conjectured non-linear stability of the Kerr family refers only to the \emph{exterior} 
of the black hole region.
Considerations in the black hole interior are of a completely different nature.
The Schwarzschild case $a=0$ terminates at a spacelike singularity, whereas
for the rotating Kerr case $0<|a|<M$, the Cauchy development of
two-ended data   can be smoothly extended beyond a Cauchy horizon.
The $s=0$ case of $(\ref{Teukphysicintro})$ in the Kerr interior (as well
as the simpler Reissner--Nordstr\"om case) has been studied
in~\cite{McNamara1, mcnamara1978instability,  annefranzen, LukOhpub, Franzen2, Hintz:2015koq, LukSbierski, DafShlap}, and both $C^0$-stability but also $H^1$-instability 
have been obtained.  See~\cite{Gajic:2015csa, Gajic:2015hyu} for the extremal case.
In the full nonlinear theory,
 it has been proven that if the Kerr exterior stability conjecture
is true, then the bifurcate Cauchy horizon is globally $C^0$-stable~\cite{DafLuk1}. 
This implies in particular
that the $C^0$ inextendibility formulation of ``strong cosmic censorship'' is false.
See~\cite{Chrmil}.

\subsubsection{Note added}
Very recently,~\cite{Ma:2017yui2} gave a related approach
to obtaining integrated local energy decay estimates for the Teukolsky equation
in the $|a|\ll M$ case, following the frequency localisation  framework
of~\cite{DafRodsmalla} and
again based on proving estimates for $\Psi$ 
defined by generalisations of the transformations
used in~\cite{holzstabofschw}.

\subsection{Outline of the paper}
\label{outlinesec}
We end this introduction with an outline of the paper. 

We begin in {\bf Section~\ref{TeukandKerrsec}} by
recalling the notation from~\cite{partiii} 
regarding the Kerr metric and presenting the Teukolsky equation
in physical space for  spin $s=\pm2$. 

We then define in {\bf Section~\ref{physspacechandrasec}} 
our generalisations to Kerr of the quantities
$P^{[\pm2]}$, the rescaled quantities $\Psi^{[\pm2]}$  and the intermediate quantities
$\uppsi^{[\pm2]}$,
 as used in~\cite{holzstabofschw}, and derive our generalisation of the
Regge--Wheeler equation for $\Psi^{[\pm2]}$, now coupled to
$\uppsi^{[\pm2]}$ and $\upalpha^{[\pm2]}$.

In {\bf Section~\ref{generalstatement}} we shall define
various energy quantities which will allow us in particular
to   formulate our definitive (non-degenerate)
boundedness and decay results, stated 
as {\bf Theorem~\ref{finalstatetheor}}.

The first step in the proof of Theorem~\ref{finalstatetheor} is to obtain
integrated local energy decay.
In {\bf Section~\ref{Physspacesecnew}}, we shall prove a conditional 
such estimate, using entirely physical space methods, for  
the coupled system satisfied by $\Psi^{[\pm2]}$, $\uppsi^{[\pm2}$, and $\upalpha^{[\pm2]}$.
In view of the way this will be used, we must allow also
inhomogeneous terms on the right hand side of the Teukolsky equation. 
We apply the physical space multiplier estimates 
and transport estimates and transport estimates directly 
from~\cite{holzstabofschw}, except
that these estimates must now be coupled. 
The resulting estimates (see the propositions of Sections~\ref{condmultestsec} 
and~\ref{condtranestsec}) contain
on their right hand side an additional
timelike boundary term  on $r=A_1$ and $r=A_2$ for 
$A_1<3M<A_2$.
To control these terms, we will have to frequency localise the estimates in the
region $r\in [A_1,A_2]$. 
We also give certain auxiliary physical space estimates for the homogeneous Teukolsky
equation and its derived quantities (Section~\ref{auxil.est.sec}).

The next three sections will thus concern frequency localisation.
{\bf Section~\ref{Separationmegasec}} 
will interpret Teukolsky's separation of $(\ref{Teukphysicintro})$ for 
 spin $s=\pm2$
in a framework generalising that  introduced in~\cite{partiii} for the $s=0$
case.
In {\bf Section~\ref{Chandrasec}}, we define the frequency
localised versions of $P^{[\pm2]}$ and derive the coupled system
of ordinary differential equations satisfied by the $P^{[\pm2]}$ and $\upalpha^{[\pm2]}$.
In {\bf Section~\ref{ODEmegasec}} 
we then obtain estimates for this coupled system of
ODE's in the region $r\in[A_1,A_2]$.  The main statement is summarised as 
{\bf Theorem~\ref{phaseSpaceILED}} and can be thought
of as a fixed frequency version of the propositions of Sections~\ref{condmultestsec}--\ref{condtranestsec}, now valid
in $r\in[A_1,A_2]$. The estimate is again conditional on controlling boundary terms,
but the energy currents will have been chosen so that the most difficult of these, when formally summed, 
exactly cancel those appearing in the proposition of Section~\ref{condmultestsec}.

In {\bf Section~\ref{iledsec}}, we shall turn in ernest to
the study of the Cauchy problem 
for $(\ref{Teukphysicintro})$ 
 to obtain  a preliminary degenerate energy boundedness and
integrated local energy decay estimate in physical space.
This is stated as {\bf Theorem~\ref{degenerateboundednessandILED}}.
To obtain this, we cut off our solution of $(\ref{Teukphysicintro})$
in the future so as to allow for frequency localisation in $r\in[A_1,A_2]$.
This allows us to apply Theorem~\ref{phaseSpaceILED} and
sum over frequencies. We apply also 
the propositions of Sections~\ref{condmultestsec}--\ref{condtranestsec} to the cutoff-solution and sum the estimates.
The cutoff generates an inhomogeneous term which
is however only supported in a compact spacetime region.
By revisiting suitable estimates,
the cutoff term can then be estimated exploiting 
the smallness of $a$, following~\cite{DafRodsmalla}.
(We note that the fact that these cutoffs are here supported in a fixed, finite region of $r$ 
leads to various simplifications.) We distill a simpler purely physical-space proof 
for the axisymmetric case in Section~\ref{axi-note}.

The final sections  will complete the proof of  Theorem~\ref{finalstatetheor}
from Theorem~\ref{degenerateboundednessandILED}, 
by first applying red-shift estimates of~\cite{DafRod2} to obtain
non-degenerate control at the horizon
({\bf Section~\ref{sec:rsim}})
and then the $r^p$-weighted energy hierarchy of~\cite{DafRodnew}
({\bf Section~\ref{rphierarchysec}}).
This part follows closely 
the analogous estimates in the Schwarzschild case~\cite{holzstabofschw}.

Some auxilliary computations are relegated to  {\bf Appendix~\ref{sec:Psiderivation}}
and~{\bf\ref{otherappendix}}.

\subsection*{Acknowledgements}
MD~acknowledges support through NSF grant DMS-1709270 and 
EPSRC grant
EP/K00865X/1.  GH~acknowledges support through an 
ERC Starting Grant.  IR~acknowledges support
through NSF grant DMS-1709270 and an investigator award of the Simons Foundation.

\section{The Teukolsky equation on Kerr exterior spacetimes}
\label{TeukandKerrsec}
We recall in this section the Teukolsky equation on Kerr spacetimes.

We begin in {\bf Section~\ref{kerrmetricsection}} with a review of the
Kerr metric. We then present the Teukolsky equation on Kerr in 
{\bf Section~\ref{newteusec}}, focussing on the case $s=\pm2$.
This will allows us to state
a general well-posedness statement in {\bf Section~\ref{wellposedsec}}. Finally, 
in {\bf Section~\ref{relwiththesystemsec}} we
shall recall the relation of the $s=\pm2$ Teukolsky equation with 
the system of gravitational perturbations around Kerr.

\subsection{The Kerr metric}
\label{kerrmetricsection}
We review here the Kerr metric and associated structures, 
 following the notation of~\cite{partiii}.

\subsubsection{Coordinates and vector fields}
For each $|a|<M$, recall that the
Kerr metric in Boyer--Lindquist coordinates $(r,t,\theta,\phi)$ takes the form
\begin{equation}
\label{Kerrmetric}
g_{a,M} =-\frac{\Delta}{\rho^2}
(dt-a\sin^2\theta d\phi)^2+\frac{\rho^2}{\Delta}dr^2
+\rho^2d\theta^2+\frac{\sin^2\theta}{\rho^2}(a dt-(r^2+a^2)d\phi)^2,
\end{equation}
where
\begin{equation}
\label{variousdefs}
r_\pm =M\pm\sqrt{M^2-a^2},\qquad
\Delta =  (r-r_+)(r-r_-) ,\qquad \rho^2 =r^2+a^2\cos^2\theta.
\end{equation}

We recall from~\cite{partiii} the fixed
ambient manifold-with-boundary $\mathcal{R}$, diffeomorphic
to $\mathbb R^+\times \mathbb R\times \mathbb S^2$ and
the coordinates $( r, t^*, \theta^*, \phi^*)$  on $\mathcal{R}$ known as Kerr star coordinates. 

We recall the relations
\[
t(t^*, r) = t^*- \bar{t}(r) \ \ \ , \ \ \ \phi (\phi^\star, r) = \phi^* - \bar{\phi}(r) \mod 2\pi \ \ \ , \ \ \  \theta = \theta^*
\]
relating Boyer--Lindquist and Kerr star coordinates.
We do not need here the explicit form of $\bar{t}(r)$ and $\bar{\phi}(r)$; see~\cite{partiii}, Section 2.1.3 but remark that they both vanish for $r\geq 9/4M$. When expressed in Kerr star coordinates, the metric (\ref{Kerrmetric}) (defined a priori only in the interior of $\mathcal{R}$) extends to a smooth metric on $\mathcal{R}$, i.e.~it extends smoothly to the event horizon $\mathcal{H}^+$ defined as the boundary $\partial {\mathcal{R}}=\{r=r_+\}$.

It is easy to see that
the coordinate vector fields
$T=\partial_{t^*}$ and $\Phi=\partial_{\phi^*}$ of the fixed coordinate system
coincide for all $a$, $M$
with the  coordinate vector fields $\partial_t$ and $\partial_\phi$ of
Boyer--Lindquist coordinates, which are Killing for the metric
 $(\ref{Kerrmetric})$. 
 We recall that $T$ is spacelike in the so-called
 \emph{ergoregion} $\mathcal{S}=\{\Delta-a^2\sin^2\theta<0\}$.
Setting 
\[
\upomega_+\doteq \frac{a}{2Mr_+},
\]
we recall that the Killing field 
\[
K=T+\upomega_+ \Phi
\]
is null on the horizon $\mathcal{H}^+$ and is
timelike in $\{r_+<r<r_++R_K\}$ for some $R_K=R_K(a_0,M)$
where $R_K\to \infty$ as $a_0\to 0$.

An additional important coordinate will be $r^*$ defined to be a function $r^*(r)$
such that  
\begin{equation}
\label{rstareq}
\frac{dr^*}{dr} = \frac{r^2+a^2}{\Delta}
\end{equation}
and centred as in~\cite{partiii} so that $r^*(3M)=0$. Note that
$r^*\to -\infty$ as $r\to r_+$, while $r^*\to \infty$ as $r\to \infty$.
Given a parameter $R$ thought of as an $r$-value, we will often
denote $r^*(R)$ by $R^*$.

The vector fields
\begin{equation}
\label{spacetimenullframe}
L = \partial_{r^*}+T+ \frac{a}{r^2+a^2}\Phi,
\qquad
\underline{L}= -\partial_{r^*} +T +\frac{a}{r^2+a^2}\Phi,
\end{equation}
where $\partial_{r^*}$ is defined with respect to $(r^*,t,\theta,\phi)$ coordinates,
define principal null directions. 
We have the normalisation
\[
g(L, \underline{L})=-2\frac{\Delta\rho^2 }{(r^2+a^2)^2}  \, .
\]
The vector field $L$ extends smoothly to $\mathcal{H}^+$
to be parallel to the null generator,
while  $\underline{L}$ extends smoothly to $\mathcal{H}^+$ so as to vanish 
identically.
The quantity  $\Delta^{-1}\underline{L}$  has a smooth
nontrivial limit  on $\mathcal{H}^+$.
The vector fields $L$ and $\underline{L}$ are again $T$-(and $\Phi$-)invariant.

\subsubsection{Foliations and the volume form} 
\label{volumeformrelations}

For all values $\tau\in \mathbb R$, we recall that the hypersurfaces 
$\Sigma_\tau=\{t^*=\tau\}$ are  spacelike (see~\cite{partiii}, Section~2.2.5). 
We will denote the unit future normal of $\Sigma_\tau$ by $n_{\Sigma_\tau}$.
We recall the notation 
\[
\mathcal{R}_0=\{t^*\ge 0\}, \qquad \mathcal{R}_{(0,\tau)}=\{0\le t^*\le \tau\},
\qquad
\mathcal{H}^+_0=\mathcal{R}_0\cap\mathcal{H}^+,
\qquad 
\mathcal{H}^+_{(0,\tau)}=\mathcal{R}_{(0,\tau)}\cap \mathcal{H}^+.
\]

For polynomial decay following the method of~\cite{DafRodnew, Moschnewmeth}, 
we will also require
hypersurfaces $\widetilde{\Sigma}_{\tau}$ which connect the event horizon
and null infinity. For this we fix some $0<\eta<1$ and define the coordinate
\begin{equation}
\tilde{t}^* = t^* - \xi \left(r^*\right) \left(r^* + 2M \left(\frac{2M}{r}\right)^\eta - R_\eta^* - 2M \left(\frac{2M}{R_\eta}\right)^\eta - M\right)
\end{equation}
where $\xi$ is a smooth cut-off function equal to zero for $r \leq R_\eta$ and equal to $1$ for $r\geq R_\eta+M$. It is straightforward if tedious to show that for $R_\eta$ sufficiently large (and a suitably chosen function $\xi$) the hypersurfaces 
$\widetilde{\Sigma}_{\tau}$ defined by
\begin{align}
\label{choiceofhyperbol}
\widetilde{\Sigma}_{\tau} := \{ \tilde{t}^* = \tau \} \, 
\end{align}
are smooth and spacelike everywhere, in fact $c_\eta r^{-\eta-1} \leq -g\left(\nabla \tilde{t}^*, \nabla \tilde{t}^*\right) \leq C_\eta r^{-\eta-1}$ indicating that the hypersurfaces become asymptotically null near infinity. We take this $R_\eta$ as fixed from now on.

We will in fact use coordinates 
$\left(\tilde{t}^*, r, \theta,\phi^*\right)$ and perform estimates in the spacetime regions
\[
\widetilde{\mathcal{R}}(\tau_1,\tau_2)=\{\tau_1\le \tilde{t}^* \le \tau_2\},
\qquad
\widetilde{\mathcal{R}}_0=\{\tilde{t}^* \ge 0\}.
\]
See Figure~\ref{sigmasfig}.

\begin{figure}
\centering{
\def\svgwidth{10pc}
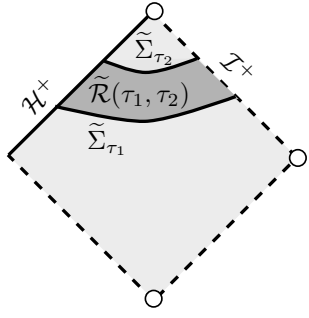}
\caption{The region $\widetilde{\mathcal{R}}(\tau_1,\tau_2)$}\label{sigmasfig}
\end{figure}

We compute the volume form in the different coordinate systems (recalling that the $r$ and $\theta$ coordinates are common to \emph{all} coordinate systems, so $\rho^2=r^2+a^2\cos^2 \theta$ is umambiguously defined)
\begin{align} \label{volform}
dV=\rho^2 dt \, dr\, \sin \theta d\theta d\phi = \frac{\rho^2 \Delta}{\left(r^2+a^2\right)} dt \, dr^* \, \sin \theta d\theta d\phi^*  = \rho^2 dt^* \, dr \, \sin \theta d\theta d\phi = \rho^2 d\tilde{t}^* \, dr \, \sin \theta d\theta d\phi^* \, .
\end{align}
We will often use the notation 
\[
d\sigma=\sin \theta d\theta d\phi.
\]

Denoting the (timelike) unit normal to the hypersurfaces (\ref{choiceofhyperbol}) by $n_{\widetilde\Sigma_{\tau}}$ we compute in coordinates $\left( r,\tilde{t}^*, \theta^*,\phi^*\right)$
\begin{align} \label{innp}
\sqrt{g_{\widetilde\Sigma_{\tau}}} g\left(\frac{r^2+a^2}{\Delta} \underline{L}, n_{\widetilde\Sigma_{\tau}}\right) =v\left(r,\theta\right) \rho^2 \sin \theta  \ \ \ \textrm{and} \ \ \ \sqrt{g_{\widetilde\Sigma_{\tau}}} g\left(L, n_{\widetilde\Sigma_{\tau}}\right) =v\left(r,\theta\right) \frac{1}{r^{1+\eta}} \rho^2 \sin \theta
\end{align}
for a function $v$ with $C^{-1} \leq v \leq C$. In particular, the volume element on slices of constant $\tilde{t}^*=\tau$ satisfies 
\[
dV_{\widetilde{\Sigma}_{\tau}} =\sqrt{g_{\widetilde{\Sigma}_\tau}} dr d\theta d\phi = v\left(r, \theta\right) r^2 r^{-\frac{1+\eta}{2}} dr d\sigma
\]
for a (potentially different) function $v$ with $C^{-1} \leq v \leq C$.

For future reference we note that, again in coordinates $\left( r,\tilde{t}^*,\theta^*,\phi^*\right)$, we have on the null hypersurfaces corresponding to the horizon and null infinity respectively the relations
\begin{align} \label{innp2}
\sqrt{g_{\mathcal{H}^+}} \ g\left(\frac{r^2+a^2}{\Delta} \underline{L}, L\right) = v\left(r,\theta\right) \sin \theta \ \ \ \textrm{and} \ \ \ \sqrt{g_{\mathcal{I}^+}} \  g\left(L,\underline{L}\right) = v\left(r,\theta\right) \rho^2 \sin \theta \,  ,
\end{align}
where the volume forms are understood to be themselves normalised by $L$ and
$\underline{L}$, respectively. The above will be the expressions that arise
in the context of the divergence theorem.

Finally, we note the covariant identities
\begin{align} \label{siid}
\nabla_a \left(\frac{1}{\rho^2} \frac{r^2+a^2}{\Delta} \underline{L}^a \right) = 0  \ \ \ 
\textrm{and} \ \ \ 
\nabla_a \left(\frac{1}{\rho^2} \frac{r^2+a^2}{\Delta} L^a \right) = 0 \, ,
\end{align}
which are most easily checked in Boyer--Lindquist coordinates.

\subsubsection{The very slowly rotating case $|a|<a_0 \ll M$}
\label{very.slowly.very.slowly}

In the present paper, we will restrict to the very slowly rotating case. 
This will allow us to exploit certain simplifications which arise
from closeness to Schwarzschild.

Recall that the hypersurface  $r=3M$ in Schwarzschild is known as the
\emph{photon sphere} and corresponds to the set where
integrated local energy decay estimates necessarily degenerate.
In the case $|a|<a_0\ll M$ the trapping is localised near $r=3M$ while
the ergoregion $\mathcal{S}$ is localised near $r=2M$.
See the general discussion in~\cite{Mihalisnotes}.
Let us quantify this below by fixing certain parameters.

We will fix parameters $A_1<3M<A_2$ sufficiently close
to $3M$.  We note already that for sufficiently small $|a|<a_0\ll M$, 
then all future trapped null geodesics will asymptote to an $r$ value
contained in $r\in [A_1,A_2]$. 
(We shall not use this fact directly, but rather, a related
property concerning the maximum of a potential function associated
to the separated wave equation. See already Lemma~\ref{lem.max}.)

We moreover can choose $a_0$ small enough so that in addition,
$R_K>A_1$ and so that the ergoregion satisfies $\mathcal{S}\subset \{r^*<4A_1^*\}$.

Fixing a cutoff function $\chi(r^*)$ which is equal to $1$
for $r^*\le 4A_1^*$ and $0$ for $r^*\ge 2A_1^*$
we define the vector field
$T+\upomega_+\chi \Phi$.
We note that by our arrangement, this vector field is now timelike for all $r>r_+$,
Killing outside $\{4A_1^*<r^*<2A_1^*\}$, null on $\mathcal{H}^+$,
and equal to $T$ on $\{r^*\ge A_1^*\}$.

Finally, let us note that, if $A_1^*$ is sufficiently small, then
restricting to small $a_0$, 
we have that
$t=t^*$ for $r^*\ge 2A_1^*$ for all $|a|<a_0$.

We note in particular
\[
t=t^*=\tilde{t}^* {\rm\ in\ the\ region\ }2A_1^*\le r^*\le 2A^*_2.
\]

\subsubsection{Parameters and conventions}
This paper will rely on fixing a number of parameters which will appear in the proof.
We have just discussed
the parameters $\eta$ and
\[
A_1<3M< A_2
\]
which have already been fixed.

We will also introduce fixed parameters
\[
\delta_1, \delta_2,\delta_3,  E
\]
which will be connected
to adding multiplier constructions on Schwarzschild,
as well as parameters $C_\sharp$, $c_\flat$, $C_\flat$ delimiting frequency ranges.
In particular, eventually,
these can be all thought of as fixed in terms of $M$ alone.

We will introduce 
an additional smallness parameter $\varepsilon$ associated
to the cutoffs in time. (This notation is retained from our~\cite{DafRodsmalla}.)
Again, eventually, this will be fixed depending only on $M$.

Finally, we will exploit the slowly rotating assumption by
employing $a_0$ as a smallness parameter, which will
only be fixed at the end of the proof. 

We introduce the following conventions regarding inequalities.
For non-negative quantities $\mathcal{E}_1$ and $\mathcal{E}_2$,
by 
\[
\mathcal{E}_1\lesssim \mathcal{E}_2
\]
we mean that
there exists a constant $C=C(M)>0$, depending
only on $M$, such that
\[
\mathcal{E}_1\le C(M) \mathcal{E}_2.
\]
We will sometimes use the notation 
\[
\mathcal{E}_1\lesssim \mathcal{Q}+\mathcal{E}_2
\]
where $\mathcal{Q}$ is not necessarily a non-negative quantity. 
In this context, this will mean that there exist constants $c(M)$, $C(M)$ such 
that
\[
c\mathcal{E}_1\le \mathcal{Q} +C\mathcal{E}_2.
\]
Note that two inequalities of the above form can be added
when the terms $\mathcal{Q}$ are identical.

Before certain parameters are fixed, say $\delta_1$,
we will use the notation
$\lesssim_{\delta_1}$ to denote the additional dependence
on $\delta_1$ of the constant $C(M,\delta_1)$ appearing in various inequalities.
Only when the parameter is definitively fixed in terms of $M$, 
can $\lesssim_{\delta_1}$  be replaced by
$\lesssim$.  

On the other hand, in the context of the restriction to $a_0\ll M$, which will appear ubiquitously,
the constant implicit
in $\ll$ may depend on all parameters yet to be fixed.
This will not cause confusion because
restriction to smaller $a$ will always be favourable in every estimate.

\subsection{The Teukolsky equation for spin weighted complex functions}
\label{newteusec}

In this section we present the Teukolsky equation on Kerr. 

We first review in Section~\ref{spinweightfu}
the notion of spin $s$-weighted complex functions and discuss some elementary properties of the spin $s$-weighted Laplacian in Section~\ref{sec:spinwl}.
We then recall in Section~\ref{itsdefinedhere} the classical form of the
Teukolsky operator for general spin. 
Finally,
specialising to $s=\pm 2$ we derive in Section~\ref{rescaletobereg} rescaled quantities
which satisfy an equation regular also on the horizon.
It is in this form that we will be able to state well-posedness in the section that
follows.

\subsubsection{Spin $s$-weighted complex functions on $S^2$ and $\mathcal{R}$}
\label{spinweightfu}

The Teukolsky equation will concern functions whose $(\theta,\phi)$
(or equivalently $(\theta,\phi^*)$) dependence
is that of a spin $s$-weighted function, for $s\in\frac12 \mathbb Z$.
We will always represent such functions as 
usual functions $\upalpha(r,t,\theta,\phi)$.

Smooth spin $s$-weighted functions on $S^2$ naturally arise, in a one-to-one fashion, from  complex-valued functions on $S^3$ (viewed as the Hopf bundle) which transform in a particular way under the group action on the $S^1$ fibres of $S^3$, as will be described now. (Note that this is indeed natural as $S^3$ can be identified with the bundle of orthonormal frames on $S^2$, and the definition of the Teukolsky null curvature components indeed depends on a choice of frame on $S^2$. See Section \ref{relwiththesystemsec}.)

Viewing $S^3$ as the Hopf bundle we have a $U(1)$ action on the $S^1$ fibres (corresponding to a rotation of the orthonormal frame in the tangent space of $S^2$). Introducing Euler coordinates\footnote{Euler coordinates cover $S^3$ everywhere except the north and southpole at $\theta=0$ and $\theta=\pi$ respectively. The ranges of the coordinates are $0<\theta<\pi$, $0\leq \phi < 2\pi$ and $0\leq \rho < 4\pi$.} $\left(\theta,\phi,\rho\right)$ on $S^3$ we denote this action by $e^{i\rho}$. Now any smooth function $F:S^3 \rightarrow \mathbb{C}$ which transforms as $F\left(p e^{i\rho}\right) = e^{-i\rho s} F\left(p\right)$ for $p \in S^3$ descends to a spin-weighted function on $S^2$ (by \emph{choosing} a frame at each point). More precisely, $F$ descends 
to a section of a complex line bundle over $S^2$ denoted traditionally by $B(R)$. See \cite{Lerner, Whale}.

Let $Z_1, Z_2, Z_3$ be a basis of right invariant vector fields constituting a global orthonormal frame on $S^3$. In Euler coordinates we have the representation
\begin{align}
Z_1 = -\sin \phi \partial_\theta + \cos \phi  \left( \csc \theta \partial_\rho - \cot \theta \partial_\phi\right) \ , \ 
Z_2 = - \cos \phi \partial_\theta - \sin \phi  \left( \csc \theta \partial_\rho - \cot \theta \partial_\phi\right) \ , \ 
Z_3 = \partial_\phi \, .
\end{align}
A complex-valued function $F$ of the Euler coordinates $\left(\theta,\phi, \rho\right)$ is smooth on $S^3$ if for any $k_1,k_2,k_3 \in \mathbb{N} \cup \{0\}$ the functions $\left(Z_1\right)^{k_1} \left(Z_2\right)^{k_2} \left(Z_3\right)^{k_3}F$ are smooth functions of the Euler coordinates\footnote{By this we understand that the restrictions $\left(0,2\pi\right) \times \left(0,\pi\right) \times \left(0,4\pi\right) \rightarrow \mathbb{C}$ are smooth in the usual sense and for any fixed $\theta \in \left(0,\pi\right)$ extend continuously with the same value respectively to $\phi=0$, $\phi=2\pi$ and $\rho=0$, $\rho=4\pi$.} and extend continuously to the poles of the coordinate system at $\theta=0$ and $\theta=\pi$. 

Since spin $s$-weighted functions on $S^2$ arise from smooth functions on $S^3$ as discussed above, there is a natural notion of the space of smooth spin $s$-weighted functions on $S^2$: 
A complex-valued function $f$ of the coordinates $(\theta,\phi)$ is called a smooth spin $s$-weighted function on $S^2$ if for any $k_1,k_2,k_3 \in \mathbb{N} \cup \{0\}$ the functions $(\tilde{Z}_1)^{k_1} (\tilde{Z}_2)^{k_2} (\tilde{Z}_3)^{k_3} f$ are smooth functions of the coordinates\footnote{By this we understand that the restrictions $\left(0,\pi\right) \times \left(0,2\pi\right) \rightarrow \mathbb{C}$ are smooth in the usual sense and for any fixed $\theta \in \left(0,\pi\right)$ extend continuously with the same value to $\phi=0$ and $\phi=2\pi$ respectively.} and extend continuously to the poles of the coordinate system at $\theta=0$ and $\theta=\pi$, where 
\begin{align}
\tilde{Z}_1 = -\sin \phi \partial_\theta + \cos \phi  \left( -is \csc \theta - \cot \theta \partial_\phi\right) \ , \ 
\tilde{Z}_2 = - \cos \phi \partial_\theta - \sin \phi  \left( -is \csc \theta - \cot \theta \partial_\phi\right) \ , \ 
\tilde{Z}_3 = \partial_\phi \, .
\end{align}
The space of smooth spin $s$-weighted functions on $S^2$ is denoted $\mathscr{S}^{[s]}_\infty$.  Note that considered as usual functions on $S^2$,
elements  of $\mathscr{S}^{[s]}_\infty$ are in general not regular at $\theta=0$.

We define the Sobolev space of smooth spin $s$-weighted functions on $S^2$, denoted ${}^{[s]}H^m(\sin\theta d\theta d\phi)$ as the completion of $\mathscr{S}^{[s]}_\infty$ with respect to the norm.
\[
\| f\|^2_{{}^{[s]}H^m(\sin\theta d\theta d\phi)} = \sum_{i=0}^m \sum_{k_1+k_2+k_3=i} 
\int_{S^2} |(\tilde{Z}_1)^{k_1} (\tilde{Z}_2)^{k_2} (\tilde{Z}_3)^{k_3} f|^2 \sin \theta d\theta d\phi \, .
\]
Note that the space $\mathscr{S}^{[s]}_{\infty}$ is dense in $L^2(\sin\theta d\theta d\phi)$.

We now define the analogous notions for functions $f$ of the spacetime coordinates $\left(t^*,r,\theta,\phi^*\right)$.

We define a smooth complex-valued spin $s$-weighted function $f$ on $\mathcal{R}$ to be a function $f: \left(-\infty,\infty \right) \times \left[2M,\infty\right) \times \left(0,\pi\right) \times \left[0, 2\pi\right)$ which is smooth in the sense that for any $k_1,k_2,k_3,k_4,k_5 \in \mathbb{N} \cup \{0\}$ the functions
\[
(\tilde{Z}_1)^{k_1} (\tilde{Z}_2)^{k_2} (\tilde{Z}_3)^{k_3} \left(\partial_{t^*}\right)^{k_4} \left(\partial_r\right)^{k_5} f
\]
are smooth functions\footnote{By this we understand that the restriction of these functions to $\left(-\infty,\infty \right) \times \left[2M,\infty\right) \times \left(0,\pi\right) \times \left(0, 2\pi\right)$ is smooth in the usual sense and that the functions extend continuously with the same value to $\phi=0$ and $\phi=2\pi$.} which extend continuously to the poles at $\theta=0$ and $\theta=\pi$. In particular, the restriction of $f$ to fixed values of $t^*, r$ is a smooth spin $s$-weighted function on $S^2$. We denote the space of smooth complex-valued spin $s$-weighted functions on $\mathcal{R}$ by $\mathscr{S}_\infty^{[s]}(\mathcal{R})$.

We similarly define a smooth complex-valued spin $s$-weighted function $f$ on a slice ${\Sigma}_\tau$ to be  a function $f: \left[2M,\infty\right) \times \left(0,\pi\right) \times \left[0, 2\pi\right)$ which is smooth in the sense that for any $k_1,k_2,k_3,k_4 \in \mathbb{N} \cup \{0\}$ the functions
\[
(\tilde{Z}_1)^{k_1} (\tilde{Z}_2)^{k_2} (\tilde{Z}_3)^{k_3} \left(\partial_{r}\right)^{k_4}  f
\]
are smooth functions extending continuously to the poles at $\theta=0$ and $\theta=\pi$. The space of such functions is denoted $\mathscr{S}_\infty^{[s]}(\Sigma_{\tau})$. The Sobolev space ${}^{[s]}H^m(\Sigma_{\tau})$ is defined as the completion of $\mathscr{S}_\infty^{[s]}(\Sigma_{\tau})$ with respect to the norm
\[
\| f\|^2_{{}^{[s]}H^m(\Sigma_{\tau})} = \sum_{i=0}^m \sum_{k_1+k_2+k_3+k_4=i} 
 \int_{\Sigma_\tau} dV_{\Sigma_{\tau}} |(\tilde{Z}_1)^{k_1} (\tilde{Z}_2)^{k_2} (\tilde{Z}_3)^{k_3} \left(\partial_r\right)^{k_4} f|^2  \, .
\]
If $\mathcal{U}$ is an open subset of $\Sigma_\tau$ we can define $\mathscr{S}_\infty^{[s]}(\mathcal{U})$ and ${}^{[s]}H^m(\mathcal{U})$  in the obvious way. This allows to define the space ${}^{[s]}H^m_{\rm loc} \left(\Sigma_\tau\right)$ as the space of functions on $\Sigma_{\tau}$ such that the restriction to any $\mathcal{U} \Subset \Sigma_\tau$ (meaning that there is a compact set $K$ with $\mathcal{U} \subset K \subset \Sigma_\tau$) is in ${}^{[s]}H^m \left(\mathcal{U}\right)$.

We finally note that we can analogously define these spaces for the slices $\widetilde{\Sigma}_{\tau}$, i.e.~define the spaces
\[
\mathscr{S}_\infty^{[s]}(\widetilde{\Sigma}_{\tau}) \ \  , \ \ {}^{[s]}H^m(\widetilde{\Sigma}_{\tau}) \ \ , \ \  {}^{[s]}H_{loc}^m(\widetilde{\Sigma}_{\tau}) \, .
\]

\subsubsection{The spin $s$-weighted Laplacian} \label{sec:spinwl}
Let us note that the operator defined in the introduction,
\begin{equation}
\label{spinweightlaplacagain}
\mathring{\slashed\triangle}^{[s]}=  -\frac{1}{\sin \theta} \frac{\partial}{\partial \theta} \left(\sin \theta \frac{\partial}{\partial \theta}\right) - \frac{1}{\sin^2 \theta} \partial_\phi^2 + 2s i\frac{ \cos \theta}{\sin^2 \theta} \partial_\phi + 4 \cot^2 \theta - s \, ,
\end{equation}
is a smooth operator on $\mathscr{S}^{[s]}_\infty$. Indeed, a computation yields 
$[(\tilde{Z}_1)^2 + (\tilde{Z}_2)^2 +(\tilde{Z}_3)^2] \Xi
= [ -\mathring{\slashed{\Delta}}^{[s]} - s -s^2 ]\Xi$.
Note also the formula $
\sum_{i=1}^3 | \tilde{Z}_i \Xi|^2 = |\partial_{\theta} \Xi|^2 + \frac{1}{\sin^2 \theta} | is \Xi \cos \theta+ \partial_\phi \Xi|^2 +s^2 |\Xi|^2$.

The eigenfunctions of $\mathring{\slashed\triangle}^{[s]}$ are
again in $\mathscr{S}^{[s]}_\infty$ and are known 
as $s$-spin weighted spherical harmonics. 
We shall discuss these (and their twisted analogues) further in 
Section~\ref{spinweigtsec}.

An integration by parts yields for $\Xi \in \mathscr{S}^{[s]}_\infty$ 
\begin{align} \label{laplace1}
 \int_0^\pi \int_0^{2\pi} d\phi \, d\theta \sin \theta \left(\mathring{\slashed\triangle}^{[+2]} \left(0\right) \Xi \right)\overline{\Xi} \nonumber \\
= \int_0^\pi \int_0^{2\pi} d\phi  \,d\theta \sin^5 \theta  \left[ \partial_\theta \left(\frac{\overline{\Xi}}{\sin^{2} \theta}\right) - \frac{i}{\sin \theta} \partial_\phi\left(\frac{\overline{\Xi}}{\sin^{2} \theta}\right)\right] \left[ \partial_\theta \left(\frac{{\Xi}}{\sin^{2} \theta}\right) + \frac{i}{\sin \theta} \partial_\phi\left(\frac{{\Xi}}{\sin^{2} \theta}\right)\right] \, ,
\end{align}
where the right hand side is manifestly non-negative.\footnote{In fact, the right hand side vanishes for the first spin-weighted spherical harmonics.} Introducing the spinorial gradient
$$
\mathring{\slashed{\nabla}}^{[\pm 2]}\Xi=   \big( \partial_\theta \Xi, 
\partial_\phi \Xi\pm 2 \cdot i \cos \theta \, \Xi \big)
$$
and defining
\begin{align}  \label{notcod}
| \mathring{\slashed{\nabla}}^{[\pm 2]}\Xi|^2  := 
   \big| \partial_\theta \Xi \big|^2 + \frac{1}{\sin^2 \theta}  \big| \partial_\phi \Xi \pm 2 \cdot i \cos \theta \Xi \big|^2  \, ,
\end{align}
we also have
\begin{align} \label{laplace2}
\int_0^\pi \int_0^{2\pi} d\phi  \,d\theta \sin \theta \left[ \mathring{\slashed\triangle}^{[\pm 2]} \left(0\right)  \pm 2 \right] \Xi \cdot \overline{\Xi}&= 
 \int_0^\pi \int_0^{2\pi} d\phi  \,d\theta \left[ \sin \theta  \big| \partial_\theta \Xi \big|^2 + \frac{1}{\sin \theta}  \big| \partial_\phi \Xi \pm 2 \cdot i \cos \theta \Xi \big|^2  \right]   \nonumber \\
 &= \int_0^\pi \int_0^{2\pi} d\phi  \,d\theta \sin \theta | \mathring{\slashed{\nabla}}^{[\pm 2]}\Xi|^2 \, .
\end{align}
We note that  for $\Xi, \Pi \in \mathscr{S}^{[s]}_\infty $
\begin{align} \label{laplacef}
\int_0^\pi \int_0^{2\pi} d\phi  \,d\theta \sin \theta \left[ \mathring{\slashed\triangle}^{[\pm 2]} \left(0\right)  \pm 2 \right] \Xi\, \cdot\overline\Pi= 
 \int_0^\pi \int_0^{2\pi} d\phi \, d\theta \sin \theta  \left[ \mathring{\slashed{\nabla}}^{[\pm 2]}\Xi\cdot
 \overline{ \mathring{\slashed{\nabla}}^{[\pm 2]}\Pi}
 \right]_{\Bbb S^2}  .
\end{align}
Directly from (\ref{laplace1}) and (\ref{laplace2}) we deduce the Poincar\'e inequality
\begin{align} \label{poincare}
 \int_0^\pi \int_0^{2\pi} d\phi  \,d\theta \sin \theta  \big|\mathring{\slashed{\nabla}}^{[\pm 2]}\Xi \big|^2 \geq 2  \int_0^\pi \int_0^{2\pi} d\phi  \,d\theta \sin \theta |\Xi|^2 \, .
\end{align}
Combining (\ref{laplace2}) and (\ref{poincare}) we also deduce
\begin{align} \label{azimuthal}
 \int_0^\pi \int_0^{2\pi} d\phi  \,d\theta \sin \theta  \big|\mathring{\slashed{\nabla}}^{[\pm 2]}\Xi \big|^2 \geq \frac{1}{8} \int_0^\pi \int_0^{2\pi} d\phi \, d\theta \sin \theta |\Phi \Xi|^2 \, .
\end{align}

\subsubsection{The Teukolsky operator for general spin $s$}
\label{itsdefinedhere}

Recall that the operator
\begin{eqnarray}
\label{Teukop}
\nonumber
\mathfrak{T}^{[s]} {\upalpha}^{[s]}&=
\Box_g {\upalpha}^{[s]} +\frac{2s}{\rho^2}(r-M)\partial_r {\upalpha}^{[s]} +\frac{2s}{\rho^2}
\left(\frac{a(r-M)}{\Delta}  +i\frac{\cos \theta}{\sin^2\theta}\right) \partial_\phi {\upalpha}^{[s]}\\
&\qquad
+\frac{2s}{\rho^2}\left(\frac{M(r^2-a^2)}{\Delta}-r-ia\cos\theta\right)\partial_t {\upalpha}^{[s]}
+\frac{1}{\rho^2}(s-s^2\cot^2\theta) {\upalpha}^{[s]}
\end{eqnarray}
is the traditional representation (see for instance~\cite{sbierskithesis})
of the Teukolsky operator with spin $s\in \frac12\mathbb Z$.
In view of the comments above,
this operator is smooth on $\mathscr{S}_\infty^{[s]}(\mathcal{R}\setminus\mathcal{H}^+)$.
We will say that such an  $\upalpha^{[s]}\in \mathscr{S}_\infty^{[s]}(\mathcal{R}\setminus
\mathcal{H}^+)$
satisfies the 
 Teukolsky equation  if the following holds:
\begin{equation}
\label{Teukphysic}
\mathfrak{T}^{[s]} {\upalpha}^{[s]}=0.
\end{equation}

The operator $(\ref{Teukphysic})$ is not smooth on $\mathscr{S}_\infty^{[s]}(\mathcal{R})$
itself.
This is because it has been derived with respect to a choice of frame which degenerates
at the horizon. See Section~\ref{relwiththesystemsec}. 
To obtain a regular equation at the horizon, we must
considered rescaled quantities. We turn to this now.

\subsubsection{Rescaled equations}
\label{rescaletobereg}

To understand regularity issues at the horizon
we must consider rescaled quantities.
We will restrict here to $s=\pm2$.

Define 
\begin{equation}
\label{rescaleddefsfirstattempt}
\tilde\upalpha^{[+2]}  = \Delta^2 (r^2+a^2)^{-\frac32}\upalpha^{[+2]}, \qquad
\tilde\upalpha^{[-2]} = \Delta^{-2}(r^2+a^2)^{-\frac32}\upalpha^{[-2]} \, .
\end{equation}

Define now the modified Teukolsky operator $\widetilde{\mathfrak{T}}^{[s]}$ by the relation
\begin{align} \label{teurefer}
\Delta\rho^{-2} \widetilde{\mathfrak{T}}^{[s]} = \frac{1}{2} \left(L \underline{L} + \underline{L} L\right)  &+ \frac{\Delta}{\left(r^2+a^2\right)^2}\left( \mathring{\slashed\triangle}^{[ s]}  + s -3 \frac{a^4+a^2r^2-2Mr^3}{(r^2+a^2)^2} +2\right) \nonumber \\
&- \frac{\Delta}{\left(r^2+a^2\right)^2}\left(2 a T \Phi +a^2 \sin^2\theta TT + 2i s a\cos \theta T  \right) + \mathfrak{t}^{[s]} \, ,
\end{align}
with $\mathring{\slashed\triangle}^{[ s]}$ denoting the spin $\pm 2$ weighted Laplacian on the round sphere defined in $(\ref{spinweightlaplacagain})$ and with the first order term $\mathfrak{t}^{[s]}$ given by
\begin{align}
\mathfrak{t}^{[+2]} = -2\frac{w^\prime}{w} \underline{L} - 8 a w\frac{r}{r^2+a^2} \Phi \ \ \ \textrm{and} \ \ \ \mathfrak{t}^{[-2]} = +2\frac{w^\prime}{w} {L} + 8 aw \frac{r}{r^2+a^2} \Phi \ \ \ \textrm{where} \ \ \ w := \frac{\Delta}{\left(r^2+a^2\right)^2}. 
\end{align}

One sees that $(\ref{Teukphysic})$ for $s=+2$
can be rewritten
as
\begin{equation}
\label{rescaledeqTeu}
\widetilde{\mathfrak{T}}^{[+2]} \tilde\upalpha^{[+2]} =0.
\end{equation}

On the other hand,  we observe that
$\widetilde{\mathfrak{T}}^{[+2]}$ now is a smooth
operator on $\mathscr{S}_\infty^{[s]}(\mathcal{R})$ and that its second order part is hyperbolic, in
fact, it is exactly equal   to 
$-\Box_g$. 

Similarly, we see that $(\ref{Teukphysic})$ for $s=-2$
can be rewritten
as
\begin{equation} \label{rescaledeqTeuaux}
\widetilde{\mathfrak{T}}^{[-2]}  \left(\Delta^{2} \tilde\upalpha^{[-2]} \right)=0 \, ,
\end{equation}
which in turn can be rewritten as
\begin{align}
\label{rescaledeqTeu2}
\left[ \widetilde{\mathfrak{T}}^{[-2]} -2\frac{\rho^2}{\Delta}\frac{w^\prime}{w} L +\mathfrak{t}^{[-2]}_{aux}\right] \tilde\upalpha^{[-2]} =  0  \, ,
\end{align}
where
\[
\mathfrak{t}^{[-2]}_{aux} = \frac{\rho^2}{\Delta}\left[- 4 \frac{\left(r^2+a^2\right)^\prime}{(r^2+a^2)}L+2\frac{\Delta^\prime}{\Delta}\underline{L} -2\left(\frac{\Delta^\prime}{\Delta}\right)^\prime +8 \frac{\left(r^2+a^2\right)^\prime}{(r^2+a^2)} \frac{\Delta^\prime}{\Delta}\right]
\]
is a first order operator acting smoothly on $\mathscr{S}_\infty^{[s]}(\mathcal{R})$. Now we observe 
that $ \widetilde{\mathfrak{T}}^{[-2]} -2\frac{\rho^2}{\Delta}\frac{w^\prime}{w} L$ also acts
smoothly on  $\mathscr{S}_\infty^{[s]}(\mathcal{R})$ and that its second order part is exactly equal   to 
$-\Box_g$. 
This will allow us to state a well-posedness proposition in the
section to follow. 

\begin{remark}
The weights in $(\ref{rescaleddefsfirstattempt})$ for $\widetilde{\upalpha}^{[+2]}$ 
will be useful for the global analysis of the equation, whereas the weights for
$\widetilde{\upalpha}^{[-2]}$ will only be useful for the well-posedness below.
For this reason, we shall define later (see Section \ref{boundcondusec}) the different rescaled quantities
$u^{[\pm2]}   = \Delta^{\pm 1}\sqrt{r^2+a^2}\upalpha^{[\pm 2]}$, and
deal mostly with the further 
rescaled quantities $u^{[\pm2]}\cdot w$. Note that
\[
u^{[+2]}\cdot w= \widetilde{\upalpha}^{[+2]}, \qquad {\rm but}\qquad u^{[-2]}\cdot w=  (r^2+a^2)^{-\frac32}\upalpha^{[-2]}  .
\]
The first quantity is finite (and generically non-zero) on the horizon $\mathcal{H}^+$
while the second quantity is finite (and generically non-zero) on null infinity $\mathcal{I}^+$
which makes them useful in the global considerations below. Note also that both quantities satisfy the simple equations (\ref{rescaledeqTeu}) and (\ref{rescaledeqTeuaux}) respectively.
\end{remark}

\subsection{Well-posedness}
\label{wellposedsec}

Standard theory yields that the Teukolsky equation in the form $(\ref{rescaledeqTeu})$, $(\ref{rescaledeqTeu2})$
is well-posed on $\mathcal{R}_0$  or $\widetilde{\mathcal{R}}_0$
with initial data $(\tilde\upalpha^{[s]}_0,\tilde\upalpha^{[s]}_1)$ defined on
$\Sigma_0$ 
in ${}^{[s]}H^j_{\rm loc}(\Sigma_0)\times {}^{[s]}H^{j-1}_{\rm loc}(\Sigma_0)$,
resp.~with $\widetilde{\Sigma}_0$ replacing $\Sigma_0$.
We state this as a proposition for reference:

\begin{proposition}[Well-posedness]
\label{wellposedstat}
For $s=\pm 2$, let
$(\tilde\upalpha^{[s]}_0, \tilde\upalpha^{[s]}_1)\in {}^{[s]}H^j_{\rm loc}(\Sigma_0)\times 
{}^{[s]}H^{j-1}_{\rm loc}(\Sigma_0)$ be
complex valued spin weighted functions with $j\ge 1$.
Then there exists a unique complex valued  $\tilde\upalpha^{[s]}$ 
on $\mathcal{R}_0$ satisfying $(\ref{rescaledeqTeu})$ 
(equivalently $\upalpha^{[s]}$ satisfying~\eqref{Teukphysic})
with
$\tilde\upalpha^{[s]} \in {}^{[s]}H^{j}_{\rm loc} (\Sigma_\tau)$, $n_{\Sigma_{\tau}} \tilde\upalpha^{[s]} \in {}^{[s]}H^{j-1}_{\rm loc}(\Sigma_\tau)$ such that $\tilde\upalpha^{[s]}\big|_{\Sigma_0}=\tilde\upalpha^{[s]}_0$,
$(n_{\Sigma_0}\tilde\upalpha^{[s]})\big|_{\Sigma_0}=\upalpha^{[s]}_1$. 
In particular, if 
 $(\tilde\upalpha^{[s]}_0, \tilde\upalpha^{[s]}_1)\in \mathscr{S}^{[s]}_{\infty}(\Sigma_0)$ 
then $\tilde\upalpha^{[s]}\in \mathscr{S}^{[s]}_{\infty}(\mathcal{R}_0)$.

The same statement holds with $\widetilde{\Sigma}_0$, $\widetilde{\Sigma}_\tau$,
$\widetilde{\mathcal{R}}_0$ in place of $\Sigma_0$, $\Sigma_\tau$, $\mathcal{R}_0$,
respectively.
\end{proposition}

\begin{proof}
cf.~Proposition 4.5.1 of~\cite{DafRodsmalla}.
\end{proof}

\subsection{Relation with the system of gravitational perturbations}
\label{relwiththesystemsec}

The Teukolsky equation (\ref{Teukphysicintro}) is traditionally derived via the 
Newman--Penrose formalism~\cite{newmanpenrose}. 
One defines the (complex) null tetrad $\left(\ell, n, m, \bar{m}\right)$ by
\begin{align}
l = \frac{r^2+a^2}{\Delta} L \ \ \ , \ \ \ n = \frac{r^2+a^2}{2\rho^2} \underline{L} \ \ \ , \ \ \ m = \frac{1}{\sqrt{2}(r+ia \cos \theta)} \left( ia\sin \theta \partial_t + \partial_{\theta} + \frac{i}{\sin \theta} \partial_\phi\right) \, ,
\end{align}
which is normalised such that
\[
g\left(l,n\right) = -1  \ \ \ , \ \ \ g \left(m, \bar{m}\right) = 1 \ \ \ , \ \ \ g\left(m,m\right) = g\left(\bar{m},\bar{m}\right) = 0 \, .
\]
Note that we can obtain an associated real spacetime null frame $\left(\ell, n, e_1,e_2\right)$ by defining $e_1 = \frac{1}{\sqrt{2}} \left(m + \bar{m}\right)$ and $e_2 =  \frac{1}{\sqrt{2}i} \left(m - \bar{m}\right)$, which then satisfies in particular $g\left(e_1,e_1\right)= g\left(e_2,e_2\right) =1$ and $g\left(e_1,e_2\right)=0$. 

The extremal Newman--Penrose curvature scalars are defined as the following components of the spacetime Weyl tensor\footnote{Recall that the Riemann tensor agrees with the Weyl tensor for a Ricci flat metric.}
\begin{align}
{\bf \Psi}_0 = - W \left(l,m,l,m\right) \ \ \ , \ \ \ {\bf \Psi}_4 = - W \left(n, \overline{m}, n, \overline{m}\right) \, .
\end{align}
Both ${\bf \Psi}_0$ and ${\bf \Psi}_4$ vanish for the exact Kerr metric. Remarkably, upon linearising the Einstein vacuum equations $(\ref{vaceqhere})$ (using the above frame) the linearised components ${\bf \Psi}_0$ and ${\bf \Psi}_4$ are gauge invariant (with respect to infinitesimal changes of both the frame and the coordinates) and moreover satisfy decoupled equations. 
Indeed, one may check that $\upalpha^{[-2]}= \left(r-ia\cos\theta\right)^4{\bf \Psi}_4$ and $\upalpha^{[+2]}= {\bf \Psi}_0$ satisfy precisely the Teukolsky
equation (\ref{Teukphysicintro}) for $s=-2$ and $s=2$ respectively. 

Instead of defining spin $s$-weighted complex functions ${\bf \Psi}_0$, ${\bf \Psi}_4$ one may (equivalently) define symmetric traceless $2$-tensors $\alpha$ and $\underline{\alpha}$ (living in an appropriate bundle of \emph{horizontal} tensors) by
\[
 \alpha \left(e_A, e_B\right) = W \left( L,e_A, L, e_B\right)  \ \ \ , \ \ \  \underline{\alpha} \left(e_A, e_B\right) = W \left( \underline{L},e_A, \underline{L}, e_B\right)  \, .
\]
Using the symmetry and the trace properties of the Weyl tensor we derive the relations
\[
\underline{\alpha} \left(e_1,e_1\right)=-\underline{\alpha} \left(e_2,e_2\right)  = -\frac{1}{2} \left(\frac{2\rho^2}{r^2+a^2}\right)^2 \left({\bf \Psi}_4 + \overline{{\bf \Psi}_4}\right)
\]
and
\[
\underline{\alpha} \left(e_1,e_2\right) = \underline{\alpha} \left(e_2,e_1\right)  = +\frac{1}{2}i\left(\frac{2\rho^2}{r^2+a^2}\right)^2 \left({\bf \Psi}_4 - \overline{{\bf \Psi}_4}\right) \, ,
\]
which relate the spin $2$-weighted complex function and the tensorial version of the curvature components.  Of course similar formulae are easily derived for $\alpha$. 

We can now connect directly to our previous~\cite{holzstabofschw} where we wrote down the 
Teukolsky equation for the symmetric traceless tensors $\alpha$ and $\underline{\alpha}$ in the Schwarzschild spacetime.

As a final remark we note that in the Schwarzschild case considered 
in~\cite{holzstabofschw} the null frame 
used to define the extremal Weyl components arose directly from a double null foliation of the spacetime. In stark contrast, the algebraically special null frame $\left(l, n, e_1,e_2\right)$ in 
Kerr for $a\ne 0$ does \emph{not} arise from a double null foliation of that spacetime.

\section{Generalised Chandrasekhar transformations for $s=\pm 2$}
\label{physspacechandrasec}
In this section, we generalise the physical space
reformulations of Chandrasekhar's transformations, given in~\cite{holzstabofschw}, to Kerr.

In accordance with the conventions of our present paper,
we will consider complex scalar spin $\pm2$ weighted quantities $\upalpha^{[\pm2]}$
in place of the tensorial ones of~\cite{holzstabofschw}.
We begin in {\bf Section~\ref{physicalspacedefsec}} with the definitions of
the quantities $P^{[\pm2]}$ associated to 
quantities $\upalpha^{[\pm2]}$.  If $\upalpha^{[\pm2]}$ satisfy
the (inhomogeneous) Teukolsky equation, then we show
in {\bf Section~\ref{inhomogRWphyspace}} that $P^{[\pm2]}$ will satisfy
an (inhomogeneous) 
Regge--Wheeler type equation, coupled to $\upalpha^{[\pm2]}$.
The latter coupling vanishes in the Schwarzschild case. 
The precise
relation with the tensorial definitions of~\cite{holzstabofschw} will be 
given in {\bf Section~\ref{relagainwithDHR}}.

\subsection{The definitions of $P^{[\pm2]}$, $\Psi^{[\pm2]}$ and $\uppsi^{[\pm2]}$}
\label{physicalspacedefsec}

Given functions $\upalpha^{[\pm2]}$, we define
\begin{equation}
\label{Pdefsbulkplus2}
 P^{[+2]} = -\frac{(r^2+a^2)^{1/2}}{2\Delta} \underline{L}^\mu \nabla_\mu \left( \frac{(r^2+a^2)^2}{\Delta} \underline{L}^\mu \nabla_\mu \left(\Delta^2 \left(r^2+a^2\right)^{-\frac{3}{2}} \upalpha^{[+2]}\right)\right) \,  , 
\end{equation}
\begin{equation}
\label{Pdefsbulkmunus2}
 P^{[-2]} =-\frac{(r^2+a^2)^{1/2}}{2\Delta} {L}^\mu \nabla_\mu \left( \frac{(r^2+a^2)^2}{\Delta} {L}^\mu \nabla_\mu \left( \left(r^2+a^2\right)^{-\frac{3}{2}} \upalpha^{[-2]}\right)\right)  .
\end{equation}
These are our physical-space generalisations to Kerr of Chandrasekhar's fixed frequency
Schwarzschild transformation
theory.

Note that if $\widetilde{\upalpha}^{[\pm2]}\in\mathscr{S}^{[\pm2]}_\infty(\mathcal{U})$
for $\mathcal{U}\subset\mathcal{R}$, then
$P^{[\pm2]}\in \mathscr{S}^{[\pm2]}_{\infty}(\mathcal{R})$.
We will typically work with the rescaled functions
\begin{equation}
\label{rescaledPSI}
\Psi^{[\pm2]} = (r^2+a^2)^{\frac32} P^{[\pm2]},
\end{equation}
which are of course again smooth.

As in~\cite{holzstabofschw}, it will be
again useful to give a name to the intermediate quantities $\uppsi^{[\pm2]}$
defined by
\begin{equation}
\label{officdeflittlepsiplus}
\uppsi^{[+2]} =
-\frac12 \Delta^{-\frac32} \left(r^2+a^2\right)^{+2} \underline{L}^\mu\nabla_\mu (\Delta^2 \left(r^2+a^2\right)^{-\frac{3}{2}} \upalpha^{[+2]}) 
\end{equation}
\begin{equation}
\label{officdeflittlepsiminus}
\uppsi^{[-2]} =+\frac12\Delta^{-\frac32}(r^2+a^2)^2 L^\mu\nabla_\mu
\left(\upalpha^{[-2]}(r^2+a^2)^{-\frac32}\right).
\end{equation}
We can rewrite $(\ref{Pdefsbulkplus2})$--$(\ref{Pdefsbulkmunus2})$
as
\begin{align} \label{auxrel1}
\underline{L}^\mu\nabla_\mu \left(\sqrt\Delta {\uppsi}^{[+2]} \right ) = \Delta(r^2+a^2)^{-2} \Psi^{[+2]},
\end{align}
\begin{align} \label{auxrel2}
L^\mu\nabla_\mu(\sqrt{\Delta} \uppsi^{[-2]})= - \Delta(r^2+a^2)^{-2}\Psi^{[-2]}.
\end{align}
Note that for $\widetilde{\upalpha}^{[\pm2]}$  smooth, it is the quantities $\sqrt{\Delta}\uppsi^{[+2]}$,
$(\sqrt{\Delta})^{-1}\uppsi^{[-2]}$
which are smooth.

\subsection{The generalised inhomogeneous Regge--Wheeler-type equation with error}
\label{inhomogRWphyspace}

The importance of the quantities $\Psi^{[\pm2]}$ arises from the following fundamental
proposition:

\begin{proposition}
\label{follfundprop}
If $\upalpha^{[\pm2]}$ satisfy the inhomogeneous equations
\begin{equation}
\label{inhomoteuk}
\widetilde{\mathfrak{T}}^{[+2]} \left(\tilde{\upalpha}^{[+2]}\right) =  F^{[+2]} \ \ \ \textrm{and} \ \ \  \widetilde{\mathfrak{T}}^{[-2]} \left(\Delta^2 \tilde{\upalpha}^{[-2]}\right) = \Delta^2 F^{[-2]}
\end{equation}
then the quantities 
$\Psi^{[\pm2]}$ satisfy the equation
\begin{equation}
\label{RWtypeinthebulk}
\mathfrak{R}^{[\pm2]}\Psi^{[\pm2]}=  -\frac{\rho^2}{\Delta} \mathcal{J}^{[\pm2]} -  \frac{\rho^2}{\Delta} \mathfrak{G}^{[\pm 2]}
\end{equation}
where
\begin{align}
\Delta \rho^{-2} \mathfrak{R}^{[s]}= \frac{1}{2} \left( L \underline{L} + \underline{L} L\right)   &+ \frac{\Delta}{\left(r^2+a^2\right)^2} \Big\{ \left(\mathring{\slashed\triangle}^{[s]} +s^2 +s\right)  -\frac{6Mr}{r^2+a^2} \frac{r^2-a^2}{r^2+a^2}  -7a^2 \frac{\Delta}{(r^2+a^2)^2} \Big\}       \nonumber \\
\label{RWoprdefhere}
& - \frac{\Delta}{\left(r^2+a^2\right)^2}\left(2 a T \Phi  +a^2 \sin^2\theta TT - 2i s a\cos \theta T  \right) \, ,
\end{align}
\begin{align}
\mathcal{J}^{[+2]} = & \frac{\Delta}{\left(r^2+a^2\right)^2} \left[ \frac{-8r^2+8a^2}{r^2+a^2} a\Phi -20a^2 \frac{r^3-3Mr^2+ra^2+Ma^2}{\left(r^2+a^2\right)^2}\right]\left( \sqrt{\Delta}\uppsi^{[+2]} \right) \nonumber \\
& +a^2 \frac{\Delta}{\left(r^2+a^2\right)^2} \left[ -12 \frac{r}{r^2+a^2}  a\Phi
+3  \left( \frac{r^4 -a^4+10Mr^3-6Ma^2r}{(r^2+a^2)^2}\right) \right] \left(\upalpha^{[+2]}\Delta^2 \left(r^2+a^2\right)^{-\frac{3}{2}}\right) \nonumber
\end{align}
\begin{align} \label{Gerror}
\mathfrak{G}^{[+ 2]}=\frac{1}{2} \underline{L} \left( \frac{\left(r^2+a^2\right)^2}{\Delta}\underline{L} \left(\frac{\Delta}{w\rho^2} F^{[+2]}\right)\right)
\end{align}
and
\begin{align}
\mathcal{J}^{[-2]} = 
& \frac{\Delta}{\left(r^2+a^2\right)^2} \left[ \frac{8r^2-8a^2}{r^2+a^2} a \Phi  -20a^2 \frac{r^3-3Mr^2+ra^2+Ma^2}{\left(r^2+a^2\right)^2} \right]\left( \sqrt{\Delta}\uppsi^{[-2]} \right) \nonumber \\
& +a^2 \frac{\Delta}{\left(r^2+a^2\right)^2} \left[  +12 \frac{r}{r^2+a^2} a \Phi
+3  \left( \frac{r^4 -a^4+10Mr^3-6Ma^2r}{(r^2+a^2)^2}\right) \right] \left(\upalpha^{[-2]} \left(r^2+a^2\right)^{-\frac{3}{2}}\right)  \, ,
\end{align}
\begin{align}
\label{Gerror2}
\mathfrak{G}^{[- 2]}=\frac{1}{2} {L} \left( \frac{\left(r^2+a^2\right)^2}{\Delta}{L} \left(\frac{\Delta^3}{w\rho^2} F^{[-2]}\right)\right).
\end{align}
\end{proposition}
\begin{proof}
Direct calculation. See Appendix.
\end{proof}

We will call the operator $\mathfrak{R}^{[s]}$
defined by $(\ref{RWoprdefhere})$
the \emph{generalised Regge--Wheeler operator}. We note that it has
smooth coefficients on $\mathcal{R}_0$ and its highest
order part is proportional to the wave operator.
The equation $(\ref{RWtypeinthebulk})$ reduces to the usual Regge--Wheeler 
equation in the case $a=0$:

\begin{corollary}
If $a=0$ and $F^{[\pm2]}=0$ then $\Psi^{[\pm2]}$ satisfies
the Regge--Wheeler equation
\begin{equation}
\label{RWinthebulkps}
L\underline{L} \Psi^{[\pm 2]} + \frac{\Omega^2}{r^2}
\left( \mathring{\slashed\triangle}^{[\pm 2]}  \Psi^{[\pm 2]}  \pm 2\right)   \Psi^{[\pm 2]} + \Omega^2\left(\frac{4}{r^2} - \frac{6M}{r^3}\right) \Psi^{[\pm 2]} = 0 ,
\end{equation}
where $\Omega^2 = 1- \frac{2M}{r}$.
\end{corollary}

As discussed already in the introduction, we see that $(\ref{RWtypeinthebulk})$,
although still coupled to $\upalpha^{[\pm2]}$, 
retains some of the good structure of $(\ref{RWinthebulkps})$.
The operator $(\ref{RWtypeinthebulk})$ has a good divergence
structure admitting estimates via integration by parts, i.e.~it
does not have the problematic first order terms of  the Teukolsky operator
$\widetilde{\mathfrak{T}}^{[\pm2]}$, cf.~(\ref{teurefer}). See 
already the divergence identities of Section~\ref{sec:multid}.
Moreover, the terms 
$\mathcal{J}^{[+2]}$ can be thought of as lower order, from the perspective
of $\Psi^{[\pm2]}$, as they only involve up to second derivatives of 
$\upalpha^{[\pm2]}$ (via the term $\Phi (\sqrt{\Delta}\uppsi^{[\pm2]})$).

\subsection{Relation with the quantities $P$ and $\protect\underline{P}$ of~\cite{holzstabofschw}}
\label{relagainwithDHR}

As with the tensorial quantities $\alpha$ and $\underline{\alpha}$ discussed in 
Section~\ref{relwiththesystemsec}, 
in \cite{holzstabofschw} the transformations to the quantities $P$ and $\underline{P}$ (corresponding to the complex functions $ P^{[+2]}$, $ {P}^{[-2]}$ in this paper) were again 
given tensorially. In particular, the 
Regge--Wheeler equation for the symmetric traceless tensor $\Psi=r^5 P$ was written tensorially using projected covariant derivatives  as (cf.~Corollary~7.1 of \cite{holzstabofschw})
\begin{equation}
\label{fromprevpaper}
\Omega \slashed{\nabla}_3 \left(\Omega \slashed{\nabla}_4 \Psi \right) - \Omega^2 \slashed{\Delta} \Psi +  \Omega^2 V \Psi = 0 \ \ \ \textrm{with} \ \  V=\frac{4}{r^2} - \frac{6M}{r^3} \, ,
\end{equation}
where $\slashed{\nabla}_3$ and $\slashed{\nabla}_4$ are projected (to the spheres of symmetry) covariant derivatives in the null directions,
$\slashed{\Delta}$ is the covariant Laplacian associated with the metric on the spheres of symmetry acting on symmetric traceless tensors and $\Omega^2 = 1- \frac{2M}{r}$. Note that 
unlike the operator $\mathring{\slashed\triangle}^{[s]}$ considered 
in this paper, the operator
$\slashed{\Delta}$ was defined as a \emph{negative} operator in \cite{holzstabofschw}.

Computing the equation satisfied by the components of $\Psi$ in the standard orthonormal frame on the spheres of symmetry one obtains 
\begin{align}
&L \underline{L} \left( \Psi_{11}\right) +\Omega^2 \left(-\slashed{\Delta} \left(\Psi_{11}\right) +4\frac{\cos \theta}{\sin^2 \theta} \partial_\phi \Psi_{12} +4 \cot^2 \theta \Psi_{11} \right)  +   \Omega^2 V \Psi_{11} = 0 \, , \nonumber \\
&L \underline{L} \left( \Psi_{12}\right) + \Omega^2 \left(-\slashed{\Delta} \left(\Psi_{12}\right) -4\frac{\cos \theta}{\sin^2 \theta} \partial_\phi \Psi_{11} +4 \cot^2 \theta \Psi_{12}\right)  +   \Omega^2 V \Psi_{12} = 0 \nonumber \, ,
\end{align}
from which one infers that the complex-valued functions $\Psi^{[\pm 2]}=  \Psi_{11} \mp i \Psi_{12}$ satisfy the Regge--Wheeler equation (\ref{RWinthebulkps}) for $s=\pm 2$.\footnote{Note that in this paper $\Psi^{[+2]} = r^3 P^{[+2]}$ for $a=0$ so when relating orthonormal components of the tensor $P$ and the complex function $P^{[2]}$ there is an additional factor of $r^2$. This factor disappears when replacing the orthonormal frame on the spheres of symmetry with an orthonormal frame on the \emph{unit} sphere to express the components of $P$.}

\section{Energy quantities and statement of the main theorem}
\label{generalstatement}

We first give certain definitions of weighted energy quantities  
 in {\bf Section~\ref{Weighandstate}}. This will allow us to give
 a precise statement of the main theorem of this paper
 ({\bf Theorem~\ref{finalstatetheor}})
in {\bf Section~\ref{statementsubsec}}.
We will finally discuss in {\bf Section~\ref{logicoftheproofsec}}
how the logic of the proof of Theorem~\ref{finalstatetheor}
is represented by the sections that follow.

\subsection{Definitions of weighted energies}
\label{Weighandstate}
We will define in this section a number of weighted energies.
In addition to those appearing in the statement of
Theorem~\ref{finalstatetheor}, we will need to consider various auxiliary
quantities.

\subsubsection{The left, right and trapped subregions}
We will in particular need to introduce  energies
localised to various subregions of $\widetilde{\Sigma}_\tau$ and
$\widetilde{\mathcal{R}}(\tau_1, \tau_2)$. In anticipation of this,
let us define the following subregions
\[
\widetilde{\mathcal{R}}^{\rm left}(\tau_1, \tau_2) =
\widetilde{\mathcal{R}}(\tau_1, \tau_2)\cap \{r\le A_1\},
\qquad \widetilde{\mathcal{R}}^{\rm right}(\tau_1, \tau_2) =
\widetilde{\mathcal{R}}(\tau_1, \tau_2)\cap \{r\ge A_2\},
\]
\[
\widetilde{\mathcal{R}}^{\rm away}(\tau_1, \tau_2)
= \widetilde{\mathcal{R}}^{\rm left}(\tau_1, \tau_2)\cup 
\widetilde{\mathcal{R}}^{\rm right}(\tau_1, \tau_2)
\]
\[
\widetilde{\mathcal{R}}^{\rm trap}(\tau_1, \tau_2)=
\widetilde{\mathcal{R}}(\tau_1, \tau_2)\cap \{A_1\le r\le A_2\}.
\]
Note that
\[
\widetilde{\mathcal{R}}^{\rm trap}(\tau_1, \tau_2) \cup 
\widetilde{\mathcal{R}}^{\rm away}(\tau_1, \tau_2)=
\widetilde{\mathcal{R}}^{\rm trap}(\tau_1, \tau_2) \cup 
\widetilde{\mathcal{R}}^{\rm left}(\tau_1, \tau_2)\cup 
\widetilde{\mathcal{R}}^{\rm right}(\tau_1, \tau_2) = 
\widetilde{\mathcal{R}}(\tau_1, \tau_2).
\]

For $\widetilde{\Sigma}_\tau$, it will be more natural to consider
\[
\widetilde{\Sigma}_\tau^{\rm left}= \widetilde{\Sigma}_\tau\cap \{r\le A_1\}, \qquad
\widetilde{\Sigma}_\tau^{\rm right}= \widetilde{\Sigma}_\tau\cap \{r\ge A_2\}, \qquad
\widetilde{\Sigma}_\tau^{\rm away}=\widetilde{\Sigma}_\tau^{\rm left}\cup  \widetilde{\Sigma}_\tau^{\rm right},
\]
\[
\widetilde{\Sigma}_\tau^{\rm overlap}= \widetilde{\Sigma}_\tau\cap \{2A_1^*\le r^*\le 2A_2^*\}.
\]
Note
\[
\widetilde{\Sigma}_\tau\subset 
\widetilde{\Sigma}_\tau^{\rm overlap}\cup \widetilde{\Sigma}_\tau^{\rm away}.
\]
See Figure~\ref{overlapfigure}.
\begin{figure}
\centering{
\def\svgwidth{10pc}
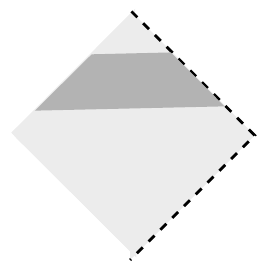}
\caption{Partioning $\widetilde{\mathcal{R}}(\tau_1, \tau_2)$ and
$\widetilde{\Sigma}_\tau$}\label{overlapfigure}
\end{figure}

\subsubsection{Weighted energies for $\Psi^{[\pm 2]}$}
The energies in this section will in general
be applied to $\Psi^{[\pm2]}$ satisfying the inhomogeneous equation 
$(\ref{RWtypeinthebulk})$.

Let $p$ be.a free parameter (which will eventually always take the
values $0, \eta, 1$ or $2$). 
We define the following weighted energies on the slices 
$\widetilde\Sigma_{\tau}$
\begin{align} \label{meP}
\mathbb{E}_{\widetilde\Sigma_{\tau},p} \left[\Psi^{[\pm2]}\right] \left(\tau\right) &= \int_{\widetilde\Sigma_{\tau}} dr d\sigma \left(  \big| L \Psi^{[\pm 2]}|^2 r^p + \big|\mathring{\slashed{\nabla}}^{[s]} \Psi^{[\pm 2]} \big|^2 r^{-2} + \big|\Psi^{[\pm 2]} \big|^2 r^{-2}   + r^{-1-\eta} \Big|\underline{L}\Psi^{[\pm 2]}\Big|^2\right)  \, ,  \\
\overline{\mathbb{E}}_{\widetilde\Sigma_{\tau},p} \left[\Psi^{[\pm2]}\right] \left(\tau\right) &= \int_{\widetilde\Sigma_{\tau}} dr d\sigma \left(  \big| L \Psi^{[\pm2]}|^2 r^p + \big|\mathring{\slashed{\nabla}}^{[s]} \Psi^{[\pm 2]} \big|^2 r^{-2} + \big|\Psi^{[\pm 2]} \big|^2 r^{-2} + r^{-1-\eta} \Big|\frac{r^2+a^2}{\Delta}\underline{L}\Psi^{[\pm 2]}\Big|^2 \right)\, . \nonumber
\end{align}
We remark that an overbar indicates that the energy has optimised weights near the horizon.

We will also consider the following  energy through $\widetilde\Sigma_{\tau}^{\rm away}$: 
\[
\mathbb{E}^{\rm away}_{\widetilde\Sigma_{\tau},p} \left[\Psi^{[\pm2]}\right] \left(\tau\right) = \int_{\widetilde\Sigma_{\tau}^{\rm away}} dr d\sigma \left(  \big| L \Psi^{[\pm 2]}|^2 r^p + \big|\mathring{\slashed{\nabla}}^{[s]} \Psi^{[\pm 2]} \big|^2 r^{-2} + \big|\Psi^{[\pm 2]} \big|^2 r^{-2} 
 + r^{-1-\eta} \Big|\underline{L}\Psi^{[\pm 2]}\Big|^2 \right) \, 
\]
and the following energy through $\widetilde\Sigma_{\tau}^{\rm overlap}$:
\[
\mathbb{E}^{\rm overlap}_{\widetilde{\Sigma}_\tau}\left[\Psi^{[\pm2]}\right](\tau)=
 \int_{\widetilde\Sigma_{\tau}^{\rm overlap}}
 dr d\sigma \left(  \big| L \Psi^{[\pm 2]}|^2 r^p + \big|\mathring{\slashed{\nabla}}^{[s]} \Psi^{[\pm 2]} \big|^2 r^{-2} + \big|\Psi^{[\pm 2]} \big|^2 r^{-2}   + r^{-1-\eta} \Big|\underline{L}\Psi^{[\pm 2]}\Big|^2 \right).
\]
Note that 
\begin{equation}
\label{just.a.note}
\mathbb{E}_{\widetilde{\Sigma}_\tau, p}\left[\Psi^{[\pm2]}\right]
\lesssim
\mathbb{E}^{\rm away}_{\widetilde{\Sigma}_\tau, p}\left[\Psi^{[\pm2]}\right]
+ \mathbb{E}^{\rm overlap}_{\widetilde{\Sigma}_\tau}\left[\Psi^{[\pm2]}\right].
\end{equation}

On the event horizon $\mathcal{H}^+$ we define the energies
\begin{align}
\mathbb{E}_{\mathcal{H}^+} \left[\Psi^{[\pm 2]}\right] \left(\tau_1,\tau_2\right) &=\int_{\tau_1}^{\tau_2} d\tau d\sigma  |L \Psi^{[\pm 2]}|^2   \, ,
\\
\overline{\mathbb{E}}_{\mathcal{H}^+} \left[\Psi^{[\pm 2]}\right] \left(\tau_1,\tau_2\right) &=\int_{\tau_1}^{\tau_2} d\tau d\sigma \left( | \Psi^{[\pm 2]}|^2 + |L \Psi^{[\pm 2]}|^2 + |\mathring{\slashed{\nabla}}^{[s]}\Psi^{[\pm 2]}|^2  \right) \, . \nonumber
\end{align}

On null infinity $\mathcal{I}^+$ we define the energies
\begin{align}
\mathbb{E}_{\mathcal{I}^+,p} \left(\Psi^{[\pm2]}\right] \left(\tau_1,\tau_2\right) &=\int_{\tau_1}^{\tau_2} d\tau d\sigma \left[  |\underline{L} \Psi^{[\pm2]}|^2 + r^{p-2} |\mathring{\slashed{\nabla}}^{[s]} \Psi^{[\pm2]}|^2 +r^{p-2} |\Psi^{-2}|^2 \right) \, . \nonumber
\end{align}

In addition to the energy fluxes,
we will define the weighted spacetime energies
\begin{align}
\mathbb{I}_p \left[\Psi^{[\pm 2]}\right] \left(\tau_1,\tau_2\right) &= \int_{\tau_1}^{\tau_2} d\tau \int_{\widetilde\Sigma_{\tau}} dr d\sigma \left(\left( \frac{ \big| L \Psi^{[\pm 2]}|^2}{r^{1+
\boldsymbol{\delta}^p_0 \eta}} + \frac{\big| \mathring{\slashed{\nabla}}^{[s]} \Psi^{[\pm 2]} \big|^2}{r^{3+\boldsymbol{\delta}^p_2\eta}} + \frac{\big|\Psi^{[\pm 2]} \big|^2}{r^{3+\boldsymbol{\delta}^p_2\eta}} \right) r^p +\frac{\big|\underline{L}\Psi^{[\pm 2]}\big|^2}{ r^{1+\eta} } \right)\, , \label{nondegPsi}  \\
\overline{\mathbb{I}}_p \left[\Psi^{[\pm 2]}\right] \left(\tau_1,\tau_2\right) &= \int_{\tau_1}^{\tau_2} d\tau \int_{\widetilde\Sigma_{\tau}} dr d\sigma \left(\left( \frac{ \big| L \Psi^{[\pm 2]}|^2}{r^{1+\boldsymbol{\delta}^p_0 \eta}} + \frac{\big|\mathring{\slashed{\nabla}}^{[s]} \Psi^{[\pm 2]} \big|^2}{r^{3+\boldsymbol{\delta}^p_2\eta}} + \frac{\big|\Psi^{[\pm 2]} \big|^2}{r^{3+\boldsymbol{\delta}^p_2\eta}} \right) r^p +\frac{\Big|\frac{r^2+a^2}{\Delta}\underline{L}\Psi^{[\pm 2]}\Big|^2}{ r^{1+\eta} } \right)\, ,
\end{align}
where $\boldsymbol{\delta}^a_b$ is the Kronecker delta symbol and also the degenerate spacetime energies
\begin{align}
{\mathbb{I}}^{{\rm deg}}_p \left[\Psi^{[\pm 2]}\right] \left(\tau_1,\tau_2\right) = \int_{\tau_1}^{\tau_2} d\tau \int_{\widetilde\Sigma_{\tau}} dr d\sigma \Bigg\{ & \left( \frac{ \big| L \Psi^{[\pm 2]}|^2}{r^{1+\boldsymbol{\delta}^p_0 \eta}} + \frac{\big| \mathring{\slashed{\nabla}}^{[s]} \Psi^{[\pm 2]} \big|^2}{r^{3+\boldsymbol{\delta}^p_2\eta}} +\frac{\big|\underline{L}\Psi^{[\pm 2]}\big|^2}{ r^{1+\eta} } r^{-p}  \right) r^p \cdot \chi \nonumber \\
 &+ \frac{\big|(\underline{L}-L)\Psi^{[\pm 2]}\big|^2}{ r^{1+\eta} }+ r^p \frac{\big|\Psi^{[\pm 2]} \big|^2}{r^{3+\boldsymbol{\delta}^p_2\eta}} \Bigg\} \, ,
\end{align}
\begin{align}
\overline{\mathbb{I}}^{{\rm deg}}_p \left[\Psi^{[\pm 2]}\right] \left(\tau_1,\tau_2\right) &= {\mathbb{I}}^{{\rm deg}}_p \left[\Psi^{[\pm 2]}\right] \left(\tau_1,\tau_2\right) \ \ \textrm{but replacing $\underline{L}$ by $\frac{r^2+a^2}{\Delta}\underline{L}$ in the round bracket} \nonumber
\end{align}
with $\chi$ a radial cut-off function equal to $1$ in 
$r^*\in (-\infty, A_1^*] \cup [A_2^*,\infty)$ and vanishing in $r^*\in \left[A_1^*/4,A_2^*/4\right]$.
Finally, we shall define 
\begin{align*}
{\mathbb{I}}^{{\rm away}}_p \left[\Psi^{[\pm 2]}\right] \left(\tau_1,\tau_2\right) &= \int_{\tau_1}^{\tau_2} d\tau \int_{\widetilde\Sigma_{\tau}^{\rm away}} dr d\sigma \left[ \frac{ \big| L \Psi^{[\pm 2]}|^2}{r^{1+\boldsymbol{\delta}^p_0 \eta}} + \frac{\big|\mathring{\slashed{\nabla}}^{[s]} \Psi^{[\pm 2]} \big|^2}{r^{3+\boldsymbol{\delta}^p_2\eta}} +\frac{\big|\underline{L}\Psi^{[\pm 2]}\big|^2}{ r^{1+\eta} } r^{-p}  + \frac{\big|\Psi^{[\pm 2]} \big|^2}{r^{3+\boldsymbol{\delta}^p_2\eta}} \right] r^p 
\end{align*}
and
\begin{align}
\nonumber
\mathbb{I}^{{\rm trap}}[\Psi^{[\pm2]}](\tau_1,\tau_2) =& \int_{\tau_1}^{\tau_2}
d\tau\int_{\widetilde{\Sigma}_\tau^{\rm trap}}drd\sigma \Bigg\{
\chi 
\left( \big| L \Psi^{[\pm 2]}|^2+\big| \mathring{\slashed{\nabla}}^{[s]} \Psi^{[\pm 2]} \big|^2
+ \big|\underline{L}\Psi^{[\pm 2]}\big|^2 \right)\\
\label{I.trapped.def}
&\qquad\qquad\qquad\qquad\qquad\qquad\qquad
\big|(\underline{L}-L)\Psi^{[\pm 2]}\big|^2+\big|\Psi^{[\pm 2]} \big|^2
\Bigg\}.
\end{align}
Note that
\[
{\mathbb{I}}^{{\rm deg}}_p \left[\Psi^{[\pm 2]}\right] (\tau_1,\tau_2) \lesssim
{\mathbb{I}}^{{\rm away}}_p \left[\Psi^{[\pm 2]}\right] \left(\tau_1,\tau_2\right) +
\mathbb{I}^{{\rm trap}}[\Psi^{[\pm2]}](\tau_1,\tau_2).
\]

\subsubsection{Weighted energies for $\upalpha^{[+2]}$, $\uppsi^{[+2]}$}
The quantities in this section will in general
be applied to $\upalpha^{[+2]}$, $\uppsi^{[+2]}$ arising from
a solution $\tilde\upalpha^{[+2]}$ 
of the inhomogeneous equation 
$(\ref{inhomoteuk})$.

We define the following energy densities
\begin{align} \label{dens1}
e_p \left[\upalpha^{[+2]}\right] &= \sum_{\Gamma \in \{id, \Phi\}} \Big|  \Gamma \left( \upalpha^{[+2]} \Delta^2 \left(r^2+a^2\right)^{-1}\right)  \Big|^2   r^{-\boldsymbol{\delta}^p_2\eta} r^p + \Big|T\left(\upalpha^{[+2]} \Delta^2 \left(r^2+a^2\right)^{-1}\right)  \Big|^2   r^{2-\eta} \, ,
\\
e_p \left[\uppsi^{[+2]}\right] &= \sum_{\Gamma \in \{id, \Phi\}} \Big|  \Gamma \left(\uppsi^{[+2]} \sqrt{\Delta}\right) \Big|^2   r^{-\boldsymbol{\delta}^p_2\eta} r^p + \Big|T\left(\uppsi^{[+2]} \sqrt{\Delta}\right) \Big|^2   r^{2-\eta} \, . \label{dens2}
\end{align}
With these, we 
define the following weighted energies on the slices 
$\widetilde\Sigma_{\tau}$:
\begin{align}
\mathbb{E}_{\widetilde\Sigma_{\tau},p} \left[\upalpha^{[+2]}\right] \left(\tau\right) = \int_{\widetilde\Sigma_{\tau}} dr d\sigma \  e_p \left[\upalpha^{[+2]}\right] \ \ \ , \ \ \ 
\mathbb{E}_{\widetilde\Sigma_{\tau},p} \left[\uppsi^{[+2]}\right] \left(\tau\right) = \int_{\widetilde\Sigma_{\tau}} dr d\sigma \ e_p \left[\uppsi^{[+2]}\right]   \, .
\end{align}
\begin{remark}
\label{control.the.l.derivs1}
We remark already that while these energies contain the $T$ and the $\Phi$ derivative only, we can obtain also the $L$ and the $\underline{L}$ derivative if we control in addition the energy (\ref{meP}) of $\Psi^{[+2]}$. This is because of the relations (\ref{psirel}) and (\ref{Prel}) and the relation $L = -\underline{L} + 2T + \frac{a}{r^2+a^2}\Phi$.
\end{remark}
It will be useful to also consider separately
\begin{align}
\label{newleftdef}
\mathbb{E}^{\rm left}_{\widetilde\Sigma_{\tau},p} \left[\upalpha^{[+2]}\right] \left(\tau\right) &= \int_{\widetilde\Sigma_{\tau}^{\rm left}} dr d\sigma  e_p \left[\upalpha^{[+2]}\right]  \, , \  \
\mathbb{E}^{\rm left}_{\widetilde\Sigma_{\tau},p} \left[\uppsi^{[+2]}\right] \left(\tau\right) = \int_{\widetilde\Sigma_{\tau}^{\rm left}} dr d\sigma \ e_p \left[\uppsi^{[+2]}\right]   \, , \\
\mathbb{E}^{\rm right}_{\widetilde\Sigma_{\tau},p} \left[\upalpha^{[+2]}\right] \left(\tau\right) &= \int_{\widetilde\Sigma_{\tau}^{\rm right}} dr d\sigma \  e_p \left[\upalpha^{[+2]}\right]  \, , \
\label{newrightdef}
\mathbb{E}^{\rm right}_{\widetilde\Sigma_{\tau},p} \left[\uppsi^{[+2]}\right] \left(\tau\right) = \int_{\widetilde\Sigma_{\tau}^{\rm right}} dr d\sigma  \ e_p \left[\uppsi^{[+2]}\right]   \, . 
\end{align}
We also use the notation $\mathbb{E}^{\rm away}_{\widetilde\Sigma_{\tau},p}$ for the sum of the left and the right energies.
On (timelike) hypersurfaces of constant $r=A>r_+$ we define
\begin{align}
\mathbb{E}_{r=A} \left[\upalpha^{[+2]}\right] \left(\tau_1,\tau_2\right) &= \int_{\tau_1}^{\tau_2} d\tau d\sigma\  e_p \left[\upalpha^{[+2]}\right]  \Bigg|_{r=A}  \ ,  \ 
\mathbb{E}_{r=A} \left[\uppsi^{[+2]}\right] \left(\tau_1,\tau_2\right) =\int_{\tau_1}^{\tau_2} d\tau d\sigma \ e_p \left[\uppsi^{[+2]}\right]  \Bigg|_{r=A} .
\end{align}
In the limit $r \rightarrow r_+$ we obtain the energies the event horizon $\mathcal{H}^+$ which we denote
\begin{align}
\mathbb{E}_{\mathcal{H}^+} \left[\upalpha^{[+2]}\right] \left(\tau_1,\tau_2\right) = \int_{\tau_1}^{\tau_2} d\tau d\sigma \  e_p \left[\upalpha^{[+2]}\right]  \Bigg|_{r=r_+}  \, ,  \, \
\mathbb{E}_{\mathcal{H}^+} \left[\uppsi^{[+2]}\right] \left(\tau_1,\tau_2\right) =\int_{\tau_1}^{\tau_2} d\tau d\sigma \ e_p \left[\uppsi^{[+2]}\right]  \Bigg|_{r=r_+} \, .
\end{align}

Let us  define
\[
\mathbb{E}^{\rm overlap}_{\widetilde{\Sigma}_{\tau}}\left[\upalpha^{[\pm2]}\right]
=\int_{\widetilde{\Sigma}_{\tau}^{\rm overlap}}
dr\,d\sigma \, e_{0}\left[\upalpha^{[\pm2]}\right], \qquad
\mathbb{E}^{\rm overlap}_{\widetilde{\Sigma}_{\tau}}\left[\uppsi^{[\pm2]}\right]
=\int_{\widetilde{\Sigma}_{\tau}^{\rm overlap}}
dr\,d\sigma \, e_{0}\left[\uppsi^{[\pm2]}\right].
\]

We also define the following weighted spacetime energies
\begin{align}
\mathbb{I}_p \left[\upalpha^{[+2]}\right]  \left(\tau_1,\tau_2\right)= \int_{\tau_1}^{\tau_2} d\tau \int_{\widetilde\Sigma_{\tau}} dr d\sigma \  \frac{1}{r} e_p \left[\upalpha^{[+2]}\right]  \ \ , \ \ \nonumber 
\mathbb{I}_p \left[\uppsi^{[+2]}\right]  \left(\tau_1,\tau_2\right)= \int_{\tau_1}^{\tau_2} d\tau \int_{\widetilde\Sigma_{\tau}} dr d\sigma \ \frac{1}{r} e_p \left[\uppsi^{[+2]}\right]   \, .
\end{align} 
As with the fluxes, it will be useful to also define
\begin{align}
\label{alsointenglr1}
\mathbb{I}_p^{\rm left} \left[\upalpha^{[+2]}\right]  \left(\tau_1,\tau_2\right)&= \int_{\tau_1}^{\tau_2} d\tau \int_{\widetilde\Sigma_{\tau}^{\rm left}} dr\, d\sigma \ \frac{1}{r} e_p \left[\upalpha^{[+2]}\right] \, , \\
\mathbb{I}_p^{\rm left} \left[\uppsi^{[+2]}\right]  \left(\tau_1,\tau_2\right)&= \int_{\tau_1}^{\tau_2} d\tau \int_{\widetilde\Sigma_{\tau}^{\rm left} } dr\, d\sigma \ \frac{1}{r} e_p \left[\uppsi^{[+2]}\right]   \, , \\
\mathbb{I}_p^{\rm right} \left[\upalpha^{[+2]}\right]  \left(\tau_1,\tau_2\right)&= \int_{\tau_1}^{\tau_2} d\tau \int_{\widetilde\Sigma_{\tau}^{\rm right} } dr\, d\sigma \ \frac{1}{r} e_p \left[\upalpha^{[+2]}\right] \, , \\
\label{alsointenglr4}
\mathbb{I}_p^{\rm right} \left[\uppsi^{[+2]}\right]  \left(\tau_1,\tau_2\right)&= \int_{\tau_1}^{\tau_2} d\tau \int_{\widetilde\Sigma_{\tau}^{\rm right} } dr\, d\sigma \ \frac{1}{r} e_p \left[\uppsi^{[+2]}\right]   \, .
\end{align}
Finally, we define
\begin{equation}
\label{edw.trapped.def}
\mathbb{I}^{\rm trap}\left [\upalpha^{[+2]}\right](\tau_1,\tau_2) =\int_{\tau_1}^{\tau_2}
d\tau\int_{\widetilde{\Sigma}_\tau^{\rm trap}}
dr\, d\sigma e_{0}\left[\upalpha^{[+2]}\right] \, ,
\end{equation}
\begin{equation}
\label{edw.trapped.def.2}
\mathbb{I}^{\rm trap}\left[\uppsi^{[+2]}\right](\tau_1,\tau_2) =\int_{\tau_1}^{\tau_2}
d\tau\int_{\widetilde{\Sigma}_\tau^{\rm trap}}
dr\, d\sigma \,  e_{0}\left[\uppsi^{[+2]}\right].
\end{equation}
We note the relations
\begin{equation}
\label{poses.snmeiwseis}
\mathbb{I}_p \left[\upalpha^{[+2]}\right]  \left(\tau_1,\tau_2\right) 
\lesssim 
\mathbb{I}_p^{\rm left} \left[\upalpha^{[+2]}\right]  \left(\tau_1,\tau_2\right)+
\mathbb{I}^{\rm trap}\left[\upalpha^{[+2]}\right](\tau_1,\tau_2) 
+\mathbb{I}_p^{\rm right} \left[\upalpha^{[+2]}\right]  \left(\tau_1,\tau_2\right).
\end{equation}

\subsubsection{Weighted energies for $\upalpha^{[-2]}$, $\uppsi^{[-2]}$}
The quantities in this section will in general
be applied to $\upalpha^{[-2]}$, $\uppsi^{[-2]}$ arising from a solution
$\tilde\upalpha^{[-2]}$ of the inhomogeneous equation 
$(\ref{inhomoteuk})$.

We define the following weighted energies on the slices $\widetilde\Sigma_{\tau}$:\footnote{Note that in contrast to the $\left[+2\right]$-energies, no $p$-weights appear. The underlying reason is that the transport estimates for $\uppsi^{[-2]}$ and $\upalpha^{[-2]}$ will always be applied with the same $r$-weight. Note also in this context that the $\mathbb{E}$-energies for $\upalpha^{[-2]}$ and $\uppsi^{[-2]}$ on the slices $\widetilde{\Sigma}_{\tau}$ in (\ref{zas})--(\ref{zas2}) carry the same $r$-weight as the corresponding spacetime $\mathbb{I}$-energies in (\ref{fiid})--(\ref{finid}). This arises from the fact that the transport for the $[-2]$-quantities happens in the $L$-direction and the relation (\ref{innp}) between $L$ and the unit normal to the slices $\widetilde{\Sigma}_{\tau}$.}
\begin{align}
\mathbb{E}_{\widetilde\Sigma_{\tau}} \left[\upalpha^{[-2]}\right] \left(\tau\right) &= \int_{\widetilde\Sigma_{\tau}} dr d\sigma  \sum_{\Gamma \in \{id, T, \Phi\}} \Big|\Gamma \left( \frac{\sqrt{r^2+a^2} \alpha^{[-2]}}{\Delta^2}\right)  \Big|^2   r^{-1-\eta} \label{zas} \, , \\
\mathbb{E}_{\widetilde\Sigma_{\tau}} \left[\uppsi^{[-2]}\right] \left(\tau\right) &= \int_{\widetilde\Sigma_{\tau}} dr d\sigma  \sum_{\Gamma \in \{id, T, \Phi\}}   \Big| \Gamma \left( \frac{{\uppsi}^{[-2]}(r^2+a^2)}{\sqrt{\Delta}}   \right) \Big|^2 r^{-1-\eta}\, . \label{zas2}  
\end{align}

We also define the energies 
\[
\overline{\mathbb{E}}_{\widetilde\Sigma_{\tau}} \left[\upalpha^{[-2]}\right] \left(\tau\right) \ \ ,  \ \ \overline{\mathbb{E}}_{\widetilde\Sigma_{\tau}} \left[\uppsi^{[-2]}\right] \left(\tau\right) \ \  
\]
by adding to the set $\Gamma$ in the energies without the overbar the vectorfield $\frac{r^2+a^2}{\Delta}\underline{L}$. Hence an overbar again indicates that the energy has been improved near the horizon. 
\begin{remark}
\label{control.the.l.derivs2}
In analogy with Remark~\ref{control.the.l.derivs1}, note that in view of the relations (\ref{psibart}) and (\ref{Pbart}) controlling the energies above and in addition the energy (\ref{meP}) allows one to control also the $L$ derivative of $\upalpha^{[-2]}$ and $\uppsi^{[-2]}$. Together these allow one to control the $\underline{L}$ derivative of $\upalpha^{[-2]}$ and $\uppsi^{[-2]}$ (without the $\Delta^{-1}$-weight near the horizon) in view of the relation $\underline{L} = -{L} + 2T + \frac{a}{r^2+a^2}\Phi$.
\end{remark}
We define 
\[
\mathbb{E}^{\rm left}_{\widetilde\Sigma_{\tau}} \left[\upalpha^{[-2]}\right], \qquad
\mathbb{E}^{\rm left}_{\widetilde\Sigma_{\tau}} \left[\uppsi^{[+2]}\right], \qquad
\mathbb{E}^{\rm right}_{\widetilde\Sigma_{\tau}} \left[\upalpha^{[-2]}\right], \qquad
\mathbb{E}^{\rm right}_{\widetilde\Sigma_{\tau}} \left[\uppsi^{[-2]]}\right],
\]
by appropriately restricting the domain in $(\ref{zas})$--$(\ref{zas2})$, 
in analogy with the
definitions $(\ref{newleftdef})$--$(\ref{newrightdef})$.

On (timelike) hypersurfaces of constant $r=A>r_+$ we define
\begin{align}
\mathbb{E}_{r=A} \left[\upalpha^{[-2]}\right] \left(\tau_1,\tau_2\right) &= \int_{\tau_1}^{\tau_2} d\tau d\sigma \sum_{\Gamma \in \{id, T, \Phi\}} \Big| \Gamma \left(  \frac{\sqrt{r^2+a^2} \alpha^{[-2]}}{\Delta^2}\right)  \Big|^2 \Bigg|_{r=A} \, ,  \nonumber \\
\mathbb{E}_{r=A} \left[\uppsi^{[-2]}\right] \left(\tau_1,\tau_2\right) &=\int_{\tau_1}^{\tau_2} d\tau d\sigma \sum_{\Gamma \in \{id, T, \Phi\}} \Big|  \Gamma \left(\frac{{\uppsi}^{[-2]}(r^2+a^2)}{\sqrt{\Delta}} \right) \Big|^2 \Bigg|_{r=A} \, .
\end{align}
In the limit $r\rightarrow \infty$ we define on null infinity $\mathcal{I}^+$ 
\begin{align}
\mathbb{E}_{\mathcal{I}^+} \left[\upalpha^{[-2]}\right] \left(\tau_1,\tau_2\right) &= \int_{\tau_1}^{\tau_2} d\tau d\sigma \sum_{\Gamma \in \{id, T, \Phi\}} \Big| \Gamma \left(  \frac{\sqrt{r^2+a^2} \alpha^{[-2]}}{\Delta^2}\right)  \Big|^2 \Bigg|_{r\rightarrow \infty}\, ,  \nonumber \\
\mathbb{E}_{\mathcal{I}^+} \left[\uppsi^{[-2]}\right] \left(\tau_1,\tau_2\right) &=\int_{\tau_1}^{\tau_2} d\tau d\sigma \sum_{\Gamma \in \{id, T, \Phi\}} \Big|  \Gamma \left(\frac{{\uppsi}^{[-2]}(r^2+a^2)}{\sqrt{\Delta}} \right) \Big|^2 \Bigg|_{r\rightarrow \infty} \, .
\end{align}
We also define the following weighted spacetime energies
\begin{align}
\mathbb{I} \left[\upalpha^{[-2]}\right]  \left(\tau_1,\tau_2\right)&= \int_{\tau_1}^{\tau_2} d\tau \int_{\widetilde\Sigma_{\tau}} dr d\sigma \sum_{\Gamma \in \{id, T, \Phi\}} \Big|\Gamma \left(  \frac{\sqrt{r^2+a^2} \alpha^{[-2]}}{\Delta^2}         \right)  \Big|^2 r^{-1-\eta} 
\label{fiid} \, ,
 \\
\mathbb{I} \left[\uppsi^{[-2]}\right]  \left(\tau_1,\tau_2\right)&= \int_{\tau_1}^{\tau_2} d\tau \int_{\widetilde\Sigma_{\tau}} dr d\sigma \sum_{\Gamma \in \{id, T, \Phi\}} \Big| \Gamma \left(\frac{{\uppsi}^{[-2]}(r^2+a^2)}{\sqrt{\Delta}}\right) \Big|^2 r^{-1-\eta}   \label{finid} \, ,
\end{align}
and the energies 
\[
\overline{\mathbb{I}} \left[\upalpha^{[-2]}\right]  \left(\tau_1,\tau_2\right) \ \ , \ \ \overline{\mathbb{I}} \left[\uppsi^{[-2]}\right]  \left(\tau_1,\tau_2\right) 
\]
by adding to the set $\Gamma$ appearing in the definitions (\ref{fiid})--(\ref{finid}) the vectorfield $\frac{r^2+a^2}{\Delta}\underline{L}$.
We define again
\[
\mathbb{I}^{\rm left} \left[\upalpha^{[-2]}\right], \qquad
\mathbb{I}^{\rm left} \left[\uppsi^{[-2]}\right], \qquad
\mathbb{I}^{\rm right} \left[\upalpha^{[-2]}\right], \qquad
\mathbb{I}^{\rm right} \left[\uppsi^{[-2]}\right],
\]
by restricting the domain in $(\ref{fiid})$--$(\ref{finid})$, in analogy
with $(\ref{alsointenglr1})$--$(\ref{alsointenglr4})$.
Finally, in analogy with $(\ref{edw.trapped.def})$--$(\ref{edw.trapped.def.2})$, we define
\begin{equation}
\label{minus.2.defs}
\mathbb{I}^{\rm trap}\left [\upalpha^{[-2]}\right], \qquad
\mathbb{I}^{\rm trap}\left[\uppsi^{[-2]}\right]
\end{equation}
and we note the $[-2]$ version of $(\ref{poses.snmeiwseis})$.

\subsection{Precise statement of the main theorem: Theorem~\ref{finalstatetheor}}
\label{statementsubsec}

We are now ready to give a precise version of the main theorem
stated in Section~\ref{herethemainresultintro}:

\begin{theorem}\label{finalstatetheor}
Let
$(\tilde\upalpha^{[\pm2]}_0, \tilde\upalpha^{[\pm2]}_1)\in {}^{[\pm2]}H^j_{\rm loc}(
\widetilde{\Sigma}_0)\times 
{}^{[\pm2]}H^{j-1}_{\rm loc}(\widetilde{\Sigma}_0)$ and   $\tilde\upalpha^{[\pm2]}$ 
be as in the well-posedness  Proposition~\ref{wellposedstat}, and let $\upalpha^{[\pm2]}$,
$P^{[\pm2]}$, $\Psi^{[\pm2]}$, $\uppsi^{[\pm2]}$ be as defined
by $(\ref{rescaleddefsfirstattempt})$, $(\ref{Pdefsbulkplus2})$, $(\ref{Pdefsbulkmunus2})$,
$(\ref{rescaledPSI})$, $(\ref{officdeflittlepsiplus})$ and $(\ref{officdeflittlepsiminus})$.
Then the following estimates hold:
\begin{enumerate}
\item degenerate energy boundedness and integrated local energy decay as in Theorem \ref{degenerateboundednessandILED}
\item
red-shifted boundedness and integrated local energy decay as in Theorem \ref{prop:basicnondeg}
\item
the weighted $r^p$ hierarchy of estimates as in Propositions \ref{prop:nondegweighted}   and  \ref{prop:weightedtransport} ($s=+2$) \\
as well as Propositions \ref{prop:nondegweighted2} and \ref{prop:weightedtransport2} ($s=-2$)

\item polynomial decay of the energy as in Theorem \ref{prop:maindecayprop}.
\end{enumerate}
For each statement, the Sobolev exponent $j$ in the initial data norm is assumed large enough
so that the quantities on the right hand sides of the corresponding estimates above are well defined.
\end{theorem}

Let us note that we can easily deduce from the above
an alternative version where
initial data is posed (and weighted norms given) on $\Sigma_0$ instead of $\widetilde{\Sigma}_0$.
We suffice here with the remark that smooth, compactly supported initial data on $\Sigma_0$
trivially give rise to initial data on $\widetilde{\Sigma}_0$ satisfying the assumptions of the
above theorem.

As an example of the pointwise estimates which follow immediately
from the above theorem,
let us note the following pointwise corollary:
\begin{corollary}
\label{ptwisecor}
Let $(\tilde\upalpha^{[\pm2]}_0, \tilde\upalpha^{[\pm2]}_1)$
be smooth and of compact support. Then
the  solution $\tilde{\upalpha}$ satisfies
\[
|r^{\frac{3+\eta}{2}} \tilde{\upalpha}^{[+2]}|\le C|{\tilde{t}^*}|^{-(2-\eta)/2} \ \ \ \ , \ \ \ \ |r^4\tilde{\upalpha}^{[-2]}|\le C|{\tilde{t}^*}|^{-(2-\eta)/2}
\]
where $C$ depends on an appropriate higher Sobolev weighted norm.
\end{corollary}

The above decay rates can be improved following~\cite{Moschnewmeth}.

\begin{remark}
\label{decaycomments}
Recall that the quantities $\tilde\upalpha^{[\pm2]}$ are regular on the horizon and that near infinity  $r^{\frac{3+\eta}{2}} \tilde{\upalpha}^{[+2]} \sim r^{\frac{5+\eta}{2}} {\upalpha}^{[+2]} \sim r^{\frac{5+\eta}{2}} {\bf \Psi}_0$ and $r^4 \tilde{\upalpha}^{[-2]} \sim r^{-3} \upalpha^{[-2]} \sim r{\bf \Psi}_4$, allowing direct comparison with the null-components of curvature in an orthonormal frame (see Section \ref{relwiththesystemsec}). 
\end{remark}

\begin{remark}
Note that, in view of  Remark~\ref{decaycomments}, one sees that the decay
in $r$ provided for ${\bf \Psi}_0$ by Corollary~\ref{ptwisecor} is weaker than
peeling, consistent with the fact that, just as  in~\cite{CK}, our weighted energies  do not in fact 
impose initially
the validity of peeling. This   is important since it has been shown that peeling does not hold for
generic physically interesting data~\cite{christorome}. 
\end{remark}

\subsection{The logic of the proof}
\label{logicoftheproofsec}

The remainder of the paper concerns the proof of Theorem~\ref{finalstatetheor}.

Sections~\ref{Physspacesecnew}--\ref{ODEmegasec} are preliminary:
Section~\ref{Physspacesecnew} will prove an integrated energy estimate for
$\Psi^{[\pm2]}$, $\uppsi^{[\pm2]}$ and $\upalpha^{[\pm2]}$ arising from
general solutions to the inhomogeneous $s=\pm2$ Teukolsky equations $(\ref{inhomoteuk})$
outside of the region $r\in [A_1,A_2]$, with additional boundary
terms on $r=A_i$, as well as certain auxiliary estimates 
(Section~\ref{auxil.est.sec})
for $\Psi^{[\pm2]}$, $\uppsi^{[\pm2]}$ and $\upalpha^{[\pm2]}$ arising
from  a solution of the homogeneous equation $(\ref{Teukphysic})$.
Sections~\ref{Separationmegasec}--\ref{ODEmegasec}
will concern so-called $[A_1,A_2]$-admissible  solutions 
and will provide frequency-localised estimates  in the region $[A_1,A_2]$,
again with boundary terms on $r=A_i$.

The proof
proper of Theorem~\ref{finalstatetheor} commences in Section~\ref{iledsec}
where the degenerate integrated local energy decay and boundedness statements
are proven (statement 1.), using
the results of Sections~\ref{Physspacesecnew}--\ref{iledsec}, applied to  a
particular solution
 $\upalpha^{[\pm2]}_{\text{\Rightscissors}}$  of the inhomogeneous
equation $(\ref{inhomoteuk})$
which arises by cutting off a solution $\upalpha$ of the
homogeneous equation so that, \emph{when
restricted to  the $r$-range $[A_1,A_2]$},  $\upalpha^{[\pm2]}_{\text{\Rightscissors}}$
is compactly supported in $t^*\in [0,\tau_{\rm final}]$. 
The estimate of statement 1.~follows by appropriately summing the estimates
of Section~\ref{Physspacesecnew} and~\ref{ODEmegasec} applied to
$\upalpha_{\text{\Rightscissors}}$. We note already that when summing,
the most dangerous boundary terms on $r=A_i$
have been arranged to precisely cancel, while
the error term arising from the inhomogeneous term on the right hand
side of the equation of $\upalpha_{\text{\Rightscissors}}$ can easily be absorbed
in view of its support properties and the auxiliary estimates of
Section~\ref{auxil.est.sec}. Finally, in
Section~\ref{axi-note}, we will distill from our argument 
a simpler, purely physical space proof
of statement 1.~for the axisymmetric case.

The degenerate boundedness and integrated local energy decay
are combined with  redshift estimates in Section~\ref{sec:rsim} 
to obtain   statement 2.

Finally, the weighted $r^p$ estimates are obtained in Section~\ref{rphierarchysec},
giving statements 3.--4.

\section{Conditional physical space estimates} 
\label{Physspacesecnew}

In this section, we will derive certain physical space 
estimates for $\Psi^{[\pm2]}$, $\uppsi^{[\pm2]}$, $\upalpha^{[\pm2]}$
defined above, arising from solutions  $\upalpha^{[\pm2]}$ of the inhomogeneous version 
$(\ref{inhomoteuk})$ of
the Teukolsky equation.

We first apply  in {\bf Section~\ref{condmultestsec}} multiplier estimates for solutions
$\Psi^{[\pm2]}$ of the \emph{inhomogeneous} equation $(\ref{RWtypeinthebulk})$ outside
the region $r\in [A_1,A_2]$. Here, we use
the good divergence structure of the generalised
Regge--Wheeler operator. We then estimate in {\bf Section~\ref{condtranestsec}}
the quantities
 $\uppsi^{[\pm2]}$ and $\upalpha^{[\pm2]}$ via transport estimates. 
Taken together, these should 
 be viewed as providing a conditional
estimate stating that an integrated energy
expression for $\Psi^{[\pm2]}$, $\uppsi^{[\pm2]}$ and 
$\upalpha^{[\pm2]}$ can be controlled from initial data
provided that boundary terms on $r=A_i$ can be controlled.
(To understand the latter boundary term, this estimate must
be combined with that obtained in Section~\ref{ODEmegasec}.)

Finally, we shall need some auxiliary physical space estimates  (applied throughout
$\mathcal{R}$) for $\Psi^{[\pm2]}$, $\uppsi^{[\pm2]}$ and $\upalpha^{[\pm2]}$ arising
from a solution of the \emph{homogeneous} Teukolsky equation
$(\ref{Teukphysic})$. These
will be given in {\bf Section~\ref{auxil.est.sec}}.

{\bf \emph{Let us note that we may always assume
in what follows that any $\widetilde{\upalpha}^{[\pm2]}$
referred to is in $\mathscr{S}^{[\pm2]}_{\infty}(\widetilde{\mathcal{R}}_0)$.}}

\subsection{Multiplier estimates for $\Psi^{[\pm2]}$}
\label{condmultestsec}

We will apply multiplier estimates for $\Psi^{[\pm2]}$. The main result is

\begin{proposition}
\label{multpropnofreq}
Let $\upalpha^{[\pm2]}$ be as in Proposition~\ref{follfundprop},
and $\uppsi^{[\pm2]}$, $\Psi^{[\pm2]}$ be as defined in (\ref{Pdefsbulkplus2}), (\ref{Pdefsbulkmunus2}), (\ref{officdeflittlepsiplus}), (\ref{officdeflittlepsiminus}).
Let $\delta_1<1$, $\delta_2<1$ and $E>1$ be parameters and 
let $f_0$ be defined by $(\ref{newdefshere1})$ and
$y_0$ be defined by $(\ref{newdefshere2})$.
Then for sufficiently small $\delta_1$ and $\delta_2$ and sufficiently large
$E$, it follows that for sufficiently small $|a|<a_0\ll M$, 
then for any $0\le \tau_1\le \tau_2$, we have
\begin{align*}
&\mathbb{E}^{\rm away}_{\widetilde\Sigma_{\tau}, \eta}\left[\Psi^{[\pm2]}\right](\tau_2)
 +
\mathbb{I}^{\rm away}_{\eta}\left[\Psi^{[\pm2]}\right]
(\tau_1,\tau_2) \\
&\lesssim_{\delta_1,\delta_2, E} \mathbb{E}^{\rm away}_{\widetilde\Sigma_{\tau}, \eta}\left[\Psi^{[\pm2]}\right](\tau_1)
\\
&\qquad +\mathfrak{H}^{\rm away}[\Psi^{[\pm2]}](\tau_1,\tau_2) +\mathbb{Q}_{r=A_2}\left[\Psi^{[+2]}\right](\tau_1,\tau_2) -
 \mathbb{Q}_{r=A_1}\left[\Psi^{[+2]}\right](\tau_1,\tau_2)\\
&\qquad +|a|\, \mathbb{I}^{\rm left}_{[\eta]}\left[\uppsi^{[\pm2]}\right](\tau_1,\tau_2)
+|a|\, \mathbb{I}^{\rm right}_{[\eta]}\left[\uppsi^{[\pm2]}\right](\tau_1,\tau_2)\\
&\qquad +|a|\, \mathbb{I}^{\rm left}_{[\eta]}\left[\upalpha^{[\pm2]}\right](\tau_1,\tau_2)
+|a|\, \mathbb{I}^{\rm right}_{[\eta]}\left[\upalpha^{[\pm2]}\right](\tau_1,\tau_2).
\end{align*}
where   $\mathbb{Q}_{r=A_i}[\Psi^{[\pm2]}](\tau_1,\tau_2)$
is defined by $(\ref{blackboardQdef})$ and $\mathfrak{H}^{\rm away}[\Psi^{[\pm2]}](\tau_1,\tau_2)$
is defined by $(\ref{thisinhomodef})$. Moreover the subindex $\left[\eta\right]$ on the right hand side is equal to $\eta$ in case of $s=+2$ and it is dropped entirely in case $s=-2$.
\end{proposition}

We note already that the boundary terms $\mathbb{Q}_{r=A_2}\left[\Psi^{[+2]}\right](\tau_1,\tau_2) -
 \mathbb{Q}_{r=A_1}\left[\Psi^{[+2]}\right](\tau_1,\tau_2)$ appearing above
formally coincide with those of the fixed
frequency identity  to be obtained in  Section~\ref{multestsec}.
Thus these terms will cancel when all identities are summed in 
Section~\ref{iledsec}.

{\bf 
\emph{In what follows, our multiplier constructions will be identical for $\Psi^{[+2]}$ and $\Psi^{[-2]}$.
We will thus denote these  simply as $\Psi$. 
The spin weight will be explicitly denoted however for the terms arising from the right
hand side of $(\ref{RWtypeinthebulk})$.}}

\subsubsection{Multiplier identities} \label{sec:multid}
The proof of Proposition~\ref{multpropnofreq} will rely on various multiplier
identities for $(\ref{RWtypeinthebulk})$. These are analogous for standard
multiplier estimates proven for solutions of the  scalar wave equation
and in particular 
generalise  specific estimates which have been proven for
the Regge--Wheeler equation $(\ref{fromprevpaper})$ on 
Schwarzschild in~\cite{holzstabofschw}.

\paragraph{The $T+\upomega_+ \chi \Phi$ identity.}

Multiplying (\ref{RWtypeinthebulk}) by $\left(T+\upomega_+ \chi \Phi\right)\overline{\Psi}$ (recall $\chi$ was fixed in Section \ref{very.slowly.very.slowly}) and taking the real part leads to (use the formulae of Appendix \ref{sec:compT} and \ref{sec:compphi} and (\ref{convert}))
\begin{align}
\label{Hawkinident}
 \left(L+\underline{L}\right) \Big\{F^{T+\upomega_+\chi \Phi}_{L+\underline{L}}\Big\} +\left(L-\underline{L}\right) &\Big\{F^{T+\upomega_+\chi \Phi}_{L-\underline{L}}  \Big\}  + I^{T+\upomega_+\chi \Phi} \equiv {\rm Re}\left(\left(-\left(T+\upomega_+ \chi \Phi\right)\overline{\Psi}\right)\left(\mathcal{J}^{[s]} + \mathfrak{G}^{[s]}\right) \right)
\end{align}
where $\equiv$ denotes equality after integration with respect to the measure $\sin \theta d\theta d\phi$ and
\begin{align} \label{tidi}
F^{T+\omega_+\chi \Phi}_{L+\underline{L}} &= \frac{1}{16}\Big\{   |\left(L+\underline{L}\right)\Psi|^2+ |\left(L-\underline{L}\right)\Psi|^2 +4w |\mathring{\slashed{\nabla}}^{[s]} \Psi |^2
 + 4w \left[s^2 -\frac{6Mr}{r^2+a^2} \frac{r^2-a^2}{r^2+a^2} -\frac{7a^2\Delta}{(r^2+a^2)^2}  \right] |\Psi|^2
 \nonumber \\
&\qquad \qquad    -4a^2 w \sin^2 \theta|T\Psi|^2  - 8wa \upomega_+\chi |\Phi \Psi|^2 -8w a^2 \sin^2\theta \upomega_+\chi {\rm Re}\left((T\Psi)  \Phi \overline\Psi\right) \nonumber \\
& \qquad \qquad +4 \upomega_+ \left( \chi - \frac{r_+^2+a^2}{r^2+a^2}\right)\textrm{Re}\left(\Phi\Psi \left(L+\underline{L}\right) \overline\Psi \right)- 8swa \upomega_+\chi \cos \theta  \textrm{Im} \left(\Psi \Phi \overline{\Psi}\right)\Big\},
\nonumber  \\
 F^{T+\upomega_+\chi \Phi}_{L-\underline{L}} &= \frac{1}{8}\left(-2 {\rm Re}\left( \left(L - \underline{L}\right) \Psi T \overline\Psi\right) - 2\upomega_+\chi  \textrm{Re} \left(\left(L - \underline{L}\right) \Psi (\Phi \overline\Psi) \right) \right),
 \\
 I^{T+\upomega_+\chi \Phi} &= \frac{1}{2}\upomega_+ \chi^\prime  {\rm Re}\left( \left((L-\underline{L})\Psi\right)(\Phi \overline\Psi)\right). \nonumber
\end{align}

\paragraph{The $y$ identity.}
Multiplying (\ref{RWtypeinthebulk}) by $ y\left(L-\underline{L}\right)\overline{\Psi}$ for a smooth radial function $y$ and taking the real part produces (use the formulae of Appendix \ref{sec:compy})
\begin{equation}
\label{yidentform}
\left(L+\underline{L}\right) \Big\{F^{y}_{L+\underline{L}}\Big\} +\left(L-\underline{L}\right) \Big\{F^{y}_{L-\underline{L}}  \Big\} + I^y \equiv {\rm Re}\left(\left(\mathcal{J}^{[s]} + \mathfrak{G}^{[s]}\right)  \left(- y\left(L-\underline{L}\right) \overline{\Psi}\right)\right)
\end{equation}
where $\equiv$ denotes equality after integration with respect to the measure $\sin \theta d\theta d\phi$ and
\begin{align} \label{yidi}
F^y_{L+\underline{L}} &= \frac{1}{4}  {\rm Re} \Big\{  y \left(L + \underline{L}\right) \Psi \left( (L-\underline{L})\overline{\Psi}\right) + 2way \Phi \Psi \underline{L} \overline{\Psi} -2way \Phi \Psi L \overline{\Psi} \nonumber \\
& \ \ \ \ \ \ \ \  \ \  \ -2w a^2 \sin^2\theta T\Psi \left(y (L-\underline{L})\overline{\Psi}\right)-4sa\cos\theta wy \textrm{Im} \left( \Psi \overline{\Psi}^\prime \right)\Big\},
\nonumber \\
F^y_{L-\underline{L}} &=\frac{1}{4} \Big\{ -\frac{y}{2}   \left| (L+\underline{L})\Psi\right|^2 -\frac{y}{2}   \left| (L-\underline{L})\Psi\right|^2 + 2w y |\mathring{\slashed{\nabla}} \Psi |^2 +2w y\left[s^2 -\frac{6Mr}{r^2+a^2} \frac{r^2-a^2}{r^2+a^2} - \frac{7a^2\Delta}{(r^2+a^2)^2}  \right] |\Psi|^2\nonumber \\
  & \qquad  \ \ \  + 4way {\rm Re} \left( \Phi \Psi T \overline{\Psi} \right)+2w y a^2 \sin^2\theta |T\Psi|^2  +4sa \cos \theta wy \textrm{Im} \left(\Psi T \overline{\Psi}\right)\Big\},
\nonumber \\
I^y &= \frac{y^\prime}{4} \left[ \left( (L+\underline{L})\Psi\right)^2 + \left( (L-\underline{L})\Psi\right)^2 \right]  - \left(wy\left[s^2 -\frac{6Mr}{r^2+a^2} \frac{r^2-a^2}{r^2+a^2} - \frac{7a^2\Delta}{(r^2+a^2)^2}  \right] \right)^\prime |\Psi|^2     - \left(w y \right)^\prime |\mathring{\slashed{\nabla}} \Psi |^2   \nonumber \\
& \ \ +2\left( \frac{a^2}{r^2+a^2}  w y \right)^\prime |\Phi\Psi|^2 + \frac{4ra^2}{(r^2+a^2)^2} \frac{\Delta}{r^2+a^2} w y |\Phi \Psi|^2 \nonumber \\
& \ \  - a \left[ L \left(wy\right) \right] {\rm Re} \left((\Phi\Psi)(\underline{L} \overline{\Psi}) \right)+ a \left[ \underline{L} \left(wy\right) \right] {\rm Re} \left((\Phi\Psi)({L}\overline{\Psi})\right) - (wy)^\prime a^2 \sin^2\theta | T\Psi|^2 \nonumber \\
& \ \ -\frac{1}{2}y \frac{r a}{\left(r^2+a^2\right)^2} \frac{\Delta}{r^2+a^2} {\rm Re} \left(\Phi \Psi \left(L+\underline{L}\right)\overline{\Psi} \right)  - 2sa\cos \theta (w y)^\prime \textrm{Im} \left(\Psi T\overline{\Psi}\right).
\end{align}

\paragraph{The $h$ identity.}
Multiplying (\ref{RWtypeinthebulk}) by $h\overline{\Psi}$ for a smooth radial function $h$ and taking real parts leads to (use the formulae of Appendix \ref{sec:compL})
\begin{align}
\left(L+\underline{L}\right) \Big\{F^{h}_{L+\underline{L}}\Big\} +\left(L-\underline{L}\right) \Big\{F^{h}_{L-\underline{L}}  \Big\} + I^h \equiv {\rm Re} \left(-\left(\mathcal{J}^{[s]} + \mathfrak{G}^{[s]}\right) h \overline{\Psi}\right)
\end{align}
where $\equiv$ denotes equality after integration with respect to the measure $\sin \theta d\theta d\phi$ and
\begin{align}
F^h_{L+\underline{L}} &= \frac{1}{4} {\rm Re} \Big\{   \left(L + \underline{L}\right) \Psi h \overline{\Psi}  -2w a^2 \sin^2\theta T\Psi h\overline{\Psi}\Big\} \nonumber \\
F^h_{L-\underline{L}} &=
\frac{1}{4} \textrm{Re} \Big\{ -\left(L-\underline{L}\right)  \Psi h \overline{\Psi} + h^\prime |\Psi|^2   \Big\} 
\nonumber \\
I^h&= \frac{h}{4} \Big[  |(L-\underline{L})\Psi|^2 -  |(L+\underline{L})\Psi|^2 +4  w  |\mathring{\slashed{\nabla}}^{[s]} \Psi |^2 \Big] + \Big[wh \left(s^2 -\frac{6Mr}{r^2+a^2} \frac{r^2-a^2}{r^2+a^2}  - \frac{7a^2\Delta}{(r^2+a^2)^2} \right)- \frac{h^{\prime \prime}}{2}\Big] |\Psi|^2  \nonumber \\
& \ \ \  \  + 2w a h {\rm Re} \left((T  \Psi)  (\Phi \overline{\Psi}
 )\right) +  w a^2 \sin^2\theta h |T\Psi|^2 -2sa\cos \theta h w \textrm{Im} \left(T\Psi \overline{\Psi}\right) .
\end{align}

\paragraph{The $f$ identity.}
Adding the $y$-identity with $y=f$ 
and the $h$-identity with $h=f^\prime$ for $f$ a smooth radial function yields the identity (recall (\ref{convert}))
\begin{align}
\label{fidentityhere}
\left(L+\underline{L}\right) \Big\{F^{f}_{L+\underline{L}}\Big\} +\left(L-\underline{L}\right) \Big\{F^{f}_{L-\underline{L}}  \Big\}  + I^{f} \equiv {\rm Re} \left(-\left(\mathcal{J}^{[s]} + \mathfrak{G}^{[s]}\right) \left(f^\prime \overline{\Psi} + 2f\overline{\Psi}^\prime\right) \right)
\end{align}
where $\equiv$ denotes equality after integration with respect to the measure $\sin \theta d\theta d\phi$ and
\begin{align} \label{Xf}
F^{f}_{L+\underline{L}} &= \frac{1}{4} {\rm Re} \Big\{  f \left(L + \underline{L}\right) \Psi \left( (L-\underline{L})\overline{\Psi}\right) + 2waf \Phi \Psi \underline{L} \overline{\Psi} -2waf \Phi \Psi L \overline{\Psi} -2w a^2 \sin^2\theta T\Psi \left(f (L-\underline{L})\overline{\Psi}\right) \nonumber \\
\ \ \ \ & \qquad \ \ \ \ \  \  \left(L + \underline{L}\right) \Psi f^\prime \overline{\Psi}  -2w a^2 \sin^2\theta T\Psi f^\prime\overline{\Psi}-4sa\cos\theta wf \textrm{Im} \left( \Psi \overline{\Psi}^\prime \right) \Big\}, 
\nonumber \\
F^{f}_{L-\underline{L}} &=  \frac{1}{4} \Big\{ -\frac{f}{2}   \left| (L+\underline{L})\Psi\right|^2 -\frac{f}{2}   \left| (L-\underline{L})\Psi\right|^2 + 2w f |\mathring{\slashed{\nabla}}^{[s]} \Psi |^2 +2w f\left[s^2 -\frac{6Mr}{r^2+a^2} \frac{r^2-a^2}{r^2+a^2} - \frac{7a^2\Delta}{(r^2+a^2)^2}  \right] |\Psi|^2\nonumber \\
  & \qquad  \  + 4waf {\rm Re} \left( \Phi \Psi T \overline{\Psi} \right)+2w f a^2 \sin^2\theta |T\Psi|^2 
  -f^\prime \textrm{Re} \left(\left(L-\underline{L}\right)  \Psi  \overline{\Psi}\right) + f^{\prime \prime} |\Psi|^2 + 4sa \cos \theta w f \textrm{Im} \left(\Psi T\overline{\Psi}\right)  \Big\},
\nonumber \\
I^{f} &= \frac{f^\prime}{4} \left[  + 2\left( (L-\underline{L})\Psi\right)^2 \right] - w^\prime f |\mathring{\slashed{\nabla}}^{[s]} \Psi |^2 + \left[- f \left(w\left[s^2 -\frac{6Mr}{r^2+a^2} \frac{r^2-a^2}{r^2+a^2} - \frac{7a^2\Delta}{(r^2+a^2)^2}  \right] \right)^\prime - \frac{f^{\prime \prime \prime}}{2} \right] |\Psi|^2        \nonumber \\
& \ \ +2\left( \frac{a^2}{r^2+a^2}  w f \right)^\prime |\Phi\Psi|^2 + \frac{4ra^2}{(r^2+a^2)^2} \frac{\Delta}{r^2+a^2} w f |\Phi \Psi|^2 \nonumber \\
& \ \  - a \left[ L \left(wf\right) \right] {\rm Re} \left((\Phi\Psi)(\underline{L} \overline{\Psi}) \right)+ a \left[ \underline{L} \left(wf\right) \right] {\rm Re} \left((\Phi\Psi)({L}\overline{\Psi})\right) - (fw)^\prime a^2 \sin^2\theta | T\Psi|^2 \nonumber \\
& \ \ -\frac{1}{2}f \frac{r a}{\left(r^2+a^2\right)^2} \frac{\Delta}{r^2+a^2} {\rm Re} \left(\Phi \Psi \left(L+\underline{L}\right)\overline{\Psi} \right)  - 2sa\cos \theta f w^\prime \textrm{Im} \left(\Psi T\overline{\Psi}\right) \nonumber \\
& \ \ + 2w a f^\prime {\rm Re} \left((T  \Psi)  (\Phi \overline{\Psi}
 )\right) +  w a^2 \sin^2\theta f^\prime |T\Psi|^2.
\end{align}

\paragraph{The $r^p$-weighted identity.}
We multiply (\ref{RWtypeinthebulk}) by $r^p \beta_4 \xi L\overline{\Psi}$ with $\beta_4=1+4\frac{M}{r}$ and $\xi$ a smooth radial cut-off satisfying $\xi=0$ for $r \leq R$ and $\xi=1$ for $r \geq R+M$ with $R$ is chosen directly below (\ref{ipb}) depending only on $M$.  After taking the real parts of the resulting identity we obtain (use the formulae of Appendix \ref{sec:comprp})
\begin{align} \label{rpphysspaceid}
\underline{L} \big\{F^{r^p}_{\underline{L}}\big\} +L \big\{F^{r^p}_{L}  \big\}+ I^{r^p} \equiv {\rm Re} \left(-\left(\mathcal{J}^{[s]} + \mathfrak{G}^{[s]}\right) r^p \beta_4 \xi L\overline{\Psi}\right)
\end{align}
where $\equiv$ denotes equality after integration with respect to the measure $\sin \theta d\theta d\phi$ and
\begin{align}
F^{r^p}_{\underline{L}} 
&=   \frac{1}{2}\xi r^p \beta_4 |{L}\Psi|^2 + \frac{1}{2}\textrm{Re} \left(a \xi w r^p \beta_4 \Phi \Psi {L} \overline{\Psi}\right) + \frac{1}{2}w a^2 \sin^2\theta r^p \beta_4 \xi \textrm{Re} \left( T\Psi {L}\overline{\Psi}\right) ,
 \\
 F^{r^p}_{L}  &=
\frac{1}{2}w \xi r^p \beta_4 |\mathring{\slashed{\nabla}}^{[s]} \Psi |^2 + \frac{1}{2} \left(w r^p \beta_4 \xi \left[s^2 -\frac{6Mr}{r^2+a^2} \frac{r^2-a^2}{r^2+a^2} - \frac{7a^2\Delta}{(r^2+a^2)^2}  \right] \right) |\Psi|^2 + \frac{a^2 w r^p}{r^2+a^2}  \xi  \beta_4 |\Phi\Psi|^2 \nonumber \\
& \ \ \ \ \ \ \ \ \ \  -\frac{1}{2} \textrm{Re} \left(a \xi w r^p \beta_4 \Phi \Psi \underline{L}\overline{\Psi}\right) - \frac{1}{2} w a^2 \sin^2\theta \xi  r^p\beta_4  |T\Psi|^2  + \frac{1}{2}w a^2 \sin^2\theta r^p \beta_4 \xi \textrm{Re} \left( T\Psi {L}\overline{\Psi}\right),
\\
 \label{ipb}
I^{r^p} &=+\frac{1}{2} \left( \xi \left(pr^{p-1} + \mathcal{O}\left(r^{p-2}\right)\right) + \xi^\prime r^p \beta_4 \right) |{L}\Psi|^2 \\
& \ \ \ +\frac{1}{2}\left(\xi \left[\frac{\left(2-p\right)}{r^{3-p}} + \frac{(3-p)2M}{r^{4-p}} + \mathcal{O}\left(r^{p-5}\right) \right]+  \frac{\xi^\prime w}{r^{-p}}\right) |\mathring{\slashed{\nabla}}^{[s]}  \Psi |^2  \nonumber \\
& \ \ \ - \frac{1}{2} \left(w r^p \beta_4 \xi \left[s^2 -\frac{6Mr}{r^2+a^2} \frac{r^2-a^2}{r^2+a^2} - \frac{7a^2\Delta}{(r^2+a^2)^2}  \right] \right)^\prime |\Psi|^2 \nonumber \\
& \ \ \ - \frac{2r a \xi r^p \beta_4w}{r^2+a^2}  \textrm{Re} \left( \Phi \Psi {L}\overline{\Psi}\right)-  \left( \frac{a^2}{r^2+a^2}  \xi w r^p \beta_4 \right)^\prime |\Phi\Psi|^2 + \frac{1}{2} \textrm{Re} \left( \left(a \xi w r^p \beta_4\right)^\prime \Phi \Psi \left( {L} + \underline{L}\right) \overline{\Psi}  \right) \nonumber \\ 
&\ \ \ + \frac{1}{2} \textrm{Re} \left(a\xi w r^p \beta_4 \Phi \Psi \left[L, \underline{L} \right] \overline{\Psi}\right) +\frac{1}{2} \left(\xi \left(w r^p \beta_4\right)^\prime + \xi^\prime w r^p \beta_4 \right) a^2 \sin^2\theta | T\Psi|^2 -  2sa \cos \theta w r^p \beta_4 \xi \textrm{Im} \left(T\Psi L \overline{\Psi}\right). \nonumber
\end{align}
It is easy to see that we can choose $R$ in the cut-off function such that the coefficients of $ |{L}\Psi|^2$, $|\mathring{\slashed{\nabla}}^{[s]}  \Psi |^2$ and $|\Psi|^2$ in (\ref{ipb}) are all non-negative in $r\geq R+M$ for $p \in \left[0,2\right]$ and we henceforth make that choice.

\begin{remark}{(Conversion into divergence identities)} \label{rem:convert}
To convert the identities derived in this section into proper spacetime divergence identities (from which the boundary contributions, etc.,~are most easily assessed) we recall the identities (\ref{siid}).
Since the left hand side of any multiplier identity above has the schematic form
\[
L \big\{ F_L\big\} + \underline{L} \big\{ F_{\underline{L}} \big\} 
+ I + \mathcal{E} = RHS \, 
\]
with $\int \mathcal{E} \sin \theta d\theta d\phi=0$, we can use (\ref{siid}) to convert them into the divergence form 
\[
\nabla_a \left(L^a f \frac{1}{\rho^2} \frac{r^2+a^2}{\Delta} F_L + \underline{L}^a \frac{1}{\rho^2} \frac{r^2+a^2}{\Delta} F_{\underline{L}} \right) + I \frac{1}{\rho^2} \frac{r^2+a^2}{\Delta}+ \mathcal{E} \frac{1}{\rho^2} \frac{r^2+a^2}{\Delta} = RHS \frac{1}{\rho^2} \frac{r^2+a^2}{\Delta} \, .
\]
This is easily integrated using Stokes' theorem and making use of the formulae (\ref{innp}) and (\ref{innp2}) for the normals to the spacelike hypersurfaces (and the horizon and null infinity). Therefore it is the above identity which provides the precise sense in which the $F$'s in the identities indeed 
correspond to
boundary terms. Note the term involving $\mathcal{E}$  disappears after integration with
respect to the spacetime volume form (\ref{volform}).
\end{remark}

\subsubsection{Proof of Proposition~\ref{multpropnofreq}}

We define (cf.~\cite{holzstabofschw})
\begin{equation}
\label{newdefshere1}
f_0 = \left(1-\frac{3M}r\right)\left(1+\frac{M}r\right),
\end{equation}
and
\begin{equation}
\label{newdefshere2}
y_0=\delta_1 (f_0- \delta_1 \tilde\chi (r) r^{-\eta})
\end{equation}
where $\tilde\chi$ is a cutoff function such that $\tilde\chi=0$
for $r\le 9M$ and $\tilde\chi=1$ for $r\ge 10M$.  
We note the following
Schwarzschild proposition
\begin{proposition}[\cite{holzstabofschw}]
\label{coerc.in.schw}
In the Schwarzschild case $a=0$, then
\[
r^{-2} |(L-\bar{L}) \Psi|^2+ (1-3M/r)^2r^{-3} |\mathring{\slashed{\nabla}}^{[s]} \Psi |^2   
+r^{-3} |\Psi|^2 \lesssim 
I^{f_0}.
\]
As a consequence, for $\delta_1$ and $\delta_2$ sufficiently small and arbitrary $E$
we have
\begin{align*}
&(r^{-2}+\delta_1^2 r^{-1-\eta}) |(L-\bar{L}) \Psi|^2 +(1-3M/r)^2\delta_1^2 r^{-1-\eta} |(L+\bar{L}) \Psi|^2 
+\delta_2 \xi r^{\eta-1}|L\psi|^2  
+ (1-3M/r)^2r^{-3} |\mathring{\slashed{\nabla}}^{[s]} \Psi |^2\\
&\qquad\qquad +\delta_2 r^{\eta-1}|\Psi |^2
\lesssim  I^{f_0}+I^{y_0}+I^E +\delta_2 I^{r^\eta}.
\end{align*}
Note that in view of Remark~\ref{rem:convert}, 
upon application of the divergence theorem,
the left hand side leads to a term
which controls the integrand of $\mathbb I^{\rm deg}_\eta$.
\end{proposition}

Returning  to the Kerr case, we add
\begin{enumerate}
\item
the $f$-identity $(\ref{fidentityhere})$ applied with $f=f_0$, 
\item
the $y$-identity $(\ref{yidentform})$ applied with
$y=y_0$, 
\item
$E$ times the $T+\upomega_+\chi\Phi$ identity $(\ref{Hawkinident})$ 
\item
$\delta_2$ times the $r^{\eta}$ identity $(\ref{rpphysspaceid})$
\end{enumerate}
integrated in the region 
\[
\widetilde{\mathcal{R}}^{\rm away}(\tau_1,\tau_2)=\widetilde{\mathcal{R}}(\tau_1,\tau_2)\setminus \{A_1\le r\le A_2\}
\]
with respect to the spacetime volume form, and
apply Remark~\ref{rem:convert}.
We always will assume $E>1$
and $\delta_1<1$, $\delta_2<1$.

We have:
\begin{enumerate}
\item
Given any $E>1$, and sufficiently small $\delta_1$, $\delta_2$, then 
for $|a|<a_0\ll M$ sufficiently small, 
the resulting bulk term is nonnegative and in fact satisfies the coercitivity estimate
\begin{equation}
\label{fundcoerci}
\int_{\widetilde{\mathcal{R}}^{\rm away}(\tau_1,\tau_2)} \left( I^{f} +I^{y}+ EI^{T+\upomega_+\chi\Phi}+\delta_2 I^{r^{\eta}}\right)
 \frac{1}{\rho^2} \frac{r^2+a^2}{\Delta} dVol
\gtrsim_{\delta_1,\delta_2} \mathbb{I}^{\rm away}_{\eta} \left[\Psi^{[\pm2]}\right](\tau_1,\tau_2).
\end{equation}
This follows from (a)
Proposition~\ref{coerc.in.schw}, (b) smooth dependence on $a$
to infer coercivity away form the horizon and away from infinity,
(c)
the fact that for all $a$, the term  $I^{r^\eta}$  manifestly controls
the integrand  of $\mathbb{I}^{\rm away}_\eta$ for large $r$,
(d)  the fact that by direct inspection, for sufficiently
small $|a|<a_0\ll M$, the term $I^f+I^y$ 
controls the  integrand of $\mathbb{I}^{\rm away}_{\eta}$
near the horizon.
\item
For sufficiently large $E>1$, then for all $\delta_1<1$, $\delta_2<1$ the total flux terms
on $\mathcal{H}^+$ and $\mathcal{I}^+$ are nonnegative.
This follows from Remark~\ref{rem:convert}
and direct inspection of the boundary terms $F$ thus generated, together with
the relations concerning the volume form given in Section~\ref{volumeformrelations}.
(To avoid appealing to the fact that the flux to $\mathcal{I}^+$ is well defined,
we may argue as follows: The identity can be applied in a region
bounded by  a finite ingoing null boundary,
making the region of integration compact. The flux term on this boundary
is manifestly nonnegative by the choice of the multipliers. One then takes
this null boundary to the limit.)
\item
Again, by Remark~\ref{rem:convert}, inspection and the relations of Section~\ref{volumeformrelations},
it follows that for sufficiently large $E>1$, then
for all $\delta_1<1$, $\delta_2<1$,
 the 
arising flux term on $\tilde{t}^*=\tau_2$ controls
the energy 
$\mathbb{E}^{\rm away}_{\widetilde{\Sigma}_\tau,\eta}\left[\Psi^{[\pm2]}\right](\tau_2)$
with a uniform constant. 
\item
Similarly, for sufficiently large $E>1$, then for all $\delta_1<1$, $\delta_2<1$, the
 initial flux term on $\tilde{t}^*=\tau_1$ is controlled
by the energy 
 $\mathbb{E}^{\rm away}_{\widetilde{\Sigma}_\tau,\eta}\left[\Psi^{[\pm2]}\right](\tau_1)$,
with a constant depending on $E$.
 \item
 The remaining flux terms on $r=A_1$ and $r=A_2$ produce exactly the expression
 \[
\mathbb{Q}_{r=A_2}\left[\Psi^{[\pm2]}\right](\tau_1,\tau_2) -
 \mathbb{Q}_{r=A_1}\left[\Psi^{[\pm2]}\right](\tau_1,\tau_2)
\]
where (recalling (\ref{tidi}), (\ref{yidi}) and (\ref{Xf}))
\begin{equation}
\label{blackboardQdef}
\mathbb{Q}_{r=A_i}(\tau_1,\tau_2)
= \int_0^{\tau_2} dt \int_0^\pi d\theta \sin \theta \int_0^{2\pi} d\phi \Big\{ 2F^{f_0}_{L-\underline{L}} +2F^{y_0}_{L-\underline{L}} + 2F^{T+\omega_+\chi\Phi}_{L-\underline{L}} \Big\}.
\end{equation}
This again follows from Remark~\ref{rem:convert}: In $\left(t,r^*,\theta,\phi\right)$-coordinates we have that $\frac{1}{\sqrt{g_{r^* r^*}}} \partial_{r^*}$ is the unit normal to constant $r^*$ hypersurfaces and $\rho^2 \frac{1}{\sqrt{g_{r^* r^*}}} \frac{\Delta}{r^2+a^2} \sin \theta d\theta d\phi dt$ is the induced volume element. Using that $2\partial_{r^*} = L - \underline{L}$ and that $\partial_{r^*}$ is orthogonal to $L+\underline{L}$ the result follows. Observe that there is no contribution from $F^{r^\eta}$ in (\ref{blackboardQdef}) because that multiplier is supported away from $A_2$.
\item
The inhomogeneous term involving $\mathfrak{G}^{[\pm2]}$ 
on the right hand side of $(\ref{RWtypeinthebulk})$ generates the term
\begin{equation}
\label{thisinhomodef}
\mathfrak{H}^{\rm away}[\Psi^{[\pm2]}](\tau_1,\tau_2)=
\int_{\widetilde{\mathcal{R}}^{\rm away}(\tau_1,\tau_2)}
\mathfrak{G}^{[\pm 2]}\cdot (f,y,E, \delta_1,\delta_2) \, dVol 
\end{equation}
where (recall again Remark \ref{rem:convert})
\begin{align}
\mathfrak{G}^{[\pm 2]}\cdot (f,y,E, \delta_1,\delta_2) \doteq \frac{r^2+a^2}{\rho^2\Delta} \Big\{ E\cdot  {\rm Re}\left(\left(-\left(T+\omega_+ \chi \Phi\right)\overline{\Psi}\right) \mathfrak{G}^{[\pm2]} \right)  + {\rm Re} \left(-  \left(f_0^\prime \overline{\Psi} + 2f_0\overline{\Psi}^\prime\right)\mathfrak{G}^{[\pm2]} \right) \nonumber \\
+\delta_1{\rm Re}\left( \left(- 2f_0 \overline{\Psi}^\prime\right)\mathfrak{G}^{[\pm2]} \right)+\delta_2\cdot {\rm Re} \left(- \left(r^\eta \beta_k \xi L\overline{\Psi}\right) \mathfrak{G}^{[\pm2]}\right)
\Big\} . \nonumber
\end{align}
 \item
By Cauchy--Schwarz,
the term generated by the inhomogeneous term involving
$\mathcal{J}^{[\pm2]}$ on the right hand side of $(\ref{RWtypeinthebulk})$
can be bounded (with a constant depending on $E$) by the expression
\begin{align*}
|a|\, \mathbb{I}^{\rm left}_{[\eta]}\left[\uppsi^{[\pm2]}\right](\tau_1,\tau_2)
+|a|\, \mathbb{I}^{\rm right}_{[\eta]}\left[\uppsi^{[\pm2]}\right](\tau_1,\tau_2)\\
+|a|\, \mathbb{I}^{\rm left}_{[\eta]}\left[\upalpha^{[\pm2]}\right](\tau_1,\tau_2)
+|a|\, \mathbb{I}^{\rm right}_{[\eta]}\left[\upalpha^{[\pm2]}\right](\tau_1,\tau_2)
+
|a| \mathbb{I}^{\rm away}_{\eta}\left[\Psi^{[\pm2]}\right](\tau_1,\tau_2) \, ,
\end{align*}
with the subindex $[\eta]=\eta$ in case of $+2$, and $[\eta]$ being dropped entirely in case of $s=-2$. Note that the last term can be absorbed in view of $(\ref{fundcoerci})$, for sufficiently
small $|a|<a_0\ll M$.
 \end{enumerate}
 Thus, for $E$ sufficiently large, and $\delta_1$, $\delta_2$ sufficiently small,
 one obtains immediately
the statement of Proposition~\ref{multpropnofreq}.

 {\bf In what follows, we will now consider $E$ as fixed in terms of $M$, and thus
 incorporate the $E$ dependence into the $\lesssim$, etc.}
  We will further constrain $\delta_1$ and $\delta_2$ 
in Section~\ref{multestsec} and thus we will continue to denote explicitly
dependence of constants on $\delta_1$, $\delta_2$.

\subsection{Transport estimates for $\uppsi^{[\pm2]}$ and $\alpha^{[\pm2]}$}
\label{condtranestsec}

For transport estimates, it is natural to consider the spin $\pm2$ cases separately.

\subsubsection{Transport estimates for $\uppsi^{[+2]}$ and $\alpha^{[+2]}$}

\begin{proposition} \label{transportplusnof}
Let $\upalpha^{[+2]}$ be as in Proposition~\ref{follfundprop},
and $\uppsi^{[+2]}$, $\Psi^{[+2]}$ be as defined in (\ref{Pdefsbulkplus2}),
(\ref{officdeflittlepsiplus}).
Then we have for any $p \in \{\eta,1,2\}$ the following estimate
in $\widetilde{\mathcal{R}}^{\rm right}(\tau_1,\tau_2)$: 
\begin{align} \label{weightedtransportright}
 &  \ \ \ \mathbb{E}^{\rm right}_{\widetilde\Sigma_{\tau},p} \left[\upalpha^{[+2]}\right](\tau_2)
+\mathbb{I}_p^{\rm right} \left[\upalpha^{[+2]}\right]  \left(\tau_1,\tau_2 \right)
+ \mathbb{E}_{r=A_2} \left[\upalpha^{[+2]}\right] \left(\tau_1,\tau_2\right) \nonumber \\
&+\mathbb{E}^{\rm right}_{\widetilde\Sigma_{\tau},p} \left[\uppsi^{[+2]}\right](\tau_2) 
+\mathbb{I}_p^{\rm right} \left[\uppsi^{[+2]}\right]  \left(\tau_1,\tau_2 \right) 
+\mathbb{E}_{r=A_2} \left[\uppsi^{[+2]}\right] \left(\tau_1,\tau_2\right) \nonumber \\
&\qquad \lesssim
 \mathbb{I}^{\rm away}_{p}\left[\Psi^{[+2]}\right](\tau_1, \tau_2)
 +
\mathbb{E}^{\rm right}_{\widetilde\Sigma_{\tau},p} \left[\upalpha^{[+2]}\right](\tau_1)
+\mathbb{E}^{\rm right}_{\widetilde\Sigma_{\tau},p} \left[\uppsi^{[+2]}\right](\tau_1)
\end{align}
and the following estimate in $\widetilde{\mathcal{R}}^{\rm left}(\tau_1,\tau_2)$:
\begin{align}\label{weightedtransport.left}
\nonumber
 &  \ \ \ \mathbb{E}^{\rm left}_{\widetilde\Sigma_{\tau},p} \left[\upalpha^{[+2]}\right](\tau_2)
+\mathbb{I}_p^{\rm left} \left[\upalpha^{[+2]}\right]  \left(\tau_1,\tau_2 \right)
+ \mathbb{E}_{\mathcal{H}^+} \left[\upalpha^{[+2]}\right] \left(\tau_1,\tau_2\right) \\
\nonumber
&+\mathbb{E}^{\rm left}_{\widetilde\Sigma_{\tau},p} \left[\uppsi^{[+2]}\right](\tau_2) 
+\mathbb{I}_p^{\rm left} \left[\uppsi^{[+2]}\right]  \left(\tau_1,\tau_2 \right) 
+\mathbb{E}_{\mathcal{H}^+} \left[\uppsi^{[+2]}\right] \left(\tau_1,\tau_2\right) \\
\nonumber
&\qquad \lesssim
 \mathbb{I}^{\rm away}_{ p}\left[\Psi^{[+2]}\right](\tau_1, \tau_2)
 +
\mathbb{E}^{\rm left}_{\widetilde\Sigma_{\tau},p} \left[\upalpha^{[+2]}\right](\tau_1)
+\mathbb{E}^{\rm left}_{\widetilde\Sigma_{\tau},p} \left[\uppsi^{[+2]}\right](\tau_1) \nonumber \\
&\qquad \ \ +  \mathbb{E}_{r=A_1} \left[\upalpha^{[+2]}\right] \left(\tau_1,\tau_2\right) +  \mathbb{E}_{r=A_1} \left[\uppsi^{[+2]}\right] \left(\tau_1,\tau_2\right).
\end{align}
\end{proposition}

\begin{proof}
We recall the relations
\begin{align}
-2 \frac{\Delta}{(r^2+a^2)^2} \sqrt{\Delta} \uppsi^{[+2]} &= \underline{L}^a \nabla_a \left(\Delta^2 \left(r^2+a^2\right)^{-\frac{3}{2}} \upalpha^{[+2]}\right) \label{psirel} \, ,
 \\
+\frac{\Delta}{(r^2+a^2)^2} \left(r^2+a^2\right)^{3/2} P^{[+2]} &= \underline{L}^a \nabla_a \left( \sqrt{\Delta} \uppsi^{[+2]}\right) \, .\label{Prel}
\end{align}
From (\ref{psirel}) we derive for $n\geq 0$
\begin{align}
\nabla_a \left(r^n \frac{1}{\rho^2}\frac{r^2+a^2}{\Delta}\underline{L}^a \Big|\upalpha^{[+2]} \Delta^2 \left(r^2+a^2\right)^{-\frac{3}{2}}\Big|^2\right) +n \frac{ r^{n-1}}{\rho^2} \Big|\upalpha^{[+2]} \Delta^2 \left(r^2+a^2\right)^{-\frac{3}{2}}\Big|^2 \nonumber \\
= -2 \frac{\left(r^2+a^2\right)^2}{\Delta \rho^2} w^\frac{3}{2} r^n  \left(\uppsi^{[+2]} \cdot \overline{ \upalpha^{[+2]} \Delta^2 \left(r^2+a^2\right)^{-\frac{3}{2}}} + \overline{\uppsi^{[+2]}} \cdot  \upalpha^{[+2]} \Delta^2 \left(r^2+a^2\right)^{-\frac{3}{2}}\right) \, ,
\end{align}
and hence
\begin{align} \label{idp1}
\nabla_a \left(r^n \frac{1}{\rho^2}\frac{r^2+a^2}{\Delta}\underline{L}^a \Big| \frac{\upalpha^{[+2]} \Delta^2}{ \left(r^2+a^2\right)^{\frac{3}{2}}}\Big|^2\right) +\frac{n}{2} \frac{ r^{n-1}}{\rho^2} \Big| \frac{\upalpha^{[+2]} \Delta^2}{ \left(r^2+a^2\right)^{\frac{3}{2}}}\Big|^2\leq C \frac{1}{\rho^2} \frac{r^{n+1}}{(r^2+a^2)^2} | \sqrt{\Delta}\uppsi^{[+2]}|^2 \, .
\end{align}
Moreover, the same estimate 
(\ref{idp1})
holds replacing $\upalpha^{[+2]}$ by $T \upalpha^{[+2]}$ ($\Phi \upalpha^{[+2]}$) on the left and $\uppsi^{[+2]}$ by $T \uppsi^{[+2]}$ ($\Phi \uppsi^{[+2]}$) on the right since the relation (\ref{psirel}) trivially commutes with the Killing fields $T$ and $\Phi$ respectively. We will refer to those estimates as the ``$T$-commuted and $\Phi$-commuted (\ref{idp1})'' below.

Similarly from (\ref{Prel}),
\begin{align} \label{idp2}
\nabla_a \left(r^n \frac{1}{\rho^2}\frac{r^2+a^2}{\Delta} \underline{L}^a |\uppsi^{[+2]} \sqrt{\Delta}|^2\right) +\frac{n}{2} \frac{ r^{n-1}}{\rho^2} |\uppsi^{[+2]} \sqrt{\Delta}|^2  \leq C_n \frac{1}{\rho^2} \frac{r^{n+1}}{\left(r^2+a^2\right)^2}  \Big|P^{[+2]} \left(r^2+a^2\right)^{\frac{3}{2}}\Big|^2  \, 
\end{align}
and the same estimate replacing $\uppsi^{[+2]}$ by $T \uppsi^{[+2]}$ ($\Phi \uppsi^{[+2]}$) on the left and $\Psi^{[+2]}$ by $T\Psi^{[+2]}$ ($\Phi \Psi^{[+2]}$) on the right. We again refer to the
latter as the ``$T$-commuted and $\Phi$-commuted (\ref{idp2})'' below.

Let us first obtain the estimate in $\widetilde{\mathcal{R}}^{\rm right}(\tau_1,\tau_2)$. 
The case in $\widetilde{\mathcal{R}}^{\rm left}(\tau_1,\tau_2)$
is analogous but easier since weights in $r$ do not play a role. We add
\begin{itemize}
\item (\ref{idp2})  
with $n \in \{\eta,1,2-\eta\}$
\item the $\Phi$-commuted (\ref{idp2}) with $n \in \{\eta,1,2-\eta\}$
\item the $T$-commuted (\ref{idp2}) with $n =2-\eta$
\end{itemize}
integrated over $\widetilde{\mathcal{R}}^{\rm right}(\tau_1,\tau_2)$. 
Combining the above we conclude 
 for $p \in \{\eta,1,2\}$ the estimate
\begin{align} \label{sety}
\mathbb{E}^{\rm right}_{\widetilde\Sigma_{\tau},p} \left[\uppsi^{[+2]}\right] \left(\tau_2\right) +\mathbb{E}_{r=A_2} \left[\uppsi^{[+2]}\right] \left(\tau_1,\tau_2\right) + \mathbb{I}^{\rm right}_p \left[\uppsi^{[+2]}\right]  \left(\tau_1,\tau_2\right) \nonumber \\
\lesssim \mathbb{I}^{\rm away}_p \left[\Psi^{[+2]}\right]  \left(\tau_1,\tau_2\right) + \mathbb{E}^{\rm right}_{\widetilde\Sigma_{\tau},p} \left[\uppsi^{[+2]}\right] \left(\tau_1\right)   \, .
\end{align}
Turning to the estimate (\ref{idp1}) we add
\begin{itemize}
\item (\ref{idp1}) with $n \in \{2+\eta,3,4-\eta\}$
\item the $\Phi$-commuted (\ref{idp1}) with $n \in \{2+\eta,3,4-\eta\}$
\item the $T$-commuted (\ref{idp1}) with $n =4-\eta$
\end{itemize}
integrated over $\widetilde{\mathcal{R}}(\tau_1,\tau_2) \cap \{r \geq A_2\}$. 
Combining the above we conclude for $p \in \{\eta,1,2\}$ (note that for $p=2$ there is an $\eta$-loss in the definition of the densities (\ref{dens1}), (\ref{dens2}), ensuring that we can indeed set $p=2$)
\begin{align} \label{setx}
\mathbb{E}^{\rm right}_{\widetilde\Sigma_{\tau},p} \left[\upalpha^{[+2]}\right] \left(\tau_2\right) + \mathbb{E}_{r=A_2} \left[\upalpha^{[+2]}\right] \left(\tau_1,\tau_2\right) + \mathbb{I}^{\rm right}_p \left[\upalpha^{[+2]}\right]  \left(\tau_1,\tau_2\right) \nonumber \\
\lesssim \mathbb{I}^{\rm right}_p \left[\uppsi^{[+2]}\right]   \left(\tau_1,\tau_2\right) + \mathbb{E}^{\rm right}_{\widetilde\Sigma_{\tau},p} \left[\upalpha^{[+2]}\right] \left(\tau_1\right) \, .
\end{align}
Combining (\ref{setx}) and (\ref{sety}) yields the desired estimate to the right of trapping.

As remarked above,
the estimate in the ``left region'' $\widetilde{\mathcal{R}}^{\rm left}(\tau_1,\tau_2)$ 
is easier  and left to the reader.
\end{proof}

\subsubsection{Transport estimates for $\uppsi^{[-2]}$ and $\upalpha^{[-2]}$}

\begin{proposition} \label{transportminusnof}
Let $\upalpha^{[-2]}$ be as in Proposition~\ref{follfundprop},
and $\uppsi^{[-2]}$, $\Psi^{[-2]}$ be as defined in (\ref{Pdefsbulkmunus2}), (\ref{officdeflittlepsiminus}).
Then we have the following estimate in $\widetilde{\mathcal{R}}^{\rm right}(\tau_1,\tau_2)$: 
\begin{align} \label{fiete}
 &  \ \ \ \mathbb{E}^{\rm right}_{\widetilde\Sigma_{\tau}} \left[\upalpha^{[-2]}\right](\tau_2)
+\mathbb{I}^{\rm right} \left[\upalpha^{[-2]}\right]  \left(\tau_1,\tau_2 \right) +  \mathbb{E}_{\mathcal{I}^+} \left[\upalpha^{[-2]}\right] \left(\tau_1,\tau_2\right) \nonumber
 \\
&+\mathbb{E}^{\rm right}_{\widetilde\Sigma_{\tau}} \left[\uppsi^{[-2]}\right](\tau_2) 
+\mathbb{I}^{\rm right} \left[\uppsi^{[-2]}\right]  \left(\tau_1,\tau_2 \right) 
+ \mathbb{E}_{\mathcal{I}^+} \left[\uppsi^{[-2]}\right] \left(\tau_1,\tau_2\right) \nonumber
 \\
&\qquad \lesssim
 \mathbb{I}^{\rm away}_{ \eta}\left[\Psi^{[-2]}\right](\tau_1, \tau_2)
 +
\mathbb{E}^{\rm right}_{\widetilde\Sigma_{\tau}} \left[\upalpha^{[-2]}\right](\tau_1)
+\mathbb{E}^{\rm right}_{\widetilde\Sigma_{\tau}} \left[\uppsi^{[-2]}\right](\tau_1) \nonumber \\
& \qquad  \ \ \ + \mathbb{E}_{r=A_2} \left[\upalpha^{[-2]}\right] \left(\tau_1,\tau_2\right)
+\mathbb{E}_{r=A_2} \left[\uppsi^{[-2]}\right] \left(\tau_1,\tau_2\right)
\end{align}
and the following estimate in $\widetilde{\mathcal{R}}^{\rm left}(\tau_1,\tau_2)$:
\begin{align}\label{fiete.2}
\nonumber
 &  \ \ \ \mathbb{E}^{\rm left}_{\widetilde\Sigma_{\tau}} \left[\upalpha^{[-2]}\right](\tau_2)
+\mathbb{I}^{\rm left} \left[\upalpha^{[-2]}\right]  \left(\tau_1,\tau_2 \right) +  \mathbb{E}_{r=A_1} \left[\upalpha^{[-2]}\right] \left(\tau_1,\tau_2\right)
\nonumber
 \\
&+\mathbb{E}^{\rm left}_{\widetilde\Sigma_{\tau}} \left[\uppsi^{[-2]}\right](\tau_2) 
+\mathbb{I}^{\rm left} \left[\uppsi^{[-2]}\right]  \left(\tau_1,\tau_2 \right) 
+ \mathbb{E}_{r=A_1} \left[\uppsi^{[-2]}\right] \left(\tau_1,\tau_2\right)
 \\
&\qquad \lesssim
 \mathbb{I}^{\rm away}_{ \eta}\left[\Psi^{[-2]}\right](\tau_1, \tau_2)
 +
\mathbb{E}^{\rm left}_{\widetilde\Sigma_{\tau}} \left[\upalpha^{[-2]}\right](\tau_1)
+\mathbb{E}^{\rm left}_{\widetilde\Sigma_{\tau}} \left[\uppsi^{[-2]}\right](\tau_1)  \nonumber
\end{align}
\end{proposition}

\begin{remark} \label{remark:avoideps}
As the proof will show, these estimates also hold replacing $\mathbb{I}^{\rm away}_{ \eta}\left[\Psi^{[-2]}\right]$ by $\mathbb{I}^{\rm away}_{0}\left[\Psi^{[-2]}\right]$ provided we drop the two terms on null infinity $\mathcal{I}^+$ in (\ref{fiete}) and weaken the $r$-weight in the energies $\mathbb{E}^{\rm right}_{\widetilde\Sigma_{\tau}} \left[\upalpha^{[-2]}\right]$ and $\mathbb{E}^{\rm right}_{\widetilde\Sigma_{\tau}} \left[\uppsi^{[-2]}\right]$ from $r^{-1-\eta}$ to $r^{-1-2\eta}$; see (\ref{zas}), (\ref{zas2}). This way one could avoid the $r^\eta$ multiplier for $\Psi^{[-2]}$ (at the cost of losing control over the generically non-vanishing fluxes on null infinity).
\end{remark}

\begin{proof}
We recall the relations
\begin{align}
2 \frac{\Delta}{(r^2+a^2)^2} \sqrt{\Delta}{\uppsi}^{[-2]} &= L^a \nabla_a \left(\upalpha^{[-2]} \left(r^2+a^2\right)^{-\frac{3}{2}}\right) \label{psibart}
\, ,  \\
-\frac{\Delta}{(r^2+a^2)^2}\left(r^2+a^2\right)^{3/2} {P}^{[-2]} &=L^a \nabla_a \left( \sqrt{\Delta} {\uppsi}^{[-2]}\right) \, . \label{Pbart}
\end{align}
From (\ref{psibart}) we derive (recall $\rho^2=r^2+a^2\cos^2\theta$) for any $n, \eta \in \mathbb{R}$
\begin{align} 
\nabla_a \left(\left(\frac{\Delta}{r^2+a^2}\right)^{-n-1+4} \left(1+\frac{1}{r^\eta}\right) \frac{1}{\rho^2} L^a \Big|\frac{\sqrt{r^2+a^2} \alpha^{[-2]}}{\Delta^2} \Big|^2\right) \nonumber \\
+\left[    \left(1+\frac{1}{r^\eta}\right) \frac{2Mn \left(r^2-a^2\right)}{ (r^2+a^2)^2} + \frac{\eta}{r^{1+\eta}} \frac{\Delta}{r^2+a^2} \right] \frac{1}{\rho^2} \left(\frac{\Delta}{r^2+a^2}\right)^{-n-1+4}\Big|\frac{\sqrt{r^2+a^2} \alpha^{[-2]}}{\Delta^2} \Big|^2 \nonumber \\
= -2 \left(r^2+a^2\right) w^\frac{3}{2} \left(\frac{\Delta}{r^2+a^2}\right)^{-n-1+2} \frac{1}{\rho^2}  \left(1+\frac{1}{r^\eta}\right)  \left(\uppsi^{[-2]} \cdot \frac{\sqrt{r^2+a^2}  \overline{\alpha^{[-2]}}}{\Delta^2}  + \overline{\uppsi^{[-2]}} \cdot \frac{\sqrt{r^2+a^2}  {\alpha^{[-2]}}}{\Delta^2}\right) \, , \nonumber
\end{align}
and hence, choosing $n=3$, we have for any $\eta>0$ the estimate
\begin{align} \label{klp}
\nabla_a \left( \left(1+\frac{1}{r^\eta}\right) \frac{1}{\rho^2} L^a \Big|\frac{\sqrt{r^2+a^2} \alpha^{[-2]}}{\Delta^2} \Big|^2\right) &  \\
+\frac{1}{2} \left[    \left(1+\frac{1}{r^\eta}\right) \frac{6M \left(r^2-a^2\right)}{ (r^2+a^2)^2} + \frac{\eta}{r^{1+\eta}} \frac{\Delta}{r^2+a^2} \right] \frac{1}{\rho^2}\Big|\frac{\sqrt{r^2+a^2} \alpha^{[-2]}}{\Delta^2} \Big|^2
&\leq C_\eta \frac{1}{\rho^2} \Big|\frac{ \left(r^2+a^2\right) {\uppsi^{[-2]}}}{\sqrt{\Delta}}\Big|^2  \frac{r^{1+\eta}}{\left(r^2+a^2\right)^2} \, . \nonumber
\end{align}
Moreover, the same estimate holds replacing $1+ \frac{1}{r^\eta}$ by $\frac{1}{r^\eta}$ on the left and $r^{1+\eta}$ by $r^{1-\eta}$ on the right (cf.~Remark~\ref{remark:avoideps}). 
Note also that the estimate (\ref{klp}) also holds replacing $\upalpha^{[-2]}$ by $T\upalpha^{[-2]}$ ($\Phi \upalpha^{[-2]}$) and $\uppsi^{[-2]}$ by $T\uppsi^{[-2]}$ ($\Phi \uppsi^{[-2]}$) in view of the relation (\ref{psibart}) commuting trivially with the Killing field $T$ and $\Phi$. We will refer to those estimates as the $T$- and $\Phi$-commuted (\ref{klp}) below.

From (\ref{Pbart}) we derive
\begin{align} \label{klo}
\nabla_a \left( \left(1+\frac{1}{r^\eta}\right) \frac{1}{\rho^2} L^a \Bigg|\frac{{\uppsi}^{[-2]}(r^2+a^2)}{\sqrt{\Delta}}\Bigg|^2\right) & \\
+\frac{1}{2} \left[    \left(1+\frac{1}{r^\eta}\right) \frac{2M \left(r^2-a^2\right)}{ (r^2+a^2)^2} + \frac{\eta}{r^{1+\eta}} \frac{\Delta}{r^2+a^2} \right] \frac{1}{\rho^2}  \Bigg|\frac{{\uppsi^{[-2]}}(r^2+a^2)}{\sqrt{\Delta}}\Bigg|^2 
&\leq C_\eta \frac{1}{\rho^2} \Big| \left(r^2+a^2\right)^{3/2}P^{[-2]}\Big|^2\frac{r^{1+\eta}}{\left(r^2+a^2\right)^2} \, .  \nonumber
\end{align}
Moreover, the same estimate holds replacing $1+ \frac{1}{r^\eta}$ by $\frac{1}{r^\eta}$ on the left and $r^{1+\eta}$ by $r^{1-\eta}$ on the right (cf.~Remark~\ref{remark:avoideps}). 
Note also that the estimate (\ref{klo}) also holds replacing $\uppsi^{[-2]}$ by $T\uppsi^{[-2]}$ ($\Phi \uppsi^{[-2]}$) and $\Psi^{[-2]}$ by $T\Psi^{[-2]}$ ($\Phi \Psi^{[-2]}$) in view of the relation (\ref{Pbart}) commuting trivially with the Killing field $T$ and $\Phi$. We will refer to this estimates as the $T$- and $\Phi$-commuted (\ref{klo}) below.

We are now ready to prove the estimate in $\widetilde{\mathcal{R}}^{\rm left}(\tau_1,\tau_2)$. 

Integrating (\ref{klo}) and the $T$-commuted and $\Phi$-commuted (\ref{klo}) over $\widetilde{\mathcal{R}}^{\rm left}(\tau_1,\tau_2)$ produces 
\begin{align} 
\mathbb{E}^{\rm left}_{\widetilde\Sigma_{\tau}} \left[\uppsi^{[-2]}\right] \left(\tau_2\right)  + \mathbb{I}^{\rm left} \left[\uppsi^{[-2]}\right]  \left(\tau_1,\tau_2\right)
+ \mathbb{E}_{r=A_1} \left[\uppsi^{[-2]}\right]  \left(\tau_1,\tau_2\right) \nonumber \\
\lesssim \mathbb{I}_{\eta}^{\rm away} \left[\Psi^{[-2]}\right]  \left(\tau_1,\tau_2\right) + \mathbb{E}^{\rm left}_{\widetilde\Sigma_{\tau}} \left[\uppsi^{[-2]}\right] \left(0\right)  .
\end{align}
Integrating (\ref{klp}) and the $T$-commuted and $\Phi$-commuted (\ref{klp}) over $\widetilde{\mathcal{R}}^{\rm trap}(\tau_1,\tau_2)$ produces 
\begin{align} 
\mathbb{E}^{\rm left}_{\widetilde\Sigma_{\tau}} \left[\upalpha^{[-2]}\right] \left(\tau_2\right)  + \mathbb{I}^{\rm left} \left[\upalpha^{[-2]}\right]  \left(\tau_1,\tau_2\right)
+ \mathbb{E}_{r=A_1} \left[\upalpha^{[-2]}\right]  \left(\tau_1,\tau_2\right) \nonumber \\
\lesssim \mathbb{I}^{\rm left} \left[\uppsi^{[-2]}\right]  \left(\tau_1,\tau_2\right) + \mathbb{E}^{\rm left}_{\widetilde\Sigma_{\tau}} \left[\upalpha^{[-2]}\right] \left(0\right)  .
\end{align}
Combining the last two estimate produces the desired estimate
in $\widetilde{\mathcal{R}}^{\rm left}(\tau_1,\tau_2)$.
The estimate in $\widetilde{\mathcal{R}}^{\rm right}(\tau_1,\tau_2)$ is proven entirely analogously and is again left to the reader. The only important observation is that the good $\uppsi$-spacetime term generated from (\ref{klo}) is stronger (in terms of $r$-weight) than what is needed on the left hand side of (\ref{klp}).
\end{proof}

\subsection{Auxiliary estimates}
\label{auxil.est.sec}
We collect a number of  auxiliary estimates we shall require.

\subsubsection{The homogeneous $T+\upomega_+\chi \Phi$ estimate}

\begin{proposition}
\label{proptocontrolcutoff}
Let $\upalpha^{[\pm2]}$ satisfy the \underline{homogeneous} Teukolsky
equation $(\ref{Teukphysic})$ and let $\uppsi^{[\pm2]}$, $\Psi^{[\pm2]}$ be as defined in (\ref{Pdefsbulkplus2}), (\ref{Pdefsbulkmunus2}), (\ref{officdeflittlepsiplus}), (\ref{officdeflittlepsiminus}). Then we have
for any $0\leq \tau_1\leq \tau_2$
\begin{equation}
\label{firstinequalityhere}
\mathbb{E}_{\widetilde{\Sigma}_\tau,0}\left [\Psi^{[\pm2]}\right]
(\tau_2)
\lesssim
|a| \mathbb{I}^{\rm deg}_0\left[\Psi^{[\pm2]}\right](\tau_1,\tau_2)
+|a| \mathbb{I}_{[\eta]}\left[\uppsi^{[\pm2]}\right](\tau_1,\tau_2)
+|a| \mathbb{I}_{[\eta]}\left[\upalpha^{[\pm2]}\right](\tau_1,\tau_2)
+\mathbb{E}_{\widetilde{\Sigma}_\tau,0}\left [\Psi^{[\pm2]}\right](\tau_1).
\end{equation}
Here the subindex $\left[\eta\right]$ is equal to $\eta$ in case of $s=+2$ and it is dropped entirely in case $s=-2$.
\end{proposition}

\begin{proof}
The  inequality $(\ref{firstinequalityhere})$
 follows from integrating the identity (\ref{Hawkinident}) associated with the multiplier $T+\upomega_+\chi \Phi$ over the region $\widetilde{\mathcal{R}} \left(\tau_1,\tau_2\right)$. Note that 
 $\mathfrak{G}^{[s]}=0$, the additional boundary terms arising are nonnegative,
  and  the terms arising from $\mathcal{J}^{[s]}$ are easily controlled by the
  right hand side of $(\ref{firstinequalityhere})$ using the
 Cauchy--Schwarz inequality, in view of the fact that the support of
 $\chi$ is away from the degeneration of $\mathbb{I}^{\rm deg}$.
\end{proof}

\subsubsection{Local in time estimates}

\begin{proposition}
\label{localintimehereprop}
Let $\upalpha^{[\pm2]}$ satisfy the \underline{homogeneous} Teukolsky
equation and let $\uppsi^{[\pm2]}$, $\Psi^{[\pm2]}$ be as defined in (\ref{Pdefsbulkplus2}), (\ref{Pdefsbulkmunus2}), (\ref{officdeflittlepsiplus}), (\ref{officdeflittlepsiminus}). 
Then for any $\tau_{\rm step}>0$ there exists an $a_0\ll M$ such that
for $|a|<a_0$ 
we have for any $\tau_1 >0$
\begin{align}
\label{first.of.the.haves}
\nonumber
&\sup_{\tau_1 \le \tau\le \tau_1+\tau_{\rm step}}\mathbb{E}_{\widetilde{\Sigma}_\tau, 0}\left [\Psi^{[\pm2]}\right]
(\tau)\\
&\qquad \lesssim
\mathbb{E}_{\widetilde{\Sigma}_\tau, 0}\left [\Psi^{[\pm2]}\right]
(\tau_1)+
|a|\tau_{\rm step}e^{C\tau_{\rm step}} \mathbb{E}_{\widetilde{\Sigma}_\tau, [\eta]}\left [\uppsi^{[\pm2]}\right](\tau_1)
+
|a|\tau_{\rm step}e^{C\tau_{\rm step}}
\mathbb{E}_{\widetilde{\Sigma}_\tau, [\eta]}\left [\upalpha^{[\pm2]}\right](\tau_1),
\end{align}
\begin{align}
\label{second.of.the.haves}
\nonumber
\mathbb{I}_{0}\left [\Psi^{[\pm2]}\right]
(\tau_1, \tau_1+\tau_{\rm step})+
\mathbb{I}_{[\eta]} \left [\uppsi^{[\pm2]}\right](\tau_1, \tau_1+\tau_{\rm step})
+
\mathbb{I}_{[\eta]}\left [\upalpha^{[\pm2]}\right]
(\tau_1, \tau_1+\tau_{\rm step})
\\
\lesssim \tau_{\rm step} \mathbb{E}_{\widetilde{\Sigma}_\tau, 0}\left [\Psi^{[\pm2]}\right]
(\tau_1)+
\mathbb{E}_{\widetilde{\Sigma}_\tau, [\eta]}\left [\uppsi^{[\pm2]}\right](\tau_1)
+
\mathbb{E}_{\widetilde{\Sigma}_\tau, [\eta]}\left [\upalpha^{[\pm2]}\right](\tau_1)
\end{align}
where $C=C(M)$ (and the implicit constant in $\lesssim$ is independent of both
$\tau_{\rm step}$ and $\tau_1$, according to our general conventions). Here the subindex $\left[\eta\right]$ is equal to $\eta$ in case of $s=+2$ and it is dropped entirely in case $s=-2$.
\end{proposition}

\begin{proof}
We first note that 
\begin{align}
\label{noting.hill}
\nonumber
&\sup_{\tau_1\le \tau\le \tau_1+\tau_{\rm step}} \left(\mathbb{E}_{\widetilde{\Sigma}_\tau, 0}\left [\Psi^{[\pm2]}\right](\tau)
+ \mathbb{E}_{\widetilde{\Sigma}_\tau, [\eta]}\left [\uppsi^{[\pm2]}\right](\tau)
+\mathbb{E}_{\widetilde{\Sigma}_\tau, [\eta]}\left [\upalpha^{[\pm2]}\right](\tau)\right)\\
&\qquad \lesssim
e^{C\tau_{\rm step}}\left(
\mathbb{E}_{\widetilde{\Sigma}_\tau, 0}\left [\Psi^{[\pm2]}\right](\tau_1)
+ \mathbb{E}_{\widetilde{\Sigma}_\tau, [\eta]}\left [\uppsi^{[\pm2]}\right](\tau_1)
+
\mathbb{E}_{\widetilde{\Sigma}_\tau, [\eta]}\left [\upalpha^{[\pm2]}\right](\tau_1)\right).
\end{align}
This follows easily by the estimates of the previous sections.

We now apply $(\ref{firstinequalityhere})$ with
$\tau_2$ taken in $\tau_1\le \tau_2 \le \tau_1+\tau_{\rm step}$,
noting that the first three terms on the
right hand side can be bounded by $|a|\tau_{\rm step}$ times the right
hand side of $(\ref{noting.hill})$. Restricting $a_0$ so that in particular
$|a|\tau_{\rm step}e^{C\tau_{\rm step}} < 1$ we obtain $(\ref{first.of.the.haves})$.

We  note that we can repeat the transport estimates of
Section~\ref{condtranestsec}, 
now for the \underline{homogeneous} equations, and 
applied globally
in $\widetilde{\mathcal{R}}(\tau_1,\tau_1+\tau_{\rm step})$,
obtaining
\begin{align*}
\mathbb{I}_{[\eta]}\left [\uppsi^{[\pm2]}\right](\tau_1, \tau_1+\tau_{\rm step})
+
\mathbb{I}_{[\eta]}\left [\upalpha^{[\pm2]}\right]
(\tau_1, \tau_1+\tau_{\rm step})
&\lesssim \mathbb{E}_{\widetilde{\Sigma}_\tau, [\eta]}\left [\uppsi^{[\pm2]}\right](\tau_1)
+
\mathbb{E}_{\widetilde{\Sigma}_\tau, [\eta]}\left [\upalpha^{[\pm2]}\right](\tau_1)
\\
&\qquad+ \mathbb{I}_\eta\left[\Psi^{[\pm2]}\right](\tau_1, \tau_1+\tau_{\rm step}).
\end{align*}
Note that the term is $\mathbb{I}_\eta$ and not $\mathbb{I}^{\rm deg}_\eta$.

In view of
\[
\mathbb{I}_\eta\left[\Psi^{[\pm2]}\right](\tau_1, \tau_1+\tau_{\rm step})
\lesssim
\int_{\tau_1}^{\tau_1+\tau_{\rm step}}
\mathbb{E}_{\widetilde{\Sigma}_\tau, 0}\left [\Psi^{[\pm2]}\right](\tau) d\tau
\lesssim 
\tau_{\rm step} \sup_{\tau_1\le \tau\le \tau_1+\tau_{\rm step}} \mathbb{E}_{\widetilde{\Sigma}_\tau, 0}\left [\Psi^{[\pm2]}\right](\tau)
\]
(note the $\eta$ on the left but the $0$ on the right hand side),
we obtain  $(\ref{second.of.the.haves})$ for sufficiently small $a$.
\end{proof}

\begin{remark}
We note that  a more careful examination of the Schwarzschild case
and Cauchy stability yields that  the
inequality $(\ref{second.of.the.haves})$ can be proven without the $\tau_{\rm step}$ factor
on the first term of  right hand side, provided $\mathbb{I}_0$ is
replaced by $\mathbb{I}_0^{\rm deg}$. We shall not however require this  here.
\end{remark}

\section{The admissible class and Teukolsky's separation}
\label{Separationmegasec}

In this section we will implement Teukolsky's separation~\cite{teukolsky1973} of~$(\ref{Teukop})$
for $s=\pm2$.

To make sense \emph{a priori} of the formal separation of~\cite{teukolsky1973}, 
one must in particular work in a class of functions for
which  one can indeed take the
Fourier transform in time.  This requires applying
the analysis to functions which satisfy certain 
time-integrability properties. A useful such class is the
``sufficiently integrable, outgoing'' class defined in~\cite{partiii, Dafermos:2014jwa}
for the $s=0$ case.

In the present paper, it turns out that we shall only require Fourier analysis in
the region $r\in [A_1,A_2]$. We may thus consider the more elementary setting
of what we shall call the \emph{$[A_1,A_2]$-admissible class}
where time square integrability is only required for $r\in[A_1,A_2]$. (We will in fact
assume compact support in $t^*$ in this $r$-range.)
This leads to a number of useful simplifications. In particular, we need not refer
to the asymptotic analysis of the ODE's as $r^*\to \pm\infty$,
as was done in~\cite{partiii, Dafermos:2014jwa}, in order to infer boundary behaviour.

The section is organised as follows: We will define our elementary notion 
of $[A_1,A_2]$-admissible class
in {\bf Section~\ref{siosec}}.  
We will then implement
Teukolsky's separation in
{\bf Section~\ref{revofteuksepsec}}, deriving the radial ODE, valid for $r\in [A_1,A_2]$.

(We note already that, in practice, the results of this section will be applied to solutions of the
inhomogeneous Teukolsky equation which arises from applying a suitable cutoff
to solutions of $(\ref{Teukphysic})$. The restriction of Fourier analysis to the
range $r^*\in [A_1^*,A_2^*]$  will
allow us to use a cutoff whose derivatives are supported in a region of
finite $r^*\in[2A_1^*,2A_2^*]$, leading to additional simplifications with respect
to~\cite{partiii}. We will only turn to this in Section~\ref{iledsec}.)

\subsection{The  $[A_1,A_2]$-admissible class}
\label{siosec}

We define an admissible class of functions for our  frequency analysis.
This is to be compared with 
 the class of \emph{sufficiently integrable} functions
  from~\cite{partiii, Dafermos:2014jwa}.
Since we will only apply frequency localisation in a neighbourhood of trapping,
we only consider the behaviour in the fixed $r$-region $[A_1,A_2]$
with $r_+<A_1<A_2<\infty$ defined in 
Section~\ref{very.slowly.very.slowly}. (Recall in this region that
$t=t^*=\tilde{t}^*$.)
On the other hand, for convenience, we will assume
compact support in $t$ for these $r$-values, as this is what we shall indeed
obtain after applying cutoffs.

\begin{definition}
\label{suffintdef}
Let $a_0<M$, $|a|<a_0$ and let $g=g_{a,M}$.
We say that a smooth complex valued spin$\pm2$ weighted 
function $\tilde\upalpha:\mathcal{R}\cap \{A_1\le r\le A_2\}
\to \mathbb C$ is \emph{$[A_1,A_2]$-admissible}
if it is compactly supported in $t$.
\end{definition}

\begin{remark}
One could work with the weaker condition
that (cf.~\cite{partiii}) for all $j\ge 1$, the following holds
\begin{equation}
\label{firstcondhere}
\sup_{r\in [A_1,A_2]}
\int_{-\infty}^\infty
\int_{\mathbb S^2}
\sum_{0\le i_1+i_2+i_3+i_4+i_5\le j}
\left| (\tilde{Z}_1)^{i_1}(\tilde{Z}_2)^{i_2}(\tilde{Z}_3)^{i_3}T^{i_4}(\partial_r)^{i_5}\tilde\upalpha \right|^2 \sin
\theta\, dt\, d\theta\, d\phi <\infty,
\end{equation}
with the only caveat that in the frequency analysis we would have to 
restrict to generic frequency $\omega$ for the ODE to be satisfied
in the classical sense.
\end{remark}

\subsection{Teukolsky's separation}
\label{revofteuksepsec}
We will now implement Teukolsky's formal separation  of the operator~$(\ref{Teukop})$
in the context of $[A_1,A_2]$-admissible spin-$s$ weighted functions $\upalpha^{[s]}$
for $s=\pm2$.

We begin in Section~\ref{spinweigtsec} with a review of the basic properties
of spin-weighted oblate spheroidal harmonics and their associated eigenvalues
$\lambda^{[s]}_{m\ell}(\nu)$. We will then turn immediately in Section~\ref{sec:swest}
to some elementary estimates
for the eigenvalues $\lambda^{[s]}_{m\ell}(\nu)$ which will be
useful later in the paper.
Next, we shall  apply these oblate spheroidals  together with the Fourier transform
in time in Section~\ref{coeffssec}
to  define coefficients $\upalpha^{[s],(a\omega)}_{m\ell}(r)$ associated
to $[A_1,A_2]$-admissible $\upalpha^{[s]}$. We  then 
give Proposition~\ref{radodeprop} 
in 
Section~\ref{radialodesec}, stating 
that these coefficients  satisfy an ordinary differential equation with respect
to $r^*$; this is the content of Teukolsky's remarkable separation of $(\ref{Teukop})$.

\subsubsection{Spin-weighted oblate spheroidal harmonics}
\label{spinweigtsec}

Let $\nu\in \mathbb R$, $s=0,\pm2$ and consider
the self-adjoint operator $\mathring{\slashed{\triangle}}^{[s]}(\nu)$ defined by
\[
\mathring{\slashed{\triangle}}^{[s]}(\nu) \Xi = -\frac{1}{\sin \theta}\frac\partial{\partial  \theta}\left(\sin \theta \frac{\partial\Xi }{\partial\theta}\right)-\left(\frac{\partial^2 \Xi}{\partial\phi^2}+2s\cos\theta 
i \frac{\partial \Xi}{\partial\phi} \right)\frac{1}{\sin^2\theta} -\nu^2\cos^2\theta \Xi
+2\nu s\cos \theta \Xi+s^2\cot^2 \theta \Xi  -s \Xi
\]
on $\mathscr{S}^{[s]}_\infty$, which we recall is a dense subset of $L^2(\sin \theta\,  d\theta \, d\phi)$.

This has a complete collection of eigenfunctions
\begin{align} \label{setofef}
 \{
S^{[s]}_{m\ell}(\nu, \cos\theta)e^{im\phi}\}_{m\ell}
\end{align}
with eigenvalues $\lambda^{[s]}_{m\ell}\in \mathbb R$, indexed by $m\in \mathbb Z$, $\ell \ge |m|+|s|$.
These are known as
 the spin-weighted oblate\footnote{The prolate case corresponds to the $\xi$ being purely imaginary.} spheroidal harmonics. For each fixed $m\in \mathbb Z$, the $S^{[s]}_{m\ell}$  themselves 
 form a complete collection of eigenfunctions of the following self-adjoint operator with eigenvalues $\lambda^{[s]}_{m \ell} \left(\nu\right)$:
\begin{align}
\mathring{\slashed\triangle}^{[s]}_m \left(\nu\right) :=- \frac{1}{\sin \theta} \frac{d}{d\theta} \left(\sin \theta \frac{d}{d\theta}\right) +\left(-\nu^2 \cos^2\theta + \frac{m^2}{\sin^2 \theta} + 2\nu s \cos \theta + \frac{2ms \cos \theta}{\sin^2 \theta} + s^2 \cot^2 \theta -s \right)
\end{align}
\begin{align}  \label{deflam}
\mathring{\slashed\triangle}^{[s]}_m \left(\nu\right) S_{m\ell}^{[s]}  =  \lambda^{[s]}_{m\ell} \left(\nu\right) S_{m\ell}^{[s]} \, . 
\end{align}
The eigenfuctions themselves satisfy
\[
S^{[s]}_{m\ell}(\nu, \cos\theta)e^{im\phi} \in \mathscr{S}_{\infty}^{[s]}
\]
for all $\nu\in \mathbb R$.

We note the following familiar special cases:
\begin{enumerate}
\item For $s=0$ one obtains the oblate spheroidal harmonics familiar from the angular part of the separation equation of the scalar wave equation on Kerr \cite{partiii}. The case $s=0$ and $\nu=0$ recovers the standard spherical harmonics $S_{m \ell}^{[0]}(0,\cos\theta)e^{im\phi}=Y_{m \ell}$ with eigenvalues $\ell \left(\ell+1\right)$.

\item For $\nu=0$, then $\mathring{\slashed\triangle}^{[s]}(0)$ is the spin-$s$-weighted Laplacian
and one obtains the spin-weighted spherical harmonics, whose eigenvalues can also be determined explicitly
\begin{align} \label{hemo}
\lambda_{m \ell}^{[s]} \left(0\right)+ s = \lambda_{m \ell}^{[-s]} \left(0\right)-s = \ell \left(\ell+1\right) - s^2 \geq 2
\end{align}
where the last inequality follows from the relation $|\ell|\geq |s|$. For future reference we note the relation
\begin{equation}
\label{forfutref}
\mathring{\slashed\triangle}^{[s]}_m \left(\nu\right) = \mathring{\slashed\triangle}^{[s]}_m \left(0\right) -\nu^2 \cos^2\theta + 2\nu s \cos \theta \, .
\end{equation}

\end{enumerate}
We finally remark also the general relation 
\begin{align} \label{genrel}
\lambda_{m \ell}^{[s]} \left(\nu\right)+ s = \lambda_{m \ell}^{[-s]}  \left(\nu\right)-s 
\end{align}
allowing us to restrict to $s=+2$ without loss of generality when obtaining estimates on the $\lambda_{m \ell}^{[s]} \left(\nu\right)$.

For various asymptotics concerning the behaviour of $\lambda_{m\ell}^{[s]}$
see~\cite{Berti:2005gp}.

\subsubsection{Estimates on $\lambda^{[s]}_{m\ell} \left(\nu\right)$
and $\widetilde\Lambda^{[s]}_{m \ell}\left(\nu\right)$}
\label{sec:swest}
To estimate $ \lambda_{m\ell}^{[s]} \left(\nu \right)$ we compute from (\ref{deflam})
\begin{align} \label{imdf2}
 \lambda^{[s]}_{m \ell}\left(\nu\right) +s &= 
 \int_0^\pi \int_0^{2\pi} d\phi \, d\theta \sin \theta  \left[  \big| \partial_\theta \Xi^{[s]} \big|^2 + \left( \frac{ \left( m+ s \cos \theta  \right)^2}{\sin^2 \theta}  -(a\nu)^2 \cos^2\theta +2s \cos \theta \nu  \right) |\Xi^{[s]}|^2\right]   \, ,
\end{align}
where $\Xi^{[s]}$ denotes (shorthand instead of the full (\ref{setofef})) a normalised eigenfunction of the operator $\mathring{\slashed\triangle}^{[s]}_m \left(a\omega\right)$ with eigenvalue $\lambda_{m \ell}^{[s]}\left(a\omega\right)$. Using the variational characterisation of  the lowest eigenvalue of the operator $\mathring{\slashed\triangle}^{[s]}_m \left(0\right)$ (which is $2$ for $m=0,1$ and $m\left(m+1\right)-4$ for $m\geq 2$ by (\ref{hemo}) and the relation $|m|\leq \ell$) we conclude for 
\begin{equation}
\label{workingdef}
\widetilde{\Lambda}^{[\pm2]}_{m \ell}\left(\nu\right) := \lambda^{[s]}_{m \ell}\left(\nu\right)
+s + \nu^2 + 4|\nu |  \, 
\end{equation}
the bound
\begin{align} \label{eli1e} 
\widetilde{\Lambda}^{[\pm2]}_{m \ell}\left(\nu\right) \geq  \max \left(2, m(m+1)-4\right) \, .
\end{align}

Our ode estimates in Section~\ref{ODEmegasec} will only require $(\ref{eli1e})$.
This motivates the following
\begin{definition}
\label{admiss.freq.def}
A triple $(\omega, m, \widetilde{\Lambda})$ will be said to be \underline{admissible} 
if $\omega\in \mathbb R$, $m\in \mathbb Z$ and $\widetilde\Lambda\in \mathbb R$ satisfies
$\widetilde{\Lambda} \ge {\rm max}(2, m(m+1)-4)$.
\end{definition}

\subsubsection{The coefficients $\upalpha^{[s],(a\omega)}_{m\ell}$ and the Plancherel relations}
\label{coeffssec}

Given parameters $a$, $M$ and $s$, we let
$\upalpha^{[s]}$ be $[A_1,A_2]$-admissible according to 
Definition~\ref{suffintdef}.

We have
\begin{equation}
\label{fouriertra}
\upalpha^{[s]}(t,r,\theta,\phi)=\frac{1}{2\pi} \int_{-\infty}^\infty e^{-i\omega t} \hat \upalpha^{[s]} (\omega, r,\theta,\phi)d\omega.
\end{equation}
Setting $\nu =a\omega$, for each $\omega\in \mathbb R$ we may decompose
\begin{equation}
\label{defofcoeffs}
\hat\upalpha^{[s]}(\omega, r, \theta, \phi) =\sum_{m\ell} \upalpha^{[s],(a\omega)}_{m\ell} S^{[s]}_{m,\ell}
(a\omega,\cos\theta)e^{im\phi}.
\end{equation}
We obtain then the representation
\begin{equation}
\label{firstrep}
\upalpha^{[s]} (t,r,\theta, \phi) = \frac{1}{\sqrt{2\pi}}
\int_{-\infty}^\infty\sum_{m\ell} e^{-i\omega t} \upalpha^{[s], (a\omega)}_{m\ell}(r)
S^{[s]}_{m\ell}(a\omega, \cos\theta)e^{im\phi} d\omega.
\end{equation}
As in~\cite{partiii}, we remark that for each fixed $r$,
$(\ref{fouriertra})$ and $(\ref{firstrep})$ 
are to be understood as holding in $L^2_tL^2_{\mathbb S^2}$,
while $(\ref{defofcoeffs})$ is to be understood in $L^2_\omega L^2_{\mathbb S^2}$.
Note that if $\upalpha^{[s]}$ satisfies Definition~\ref{suffintdef}, then so do
$\partial_t \upalpha^{[s]}$ and $\partial_\phi \upalpha^{[s]}$ and we have
\[
\partial_t \upalpha^{[s]}(t,r,\theta,\phi)=\frac{-i}{2\pi} \int_{-\infty}^\infty \omega e^{-i\omega t} \hat \upalpha^{[s]} (\omega, r,\theta,\phi)d\omega \, , 
\]
\[
\partial_\phi \upalpha^{[s]}(t,r,\theta,\phi)=\frac{i}{2\pi} \int_{-\infty}^\infty m e^{-i\omega t} \hat \upalpha^{[s]} (\omega, r,\theta,\phi)d\omega \, ,
\]
where these relations are to be interpreted in $L^2_tL^2_{\mathbb S^{2}}$.

We also recall as in~\cite{DafRodsmalla, partiii}  the following  Plancherel relations
\[
\int_0^{2\pi}\int_0^\pi\int_{-\infty}^\infty
\left|\upalpha^{[s]}\right|^2(t,r,\theta,\phi)\sin\theta\, d\phi\, d\theta \, dt
=\int_{-\infty}^{\infty} \sum_{m\ell}\left|\upalpha^{[s], (a\omega)}_{m\ell}(r)\right|^2 d\omega \, ,
\]
\[
\int_0^{2\pi}\int_0^\pi\int_{-\infty}^\infty
{}_{1}\upalpha^{[s]}\cdot {}_2\bar{\upalpha}^{[s]} \sin\theta\, d\phi\, d\theta \, dt
=\int_{-\infty}^{\infty} \sum_{m\ell}{}_1\upalpha^{[s], (a\omega)}_{m\ell}\cdot
{}_2\bar{\upalpha}^{[s], (a\omega)}_{m\ell}  d\omega \, ,
\]
\[
\int_0^{2\pi}\int_0^\pi\int_{-\infty}^\infty
\left|\partial_r \upalpha^{[s]}\right |^2(t,r,\theta,\phi)\sin\theta\, d\phi\, d\theta \, dt
=\int_{-\infty}^{\infty} \sum_{m\ell}\left|\frac{d}{dr}\upalpha^{[s], (a\omega)}_{m\ell}(r)\right|^2 d\omega \, ,
\]
\[
\int_0^{2\pi}\int_0^\pi\int_{-\infty}^\infty
\left|\partial_t \upalpha^{[s]}\right|^2(t,r,\theta,\phi)\sin\theta\, d\phi\, d\theta \, dt
=\int_{-\infty}^{\infty} \sum_{m\ell}\omega^2\left|\upalpha^{[s], (a\omega)}_{m\ell}(r)\right|^2 d\omega \, ,
\]
as well as
\begin{align}
\int_0^{2\pi}\int_0^\pi\int_{-\infty}^\infty
\left(\left|\frac{\partial \upalpha^{[s]}}{\partial \theta}\right|^2+\left|\left(\frac{\partial \upalpha^{[s]}}{\partial \phi}+is \cos \theta \upalpha^{[s]} \right)\sin^{-1}\theta\right|^2\right)(t,r,\theta,\phi)\sin\theta\, d\phi\, d\theta \, dt
= \nonumber \\
\int_{-\infty}^\infty \sum_{m \ell} \left( \lambda_{m \ell}^{[s]} \left(a\omega \right) +s +a^2 \omega^2 \cos^2 \theta -2sa\omega \cos \theta\right) \left|\upalpha^{[s], (a\omega)}_{m\ell}(r)\right|^2  d\omega \, .
\end{align}
From the inequalities of Section \ref{sec:swest} it follows that 
\begin{align}
\lambda_{m \ell}^{[s]} \left(a\omega \right) +s +a^2 \omega^2 \cos^2 \theta -2sa\omega \cos \theta \lesssim \widetilde{\Lambda}^{[\pm2]}_{m \ell}(a\omega)+\omega^2 \, .
\end{align}
In what follows, we shall often write
$\lambda_{m \ell}^{[s],(a\omega)}$ for $\lambda_{m \ell}^{[s]} \left(a\omega \right)$
and  ${\widetilde{\Lambda}}_{m \ell}^{[s],(a\omega)}$ 
for ${\widetilde{\Lambda}}_{m \ell}^{[s]} \left(a\omega \right)$.

\subsubsection{The radial ODE}
\label{radialodesec}
We here state a proposition that implements Teukolsky's formal separation of $(\ref{Teukop})$ in the
context of $[A_1,A_2]$-admissible spin-weighted functions.

Fix $|a|<M$ and $s=0,\pm2$. 
Let $\upalpha^{[s]}$ be an $[A_1,A_2]$-admissible spin weighted functions
and $\upalpha^{[s], (a\omega)}_{m\ell}$ be as defined in Section~\ref{coeffssec}.
Note that (recall (\ref{rescaleddefsfirstattempt}))
defining 
\begin{equation}
\label{inhomogdefin}
F^{[+2]} = \widetilde{\mathfrak{T}}^{[+2]} \tilde{\upalpha}^{[+2]}  \ \ \ , \ \ \ F^{[-2]} = \widetilde{\mathfrak{T}}^{[-2]} \left(\Delta^2 \tilde{\upalpha}^{[-2]}\right)
\end{equation}
we have   that $F^{[s]}$ is also $[A_1,A_2]$-admissible  and
the coefficients $F^{[s],(a\omega)}_{m\ell}$
can be defined.

Let us first introduce the following shorthand notation
\[
\kappa = \left(r^2+a^2\right)\omega - am 
\]
and 
\begin{equation}
\label{defofbigLambda}
\Lambda^{[s],(a\omega)}_{m \ell}  =  \lambda_{m\ell}^{[s], (a\omega)} + a^2 \omega^2 -2am\omega .
\end{equation}

We have the following
\begin{proposition}
\label{radodeprop}
Fix $|a|<M$ and $s=0,\pm2$. 
Let $\upalpha^{[s]}$ be an $[A_1,A_2]$-admissible spin weighted function, $F^{[s]}$ be as
defined in~$(\ref{inhomogdefin})$,   
with coefficients $\upalpha^{[s], (a\omega)}_{m\ell}$, $(\rho^2F)^{[s], (a\omega)}_{m\ell}$
as defined above.
Then $\upalpha^{[s], (a\omega)}_{m\ell}$ is smooth in $r\in [A_1,A_2]$ and satisfies the 
ordinary differential equation 
\begin{equation}
\label{odeorigform}
\frac{1}{\Delta^s} \frac{d}{dr} \left(\Delta^{s+1} \frac{d\upalpha^{[s], (a\omega)}_{m\ell}}{dr}\right)   +  \left(\frac{\kappa^2-2is\left(r-M\right)
\kappa}{\Delta} + 4is\omega r - \Lambda^{[s], (a\omega)}_{m\ell} \right)
\upalpha^{[s], (a\omega)}_{m\ell}
 = \frac{\left(r^2+a^2\right)^{7/2}}{\rho^2 \Delta^{1+s/2}} F^{[s], (a\omega)}_{m\ell}.
\end{equation}
\end{proposition}

In view of our definitions, the proof is immediate from the usual formal derivation
of $(\ref{odeorigform})$.
See~\cite{hartle1974analytic}.
 The $s=0$ case corresponds precisely to Proposition~5.2.1 of~\cite{partiii}.

Note the difference between $(\ref{defofbigLambda})$ and 
our $\widetilde{\Lambda}^{[s],(a\omega)}_{m \ell}$ in (\ref{workingdef}). It is only the latter quantity
which will appear in the estimates of this paper. 
We have retained $(\ref{defofbigLambda})$
to faciliate comparison with the literature.

\subsubsection{The rescaled coefficients $u$}
\label{boundcondusec}

Let us fix parameters $|a|<M$ and $s$, and consider $\upalpha^{[s]}$ as in 
the statement of Proposition~\ref{radodeprop}.

Define the rescaled\footnote{We note that this renormalisation is slightly different from~\cite{hartle1974analytic}.} quantities
\begin{equation}
\label{rescaled}
u^{[s], (a\omega)}_{m\ell}(r) = 
\Delta^{s/2} \sqrt{r^2+a^2}\, \upalpha^{[s], (a\omega)}_{m\ell}\left(r\right) ,
\end{equation}
\begin{equation}
\label{rescaledrhs}
H^{[s], (a\omega)}_{m\ell} =\frac{\Delta}{\rho^2 w} F^{[s], (a\omega)}_{m\ell} .
\end{equation}
Equation $(\ref{odeorigform})$  then reduces to
\begin{align} \label{basic}
\frac{d^2}{(dr^*)^2} u^{[s], (a\omega)}_{m\ell}  + V^{[s], (a\omega)}_{m\ell}\left(r^*\right) u =H^{[s], (a\omega)}_{m\ell}
\end{align}
for 
\begin{align*}
V^{[s], (a\omega)}_{m\ell}\left(r^*\right) =\frac{\Delta}{\left(r^2+a^2\right)^2} \tilde{V}^{[s], (a\omega)}_{m\ell}+
 V^{[s]}_{0},
\end{align*}
with
\begin{align*}
\tilde{V}^{[s], (a\omega)}_{m\ell} :=  \frac{\kappa^2-2is\left(r-M\right)\kappa}{\Delta} + 4is\omega r - \Lambda^{[s]}_{m\ell}, \qquad
V_{0}^{[s]} := \frac{\Delta^{-s/2+1}}{\left(r^2+a^2\right)^\frac{3}{2}} \frac{d}{dr} \left(\Delta^{s+1} \frac{d}{dr} \left(\frac{\Delta^{-s/2}}{\sqrt{r^2+a^2}}\right)\right).
\end{align*}

For $s=0$, this reduces to the form of the separated wave equation
used in~\cite{partiii}.

\section{The frequency-localised transformations}
\label{Chandrasec}

In this section, we will define frequency localised versions of the quantities
$P^{[\pm2]}$, $\Psi^{[\pm2]}$, $\uppsi^{[\pm2]}$ 
of Section~\ref{physspacechandrasec} and the Regge--Wheeler type equation 
$(\ref{RWtypeinthebulk})$.

We begin in {\bf Section~\ref{sepnullframsec}}  with 
 the definitions of the frequency
localised version of the null frame $L$, $\underline{L}$. We then derive
in {\bf Section~\ref{newnamehere}}
the frequency localised expression for $\Psi^{[\pm2]}$ 
followed in {\bf Section~\ref{flRWeqsec}}
with the frequency localised form of $(\ref{RWtypeinthebulk})$.
%Finally, we derive boundary conditions for  $\Psi^{[\pm2]}$
%and $\uppsi^{[\pm2]}$  in 
%{\bf Section~\ref{moreboundarysec}}
%in the case where $\upalpha^{[\pm2]}$
%is assumed outgoing.

\vskip1pc

{\bf \emph{In what follows in this section,
we will always assume $\upalpha^{[\pm2]}$ is as in Proposition~\ref{radodeprop} 
with corresponding $u^{[\pm2],  (a\omega)}_{m\ell}$.
%, and moreover, we shall
%assume that $\omega$ is a generic frequency.
}}

\subsection{The separated null frame}
\label{sepnullframsec}
Note that (following the conventions in~\cite{partiii}) we have the following formal analogues:
\[
-i\omega \sim \partial_t,
\]
\[
im \sim \partial_\phi.
\]
We define the separated frame operators (corresponding to the principal null 
directions~$(\ref{spacetimenullframe})$) by
\begin{equation}
\label{Ldefin}
L=\frac{d}{dr^*}-i\omega+\frac{iam}{r^2+a^2},
\end{equation}
\begin{equation}
\label{Lbardefin}
-\underline{L}=\frac{d}{dr^*} +i\omega-\frac{iam}{r^2+a^2}.
\end{equation}
We have retained the notation of $(\ref{spacetimenullframe})$ without fear
of confusion.

Also note that
$(\ref{forfutref})$ 
implies  the following formal analogue: 
\[
\mathring{\slashed\triangle}^{[s]}_m \left(a\omega \right) \sim \mathring{\slashed\triangle}^{[s]}_m \left(0\right) + a^2 \cos^2\theta \partial_t^2 - 2i s a\cos \theta \partial_t \, .
\]

\subsection{The frequency localised coefficients $P^{[\pm2], (a\omega)}_{m\ell}$, 
$\Psi^{[\pm 2], (a\omega)}_{m\ell}$ and $\psi^{[\pm 2],(a\omega)}_{m\ell}$}
\label{newnamehere}

We may now understand the relations between the quantities
of Section~\ref{physicalspacedefsec}
at the frequency localised level.

\begin{proposition}
\label{newstyleprop}
Let ${\upalpha}^{[\pm2]}$ be as in Proposition~\ref{radodeprop}
and consider $P^{[+2]}$, $\Psi^{[+2]}$ and $\uppsi^{[+2]}$ defined by
$(\ref{Pdefsbulkplus2})$, $(\ref{rescaledPSI})$ 
and $(\ref{officdeflittlepsiplus})$, respectively,
and consider $P^{[-2]}$, $\Psi^{[-2]}$ and $\uppsi^{[-2]}$ 
defined by $(\ref{Pdefsbulkmunus2})$, $(\ref{rescaledPSI})$ 
and $(\ref{officdeflittlepsiminus})$, respectively.

Let  $u^{[\pm 2], (a\omega)}_{m\ell}$ be the arising 
coefficient of $\upalpha^{[\pm 2]}$.
Then $P^{[\pm 2]}$, $\Psi^{[\pm 2]}$ and $\uppsi^{[\pm 2]}$ are 
$[A_1,A_2]$-admissible spin weighted functions
and their coefficients $P^{[\pm 2], (a\omega)}_{m\ell}$, $\Psi^{[\pm 2], (a\omega)}_{m\ell}$ and
$\psi^{[\pm 2], (a\omega)}_{m\ell}$ are related by
\begin{align}
\label{plustwofwoflr1}
\left(r^2+a^2\right)\sqrt{w} \cdot {\psi}^{[+2], (a\omega)}_{m\ell} &= - \frac{1}{2} \frac{1}{w}\underline{L} \left(u^{[+2], (a\omega)}_{m\ell} \cdot w\right),
 \\
\label{plustwofwoflr2}
 \Psi^{[+2], (a\omega)}_{m\ell} = \left(r^2+a^2\right)^{3/2} {P}^{[+2], (a\omega)}_{m\ell} &= \frac{1}{w} \underline{L} \left( \left(r^2+a^2\right)\sqrt{w} \cdot \psi^{[+2], (a\omega)}_{m\ell}  \right) = -\frac{1}{2} \frac{1}{w} \underline{L} \left(\frac{1}{w}\underline{L} \left(u^{[+2], (a\omega)}_{m\ell} \cdot w\right)\right),
\end{align}
\begin{align}
\label{minustwofwoflr1}
\left(r^2+a^2\right)\sqrt{w} \cdot {\psi}^{[-2], (a\omega)}_{m\ell} &= \frac{1}{2} \frac{1}{w} L \left({u}^{[-2],(a\omega)}_{m\ell} \cdot w\right) \, ,
 \\
\label{minustwofwoflr2}
\Psi^{[-2],(a\omega)}_{m\ell} = \left(r^2+a^2\right)^{3/2} {P}^{[-2],(a\omega)}_{m\ell} &=- \frac{1}{w} L \left( \left(r^2+a^2\right)\sqrt{w} \cdot {\psi}^{[-2],(a\omega)}_{m\ell}  \right) = -\frac{1}{2} \frac{1}{w} L \left(\frac{1}{w} L
 \left({u}^{[-2],(a\omega)}_{m\ell} \cdot w\right)\right).
\end{align}
\end{proposition}

\subsection{The frequency localised Regge--Wheeler equation $(\ref{RWtypeinthebulk})$ 
for $\Psi^{[\pm 2],(a\omega)}_{m\ell}$}
\label{flRWeqsec}

A straightforward computation now leads to

\begin{proposition} \label{psieq}
Under the assumptions of Proposition~\ref{newstyleprop}, 
the  ${\Psi}^{[\pm2], (a\omega)}_{m\ell}$ 
satisfy the equation
\begin{align}  \label{rwPsi2}
\left(\Psi^{[s],(a\omega)}_{m\ell}\right)^{\prime \prime} &+ \left(\omega^2 - \mathcal{V}^{[s],(a\omega)}_{m\ell} \right)\Psi^{[s],(a\omega)}_{m\ell}=  \mathcal{J}^{[s],(a\omega)}_{m\ell} + \mathfrak{G}^{[s],(a\omega)}_{m\ell} \, ,
\end{align}
where the potential $\mathcal{V}^{[s],(a\omega)}_{m\ell}$ is \underline{real} and defined by
\begin{align} \label{Vdef}
\mathcal{V}^{[s],(a\omega)}_{m\ell}&=\frac{ \Delta \left(\lambda^{[s]}_{m \ell} +a^2\omega^2+s^2+s\right) + 4Mram\omega -a^2 m^2 }{\left(r^2+a^2\right)^2} -\frac{\Delta}{(r^2+a^2)^2}\frac{6Mr (r^2-a^2)}{(r^2+a^2)^2} -7 a^2 \frac{\Delta^2}{(r^2+a^2)^4} \nonumber \\
&= \mathcal{V}^{[s]}_0 + \mathcal{V}_1 + \mathcal{V}_2 \, .
\end{align}
and the inhomogeneous terms by
\begin{align}
\mathcal{J}^{[s],(a\omega)}_{m\ell} = & \ \ \ \ \ \  \frac{\Delta}{\left(r^2+a^2\right)^2} \left[ s\frac{-4r^2+4a^2}{r^2+a^2} aim -20a^2 \frac{r^3-3Mr^2+ra^2+Ma^2}{\left(r^2+a^2\right)^2}\right]\left( \sqrt{\Delta}\psi^{[s],(a\omega)}_{m\ell} \right) \nonumber \\
& +a^2 \frac{\Delta}{\left(r^2+a^2\right)^2} \left[ -6s \frac{r}{r^2+a^2} aim
+3  \left( \frac{r^4 -a^4+10Mr^3-6Ma^2r}{(r^2+a^2)^2}\right) \right] \left(u^{[s],(a\omega)}_{m\ell}\frac{\Delta}{\left(r^2+a^2\right)^2}\right), 
\end{align}
\[
 \mathfrak{G}^{[+2],(a\omega)}_{m\ell} =\frac{1}{2} \underline{L} \left( \frac{\left(r^2+a^2\right)^2}{\Delta}\underline{L} \left(\frac{\Delta}{w\rho^2} F^{[+2],(a\omega)}_{m\ell}\right)\right),  
 \qquad
\mathfrak{G}^{[-2],(a\omega)}_{m\ell} =\frac{1}{2} {L} \left( \frac{\left(r^2+a^2\right)^2}{\Delta}{L} \left(\frac{\Delta^3}{w\rho^2} F^{[-2],(a\omega)}_{m\ell}\right)\right).
\]
\end{proposition}

\begin{proof}
See Appendix \ref{sec:Psiderivation}.
\end{proof}

\begin{remark}
Note that $\mathcal{J}^{[s]}$ vanishes for $a=0$. The second line of $\mathcal{J}^{[s]}$ contains only linear terms in $m$
(i.e.~corresponding to only first derivatives in physical space). The first line contains in this sense ``first'' and ``zero'' derivatives of $\psi^{[s]}$ and hence at most (certain) ``second" derivatives of $u^{[s]}$.
\end{remark}

\begin{remark}
We may rewrite the potential
\begin{equation}
\label{rewritepot}
\mathcal{V}^{[\pm2]}_0 = \frac{ \Delta \left(\widetilde{\Lambda}^{[\pm2]}-4|a\omega| +4\right) + 4Mram\omega -a^2 m^2 }{\left(r^2+a^2\right)^2}.
\end{equation}
Here we see the dependence in the spin is entirely contained in the
definition of $\widetilde{\Lambda}^{[\pm2]}$.
\end{remark}

\begin{remark}
\label{stillhold}
Let us note finally that if, for a fixed frequency triple $(\omega, m, \widetilde{\Lambda})$, 
$u$ is simply assumed to be a smooth 
solution of the ODE $(\ref{basic})$ where $\lambda_{m\ell}^{[s]}(a\omega)$ is replaced by the quantity defined
by $\widetilde{\Lambda}-s-(a\omega)^2-4|a\omega|$ in view of $(\ref{workingdef})$, 
and $P$, $\Psi$, $\psi$ are \underline{defined} by relations 
$(\ref{plustwofwoflr1})$, $(\ref{plustwofwoflr2})$, $(\ref{minustwofwoflr1})$,
$(\ref{minustwofwoflr2})$,
then the identities of Proposition~\ref{psieq} again hold.
\end{remark}

\section{Frequency-localised estimates in $r\in[A_1,A_2]$}
\label{ODEmegasec}

The present section deals entirely with the system of relations satisfied
by 
\[
u^{(a\omega)}_{m\ell}, \qquad \psi^{(a\omega)}_{m\ell}, \qquad \Psi^{(a\omega)}_{m\ell}
\]
at fixed frequency in the region $r\in[A_1,A_2]$, for given inhomogeneous terms. 
The main result will be {\bf Theorem~\ref{phaseSpaceILED}},
stated in {\bf Section~\ref{sectionforstatement}}, 
which can be thought of as a fixed frequency
version of an integrated local energy estimate for all quantities near trapping,
with boundary terms  ${\rm Q}(A_i)$
which will eventually cancel the boundary terms appearing
on the right hand side of Proposition~\ref{multpropnofreq}
of Section~\ref{Physspacesecnew}.

We shall prove multiplier estimates for $(\ref{rwPsi2})$ in {\bf Section~\ref{multestsec}}
and
transport estimates for $(\ref{plustwofwoflr1})$--$(\ref{minustwofwoflr2})$ 
in {\bf Section~\ref{transestsec}}.
Together with an integration by parts argument, the transport estimates will allow
us to bound in {\bf Section~\ref{completing}} 
the inhomogeneous terms on the right hand side of
$(\ref{rwPsi2})$ arising from the coupling of the Regge--Wheeler
equation for $\Psi^{(a\omega)}_{m\ell}$ with $u^{(a\omega)}_{m\ell}$ and 
$\psi^{(a\omega)}_{m\ell}$, 
thus will allow to complete the proof of Theorem~\ref{phaseSpaceILED}

Just like with the analogous Theorem~8.1 of~\cite{partiii},
the results of this section can be understood as results about ODE's,
independently of the particular framework of Section~\ref{Separationmegasec}.
We have thus tried to give as self-contained a statement as possible.

\vskip1pc

\subsection{Statement of Theorem~\ref{phaseSpaceILED}: the main fixed frequency estimates}
\label{sectionforstatement}

In the present section  we consider the 
coupled system of ODEs 
satisfied by $u$, $\psi$ and $\Psi$ and state a fixed frequency analogue of local integrated
energy decay, in the region $r\in[A_1,A_2]$ near trapping.

\subsubsection{Frequency localised norms}
\label{frequencylocalnormssec}

Before formulating the theorem, we define certain energy norms.

In view of Remark~\ref{stillhold},
the natural setting of the theorem
refers only to  an admissible frequency triple $(\omega, m, \widetilde{\Lambda})$
(cf.~Definition~\ref{admiss.freq.def})
and associated
solutions $u^{[\pm2]}$
of~\eqref{basic} on $[A_1,A_2]$ 
 and $\psi^{[\pm2]}$, $\Psi^{[\pm2]}$
defined by 
$(\ref{plustwofwoflr1})$--$(\ref{minustwofwoflr2})$,
where $\lambda_{m\ell}^{[s]}(a\omega)$ is replaced by the quantity defined
by $\widetilde{\Lambda}-s-(a\omega)^2-4|a\omega|$ in view of $(\ref{workingdef})$.
Recall that all derived ordinary differential identities follow, in particular $(\ref{RWtypeinthebulk})$,
as does  the estimate $(\ref{eli1e})$ of Section~\ref{sec:swest}.
In practice, of course, we will always apply this for $u^{[\pm2]}$ equal
to  $u^{[\pm2],(a\omega)}_{m\ell}$ and $\widetilde{\Lambda}$ equal 
to $\widetilde{\Lambda}^{[s],(a\omega)}_{m\ell}$.

Given the above,
let us define the quantities
\begin{align*}
\|\mathfrak{d} \Psi^{[\pm2]}\|^2 &= 
\int_{A_1^*}^{A_2^*}\left[ \left|(\Psi^{[\pm2]})'\right|^2 + \left( \left(1-r^{-1}r_{\rm trap}\right)^2\left(  \omega^2 + \widetilde\Lambda\right) + 1\right)\left|(\Psi^{[\pm2]})\right|^2\right]\, dr^*,
\\
\|\mathfrak{d}\psi^{[\pm2]}\|^2 &=\int_{A_1^*}^{A_2^*}(\omega^2+m^2+1)|\psi^{[\pm2]}|^2 dr^*,
\\
\|\mathfrak{d}u^{[\pm2]}\|^2 &=\int_{A_1^*}^{A_2^*}(\omega^2+m^2+1)|u^{[\pm2]}|^2 dr^*,
\end{align*}
as well as the boundary energies for $i=1,2$:
\[
\|\mathfrak{d}\psi^{[\pm2]}\|^2(A_i)=
(\omega^2+m^2+1)|\psi^{[\pm2]}(A_i)|^2, \qquad
\|\mathfrak{d}u^{[\pm2]}\|^2(A_i)=
(\omega^2+m^2+1)|u^{[\pm2]}(A_i)|^2.
\]
In the above, $r_{\rm trap}$ is a parameter depending on $M$, $a$ and
the frequency triple $(\omega, m, \widetilde{\Lambda})$ to be determined later.
For ``trapped'' frequencies, we will have $r^*_{\rm trap} \in [A_1^*/4,A_2^*/4]$,
but it will be important that in various high frequency but untrapped frequency
ranges,
we can take $r_{\rm trap}=0$.

Note that since this is a region 
of fixed finite $r$, bounded away from infinity and the horizon, 
no $r$-weights or $\Delta$-factors need  appear in the above norms.

Finally, it will be convenient if we introduce the alternate
notation 
\[
A_+:=A_1, \qquad A_-:=A_2
\]
which will be useful when
referring to boundary terms in contexts where the choice
of term depends on the spin.

\subsubsection{Statement of the theorem}

\begin{theorem}\label{phaseSpaceILED}
Given  $0\le a_0 \ll M$ sufficiently small,
then the following is true.

Let $0\le a\le a_0$ and let $(\omega, m, \widetilde{\Lambda})$ be an admissible frequency 
triple.  Let $E>1$ be the parameter fixed after
Proposition~\ref{multpropnofreq}.
Given a parameter $\delta_1<1$,
let $f_0$,  $y_0$  be defined by $(\ref{newdefshere1})$ and $(\ref{newdefshere1})$ as in the proof
of Proposition~\ref{multpropnofreq}.

Then one can choose sufficiently small $\delta_1<1$
depending only on $M$, and
functions $f$,  $y$ and
an $r$-value $r_{\rm trap}$,
depending on the parameters $a$, $M$
and the frequency triple
$(\omega, m, \widetilde{\Lambda})$ but
satisfying the uniform bounds
\begin{equation}
\label{rtrapunifm}
r_{\rm trap}=0 {\rm \ or\ } r_{\rm trap}\in [A^*_1/4,A^*_2/4]
\end{equation}
\begin{equation}
\label{unif.f.bnds}
\left|f\right| +\left|f'\right| + \left|y\right|  \lesssim 1,
\end{equation}
\begin{equation}
\label{theymatch}
f = f_0(r), \qquad y=y_0(r)\qquad \text{ for }r^*\in [A_1^*/2,A_2^*/2]^c,
\end{equation}
such that, for all smooth solutions $u^{[\pm2]}$ of \eqref{basic}
on $[A_1,A_2]$
and associated $\psi^{[\pm2]}$ and $\Psi^{[\pm2]}$,
then
\begin{align}\label{fromPhaseSpace2}
\|\mathfrak{d} \Psi^{[\pm2]}\|^2
&\lesssim \mathfrak{H}^{[\pm 2]}
 +{\rm Q}(A_2)-{\rm Q}(A_1) +
|a| \sum_{i=1}^2 ( \|\mathfrak{d}\psi^{[\pm2]}\|^2(A_i)
 + \|\mathfrak{d}u^{[\pm2]}\|^2(A_{i})),
\end{align}
\begin{align}\label{fromPhaseSpaceoquant}
\nonumber
&\|\mathfrak{d}\psi^{[\pm2]}\|^2(A_{\mp})
 + \|\mathfrak{d}u^{[\pm2]}\|^2(A_{\mp})+
 \|\mathfrak{d}\psi^{[\pm2]}\|^2
+\|\mathfrak{d}u^{[\pm2]}\|^2 \\
&\qquad 
\lesssim   \mathfrak{H}^{[\pm 2]}+{\rm Q}(A_2)-{\rm Q}(A_1)
+ \|\mathfrak{d}\psi^{[\pm2]}\|^2(A_{\pm})
 + \|\mathfrak{d}u^{[\pm2]}\|^2(A_{ \pm}),
\end{align}
where 
\[
\mathfrak{H}^{[\pm2]}=\int_{A_1^*}^{A_2^*} {\mathfrak G}^{[\pm2]} \cdot  (f, y, E) \cdot (\Psi^{[\pm2]}, {\Psi^{[\pm2]}}')\, dr^* ,
\]
\begin{align}
{\mathfrak G}^{[\pm2]} \cdot  (f, y, E) \cdot (\Psi^{[\pm2]}, {\Psi^{[\pm2]}}') \doteq &-2f{\rm Re}\left({\Psi^{[\pm2]}}'\overline{{\mathfrak G}^{[\pm2]} }\right) -f'{\rm Re}\left(\Psi^{[\pm2]}\overline{{\mathfrak G}^{[\pm2]} }\right) -2y\text{Re}\left({\Psi^{[\pm2]}}'\overline{{\mathfrak G}^{[\pm2]} }\right) \notag\\
\label{another.def.here}
 &\qquad-E\omega{\rm Im}\left({\mathfrak G}^{[\pm2]} \overline{\Psi^{[\pm2]}}\right),
\end{align}
and ${\rm Q}$ is given by $(\ref{totalcurrent})$.
\end{theorem}

\subsection{Multiplier estimates for $\Psi^{[\pm2]}$}
\label{multestsec}

We begin in this section with frequency localised
bounds for $\Psi^{[\pm2]}$.
Frequency localisation is necessary to capture trapping,
in the style of our previous~\cite{DafRodsmalla}. The multipliers will be
frequency independent at $r=A_1$ and $r=A_2$  and will in fact
match exactly those  applied in Section~\ref{condmultestsec}.
This is ensured by $(\ref{theymatch})$.
As a result, in the setting of Section~\ref{iledsec},
the boundary terms ${\rm Q}(A_i)$ which will appear below,
after summation over frequencies, will exactly cancel the terms
$\mathbb Q(A_i)$ appearing in  Proposition~\ref{multpropnofreq}.

Recall the quantity
$\|\mathfrak{d} \Psi^{[\pm2]}\|^2$ defined in Section~\ref{frequencylocalnormssec}.
The main result of the section is the following:
\begin{proposition}
\label{multestforPsi}
With the assumptions of Theorem~\ref{phaseSpaceILED}, 
we have
\begin{equation}
\label{tothelefths0}
\|\mathfrak{d} \Psi^{[\pm2]}\|^2  \lesssim \mathfrak{H}^{[\pm2]}+  \mathcal{K}^{[\pm2]}  +
{\rm Q}(A_2)-{\rm Q}(A_1)
\end{equation}
where $\mathcal{K}^{[\pm2]}$ is defined by
\[
\mathcal{K}^{[\pm2]}=\int_{A_1^*}^{A_2^*} {\mathcal J}^{[\pm2]} \cdot  (f, y, E) \cdot (\Psi^{[\pm2]}, {\Psi^{[\pm2]}}')\, dr^* ,
\]
where 
\begin{align} \label{Kexpl}
{\mathcal J}^{[\pm2]} \cdot  (f, y, E) \cdot (\Psi^{[\pm2]}, {\Psi^{[\pm2]}}') \doteq &-2f\text{Re}\left({\Psi^{[\pm2]}}'\overline{{\mathcal J}^{[\pm2]} }\right) -f'\text{Re}\left(\Psi^{[\pm2]}\overline{{\mathcal J}^{[\pm2]} }\right) -2y\text{Re}\left({\Psi^{[\pm2]}}'\overline{{\mathcal J}^{[\pm2]} }\right) \notag\\ &\qquad -E\omega\text{Im}\left({\mathcal J}^{[\pm2]} \overline{\Psi^{[\pm2]}}\right)
\end{align}
and
${\rm Q}$ is given by $(\ref{totalcurrent})$.
\end{proposition}

The  estimate  above   differs from the estimate for
$\|\mathfrak{d} \Psi^{[\pm2]}\|^2$
given by  $(\ref{fromPhaseSpace2})$ 
as it is
still coupled with  $u^{[\pm2]}$ and $\psi^{[\pm2]}$
in view of the presence of the term $\mathcal{K}^{[\pm2]}$.
We will be able to replace $\mathcal{K}^{[\pm2]}$
with  $\mathfrak{H}^{[\pm2]}$ and the additional
boundary term  $|a|\| \mathfrak{d}\psi^{[\pm2]}\|^2(A_{ \pm})
 + |a|\| \mathfrak{d}u^{[\pm2]}\|^2(A_{ \pm})$
 appearing in $(\ref{fromPhaseSpace2})$
 in   Section~\ref{completing}.

\begin{proof}

The estimate $(\ref{tothelefths0})$ will be proven by using multiplier identities. 
The relevant frequency-localised current templates, corresponding precisely
to the physical space multiplier identities used in Section~\ref{condmultestsec},
 will be defined in Section~\ref{templatesec} below.
For a specific combination of these currents,  the bulk term will control the integrand
of the left hand side of $(\ref{tothelefths0})$ whereas the boundary terms (after summation
over frequencies)
will correspond precisely to the boundary terms of  Proposition~\ref{multpropnofreq}. 
This coercivity is stated as 
Proposition~\ref{positivprop}
in Section~\ref{totalcurrentpospropsec}. 
The precise choice of the functions $f$ and $y$ will be frequency
dependent and is carried out separately for the frequency ranges $\mathcal{G}_1$
and $\mathcal{G}_2$ in Sections~\ref{gonerangesec} and~\ref{gtworangesec} respectively.

{\bf \emph{In the rest of this subsection, we will always write $\Psi$ in the place
of $\Psi^{[\pm2]}$, as the choice of the multipliers will not depend on the spin. 
We will write $\mathcal{V}$ in place of $\mathcal{V}^{[\pm2]}$, 
and $\widetilde{\Lambda}$ for ${\widetilde{\Lambda}}^{[\pm2]}$,
remembering
that the dependence of $\mathcal{V}^{[\pm2]}$ on the spin in the
context of the separation is completely
contained in the different definition of ${\widetilde{\Lambda}}^{[\pm2]}$; see
formula $(\ref{rewritepot})$.}}
We will only refer explicitly to $s=\pm2$ when discussing the inhomogeneous
terms on the right hand side of $(\ref{rwPsi2})$.

\subsubsection{The frequency-localised multiplier current templates}
\label{templatesec}

Let us define the frequency localised multiplier currents 
which correspond to the physical space multipliers of Section~\ref{condmultestsec}:
\begin{eqnarray*}
{\rm Q}^f[\Psi]&=&f\left(|\Psi'|^2+(\omega^2-\mathcal{V})|\Psi|^2\right)+f^\prime{\rm Re}\left(\Psi'\bar{\Psi}\right)-\frac12 f''|\Psi|^2,\\
{\rm Q}^y[\Psi]&=& y\left(|\Psi'|^2+(\omega^2-\mathcal{V})|\Psi|^2\right),\\
{\rm Q}^T[\Psi]&=& \omega{\rm Im}(\Psi' \bar{\Psi}).
\end{eqnarray*}

If $\Psi$ satisfies
\[
\Psi '' +\mathcal{V}\Psi = H
\]
for an admissible frequency triple $(\omega, m, \widetilde{\Lambda})$,
then, since $\mathcal{V}$ is real, 
we have
\begin{align*}
({\rm Q}^f[\Psi])' &= 2f' |\Psi'|^2-f\mathcal{V}' |\Psi|^2 -\frac12f'''|\Psi|^2 &
+{\rm Re}(2f\bar{H}\Psi ' +f'\bar{H} \Psi),\\
({\rm Q}^y[\Psi])' &= y'(|\Psi'|^2+(\omega^2-\mathcal{V})|\Psi|^2)-y\mathcal{V}'|\Psi|^2&+2y{\rm Re}(\bar{H}\Psi'),\\
({\rm Q}^T[\Psi])' &= &\omega {\rm Im}(H\bar\Psi).
\end{align*}

Let us remark already that if $\upalpha$ is an $[A_1,A_2]$-admissible
solution of the inhomogeneous Teukolsky equation $(\ref{inhomoteuk})$,
such that the restriction of $\upalpha$ to $r\in[A_1,A_2]$ is  supported
in $t=t^*=\tilde{t}^*\in (\tau_1,\tau_2)$, then the identity corresponding to applying
\[
\int d\omega \sum_{m\ell}
\]
to
\[
{\rm Q}^f(A_1)+\int_{A_1^*}^{A_2^*} ({\rm Q}^f)'(r^*) dr^*={\rm Q}^f(A_2),
\]
resp.~with ${\rm Q}^y$, ${\rm Q}^T$,
yields precisely the
identities of Section~\ref{sec:multid} applied in the region 
$\widetilde{\mathcal{R}}^{\rm trap}(\tau_1, \tau_2)$. 
(Note that by our choices from Section~\ref{very.slowly.very.slowly},
we have $T=T+\upomega_+\chi\Phi$ in this region, and note
moreover that the boundary terms on $\tilde{t}^*=\tau_i$ vanish by the restriction on the
support.)

\subsubsection{The total current ${\rm Q}$ and its coercivity properties}
\label{totalcurrentpospropsec}

For all frequencies, we will apply the identity corresponding to a current
of the form 
\begin{equation}
\label{totalcurrent}
{\rm Q}= {\rm Q}^f+ {\rm Q}^y
+ E{\rm Q}^T,
\end{equation}
for appropriate choices of functions $f$, $y$.
The coercivity statement is given by the following:

\begin{proposition}
\label{positivprop}
Let $E$ and $f_0$ be as fixed in the proof of Proposition~\ref{multpropnofreq}.
Then one can choose $\delta_1<1$ sufficiently small, depending only on $M$,
such that the following is true:

There exist  functions $f$ and $y$ and a parameter $r_{\rm trap}$
depending on the parameters $a$, $M$
and the frequency triple
$(\omega, m, \widetilde{\Lambda})$, 
satisfying 
$(\ref{rtrapunifm})$, $(\ref{unif.f.bnds})$ and 
$(\ref{theymatch})$ 
and
such that ${\rm Q}$ defined by $(\ref{totalcurrent})$
satisfies
\begin{align}
\label{pwiseposit}
\nonumber 
&  \left|\Psi'\right|^2 +  \left(\left(1-r_{\rm trap}r^{-1}\right)^2\left(\omega^2 + \widetilde\Lambda\right) + 1\right)\left|\Psi\right|^2
\\
&\qquad \lesssim {\rm Q}'  - \mathcal{J}^{[\pm2]}\cdot (f,y,E)\cdot (\Psi, \Psi') 
-\mathfrak{G}^{[\pm2]}\cdot (f,y,E)\cdot (\Psi, \Psi').
\end{align}
\end{proposition}

\begin{proof}
See Sections~\ref{gonerangesec} and~\ref{gtworangesec}.
\end{proof}

Let us note that integrating the equation
\[
{\rm Q}(A_1) +\int_{A_1^*}^{A_2^*} {\rm Q}'(r^*) dr^*  = {\rm Q}(A_2)
\]
we infer from $(\ref{pwiseposit})$ the inequality $(\ref{tothelefths0})$.

\subsubsection{The $\mathcal{G}_1$ range}
\label{gonerangesec}

We define the range
\begin{equation}
\label{range1}
\mathcal{G}_1=
\{\widetilde\Lambda \ge c_{\flat} \omega^2 \} \cup \{ \widetilde\Lambda 
+\omega^2 +m^2 \le C_{\sharp}\}
\end{equation}
for some 
$0<c_{\flat}<1$  and $C_{\sharp}>1$ which can be chosen finally
to  depend
only on $M$.
The frequency range $\mathcal{G}_1$
includes thus ``angular-dominated frequencies''
$\widetilde{\Lambda}\gg \omega^2$, 
``trapped frequencies'' $\widetilde{\Lambda}\sim  \omega^2$
and ``low frequencies'' $\widetilde\Lambda 
+\omega^2 +m^2 \lesssim 1$.
We have the following:

\begin{proposition}
\label{addedpropos}
For sufficiently small $|a|<a_0\ll M$, then for all frequency triples
in $\mathcal{G}_1$, 
there exists a function $f$ and a parameter $r_{\rm max}$ 
with the following properties for $r^*\in [A_1^*,A_2^*]$:
\begin{enumerate}
\item
$f=f_0$ for $r^*\in[A_1^*/2,A_2^*/2]^c$ and 
$|f|\lesssim 1$, $|f'| \lesssim 1$ in $[A_1^*,A_2^*]$,
\item
$|r_{\rm max}-3M | \le c(a,M) $ with $c(a, M)\to 0$ as $a\to 0$,
in particular $a_0$ can be chosen so that $r_{\rm trap}^*\in [A_1^*/4,A_2^*/4]$;
for $m=0$, $r_{\rm max}$ is independent of $\omega$ and $\widetilde{\Lambda}$,
\item
$f' \gtrsim 1 $,
\item
$-f\mathcal{V}'-\frac12 f'' \gtrsim \left( \widetilde\Lambda (1-r_{\rm max}r^{-1})^2 +1 \right)$.
\end{enumerate}
\end{proposition}

\begin{proof}
Let  $\mathcal{V}_{Schw}^{[\pm2]}$ denote the potential $\mathcal{V}$ of
$(\ref{Vdef})$ in the $a=0$ Schwarzschild case. 
Writing this potential as in $(\ref{Vdef})$
as 
\[
\mathcal{V}_{Schw}=(\mathcal{V}_{Schw})_0+(\mathcal{V}_{Schw})_1,
\]
we see easily that $(\mathcal{V}_{Schw})_0$ has a unique maximum
at $r=3M$, while
\[
f_0'\gtrsim r(r-2M) r^{-4}, \qquad -f_0\mathcal{V}_{Schw}'-\frac12f_0''' 
\gtrsim c r(r-2M) \left( \frac{(r-3M)^2}{r^2}\ell(\ell+1)+1 \right) r^{-5},
\]
so in particular, in the region $r^*\in [A_1^*,A_2^*]$, we have
\[
f_0'\gtrsim 1 ,\qquad  -f_0\mathcal{V}_{Schw}'-\frac12f_0''' \gtrsim (1-3M/r)^2\ell(\ell+1) +1.
\]

We begin with a lemma concerning the behaviour of the potential
$\mathcal{V}$  in the $\mathcal{G}_1$ frequency range.

\begin{lemma}
\label{lem.max}
Let $0<c_{\flat}<1$ and $C_\sharp>1$ be arbitrary.
For sufficiently small $|a|<a_0\ll M$, then
for all frequency triples in the range $\mathcal{G}_1$,
the potential $\mathcal{V}_0$ of $(\ref{Vdef})$
has a unique maximum $r_{\rm max}$ satisfying property 2.~and
\begin{equation}
\label{lemmasta}
(r-r_{\rm max})^{-1}\mathcal{V}'_0 \gtrsim \widetilde{\Lambda}
\end{equation}
in $[A_1,A_2]$. If $m=0$, then $r_{\rm max}$ is manifestly independent
of $\omega$ and $\widetilde{\Lambda}$.
\end{lemma}

\begin{proof}
This is an easy computation in view of $(\ref{rewritepot})$.
For the region $\mathcal{G}_1\setminus \{\widetilde\Lambda 
+\omega^2 +m^2 \le C_{\sharp} \}$, one
 uses the bound 
\[
\widetilde{\Lambda}-4|a\omega| \geq \frac{1}{2}\widetilde{\Lambda} + \frac{1}{4}c_{\flat} \omega^2
\geq \frac14 \widetilde{\Lambda}+\frac14 c_{\flat}\omega^2 +\frac1{16}m^2
\qquad {\rm\ in\ } \qquad \mathcal{G}_1\setminus \{\widetilde\Lambda 
+\omega^2 +m^2 \le C_{\sharp} \}
\]
 and the smallness of $a$. For the region
 $\{\widetilde\Lambda 
+\omega^2 +m^2 \le C_{\sharp} \}$
it suffices to use the
general
bound $\widetilde{\Lambda}\ge 1$ and the smallness of $a$.
 Notice that according to our conventions,
 the constant in the $\gtrsim$ indeed only depends on $M$,
 since smallness of $a$ can be used to absorb the $c_{\flat}$ and $C_\sharp$ 
 dependence.
\end{proof}

Let $\chi(r^*)$ be a cutoff function such that
$\chi=1$ in $[A_1^*/4,A_2^*/4]$ and $\chi =0$  in $[A_1^*/2,A_2^*/2]^c$. 
We define now
\begin{equation}
\label{simple.def.of.f}
f= \left(1-\frac{3M +\chi(r^*)({r}_{\rm max}-3M)}r\right)\left(1+\frac{M}r\right) .
\end{equation}
This function obviously satisfies property 1.~and is easily seen to satisfy
property 3.

It remains to show property 4.
By
$(\ref{lemmasta})$ and the definition of $f$ we have
\[
-f\mathcal{V}_0' \gtrsim \widetilde\Lambda (1-r_{\rm max}r^{-1})^2 \, .
\]
On the other hand, for $|a|\ll a_0<M$ sufficiently small,
we have that $|f_0'''-f''' |\le c(a)$, and thus
\[
-f\mathcal{V}_0' -\frac12f''' \gtrsim (\widetilde\Lambda (1-r_{\rm max}r^{-1})^2+1).
\]
Finally, we note that $\mathcal{V}=\mathcal{V}_0+\mathcal{V}_1+\mathcal{V}_2$,
and   we
have $|\mathcal{V}_1-(\mathcal{V}_{Schw})_1|\le c(a)$, 
$|\mathcal{V}_2|\le c(a)$ 
 with $c(a)\to0$.
 
 We have 
  \[
  -f\mathcal{V} '-\frac12f''' =
  -f\mathcal{V}_0'-\frac12f''' 
  -f(\mathcal{V}_{Schw})_1' 
+f(\mathcal{V}_1'-(\mathcal{V}_{Schw})_1' )
-f\mathcal{V}_2'
\]
It follows readily that property 4.~indeed holds
for frequencies in $\mathcal{G}_1$.
\end{proof}

Now, given a parameter $\delta_1<1$, we define the function
\begin{equation}
\label{y1def.here}
y_1= \delta_1((1-\chi)f+\chi f^3)).
\end{equation}
Note that this function satisfies 
$(\ref{theymatch})$. We compute
\begin{equation}
\label{positivityofy1}
y'_1 =  \delta_1( (1-\chi)f' + 2\chi f^2 f'  - \chi' f +\chi' f^3) \gtrsim  \delta_1 (r-r_{\rm max})^2
\end{equation}
where we are using also that $|f|\le 1$ implies that $|f^3|\le |f|$.

Note on the other hand that for sufficiently small $|a|<a_0\ll M$, we have
\[
|\mathcal{V} |\lesssim \widetilde{\Lambda}+1  \, ,\qquad  |\mathcal{V}'| \lesssim 
\widetilde{\Lambda}+1
\]
in $r^*\in[A_1^*,A_2^*]$ for all frequencies in $\mathcal{G}_1$,
in view of the 
general bound
\begin{equation}
\label{gen.bound.all.freq}
\frac{1}{4}m^2 + 1  \le \widetilde{\Lambda}
\end{equation}
and the bound 
\[
\omega^2 \le c_{\flat}^{-1} \widetilde{\Lambda} +C_\sharp,
\]
which holds in $\mathcal{G}_1$.
Thus
\[
y'_1 \mathcal{V} - y_1\mathcal{V}' \lesssim  \delta_1(\widetilde{\Lambda}(
1-r_{\rm max}r^{-1})^2+1).
\]

It follows that we may choose $\delta_1$ sufficiently small
so as for
\begin{equation}
\label{almostcoer}
-f\mathcal{V}'-\frac12f''' -y_1'\mathcal{V} +y_1\mathcal{V}' 
+y_1'\omega^2
\gtrsim (\widetilde{\Lambda}+\delta_1 \omega^2) (1-r_{\rm max}r^{-1})^2 +1.
\end{equation}
{\bf Henceforth, $\delta_1$ will be fixed.} 
In particular, according to our conventions,
we may replace the $\delta_1$ factor by $1$ on the right hand side of
$(\ref{almostcoer})$. 

In view of $(\ref{almostcoer})$ and $(\ref{positivityofy1})$,
examining the identities of Section~\ref{templatesec}, 
we have obtained the degenerate coercivity of $({\rm Q}^f+{\rm Q}^{y_1})'$.

We would like to improve this coercivity in the ``angular-dominated'' subrange
of $\mathcal{G}_1$. 
Let us now introduce a new parameter
$C_{\flat} \gg 1$ and consider the range
\begin{equation}
\label{smaller.range}
\mathcal{G}_1\cap \{ \widetilde{\Lambda} \ge C_{\flat} \omega^2\}.
\end{equation}
Noting that we have 
\[
\mathcal{V} \gtrsim \widetilde{\Lambda} +1 
\]
in $\mathcal{G}_1$,
it follows that 
for $C_{\flat}$ sufficiently large,  we have
\[
\mathcal{V} -\omega^2 \gtrsim \mathcal{V} \gtrsim \widetilde{\Lambda} \gtrsim 
\widetilde{\Lambda}+\omega^2
\]
in $(\ref{smaller.range})$. {\bf Henceforth, $C_{\flat}$ will be fixed}.
We may now define a new small parameter 
$\delta_3>0$  and define a function
\[
y_2= \delta_3 (r_{\rm max}-r^*)\chi,
\]
where $\chi$ is the cutoff from above.
We have that for frequency triples in $(\ref{smaller.range})$,
\[
y_2'(\omega^2 -\mathcal{V})\gtrsim \delta_3, \qquad -y_2\mathcal{V}' \le \delta_3
 (\widetilde{\Lambda}(1-r_{\rm max}r^{-1})^2 +1)
\]
in $[A_1^*/4,A_2^*/4]$,
while 
\[
y_2'\mathcal{V}-y_2\mathcal{V}' \le \delta_3 (\widetilde{\Lambda}(1-r_{\rm max}r^{-1})^2 +1),
\qquad |y'_2| \lesssim \delta_3
\]
in $[A_1^*, A_2^*]$.
In particular, we may choose $\delta_3$ sufficiently small,  with
the smallness requirement depending only on $M$,
so that, defining
\begin{equation}
\label{ythesum}
y=y_1 +y_2,
\end{equation}
we have 
\begin{equation}
\label{precoer}
2f'+y' \gtrsim 1,\qquad 
-f\mathcal{V}'-\frac12f''' -y'\mathcal{V} +y\mathcal{V}' +\omega^2 y'
\gtrsim  (\widetilde{\Lambda}+\omega^2)(\delta_3+  (1-r_{\rm max}r^{-1})^2) +1
\end{equation}
in $(\ref{smaller.range})$. {\bf Henceforth, $\delta_3$ will be fixed.}

We are ready now for our final definitions.
In the range $(\ref{smaller.range})$, we define $y$ by
$(\ref{ythesum})$.
Since $\delta_3$ is now fixed we may now write
\[
(\delta_3+(1-r_{\rm max}r^{-1})^2) \gtrsim 1.
\]
We thus can set $r_{\rm trap}=0$.

For the remaining frequencies in $\mathcal{G}_1$, i.e.~for frequencies
in $\mathcal{G}_1 \cap \{\widetilde{\Lambda} < C_{\flat}\omega^2\}$,
we define  simply $y=y_1$
and $r_{\rm trap} = r_{\rm max}$.

Finally, we consider  the current 
\[
E{\rm Q}^T
\]
for $E$ the parameter fixed in Section~\ref{condmultestsec}.

Thus, applying the identity corresponding
to $(\ref{totalcurrent})$
in view of $(\ref{positivityofy1})$, 
$(\ref{almostcoer})$ and $(\ref{precoer})$,
we obtain that Proposition~\ref{positivprop} 
holds for all  frequencies in $\mathcal{G}_1$.

\subsubsection{The $\mathcal{G}_2$ range}
\label{gtworangesec}

We define this frequency range to be the complement of $\mathcal{G}_1$, i.e.
\begin{equation}
\label{range2}
\mathcal{G}_2 =\{\omega^2 > c_{\flat}^{-1} \widetilde\Lambda \} \cap \{\widetilde\Lambda +\omega^2+m^2 
 > C_\sharp\}.
\end{equation}
These are the ``time-dominated'' large frequencies.

We may choose $c_{\flat}$ sufficiently small, and $C_\sharp$ sufficiently large,
so that for sufficiently small $|a|<a_0\ll M$, we have
\begin{equation}
\label{havesandhavenots}
\omega^2-\mathcal{V} \ge \frac12\omega^2, \qquad
|\mathcal{V}'| \le \frac12\omega^2\qquad
 {\rm\ in\ }\mathcal{G}_2
\end{equation}
{\bf Henceforth, $c_{\flat}$ and $C_\sharp$ will be fixed by the above restriction.}
We note that it is certainly the case that $C_{\flat} \ge c_{\flat}$.

Consider the function $f_0 $ of the previous section.
We define simply $f=f_0$ for frequencies in $\mathcal{G}_2$.

Given the parameter $\delta_1$ fixed in Section~\ref{gonerangesec}, 
we define now $y=\delta_1 f$. 
It follows from $(\ref{havesandhavenots})$ that in the range $\mathcal{G}_2$ we have
\[
(2f' +y')\gtrsim 1, \qquad
-( f\mathcal{V}' +y\mathcal{V}') -\frac12 f''' 
+y' (\omega^2-\mathcal{V}) \gtrsim  \omega^2 
\gtrsim (\omega^2 +\widetilde{\Lambda}^2 +1 ).
\]
We may define thus
the parameter $r_{\rm trap}=0$ for the frequency range $\mathcal{G}_2$.

Finally, we may  again add  
\[
E{\rm Q}^T
\]
for $E$ the parameter fixed in Section~\ref{condmultestsec}.

Thus again applying the identity to $(\ref{totalcurrent})$ with the above
definitions we obtain that
Proposition~\ref{positivprop} holds for all frequencies in $\mathcal{G}_2$.

Since $\mathcal{G}_1\cup\mathcal{G}_2$ contains all admissible frequencies,
the results of this section together with Section~\ref{gonerangesec} 
imply that Proposition~\ref{positivprop}, and thus  $(\ref{tothelefths0})$, indeed holds.

The proof of Proposition~\ref{multestforPsi} is now complete.
\end{proof}

{\bf Let us recall that in the course of the above proof, we have
fixed the parameter
$\delta_1$. This allows us to fix also $\delta_2$ of Proposition~\ref{multpropnofreq}. 
Since $E$ has been fixed previously, it follows
that all dependences on parameters can be removed from the
$\lesssim$ in the
statement of Proposition~\ref{multpropnofreq}.}

\subsection{Transport estimates for $\psi^{[\pm2]}$ and $u^{[\pm2]}$}
\label{transestsec}

In this section we will prove frequency-localised versions for the transport
estimates of~\cite{holzstabofschw}
to obtain estimates for $u^{[+2]}$ and $\psi^{[+2]}$ from $\Psi^{[+2]}$
as well as for $u^{[-2]}$ and $\psi^{[-2]}$ from $\Psi^{[-2]}$,
localised in $r\in [A_1,A_2]$.

The main result of the section is:
\begin{proposition} \label{prop:synth}
With the assumptions of Theorem~\ref{phaseSpaceILED}, 
we have the following estimates:
\begin{align} \label{desi1}
 \|\mathfrak{d}\psi^{[\pm2]}\|^2(A_{ \mp})
 + \|\mathfrak{d}u^{[\pm2]}\|^2(A_{ \mp})+ \|\mathfrak{d}\psi^{[\pm2]}\|^2 +\|\mathfrak{d}u^{[\pm2]}\|^2 
\nonumber \\
 \lesssim \|\mathfrak{d} \Psi^{[\pm 2]}\|^2  + \|\mathfrak{d}\psi^{[\pm2]}\|^2(A_{ \pm})
 + \|\mathfrak{d}u^{[\pm2]}\|^2(A_{ \mp}).
\end{align}
\end{proposition}

\begin{proof}
We consider first the case $+2$ of (\ref{desi1}).
 
 Adding the identity arising from multiplying $(\ref{plustwofwoflr2})$
  by $r \sqrt{\Delta} \ \overline{\psi^{[+2]}}$ and its complex conjugate by 
 $r \sqrt{\Delta}\psi^{[+2]}$ leads
 after integration and applying Cauchy--Schwarz on the right hand side to the estimate
\begin{align} \label{yu1}
r |\sqrt{\Delta} \psi^{[+2]}|^2 \left(A_1^*\right) + \int_{A_1^*}^{A_2^*} dr^*  |\sqrt{\Delta} \psi^{[+2]}|^2 \lesssim \int_{A_1^*}^{A_2^*}dr^* |\Psi^{[+2]}|^2 +r |\sqrt{\Delta} \psi^{[+2]}|^2 \left(A_2^*\right)  \, .
\end{align}
Similarly, adding the identity arising from multiplying $(\ref{plustwofwoflr1})$ by $r\overline{u^{[+2]}}w$ and its complex conjugate by $r u^{[+2]}w$ leads after integration and applying 
Cauchy--Schwarz on the right hand side to the estimate
\begin{align} \label{yu2}
r |u^{[+2]}w|^2 \left(A_1^*\right) + \int_{A_1^*}^{A_2^*} dr^* \   |u^{[+2]}w|^2 \lesssim  \int_{A_1^*}^{A_2^*} dr^*  | \psi^{[+2]}|^2  +r |u^{[+2]}w|^2 \left(A_2^*\right)  \, .
\end{align}
Combining (\ref{yu1}) and (\ref{yu2}) yields (\ref{desi1}) 
without the $m^2$ and $\omega^2$ terms in the norms on the left. 

To obtain the estimate with the $m^2$ and $\omega^2$ terms we define the frequency ranges
\[
{\cal F}^{\sharp}=\{\omega^2\ge   \frac14C_{\flat}^{-1} m^2 \},\qquad 
{\cal F}^{\flat}=\{\omega^2<\frac14C_{\flat}^{-1}  m^2\}
\]
where  $C_{\flat}$ is the constant of Section~\ref{gonerangesec}.
In view of the general bound $(\ref{gen.bound.all.freq})$
which holds for all admissible frequencies, it follows that in the frequency
range ${\cal F}^{\flat}$, we have
\[
C_{\flat} \omega^2  < \frac14 m^2 \le  \widetilde{\Lambda} 
\]
and thus ${\cal F}^{\flat}$ is contained in the frequency range $(\ref{smaller.range})$.
It follows that $r_{\rm trap}=0$ for ${\cal F}^{\flat}$,
i.e.~these frequencies are \underline{not} ``trapped''.

Suppose  first that $(\omega,m)$ lie in the frequency range ${\cal F}^{\flat}$. 
Since $r_{\rm trap}=0$, we have
\begin{equation}
\label{non.deg.cont}
\int_{A_1^*}^{A_2^*}\left[  |(\Psi^{[\pm2]})'|^2+ ( \widetilde{\Lambda}^2+m^2+\omega^2+1)|\Psi^{[\pm2]}|^2\right]
dr^*
\lesssim 
\| \mathfrak{d}\Psi^{[\pm2]} \|^2.
\end{equation}
Multiplying  thus (\ref{plustwofwoflr2}) and (\ref{plustwofwoflr1}) by $m$ and $\omega$ and repeating the argument leading to (\ref{yu1}) and (\ref{yu2}) immediately leads to (\ref{desi1}). 

Suppose on the other hand
that  $(\omega, m)$ lie in the frequency range ${\cal F}^{\sharp}$. Here
we do not have the $m^2$ and $\omega^2$ in $(\ref{non.deg.cont})$ and thus
 we proceed as follows. Commuting (\ref{plustwofwoflr2}) by $\frac{d}{dr^*}$ leads to the identity
\begin{align}
\left(\frac{d}{dr^*} - i \omega + \frac{iam}{r^2+a^2} \right) \left(\sqrt{\Delta} \psi^{[+2]}\right)^\prime = -2 w \left(\Psi^{[+2]}\right)^\prime - 2w^\prime  \Psi^{[+2]} + 2r\frac{iam}{r^2+a^2} w \cdot \sqrt{\Delta}\psi^{[+2]} \, .
\end{align}
Multiplying this by $r \left(\sqrt{\Delta} \overline{\psi^{[+2]}}\right)^\prime$ and adding the complex conjugate multiplied by $r \left(\sqrt{\Delta} {\psi^{[+2]}}\right)^\prime$ we find, upon integration and using Cauchy--Schwarz on the right hand side, the estimate 
\begin{align} \label{yu3}
r \Big|\left(\sqrt{\Delta} \psi^{[+2]}\right)^\prime\Big|^2 \left(A_1^*\right) + \int_{A_1^*}^{A_2^*} dr^* \   \Big|\left(\sqrt{\Delta} \psi^{[+2]}\right)^\prime\Big|^2 
\lesssim \ & r \Big|\left(\sqrt{\Delta} \psi^{[+2]}\right)^\prime\Big|^2 \left(A_2^*\right) +  \|\mathfrak{d} \Psi^{[\pm2]}\|^2 \nonumber \\
&+  \int_{A_1^*}^{A_2^*} dr^* a^2 m^2 |\sqrt{\Delta} \psi^{[+2]}|^2    \, .
\end{align}
Using the pointwise relation (\ref{plustwofwoflr2}) and the definition of the norm $ \|\mathfrak{d} \Psi^{[\pm2]}\|$ (as well as  the simple fact that for $i=1,2$
$
| \Psi^{\pm 2}|^2 \left(A_i^*\right) \lesssim \|\mathfrak{d} \Psi^{[\pm2]}\|^2 
$),
the estimate (\ref{yu3}) is also valid replacing on the left hand side $\Big|\left(\sqrt{\Delta} \psi^{[+2]}\right)^\prime\Big|^2$ by $\Big|\underline{L} \left(\sqrt{\Delta} \psi^{[+2]}\right)\Big|^2=|w\Psi^{[+2]}|^2$. Using the relation (\ref{Lbardefin}) we therefore deduce
\begin{align} \label{yu4}
\left(\omega - \frac{am}{r^2+a^2} \right)^2  \Big| \sqrt{\Delta} \psi^{[+2]}\Big|^2 \left(A_2^*\right) + \int_{A_1^*}^{A_2^*} dr^* \left(\omega - \frac{am}{r^2+a^2} \right)^2  \Big| \sqrt{\Delta} \psi^{[+2]}\Big|^2 \nonumber \\
 \lesssim  \|\mathfrak{d} \Psi^{[\pm2]}\|^2 + \int_{A_1^*}^{A_2^*} dr^* a^2m^2  |\sqrt{\Delta} \psi^{[+2]}|^2 + \left(\omega - \frac{am}{r^2+a^2} \right)^2  \Big| \sqrt{\Delta} \psi^{[+2]}\Big|^2 \left(A_1^*\right)   \, .
\end{align}

In the range ${\cal F}^{\sharp}$,  restricting to sufficiently small
$|a|<a_0\ll M$, we have
that
\[
\omega^2 \lesssim \left(\omega - \frac{am}{r^2+a^2} \right)^2   \lesssim \omega^2.
\]
It follows that in the inequality $(\ref{yu4})$, 
we can replace  the factor in
round bracket on the left hand side
simply by $\omega^2$ and absorb the second term on the right by the left hand side. This establishes (\ref{desi1}) for the $\psi^{[+2]}$-norm on the left. We can now multiply (\ref{yu2}) by $m^2$ and $\omega^2$ and use the estimate just obtained for $\psi^{[+2]}$ to establish the estimate (\ref{desi1}) also for the $u^{[+2]}w$-term. The proof of (\ref{desi1}) is now complete.

To prove (\ref{desi1}) for $s=-2$ 
one follows the identical argument but choosing the multiplier $\frac{1}{r}$ instead of $r$.
\end{proof}

\subsection{Controlling the inhomogeneous term $\mathfrak{K}^{[\pm 2]}$ in Proposition \ref{multestforPsi}}
\label{completing}

\begin{proposition}
\label{newsyntheprop}
The term 
\[
\mathcal{K}^{[\pm2]}=\int_{A_1^*}^{A_2^*} {\mathcal J}^{[\pm2]} \cdot  (f, y, E) \cdot (\Psi^{[\pm2]}, {\Psi^{[\pm2]}}')\, dr^* 
\]
appearing in Proposition \ref{multestforPsi} satisfies
\begin{align} \label{ihen}
\big|\mathcal{K}^{[\pm2]}\big| \lesssim |a|\|\mathfrak{d} \Psi^{[\pm2]}\|^2 + |a|\|\mathfrak{d}\psi^{[\pm2]}\|^2  + |a|\|\mathfrak{d}u^{[\pm2]}\|^2 +  |a| \sum_{i=1}^2 ( \|\mathfrak{d}\psi^{[\pm2]}\|^2(A_i)
 + \|\mathfrak{d}u^{[\pm2]}\|^2(A_{i})).
\end{align}
\end{proposition}

\begin{proof}
Since $f$, $f^\prime$ and $y$ are all uniformly bounded we have by Cauchy--Schwarz:
\begin{align}\label{lkp}
\int_{A_1^*}^{A_2^*} \big|f\text{Re}\left({\Psi^{[\pm2]}}'\overline{{\mathcal J}^{[\pm2]} }\right)\big| + \big|f'\text{Re}\left(\Psi^{[\pm2]}\overline{{\mathcal J}^{[\pm2]} }\right)\big| + \big| y\text{Re}\left({\Psi^{[\pm2]}}'\overline{{\mathcal J}^{[\pm2]} }\right)\big| \nonumber \\
\lesssim |a|\|\mathfrak{d} \Psi^{[\pm2]}\|^2 + |a|\|\mathfrak{d}\psi^{[\pm2]}\|^2  + |a|\|\mathfrak{d}u^{[\pm2]}\|^2 \, .
\end{align}
 For the last remaining term, 
$
\int_{A_1^*}^{A_2^*} \omega\text{Im}\left({\mathcal J}^{[\pm2]} \overline{\Psi^{[\pm2]}}\right)
$, we observe that we only need to estimate
\begin{align} 
\label{two.terms.here}
\Big| \int_{A_1^*}^{A_2^*} \mathfrak{c}\left(r\right) \text{Im}\left(im\psi^{[\pm 2]} \omega \overline{\Psi^{[\pm2]}}\right)\Big| \ \ \ \ \textrm{and} \ \ \ \ \Big| \int_{A_1^*}^{A_2^*}  \mathfrak{c}\left(r\right)  \text{Im}\left(imu^{[\pm 2]} \omega \overline{\Psi^{[\pm2]}}\right)\Big| \, ,
\end{align}
where $\mathfrak{c}\left(r\right)$ denotes a generic bounded real-valued function with uniformly bounded derivative in $\left[A^*_1,A^*_2\right]$ (whose explicit form may change in the estimates below). This is because the other terms appearing in $\mathcal{J}^{[\pm 2]}$ are again easily controlled via Cauchy--Schwarz  and satisfy the estimate (\ref{lkp}). We show how to estimate these terms for $s=+2$, the case $s=-2$ being completely analogous. 

For the first term of $(\ref{two.terms.here})$ we have
\begin{align}
\int_{A_1^*}^{A_2^*} \mathfrak{c}\left(r\right) \text{Im}\left(im\psi^{[+ 2]} \omega \overline{\Psi^{[+2]}}\right) = \int_{A_1^*}^{A_2^*} \mathfrak{c}\left(r\right) \text{Im}\left(m \overline{\Psi^{[+2]}} \left(- \underline{L} \psi^{[+2]} -  \left(\psi^{[+2]}\right)^\prime + \frac{iam}{r^2+a^2} \psi^{[+2]}  \right) \right) \nonumber \\
=  \int_{A_1^*}^{A_2^*}  \mathfrak{c}\left(r\right) \text{Im} \left(m\overline{\Psi^{[+2]}}  \psi^{[+2]} \right) + \mathfrak{c}\left(r\right) \text{Im} \left(\overline{\Psi^{[+2]}}  m \psi^{[+2]} \right)\Big|^{A_2^*}_{A_1^*} + \int_{A_1^*}^{A_2^*}  \mathfrak{c}\left(r\right) \text{Im} \left(\overline{\Psi^{[+2]}}^\prime  m \psi^{[+2]} \right) \nonumber \\
 +  \int_{A_1^*}^{A_2^*}\text{Im} \left( \left(- \mathfrak{c}\left(r\right)  m\overline{\psi^{[+2]}}^\prime + \mathfrak{c}\left(r\right)  m\overline{\psi^{[+2]}}\right)i am \psi^{[+2]}\right)
\end{align}
where we have used the (frequency localised) relation between ${\Psi^{[+2]}}$ and ${\psi}^{[+2]}$ twice. Now the first three terms on the right hand side are again easily controlled using 
Cauchy--Schwarz  (as well as  the simple fact that for $i=1,2$
$
| \Psi^{[\pm 2]}|^2 \left(A_i^*\right) \lesssim \|\mathfrak{d} \Psi^{[\pm2]}\|^2 
$).
For the term in the last line we integrate the first summand by parts while the second is already manifestly controlled by $\|\mathfrak{d}\psi^{[\pm2]}\|^2$. This leads immediately to (\ref{ihen}).

For the second term of $(\ref{two.terms.here})$, write
\begin{align}
\int_{A_1^*}^{A_2^*}  \mathfrak{c}\left(r\right)  \text{Im}\left(im u^{[+ 2]} \omega \overline{\Psi^{[+2]}}\right) =-\int_{A_1^*}^{A_2^*}   \text{Re}\left(m\overline{u^{[+ 2]}} \omega \left(\mathfrak{c}\left(r\right)  \underline{L} \psi^{[+2]} + \mathfrak{c}\left(r\right) \psi^{[+2]}\right)\right) \, .
\end{align}
The second term on the right is already manifestly controlled by $\|\mathfrak{d}\psi^{[\pm2]}\|^2$ and for the first we integrate by parts
\begin{align}
-\int_{A_1^*}^{A_2^*}   \text{Re}\left(m\overline{u^{[+ 2]}} \omega \left(\mathfrak{c}\left(r\right)  \underline{L} \psi^{[+2]} \right)\right)=& \text{Re} \left(m \overline{u^{[+ 2]}} \omega \mathfrak{c}\left(r\right) \psi^{[+2]}\right)\Big|_{A_1^*}^{A_2^*} \nonumber \\
&+ \int_{A_1^*}^{A_2^*} \mathfrak{c}\left(r\right) m\omega |\psi^{[+2]}|^2 + \mathfrak{c}\left(r\right) \text{Re} \left( m \omega \overline{u^{[+ 2]}}  \psi^{[+2]}\right)
\end{align}
from which the estimate (\ref{ihen}) is easily obtained.
\end{proof}

Putting together Propositions~\ref{multestforPsi},~\ref{prop:synth} and~\ref{newsyntheprop},  we obtain
Theorem~\ref{phaseSpaceILED}.

\section{Back to physical space: energy boundedness and integrated local energy decay}
\label{iledsec}

We now turn in this section in ernest
to the study of the Cauchy problem for $(\ref{Teukphysic})$ for $s=\pm2$.
The main result of this section will be a uniform (degenerate) energy boundedness
and integrated energy decay statement.
This will be  stated as {\bf Theorem~\ref{degenerateboundednessandILED}} of 
{\bf Section~\ref{herestatementsec}}. This corresponds to statement 1.~of 
the main result of the paper, Theorem~\ref{finalstatetheor}.

The remainder of the section will then be devoted to the proof of Theorem~\ref{degenerateboundednessandILED}.
We first define in {\bf Section~\ref{cutoffsecs}} 
 a cutoff version $\upalpha_{\text{\Rightscissors}}^{[\pm2]}$ 
of our solution $\upalpha^{[\pm2]}$ of $(\ref{Teukphysic})$
such that $\upalpha_{\text{\Rightscissors}}^{[\pm2]}$ 
satisfies an inhomogeneous equation $(\ref{inhomoteuk})$,
whose inhomogeneous term $F^{[\pm2}_{\text{\Rightscissors}}$ is localised in time to be supported
only ``near'' $\tilde{t}^*=0$ and ``near'' $\tilde{t}^*=\tau_{\rm final}$
and in space to be supported only in $r^*=[2A_1^*,2A_2^*]$.
The cutoff is such that restricted to $r\in [A_1,A_2]$, $\upalpha_{\text{\Rightscissors}}^{[\pm2]}$ 
is compactly supported in $\tilde{t}^*\in [0,\tau_{\rm final}]$. This allows us in
 {\bf Section~\ref{summingsection}}
to then apply the results of Section~\ref{ODEmegasec}
to such $\upalpha_{\text{\Rightscissors}}^{[\pm2]}$, 
summing the resulting estimate over frequencies. In {\bf Section~\ref{retrievingfut}}
we shall combine this estimate with the conditional estimates of 
Section~\ref{Physspacesecnew},
using also the auxiliary estimates of Section~\ref{auxil.est.sec}
to obtain a global integrated energy decay statement, with an error term, however,
on the right
side arising from the cutoff. Finally, we shall 
bound this latter error terms associated to the cutoff
in {\bf Section~\ref{controlerrorsec}}, again 
using the auxiliary estimates of Section~\ref{auxil.est.sec}, allowing us
to infer the statement of Theorem~\ref{degenerateboundednessandILED}.

As remarked in Section~\ref{axi-intro}, in the axisymmetric case, one can directly
distill from the calculations of this paper
an alternative, simpler proof
of Theorem~\ref{degenerateboundednessandILED}  expressed entirely
in physical space.
We do this in {\bf Section~\ref{axi-note}}.

\subsection{Statement of degenerate boundedness and integrated energy decay}
\label{herestatementsec}

\begin{theorem}
\label{degenerateboundednessandILED}
Let $\upalpha^{[\pm2]}$, $\Psi^{[\pm2]}$ and
$\uppsi^{[\pm2]}$ be as in Theorem~\ref{finalstatetheor}.

Then, for $\boxed{s=+2}$, we have the following estimates
\begin{itemize}
\item the basic degenerate Morawetz estimate
 \begin{align} \label{basdegmorprelim}
& \ {\mathbb{I}}^{\rm deg}_{\eta} \left[\Psi^{[+2]}\right] \left(0, \tau_{\rm final}\right) 
+\mathbb{I}_{\eta} \left[\uppsi^{[+2]}\right] \left(0, \tau_{\rm final}\right)
+\mathbb{I}_{\eta} \left[\upalpha^{[+2]}\right]  \left(0, \tau_{\rm final}\right)
\nonumber \\ 
 \lesssim & \ \ {\mathbb{E}}_{\widetilde\Sigma_{\tau},\eta} \left[\Psi^{[+2]}\right] \left(0\right)  +\mathbb{E}_{\widetilde\Sigma_{\tau},\eta} \left[\uppsi^{[+2]}\right] \left(0\right)  
+ \mathbb{E}_{\widetilde\Sigma_{\tau},\eta} \left[\upalpha^{[+2]}\right] \left(0\right)
\end{align}

\item the $\eta$-weighted energy boundedness estimate
\begin{align} \label{ewbndprelim}
& \ \ {\mathbb{E}}_{\mathcal{H}^+} \left[\Psi^{[+2]}\right]\left(0, \tau_{\rm final}\right) + {\mathbb{E}}_{\widetilde\Sigma_{\tau},\eta} \left[\Psi^{[+2]}\right] \left(\tau_{\rm final}\right) \nonumber \\
 \lesssim & \ \ {\mathbb{E}}_{\widetilde\Sigma_{\tau},\eta} \left[\Psi^{[+2]}\right] \left(0\right)  +\mathbb{E}_{\widetilde\Sigma_{\tau},\eta} \left[\uppsi^{[+2]}\right] \left(0\right)  
+ \mathbb{E}_{\widetilde\Sigma_{\tau},\eta} \left[\upalpha^{[+2]}\right] \left(0\right) \, .
\end{align}
\end{itemize}
Similarly, for $\boxed{s=-2}$, we have
\begin{itemize}
\item the basic degenerate Morawetz estimate
 \begin{align} \label{basdegmorprelim2}
& \ {\mathbb{I}}^{\rm deg}_{\eta} \left[\Psi^{[-2]}\right] \left(0, \tau_{\rm final}\right) 
+\mathbb{I} \left[\uppsi^{[-2]}\right] \left(0, \tau_{\rm final}\right)
+\mathbb{I}\left[\upalpha^{[-2]}\right]  \left(0,\tau_{\rm final}\right)
\nonumber \\ 
 \lesssim & \ \ {\mathbb{E}}_{\widetilde\Sigma_{\tau},\eta} \left[\Psi^{[-2]}\right] \left(0\right)  +\mathbb{E}_{\widetilde\Sigma_{\tau}} \left[\uppsi^{[-2]}\right] \left(0\right)  
+ \mathbb{E}_{\widetilde\Sigma_{\tau}} \left[\upalpha^{[-2]}\right] \left(0\right)
\end{align}

\item the $\eta$-weighted energy boundedness estimate
\begin{align} \label{ewbndprelim2}
& \ \ {\mathbb{E}}_{\mathcal{H}^+} \left[\Psi^{[-2]}\right]\left(0, \tau_{\rm final}\right) + {\mathbb{E}}_{\widetilde\Sigma_{\tau},\eta} \left[\Psi^{[-2]}\right] \left(\tau_{\rm final}\right) \nonumber \\
 \lesssim & \ \ {\mathbb{E}}_{\widetilde\Sigma_{\tau},\eta} \left[\Psi^{[-2]}\right] \left(0\right)  +\mathbb{E}_{\widetilde\Sigma_{\tau}} \left[\uppsi^{[-2]}\right] \left(0\right)  
+ \mathbb{E}_{\widetilde\Sigma_{\tau}} \left[\upalpha^{[-2]}\right] \left(0\right) \, .
\end{align}
\end{itemize}

\end{theorem}

\begin{remark}
In the case $s=-2$ one can prove these estimates using only the $\mathbb{E}_{\widetilde\Sigma_{\tau},0}\left[\Psi^{[-2]}\right]$-energy. However, that energy is insufficient to eventually control the energy flux of $r^{-3} \alpha^{[-2]}$ through null infinity, which is why we kept the estimate as symmetric with the $s=+2$-case as possible. See also Remark \ref{remark:avoideps}.
\end{remark}

In the proof of the theorem, we may assume for convenience that the data
 $(\tilde{\upalpha}^{[\pm2]}_0,\tilde{\upalpha}^{[\pm2]}_1)$
are smooth.
It follows that all associated appropriately rescaled
quantities $\Psi^{[\pm2]}$, etc., are smooth in $\mathcal{R}_0$.
To ease notation we define the data quantities
\begin{align}
\mathbb{D}^{[+ 2]} \left(0\right) &= {\mathbb{E}}_{\widetilde\Sigma_{\tau},\eta} \left[\Psi^{[+2]}\right] \left(0\right)  +\mathbb{E}_{\widetilde\Sigma_{\tau},\eta} \left[\uppsi^{[+2]}\right] \left(0\right)  
+ \mathbb{E}_{\widetilde\Sigma_{\tau},\eta} \left[\upalpha^{[+2]}\right] \left(0\right) \, ,\nonumber \\
\mathbb{D}^{[- 2]} \left(0\right) &= {\mathbb{E}}_{\widetilde\Sigma_{\tau},\eta} \left[\Psi^{[+2]}\right] \left(0\right)  +\mathbb{E}_{\widetilde\Sigma_{\tau}} \left[\uppsi^{[+2]}\right] \left(0\right)  
+ \mathbb{E}_{\widetilde\Sigma_{\tau}} \left[\upalpha^{[+2]}\right] \left(0\right) \, .
\end{align}

\subsection{The past and future cutoffs}
\label{cutoffsecs}

Let $\xi:\mathbb R\to \infty$ be a cutoff function such that $\xi=0$ for $x\le0$  and
$\xi=1$ for $x\ge 1$. Let $\varepsilon>0$ be a parameter to be determined.
Fix $\tau_{\rm final}>0$ and define
\begin{align} \label{alphacutoff}
\tilde\upalpha_{\text{\Rightscissors}}^{[\pm2]}(\tilde{t}^*, r, \theta, \phi)  =  \Xi(\tilde{t}^*,r)
\tilde\upalpha^{[\pm2]}(\tilde{t}^*, r, \theta, \phi) \, 
\end{align}
where
\begin{eqnarray}
\nonumber
\Xi(\tilde{t}^*,r)&=& \xi(\varepsilon t^*)\, \xi (\varepsilon( \tau_{\rm final}-\tilde{t}^*)) 
\, \xi (2A_2^* -r^*)   \,    \xi (r^*-2A_1^*)\\
\nonumber
&&\qquad
+(1-\xi)(\varepsilon \tilde{t}^*)\, \xi (A_1^*-r^*)\, \xi(r^*-A_2)\\
\label{big.Xi.def}
&&\qquad
+(1-\xi)(\varepsilon (\tau_{\rm final}-\tilde{t}^*)) \,  \xi(A_1^*-r^*)\, \xi(r^*-A_2).
\end{eqnarray}
We note that $\tilde\upalpha_{\text{\Rightscissors}}^{[\pm2]}\in\mathscr{S}^{[s]}_\infty(\mathcal{R})$
and satisfies $(\ref{inhomoteuk})$ with inhomogeneity given by
\begin{align}
\label{inhomhereterm}
\frac{\Delta}{\rho^2}\tilde{F}_{\text{\Rightscissors}}^{[\pm2]}= &(L\Xi) (\underline{L}\tilde\upalpha^{[\pm2]})) + (\underline{L} \Xi) (L \tilde\upalpha^{[\pm2]}) + (L\underline{L} \Xi)\tilde\upalpha^{[\pm2]} - 2\frac{w^\prime}{w} (\underline{L} \Xi)\tilde\upalpha^{[\pm2]}  \\
&-\frac{\Delta}{(r^2+a^2)^2} \left(2a(T\Xi)(\Phi \tilde\upalpha^{[\pm2]}) +a^2 \sin^2 \theta \left( (TT\Xi)\tilde\upalpha^{[\pm2]} + 2(T\Xi)(T\tilde\upalpha^{[\pm2]})\right) + 2isa\cos \theta (T\Xi)\tilde\upalpha^{[\pm2]}\right) \nonumber  .
\end{align}
We define now  $\upalpha_{\text{\Rightscissors}}^{[\pm2]}$ to be given by $(\ref{rescaleddefsfirstattempt})$, 
$P^{[\pm2]}_{\text{\Rightscissors}}$ to be given by $(\ref{Pdefsbulkplus2})$--$(\ref{Pdefsbulkmunus2})$, 
$\Psi^{[\pm2]}_{\text{\Rightscissors}}$ to be given by $(\ref{rescaledPSI})$
and
$\uppsi^{[\pm2]}_{\text{\Rightscissors}}$ to be given by $(\ref{officdeflittlepsiplus})$--$(\ref{officdeflittlepsiminus})$,
where all quantities now have ${\text{\Rightscissors}}$.

We note that $\tilde{F}_{\text{\Rightscissors}}^{[\pm2]}$ is supported
in the support of $\nabla \Xi$ (see the shaded regions of Figure~\ref{cutofffig}):
\begin{equation}
\label{supportofF}
\big(\{0\le \tilde{t}^* \le \varepsilon^{-1} \}\cup 
 \{\tau_{\rm final}-\varepsilon^{-1} \le \tilde{t}^*\le \tau_{\rm final}\}\big)
 \bigcap   \{2A_1^* \le r^*\le 2A_2^*\}
 \end{equation}
 while 
 \[
\upalpha_{\text{\Rightscissors}}^{[\pm2]}=0
\]
in $\{A_1\le r\le A_2\}\cap( \{\tilde{t}^*\le 0\}\cup \{\tilde{t}^*\ge \tau_{\rm final}\}$.

\begin{figure}
\centering{
\def\svgwidth{8pc}
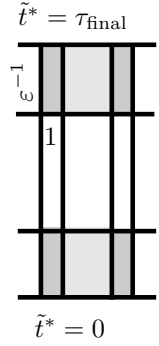}
\caption{Support of $\nabla \Xi$}\label{cutofffig}
\end{figure}

Let us already note 
the following proposition
\begin{proposition}
\label{veo.prop.edw}
Let $\Psi^{[\pm2]}_{\text{\Rightscissors}}$
be as above and let $\mathfrak{G}^{\pm2}$ 
be the inhomogeneous term associated to the generalised
Regge--Wheeler equation $(\ref{RWtypeinthebulk})$
arising from
$\tilde{F}_{\text{\Rightscissors}}^{[\pm2]}$ according to $(\ref{Gerror})$ and
$(\ref{Gerror2})$. Then we have the estimates
\begin{align}
\nonumber
\int_{\widetilde{\mathcal{R}}^{\rm trap}(0,\tau_{\rm final})} |\mathfrak{G}^{[\pm2]}|^2 
dVol &\lesssim \varepsilon^2 
\left( \mathbb{I}^{\rm trap}[\upalpha^{[\pm2]}](0,\varepsilon^{-1}) +\mathbb{I}^{\rm trap}[\uppsi^{[-2]}](0,\varepsilon^{-1})\right)
 \\
 \nonumber
 &\qquad+
\varepsilon^2\left( \mathbb{I}^{\rm trap}[\upalpha^{[\pm2]}](\tau_{\rm final}-\varepsilon^{-1},\tau_{\rm final}) +\mathbb{I}^{\rm trap}[\uppsi^{[-2]}](\tau_{\rm final}-\varepsilon^{-1},\tau_{\rm final})\right)\\
 \label{prwto.bound.veou.prop}
&\qquad+\varepsilon\sup_ {0\le \tau \le \tau_{\rm final}}\mathbb{E}_{\widetilde{\Sigma}_\tau,0} [\Psi^{[\pm2]}],
\end{align}
\begin{align}
\nonumber
\int_{\widetilde{\mathcal{R}}^{\rm away}(0,\tau_{\rm final})}
 |\mathfrak{G}^{[\pm2]}|^2 dVol &\lesssim 
 \mathbb{I}_{[\eta]}[\upalpha^{[\pm2]}](0,\varepsilon^{-1}) +\mathbb{I}_{[\eta]}[\uppsi^{[-2]}](0,\varepsilon^{-1})\\
 \nonumber
&\qquad + \mathbb{I}_{[\eta]}[\upalpha^{[\pm2]}](\tau_{\rm final}-\varepsilon^{-1},\tau_{\rm final}) +\mathbb{I}_{[\eta]}[\uppsi^{[-2]}](\tau_{\rm final}-\varepsilon^{-1},\tau_{\rm final})\\
 \label{deuto.bound.veou.prop}
&\qquad+\varepsilon^{-1}\sup_ {0\le \tau \le \tau_{\rm final}} \mathbb{E}_{\widetilde{\Sigma}_\tau,0} [\Psi^{[\pm2]}].
\end{align}
Here the subindex $\left[\eta\right]$ is equal to $\eta$ in case of $s=+2$ and it is dropped entirely in case $s=-2$.
\end{proposition}

\begin{remark}
As the proof shows and is already clear from the support of the cut-offs, only the spacetime integrals in the overlap region are needed on the right hand side of (\ref{deuto.bound.veou.prop}). 
\end{remark}

\begin{proof}
We first prove $(\ref{prwto.bound.veou.prop})$.
Note that the support of $\mathfrak{G}$ is manifestly contained
in the support $(\ref{supportofF})$ of $\tilde{F}_{\text{\Rightscissors}}^{[\pm2]}$.
Moreover, one easily sees that one obtains sum of terms containing
\begin{align*}
L\Psi^{[\pm2]}, \underline{L}\Psi^{[\pm2]}, \Psi^{[\pm2]}, T\Psi^{[\pm2]}, \Phi\Psi^{[\pm2]},
L\uppsi^{[\pm2]}, \underline{L}\uppsi^{[\pm2]},
T\uppsi^{[\pm2]}, \Phi\uppsi^{[\pm2]},
 \uppsi^{[\pm2]},\\
 L\upalpha^{[\pm2]}, \underline{L}\upalpha^{[\pm2]},
T\upalpha^{[\pm2]}, \Phi\upalpha^{[\pm2]},
  \upalpha^{[\pm2]}
\end{align*}
with $r$ and horizon weights which are  uniformly bounded in view of the support.
Since restricted to the region $r\in[A_1,A_2]$ the $r$-dependence of the cutoff
$\Xi$ defined in $(\ref{big.Xi.def})$
is trivial, while the $\tilde{t}^*$-dependence always comes with
a $\varepsilon$, it follows that 
\begin{equation}
\label{gainofeps}
|L^{k_1}\underline{L}^{k_2}T^{k_3}\Xi| \lesssim \varepsilon\qquad{\rm for\ }r\in[A_1,A_2]
\end{equation}
for any $k_1+k_2+k_3\ge 1$,
where we have used also that
$t=t^*=\tilde{t}^*$ in this region
by our choices in Section~\ref{very.slowly.very.slowly}.
It follows that  all terms in the expression
for   $\mathfrak{G}$  pick up an $\varepsilon$
factor. The inequality $(\ref{prwto.bound.veou.prop})$ now follows from Cauchy--Schwarz,
the definition of the norms and Remarks~\ref{control.the.l.derivs1} and~\ref{control.the.l.derivs2},
where in addition we have appealed to the
 coarea formula and size of $\tilde{t}^*$-support for the term involving $\Psi^{[\pm2]}$.

The proof is the same for $(\ref{deuto.bound.veou.prop})$, except
that the presence of the nontrivial $r$  cutoff in $(\ref{big.Xi.def})$ means that
$\varepsilon$ on the right hand side of
$(\ref{gainofeps})$ must now replaced by $1$ outside of $r\in[A_1,A_2]$, and thus
the $\varepsilon^2$ factor of
$(\ref{prwto.bound.veou.prop})$
is no longer present in the right hand side of the final estimate.
\end{proof}

We will in fact not use the bound $(\ref{deuto.bound.veou.prop})$  directly, but 
similar bounds for physical space terms that arise from multiplying
$\mathfrak{G}\Psi$ and $\mathfrak{G}\partial_r\Psi$.

\subsection{The summed relation}
\label{summingsection}

In view of the support of $\tilde\upalpha_{\text{\Rightscissors}}^{[\pm2]}$ and 
the smoothness of 
$(\ref{inhomhereterm})$, it follows that $\upalpha_{\text{\Rightscissors}}^{[\pm2]}$ 
manifestly satisfies the $[A_1,A_2]$-admissibility condition
of Definition~\ref{suffintdef}.
In a slight abuse of notation, we will denote the coefficients of
$\upalpha_{\text{\Rightscissors}}^{[\pm2]}$, $\Psi_{\text{\Rightscissors}}^{[\pm2]}$, etc.,
without the ${\text{\Rightscissors}}$ subscript.\footnote{This will not be a source
of confusion because we will never apply frequency analysis directly to
$\upalpha^{[\pm2]}$.}

We define thus the coefficients $u^{[\pm2], (a\omega)}_{m\ell}$
and we apply Theorem~\ref{phaseSpaceILED} 
with the admissible frequency triple
$(\omega, m, \widetilde{\Lambda}^{[\pm2], (a\omega)}_{m\ell})$.
We now, integrate over $\omega$ and sum over frequencies:
\[
\int_{-\infty}^\infty d\omega \sum_{m\ell}  .
\]
From summing the relation $(\ref{fromPhaseSpace2})$--$(\ref{another.def.here})$,
we hence obtain in view of the Plancherel
relations of Section~\ref{coeffssec} (applied to $\upalpha_{\text{\Rightscissors}}^{[\pm2]}$, 
$\uppsi_{\text{\Rightscissors}}^{[\pm2]}$ and $\Psi_{\text{\Rightscissors}}^{[\pm2]}$):
\begin{proposition}
\label{summed.near}
Let the assumptions of Theorem \ref{degenerateboundednessandILED} hold. Define the cut-off quantities $\upalpha_{\text{\Rightscissors}}^{[\pm 2]}$, $\uppsi_{\text{\Rightscissors}}^{[\pm 2]}$ and $\Psi_{\text{\Rightscissors}}^{[\pm 2]}$ as in (\ref{alphacutoff}), (\ref{rescaleddefsfirstattempt}) and (\ref{Pdefsbulkplus2})--(\ref{officdeflittlepsiminus}). Then we have the estimates 
\begin{align} \label{Step1c}
\mathbb{I}^{\rm trap}[\Psi^{[\pm2]}_{\text{\Rightscissors}}](0,\tau_{\rm final})
\lesssim \mathfrak{H}^{\rm trap}[\Psi^{[\pm 2]}_{\text{\Rightscissors}}]
+\mathbb{Q}_{r=A_2}[\Psi^{[\pm 2]}_{\text{\Rightscissors}}]-\mathbb{Q}_{r=A_1}
[\Psi_{\text{\Rightscissors}}^{[\pm 2]} ]+|a|\sum_{i=1}^2
\left( \mathbb{E}_{r=A_i}[\uppsi_{\text{\Rightscissors}}^{[\pm 2]}]+\mathbb{E}_{r=A_i}[\upalpha^{[\pm 2]}_{\text{\Rightscissors}}]\right),
\end{align}
\begin{align}
&\mathbb{I}^{\rm trap}[\uppsi^{[\pm2]}_{\text{\Rightscissors}}](0,\tau_{\rm final})+
\mathbb{I}^{\rm trap}[\upalpha^{[\pm2]}_{\text{\Rightscissors}}](0,\tau_{\rm final})
+\mathbb{E}_{r=A_{(6\mp2)/4}}[\uppsi_{\text{\Rightscissors}}^{[\pm2]}]+\mathbb{E}_{r=(6\mp2)/4}[\upalpha_{\text{\Rightscissors}}^{[\pm2]}] \nonumber
\\
\label{ovoma.edw}
&\qquad\lesssim
\mathbb{I}^{\rm trap}[\Psi^{[\pm2]}_{\text{\Rightscissors}}](0,\tau_{\rm final}) 
+ \mathbb{E}_{r=A_{(6\pm2)/4}}[\uppsi_{\text{\Rightscissors}}^{[\pm2]}]+\mathbb{E}_{r=A_{(6\pm2)/4}}[\upalpha_{\text{\Rightscissors}}^{[\pm2]}],
\end{align}
where
\[
\mathfrak{H}^{\rm trap}[\Psi^{[\pm 2]}_{\text{\Rightscissors}}]
=
\int_{-\infty}^{\infty}d\omega \sum_{m\ell} \int_{A_1^*}^{A_2^*} dr^*\, \mathfrak{G}^{(a\omega)}_{m\ell}
\cdot (f^{(a\omega)}_{m\ell}, y^{(a\omega)}_{m\ell},E)\cdot (\Psi^{[\pm2], (a\omega)}_{m\ell},
(\Psi')^{[\pm2], (a\omega)}_{m\ell}).
\]
\end{proposition}

\subsection{Global physical space estimates}
\label{retrievingfut}

Let us first combine the above estimates with the conditional physical
space estimates proven in Section~\ref{Physspacesecnew}.

\begin{proposition}
\label{firstinaline}
Let the assumptions of Theorem \ref{degenerateboundednessandILED} hold. Define the cut-off quantities $\upalpha_{\text{\Rightscissors}}^{[\pm 2]}$, $\uppsi_{\text{\Rightscissors}}^{[\pm 2]}$ and $\Psi_{\text{\Rightscissors}}^{[\pm 2]}$ as in (\ref{alphacutoff}), (\ref{rescaleddefsfirstattempt}) and (\ref{Pdefsbulkplus2})--(\ref{officdeflittlepsiminus}). Then we have the estimates
\begin{align}
\label{webeginwith}
\mathbb{E}^{\rm away}_{\widetilde{\Sigma}_\tau,\eta}[\Psi^{[\pm2]}_{\text{\Rightscissors}}]
(\tau_{\rm final})
+
\mathbb{I}_{\eta}^{\rm deg}[\Psi^{[\pm2]}_{\text{\Rightscissors}}](0,\tau_{\rm final})
+ \mathbb{I}_{[\eta]}[\uppsi^{[\pm2]}_{\text{\Rightscissors}}](0,\tau_{\rm final})
+\mathbb{I}_{[\eta]}[\upalpha^{[\pm2]}_{\text{\Rightscissors}}](0,\tau_{\rm final}) 
\lesssim
\mathfrak{H}[\Psi^{[\pm 2]}_{\text{\Rightscissors}}]
 + \mathbb{D}^{[\pm 2]} \left(0\right),
 \end{align}
 \begin{align}
 \label{manifestboundterms}
\sum_{i=1}^2  \mathbb{E}_{r=A_{i}}[\uppsi_{\text{\Rightscissors}}^{[\pm2]}](0,\tau_{\rm final})
+\mathbb{E}_{r=A_i}[\upalpha_{\text{\Rightscissors}}^{[\pm2]}] (0,\tau_{\rm final})  \lesssim
\mathfrak{H}[\Psi^{[\pm 2]}_{\text{\Rightscissors}}]
 + \mathbb{D}^{[\pm 2]} \left(0\right),
 \end{align}
 where
\[
\mathfrak{H}[\Psi^{[\pm 2]}_{\text{\Rightscissors}}]=
\mathfrak{H}^{\rm trap}[\Psi^{[\pm 2]}_{\text{\Rightscissors}}]+
\mathfrak{H}^{\rm away}[\Psi^{[\pm 2]}_{\text{\Rightscissors}}] \, .
\]
In the above, the subindex $\left[\eta\right]$ is equal to $\eta$ in case of $s=+2$ and it is dropped entirely in case $s=-2$.
\end{proposition}
\begin{proof}
We add the estimates of 
Proposition~\ref{summed.near} with those of Section~\ref{Physspacesecnew} as follows. 

Let us consider first the $+2$ case.
We first add  the first estimate $(\ref{weightedtransportright})$ 
of Proposition~\ref{transportplusnof} (applied to $\upalpha_{\text{\Rightscissors}}^{[\pm2]}$
and $\uppsi_{\text{\Rightscissors}}^{[\pm2]}$
with $\tau_1=0$, $\tau_2=\tau_{\rm final}$
and with $p=\eta$)
to a suitable constant times the estimate $(\ref{ovoma.edw})$ of
Proposition~\ref{summed.near}. The constant ensures
that the terms $\mathbb E_{r=A_2}$ on the left hand side of $(\ref{weightedtransportright})$ 
is sufficient
to absorb the analogous term on the right hand side of $(\ref{ovoma.edw})$.
Finally, we now add to the previous combination a suitable constant times the
 second estimate $(\ref{weightedtransportright})$  of 
Proposition~\ref{transportplusnof}, again so that the boundary terms
on $\mathbb E_{r=A_1}$ are now absorbed. We obtain
 thus
 \begin{equation}
 \label{thus.spoke}
 \mathbb{I}_{[\eta]}[\uppsi^{[\pm 2]}_{\text{\Rightscissors}}](0,\tau_{\rm final})
+\mathbb{I}_{[\eta]}[\upalpha^{[\pm 2]}_{\text{\Rightscissors}}](0,\tau_{\rm final}) 
\lesssim \mathbb{I}^{\rm deg}_{[\eta]}\left[\Psi^{[\pm 2]}_{\text{\Rightscissors}}\right]
+\mathbb{D}^{[\pm 2]} \left(0\right)
\end{equation}
\begin{equation}
\label{thus.spoke2}
\sum_{i=1}^2  \mathbb{E}_{r=A_{i}}[\uppsi_{\text{\Rightscissors}}^{[\pm2]}](0,\tau_{\rm final})
+\mathbb{E}_{r=A_i}[\upalpha_{\text{\Rightscissors}}^{[\pm2]}] (0,\tau_{\rm final})  \lesssim
 \mathbb{I}^{\rm deg}_{[\eta]}\left[\Psi^{[\pm 2]}_{\text{\Rightscissors}}\right]
 + \mathbb{D}^{[\pm 2]} \left(0\right)
 \end{equation}
in the case of $+2$. For the $-2$ case, we choose the relative constants in 
the reverse order, starting with  the second estimate $(\ref{fiete.2})$ of Proposition~\ref{transportminusnof}. We obtain again $(\ref{thus.spoke})$ in the $-2$ case,
as well as the estimate  $(\ref{thus.spoke2})$ for the boundary terms.

We now similarly add Proposition~\ref{multpropnofreq} 
(applied  to $\Psi^{[\pm2]}_{\text{\Rightscissors}}$ with $\tau_1=0$, $\tau_2=\tau_{\rm final}$)
to $(\ref{Step1c})$, noting that the $\mathbb Q$ boundary terms
exactly cancel.
This gives thus
\begin{equation}
\label{justbeforesum}
\mathbb{E}^{\rm away}_{\widetilde{\Sigma}_\tau,\eta}[\Psi^{[\pm2]}_{\text{\Rightscissors}}]
(\tau_{\rm final})
+
\mathbb{I}_{\eta}^{\rm deg}[\Psi^{[\pm2]}_{\text{\Rightscissors}}](0,\tau_{\rm final})
\lesssim  
\mathfrak{H}[\Psi^{[\pm 2]}_{\text{\Rightscissors}}] +
 |a|\mathbb{I}_{[\eta]}[\uppsi^{[+2]}_{\text{\Rightscissors}}](0,\tau_{\rm final})
+|a|\mathbb{I}_{[\eta]}[\upalpha^{[+2]}_{\text{\Rightscissors}}](0,\tau_{\rm final}) 
+\mathbb{D}^{[\pm 2]} \left(0\right).
\end{equation}

We fix now a sufficiently small parameter $e$ depending only on $M$. 
It follows that, restricting to $a_0\ll e$, 
we may sum  $e\times (\ref{thus.spoke})$ with $(\ref{justbeforesum})$ to absorb
both the first term on the right hand side of $(\ref{thus.spoke})$
and the middle two terms on the right hand side of $(\ref{justbeforesum})$.
The desired $(\ref{webeginwith})$ follows.  

The estimate $(\ref{manifestboundterms})$ again follows from $(\ref{thus.spoke2})$
and $(\ref{webeginwith})$.
\end{proof}

In the rest of this subsection, we proceed to remove
the ${\text{\Rightscissors}}$ from the quantities
on the left hand side of $(\ref{webeginwith})$.

Putting together the local-in-time 
Proposition~\ref{localintimehereprop} and the $(T+\chi \upomega_+ \Phi)$-energy
estimate Proposition~\ref{proptocontrolcutoff} we obtain first the following:

\begin{proposition}
With the notation of Proposition~\ref{firstinaline}, we have
the additional estimates
\begin{align}
\label{addit.one}
\nonumber
&\mathbb{I}_{\eta}^{\rm deg}[\Psi^{[\pm2]}](0,\tau_{\rm final}-\varepsilon^{-1})
+ \mathbb{I}_{[\eta]}[\uppsi^{[\pm2]}](0,\tau_{\rm final}-\varepsilon^{-1})
+\mathbb{I}_{[\eta]}[\upalpha^{[\pm2]}](0,\tau_{\rm final}-\varepsilon^{-1}) \\
&\qquad \lesssim
\mathfrak{H}[\Psi^{[\pm 2]}_{\text{\Rightscissors}}]
 + \varepsilon^{-1}\mathbb{D}^{[\pm 2]} \left(0\right),
 \end{align}
 \begin{equation}
 \label{mikrn.diafora}
 \sum_{i=1}^2  \mathbb{E}_{r=A_{i}}[\uppsi^{[\pm2]}](0,\tau_{\rm final}
 -\varepsilon^{-1})
 +\mathbb{E}_{r=A_i}[\upalpha^{[\pm2]}]
 (0,\tau_{\rm final}
 -\varepsilon^{-1})\lesssim
\mathfrak{H}[\Psi^{[\pm 2]}_{\text{\Rightscissors}}]
 +\varepsilon^{-1} \mathbb{D}^{[\pm 2]} \left(0\right),
\end{equation}
 \begin{equation}
 \label{evergeia.edw}
 \sup_{0\le t^*\le \tau_{\rm final}-\varepsilon^{-1}}
 \mathbb{E}_{\widetilde{\Sigma}_\tau,0}\left[\Psi^{[\pm2]}\right]
 (t^*)
 \lesssim
|a|\, \mathfrak{H}[\Psi^{[\pm 2]}_{\text{\Rightscissors}}]
 + \varepsilon^{-1} \mathbb{D}^{[\pm 2]} \left(0\right),
 \end{equation}
\begin{align}
\label{eivai.kiauto}
\sup_{0\le t^*\le \tau_{\rm final}-\varepsilon^{-1}} 
\mathbb{E}_{\widetilde{\Sigma}_{\tau}, [\eta]}\left[\uppsi^{[\pm2]}\right](t^*)
+ \mathbb{E}_{\widetilde{\Sigma}_{\tau}, [\eta]}\left[\upalpha^{[\pm2]}\right](t^*) 
\lesssim
\mathfrak{H}[\Psi^{[\pm 2]}_{\text{\Rightscissors}}]
 + \varepsilon^{-1}\mathbb{D}^{[\pm 2]} \left(0\right).
\end{align}
Here the subindex $\left[\eta\right]$ is equal to $\eta$ in case of $s=+2$ and it is dropped entirely in case $s=-2$.
\end{proposition}
\begin{proof}
For estimate $(\ref{addit.one})$
one applies 
Proposition~\ref{localintimehereprop} (applied with $\tau_1=0$
and with $\tau_{\rm step}=\varepsilon^{-1}$)
and Proposition~\ref{firstinaline} and
the fact that the cutoff $\Xi=1$ identically in the region 
$\tilde{t}^*\in [\varepsilon^{-1},\tau_{\rm final}-\varepsilon^{-1}]$
and in the region $\{r^*\ge 2A_2^*\}\cup \{r^*\le 2A_1^*\}$.
Estimate $(\ref{mikrn.diafora})$ follows similarly from $(\ref{manifestboundterms})$. 

Estimate $(\ref{evergeia.edw})$ now follows from $(\ref{addit.one})$ and 
Proposition~\ref{proptocontrolcutoff} applied with $\tau_1=0$ and 
$0\le \tau_2\le \tau_{\rm final}-
\varepsilon^{-1}$. 

Finally, to obtain $(\ref{eivai.kiauto})$, we argue as follows.
Revisiting the transport estimates of Section~\ref{condtranestsec}, we can estimate the  left hand side from initial
data, 
$\mathbb{I}_{\eta}^{\rm deg}[\Psi^{[\pm2]}](0,\tau_{\rm final}-\varepsilon^{-1})$,
the left hand side of $(\ref{mikrn.diafora})$
and the left hand side of $(\ref{evergeia.edw})$.
\end{proof}

Using once again the auxiliary estimates of Section~\ref{auxil.est.sec}, we can now 
improve this to:
\begin{proposition}
\label{using.auxil.prop}
\begin{equation}
\label{using.auxil}
\mathbb{I}_{\eta}^{\rm deg}[\Psi^{[\pm2]}](0,\tau_{\rm final})
+\mathbb{I}_{[\eta]} [\uppsi^{[\pm2]}](0,\tau_{\rm final})
+\mathbb{I}_{[\eta]}[\upalpha^{[\pm2]}](0,\tau_{\rm final}) 
\lesssim
\mathfrak{H}[\Psi^{[\pm 2]}_{\text{\Rightscissors}}] + \varepsilon^{-2} \mathbb{D}^{[\pm 2]} \left(0\right),
 \end{equation}
 \begin{equation}
 \label{using.auxil.again}
 \sup_{0 \le t^*
 \le \tau_{\rm final}}
 \mathbb{E}_{\widetilde{\Sigma}_\tau,0}[\Psi^{[\pm2]}]
(t^*)
\lesssim 
|a|\,\mathfrak{H}[\Psi^{[\pm 2]}_{\text{\Rightscissors}}] + \varepsilon^{-1} \mathbb{D}^{[\pm 2]} \left(0\right).
 \end{equation}
\end{proposition}
\begin{proof}
Appealing  to 
Proposition~\ref{localintimehereprop} with $\tau_1=\tau_{\rm final}-\varepsilon^{-1}$
and with $\tau_{\rm step}=\varepsilon^{-1}$,
and using $(\ref{eivai.kiauto})$, we obtain
\begin{align}
\mathbb{I}^{\rm deg}_\eta [\Psi^{[\pm2]}](0,\tau_{\rm final})
+ \mathbb{I}_{[\eta]}[\uppsi^{[\pm2]}](0,\tau_{\rm final})
+\mathbb{I}_{[\eta]}[\upalpha^{[\pm2]}](0,\tau_{\rm final}) 
\nonumber \\
\lesssim
\mathfrak{H}[\Psi^{[\pm 2]}_{\text{\Rightscissors}}] + \varepsilon^{-1} \mathbb{D}^{[\pm 2]} \left(0\right) + 
 \varepsilon^{-1}\mathbb{E}_{\widetilde{\Sigma}_\tau,0}[\Psi^{[\pm2]}](\tau_{\rm final}-\varepsilon^{-1})
\end{align}
Finally, we apply   $(\ref{evergeia.edw})$
to absorb the last term on the right hand side, obtaining thus  $(\ref{using.auxil})$.
Repeating now the proof of 
$(\ref{evergeia.edw})$ we obtain $(\ref{using.auxil.again})$.
\end{proof}

\subsection{Controlling the term $\mathfrak{H}[\Psi^{[\pm 2]}_{\text{\Rightscissors}}]$
and finishing the proof of Theorem~\ref{degenerateboundednessandILED}}

Finally, we control the error term $\mathfrak{H}[\Psi^{[\pm 2]}_{\text{\Rightscissors}}]$ arising from the cutoff.
\label{controlerrorsec}
\begin{proposition} \label{prop:soclose}
\begin{align}
\nonumber
\left|\mathfrak{H}[\Psi^{[\pm 2]}_{\text{\Rightscissors}}] \right|
&\lesssim \sup_{\{0\le \tau \le \varepsilon^{-1}\} \cup \{ \tau_{\rm final}-\varepsilon^{-1}\le 
\tau\le \tau_{\rm final}
\}  } \varepsilon^{-2}\mathbb{E}_{\widetilde{\Sigma}_\tau,0}[\Psi^{[\pm2]}](\tau)\\
\nonumber
&\qquad
+  \varepsilon \mathbb{I}_{[\eta]}[\uppsi^{[\pm2]}](\tau_{\rm final}-\varepsilon^{-1},
\tau_{\rm final})
+\varepsilon \mathbb{I}_{[\eta]}[\uppsi^{[\pm2]}](0,\varepsilon^{-1})\\
&\nonumber\qquad
+ \varepsilon \mathbb{I}_{[\eta]}[\upalpha^{[\pm2]}]
(\tau_{\rm final}-\varepsilon^{-1},\tau_{\rm final})
+ \varepsilon \mathbb{I}_{[\eta]}[\upalpha^{[\pm2]}](0,\varepsilon^{-1})
\\
\label{estimatforhpm2}
&\qquad + \varepsilon \mathbb{I}^{\rm trap}[\Psi^{[\pm2]}](0,\tau_{\rm final}).
\end{align}
\end{proposition}
\begin{proof}
Recalling
\[
\mathfrak{H}[\Psi^{[\pm 2]}_{\text{\Rightscissors}}]
=\mathfrak{H}^{\rm trap}[\Psi^{[\pm 2]}_{\text{\Rightscissors}}]+
\mathfrak{H}^{\rm away}[\Psi^{[\pm 2]}_{\text{\Rightscissors}}]
\]
let us further partition 
$\mathfrak{H}^{\rm trap}[\Psi^{[\pm 2]}_{\text{\Rightscissors}}]$
as $\mathfrak{H}^{\rm trap}[\Psi^{[\pm 2]}_{\text{\Rightscissors}}]=
\mathfrak{H}_1+\mathfrak{H}_2$ where 
we define
\begin{equation}
\label{where.we.def1}
\mathfrak{H}_1[\Psi^{[\pm 2]}_{\text{\Rightscissors}}] = \int_{-\infty}^{\infty} d\omega \sum_{m\ell}
\int_{A_1^*}^{A_2^*}
-E\omega{\rm Im}\left({\mathfrak G}^{[\pm2]} \overline{\Psi^{[\pm2]}}\right)dr^*,
\end{equation}
\begin{equation}
\label{where.we.def2}
\mathfrak{H}_2[\Psi^{[\pm 2]}_{\text{\Rightscissors}}] =
\int_{-\infty}^{\infty} d\omega \sum_{m\ell}
\int_{A_1^*}^{A_2^*}\left(
-2f{\rm Re}\left({\Psi^{[\pm2]}}'\overline{{\mathfrak G}^{[\pm2]} }\right) -f'{\rm Re}\left(\Psi^{[\pm2]}\overline{{\mathfrak G}^{[\pm2]} }\right) -2y\text{Re}\left({\Psi^{[\pm2]}}'\overline{{\mathfrak G}^{[\pm2]} }\right)\right)dr^*.
\end{equation}
We will show the above estimate for  
$\mathfrak{H}_1$, $\mathfrak{H}_2$
and $\mathfrak{H}^{\rm away}[\Psi^{[\pm 2]}_{\text{\Rightscissors}}]$.

Let us first deal with the term 
$\mathfrak{H}^{\rm away}[\Psi^{[\pm 2]}_{\text{\Rightscissors}}]$.
This is supported in
\begin{equation}
\label{where.supported.0}
\Big(\{ 0\le \tilde{t}^*\le \varepsilon^{-1} \}\cup \{ \tau_{\rm final}-\varepsilon^{-1}\le 
\tilde{t}^*\le \tau_{\rm final}  \}  \Big)
\cap \{2A_1^*\le r^*\le 2A_2^*\}
\end{equation}
and consists of quadratic terms one of which always contains a $\Psi^{[\pm2]}$-term.
Thus, by Cauchy--Schwarz this can easily be bounded by the first
three lines of the right hand side of $(\ref{estimatforhpm2})$, where an $\varepsilon^{-1}$ 
factor is introduced on the
$\Psi$ term, compensated by an $\varepsilon$ on the other terms. (The extra
$\varepsilon$ factor in $\varepsilon^{-2}$ arises from estimating a spacetime
integral by the supremum. Cf.~the proof of $(\ref{deuto.bound.veou.prop})$.)

For $\mathfrak{H}_1$, by the exact Plancherel formulae of Section~\ref{coeffssec},
the integral $(\ref{where.we.def1})$
transforms into a physical space integral supported
in 
\begin{equation}
\label{where.supported}
\Big(\{ 0\le \tilde{t}^*\le \varepsilon^{-1} \}\cup \{ \tau_{\rm final}-\varepsilon^{-1}\le \tilde{t}^*\le 
\tau_{\rm final}  \}  \Big)
\cap \{A_1^*\le r^*\le A_2^*\}
\end{equation}
which similarly to before, is obviously estimable from the first three lines of
the right hand side of $(\ref{estimatforhpm2})$. (In fact,
one could replace the factor $\varepsilon^{-2}$ with $1$, since,
just as
in the proof of $(\ref{prwto.bound.veou.prop})$,  $\tilde{t}^*$
derivatives of the cutoff $\Xi$ always generate extra $\varepsilon$ factors; we will use this 
idea below for estimating the remaining term.)

For $\mathfrak{H}_2$, we first apply Cauchy--Schwarz,
introducing a $\varepsilon^{-1}$,
\[
\left|\mathfrak{H}_2^{[\pm2]} \right|\lesssim \int_{-\infty}^{\infty}\sum_{m\ell}\int_{A_1^*}^{A_2^*}dr^*
\varepsilon^{-1} \|\mathfrak{G}^{(a\omega)}_{m\ell}\|^2
+ \varepsilon \|(\Psi^{[\pm2], (a\omega)}_{m\ell},
(\Psi')^{[\pm2], (a\omega)}_{m\ell})\|^2,
\]
where we have used $(\ref{unif.f.bnds})$ to  bound the $f$, $f'$ and $y$ factors uniformly over
frequencies.
We now apply Plancherel.
We note that by Proposition~\ref{veo.prop.edw},
the first term on the right hand side
is bounded by $\varepsilon^{-1}\times$ the right hand side of
 $(\ref{prwto.bound.veou.prop})$ while the second term is manifestly bounded
 by
 \[
 \varepsilon\mathbb{I}^{\rm trap}[\Psi^{[\pm2]}](0,\tau_{\rm final}).
 \]
We obtain $(\ref{estimatforhpm2})$ for $\mathfrak{H}_2$, finishing the proof.
\end{proof}

\begin{proposition}
For sufficiently small $a_0\ll \varepsilon\ll 1$, we have
\begin{align}
\label{almost.there}
\left|\mathfrak{H}_2^{[\pm2]} \right|
\lesssim \varepsilon^{-3}\mathbb{D}^{[\pm2]}(0)
\end{align}
\end{proposition}

\begin{proof}
Apply Proposition~\ref{proptocontrolcutoff} of 
Section~\ref{auxil.est.sec} to the estimate of Proposition \ref{prop:soclose} and combine with Proposition \ref{using.auxil.prop}.
\end{proof}

Now let $\varepsilon$ be fixed by the requirement of the above proposition.
From $(\ref{almost.there})$ and Proposition \ref{using.auxil.prop}
all statements of Theorem~\ref{degenerateboundednessandILED} now follow.

\subsection{Note on the axisymmetric case: a pure physical-space  proof}
\label{axi-note}

We note that in the axisymmetric case $\partial_\phi \upalpha^{[\pm2]}=0$, 
the physical space multiplier and transport estimates
of Section~\ref{Physspacesecnew} 
can be applied directly \emph{globally} in the region 
$\widetilde{\mathcal{R}}(\tau_1,\tau_2)$, i.e.~without the restriction
to 
$\widetilde{\mathcal{R}}^{\rm away}(\tau_1,\tau_2)$.
This leads already to a much shorter
proof of Theorem~\ref{degenerateboundednessandILED} 
which can be expressed entirely with 
physical space methods.
We explain how this physical-space proof can be distilled directly from the more
general calculations
of Section~\ref{ODEmegasec} done at fixed frequency.

Given $|a|<a_0\ll M$ sufficiently small,
let $r_{\rm trap}$ be the unique value  given by Lemma~\ref{lem.max} and define 
$f$ by~$(\ref{simple.def.of.f})$ and $y$ by $\delta_1((1-\chi )f+\chi f^3-\delta_1\tilde\chi(r)r^{-\eta})$
where $\chi$ is the cutoff appearing in $(\ref{y1def.here})$ 
and $\tilde\chi$ is the cutoff appearing in
$(\ref{newdefshere2})$. 
The calculation of Section~\ref{gonerangesec} now shows that
the coercivity property of the physical space current
$I^f+I^y$ holds globally
in $\widetilde{\mathcal{R}}^{\rm trap}(\tau_1,\tau_2)$
and thus
$(\ref{fundcoerci})$ holds when integrated globally in
$\widetilde{\mathcal{R}}(\tau_1,\tau_2)$, i.e.~without restriction to the ``away'' region
and with $\mathbb{I}^{\rm away}$ replaced by $\mathbb{I}^{\rm deg}_\eta$.
One also produces an estimate
for the future boundary term:
\begin{equation}
\label{got.produced}
\mathbb{E}_\eta \left[\Psi^{\pm2}\right](\tau_2)
\end{equation}
in view of  property 3.~of the proof of Proposition~\ref{multpropnofreq}.

We apply this estimate then
 in the region $\widetilde{\mathcal{R}}(0,\tau_2)$   directly to $\Psi^{[\pm2]}$
arising from a solution $\upalpha^{[\pm2]}$ of the homogeneous
Teukolsky equation $(\ref{Teukphysic})$.

We must estimate the error term arising from the coupling $\mathcal{J}^{[\pm2]}$.
For this we turn first to global transport estimates.

Note that in the axisymmetric case, the simple estimate
applied in Section~\ref{transestsec} for frequencies in the range
$\cal{F}^\sharp$ applies
now for all frequencies (since $\cal{F}^\flat=\emptyset$ if $m=0$) and
 corresponds to commuting the transport equations
by $\partial_{r*}$ and integrating by parts.
This physical space procedure, say in the $[+2]$ case,
allows one to obtain the estimate
\begin{align}
\label{newversion.middle}
\nonumber
\mathbb{E}_{r=A_1}\left[\upalpha^{[+2]}\right] (0,\tau_2)
+\mathbb{E}_{r=A_1}\left[\uppsi^{[+2]}\right] (0,\tau_2)
+\mathbb{E}^{\rm trap}\left[\upalpha^{[+2]}\right](\tau_2)
+\mathbb{E}^{\rm trap}\left[\upalpha^{[+2]}\right](\tau_2)\\
+\mathbb{I}^{\rm trap} \left[\upalpha^{[+2]}\right](0,\tau_2)
+\mathbb{I}^{\rm trap}\left[\uppsi^{[+2]}\right](0,\tau_2)
 \lesssim \mathbb{I}^{\rm trap}\left[\Psi^{[+2]}\right](0,\tau_2)
 +\mathbb{D}^{[+2]}(0).
\end{align}
Note that $\mathbb{I}^{\rm trap}(0,\tau_2)$ is degenerate and thus 
controlled by $\mathbb{I}^{\rm deg}_\eta(0,\tau_2)$.
Summing $(\ref{newversion.middle})$ with 
the estimates obtained from $(\ref{weightedtransportright})$ and $(\ref{weightedtransport.left})$, 
as in the proof of Proposition~\ref{firstinaline},
allows one to estimate finally
\begin{equation}
\label{9atoxoume}
\mathbb{E}_{\eta}\left[\upalpha^{[+2]}\right](\tau_2)
+\mathbb{E}_{\eta} \left[\upalpha^{[+2]}\right](\tau_2)
+\mathbb{I}_{[\eta]} \left[\upalpha^{[\pm2]}\right](\tau_1,\tau_2)
+\mathbb{I}_{[\eta]} \left[\uppsi^{[\pm2]} \right]  (\tau_1,\tau_2) \lesssim \mathbb{I}^{\rm deg}_\eta 
 \left[\Psi^{[\pm2]}\right](\tau_1,\tau_2)+
\mathbb{D}^{[\pm2]}(0).
\end{equation}

With this we estimate the new contribution to
$\mathcal{J}^{\pm2]}$ coming from
the region $\mathcal{R}^{\rm trap}(\tau_1,\tau_2)$. 
The only difficult term is the one arising from the $T$ multiplier.
In the fixed frequency estimate of Section~\ref{completing}, this corresponded to 
passing an $\omega$ from $\psi$ to $\Psi$ before applying
Cauchy--Schwarz. 
In physical space, this
corresponds simply to integration by parts in $t$.
By this physical space estimate,
we obtain that the resulting term is bounded by 
\begin{align}
\nonumber
\label{bounded.by.here.as.well}
 |a| \mathbb{I}^{\rm trap}[\Psi](0,\tau_2) + |a|
\mathbb{I}^{\rm trap} [\upalpha](0,\tau_2)
+|a| \mathbb{I}^{\rm trap} [\uppsi] (0,\tau_2) \\
+ |a| \mathbb{E}^{\rm trap}_{\widetilde{\Sigma}_{\tau}}\left[\Psi^{[\pm2]}\right](\tau_2)
+|a|\mathbb{E}^{\rm trap}_{\widetilde{\Sigma}_{\tau}}\left[\upalpha^{[\pm2]}\right](\tau_2)
+|a|\mathbb{E}^{\rm trap}_{\widetilde{\Sigma}_{\tau}}\left[\uppsi^{[\pm2]}\right](\tau_2)
+|a| \mathbb{D}^{[\pm2]}(0),
\end{align}
where the future boundary terms arise from this integration by parts.
(Note that all other terms in $\mathcal{J}^{[\pm2]}$ 
are estimated by the first line of $(\ref{bounded.by.here.as.well})$ alone.)
Combining with the original statement of Proposition~\ref{firstinaline},
this yields
\begin{equation}
\label{kovteuoume}
\mathbb{E}_{\widetilde{\Sigma}_\tau,\eta}\left[\Psi^{[\pm2]}\right]
+\mathbb{I}^{\rm deg}_\eta\left[\Psi^{[\pm2]}\right] (\tau_1,\tau_2)
\lesssim \mathbb{D}^{[\pm2]}(0) + (\ref{bounded.by.here.as.well}).
\end{equation}

In view of $ (\ref{9atoxoume})$, for sufficiently small $|a|<a_0\ll M$,
one can absorb the  terms $(\ref{bounded.by.here.as.well})$ on the right hand
side of $(\ref{kovteuoume})$ into the left hand side.
The remaining statements 
of Theorem~\ref{degenerateboundednessandILED}
follow immediately.

\section{The redshift effect and its associated Morawetz estimate} \label{sec:rsim}
In this section we will obtain statement~2.~of Theorem~\ref{finalstatetheor}
concerning the boundedness and integrated local energy
decay of the so-called red-shifted energy.
The required statement is  contained in {\bf Theorem~\ref{prop:basicnondeg}} below.

\subsection{Statement of red-shifted boundedness and integrated decay}

\begin{theorem} \label{prop:basicnondeg}
Let $\upalpha^{[\pm2]}$, $\Psi^{[\pm2]}$ and
$\uppsi^{[\pm2]}$ be as in Theorem~\ref{finalstatetheor}.
Then the following holds 
for any $\tau_2 > \tau_1\geq 0$.

For $\boxed{s=+2}$
\begin{itemize}
\item the basic degenerate Morawetz estimate
\begin{align} \label{basdegmor}
& \ \overline{\mathbb{I}}^{\rm deg}_{\eta} \left[\Psi^{[+2]}\right] \left(\tau_1,\tau_2\right) 
+\mathbb{I}_\eta \left[\uppsi^{[+2]}\right]  \left(\tau_1,\tau_2\right)
+\mathbb{I}_\eta \left[\upalpha^{[+2]}\right]  \left(\tau_1,\tau_2\right)
\nonumber \\ 
 \lesssim & \ \ \overline{\mathbb{E}}_{\widetilde\Sigma_{\tau},\eta} \left[\Psi^{[+2]}\right] \left(\tau_1\right)  +\mathbb{E}_{\widetilde\Sigma_{\tau},\eta} \left[\uppsi^{[+2]}\right] \left(\tau_1\right)  
+ \mathbb{E}_{\widetilde\Sigma_{\tau},\eta} \left[\upalpha^{[+2]}\right] \left(\tau_1\right) \, ,
\end{align}

\item the basic non-degenerate Morawetz estimate
\begin{align} \label{basnondegmor}
\overline{\mathbb{I}}_{\eta} \left[\Psi^{[+2]}\right] \left(\tau_1,\tau_2\right) 
 \lesssim & \ \ \overline{\mathbb{E}}_{\widetilde\Sigma_{\tau},\eta} \left[\Psi^{[+2]}\right] \left(\tau_1\right)  +\mathbb{E}_{\widetilde\Sigma_{\tau},\eta} \left[\uppsi^{[+2]}\right] \left(\tau_1\right)  
+ \mathbb{E}_{\widetilde\Sigma_{\tau},\eta} \left[\upalpha^{[+2]}\right] \left(\tau_1\right) \nonumber \\ & + \overline{\mathbb{E}}_{\widetilde\Sigma_{\tau},\eta} \left[T\Psi^{[+2]}\right] \left(\tau_1\right)  +\mathbb{E}_{\widetilde\Sigma_{\tau},\eta} \left[T\uppsi^{[+2]}\right] \left(\tau_1\right)  
+ \mathbb{E}_{\widetilde\Sigma_{\tau},\eta} \left[T\upalpha^{[+2]}\right] \left(\tau_1\right) \, ,
\end{align}
\item the $\eta$-weighted energy boundedness estimate
\begin{align} \label{ewbnd}
& \ \ \overline{\mathbb{E}}_{\mathcal{H}^+} \left[\Psi^{[+2]}\right] \left(\tau_1,\tau_2\right) + \overline{\mathbb{E}}_{\widetilde\Sigma_{\tau},\eta} \left[\Psi^{[+2]}\right] \left(\tau_2\right) \nonumber \\
 \lesssim & \ \ \overline{\mathbb{E}}_{\widetilde\Sigma_{\tau},\eta} \left[\Psi^{[+2]}\right] \left(\tau_1\right)  +\mathbb{E}_{\widetilde\Sigma_{\tau},\eta} \left[\uppsi^{[+2]}\right] \left(\tau_1\right)  
+ \mathbb{E}_{\widetilde\Sigma_{\tau},\eta} \left[\upalpha^{[+2]}\right] \left(\tau_1\right) \, .
\end{align}
\end{itemize}
For $\boxed{s=-2}$
\begin{itemize}
\item the basic degenerate Morawetz estimate
\begin{align} \label{basdegmor2}
&\overline{\mathbb{I}}^{\rm deg}_{\eta} \left[\Psi^{[-2]}\right] \left(\tau_1,\tau_2\right) 
+\mathbb{I} \left[\uppsi^{[-2]}\right]  \left(\tau_1,\tau_2\right)
+\mathbb{I} \left[\upalpha^{[-2]}\right]  \left(\tau_1,\tau_2\right)
\nonumber \\ 
 \lesssim & \ \ \overline{\mathbb{E}}_{\widetilde\Sigma_{\tau},\eta} \left[\Psi^{[-2]}\right] \left(\tau_1\right)  +\mathbb{E}_{\widetilde\Sigma_{\tau}} \left[\uppsi^{[-2]}\right] \left(\tau_1\right)  
+ \mathbb{E}_{\widetilde\Sigma_{\tau}} \left[\upalpha^{[-2]}\right] \left(\tau_1\right) \, ,
\end{align}

\item the basic non-degenerate Morawetz estimate
\begin{align} \label{basnondegmor2}
\overline{\mathbb{I}}_{\eta} \left[\Psi^{[-2]}\right] \left(\tau_1,\tau_2\right) 
 \lesssim & \ \ \overline{\mathbb{E}}_{\widetilde\Sigma_{\tau},\eta} \left[\Psi^{[-2]}\right] \left(\tau_1\right)  +\mathbb{E}_{\widetilde\Sigma_{\tau}} \left[\uppsi^{[-2]}\right] \left(\tau_1\right)  
+ \mathbb{E}_{\widetilde\Sigma_{\tau}} \left[\upalpha^{[-2]}\right] \left(\tau_1\right) \nonumber \\ & + \overline{\mathbb{E}}_{\widetilde\Sigma_{\tau},\eta} \left[T\Psi^{[-2]}\right] \left(\tau_1\right)  +\mathbb{E}_{\widetilde\Sigma_{\tau}} \left[T\uppsi^{[-2]}\right] \left(\tau_1\right)  
+ \mathbb{E}_{\widetilde\Sigma_{\tau}} \left[T\upalpha^{[-2]}\right] \left(\tau_1\right) \, , .
\end{align}
\item the $\eta$-weighted energy boundedness estimate
\begin{align} \label{ewbnd2}
& \ \ \overline{\mathbb{E}}_{\mathcal{H}^+} \left[\Psi^{[-2]}\right] \left(\tau_1,\tau_2\right) +\overline{\mathbb{E}}_{\widetilde\Sigma_{\tau},\eta} \left[\Psi^{[-2]}\right] \left(\tau_2\right)  \nonumber \\ \lesssim & \ \ \overline{\mathbb{E}}_{\widetilde\Sigma_{\tau},\eta} \left[\Psi^{[-2]}\right] \left(\tau_1\right)  +\mathbb{E}_{\widetilde\Sigma_{\tau}} \left[\uppsi^{[-2]}\right] \left(\tau_1\right)  
+ \mathbb{E}_{\widetilde\Sigma_{\tau}} \left[\upalpha^{[-2]}\right] \left(\tau_1\right) \, .
\end{align}
\end{itemize}

\end{theorem}

\subsection{Proof of Theorem \ref{prop:basicnondeg}}
We only prove the $s=+2$ case. The $s=-2$ case is completely analogous and slightly easier because the term $\mathcal{J}^{[-2]}$ has stronger degeneration near the event horizon. 
Note that in Section \ref{iledsec} we have already proven the estimates (\ref{basdegmor}) and (\ref{ewbnd}) provided we drop all overbars from the energies that appear. The estimate (\ref{basnondegmor}), which does not degenerate in a neighbourhood of $r=3M$ but loses a derivative, is a simple corollary of (\ref{basdegmor}) and (\ref{ewbnd}) again provided we drop all overbars from the energies.
Hence the only task left is to improve the $\underline{L}$ derivative in the energies that appear. This is achieved using the standard redshift multiplier:

\paragraph{The redshift identity.}
Recall the notational conventions of Section~\ref{sec:multid}.
Multiplying (\ref{RWtypeinthebulk}) by $Y=\frac{1}{w} \xi \underline{L}\overline{\Psi}$ (with $\xi$ a smooth radial cut-off function equal to $1$ for $r \in \left[r_+, r_+ +\frac{1}{4}M\right]$ and equal to zero for $r \geq r_+ + \frac{1}{2}M$) and taking the real parts yields (use the formulae of Appendix \ref{sec:comprs})
\begin{align} \label{redshiftphysid}
\underline{L} \big\{F^{Y}_{\underline{L}}\big\} +L \big\{F^{Y}_{L}  \big\} + I^{Y} \equiv {\rm Re} \left(-\left(\mathcal{J}^{[s]} + \mathfrak{G}^{[s]}\right) \frac{1}{w} \xi \underline{L}\overline{\Psi}\right)
\end{align}
where
\begin{align} \label{redshiftFY}
F^{Y}_L 
&=  \frac{1}{2}\frac{1}{w} \xi |\underline{L}\Psi|^2+\frac{1}{2} a \xi \textrm{Re} \left( \Phi \Psi \underline{L} \overline{\Psi} \right) + \frac{1}{2}  \xi a^2 \sin^2\theta \textrm{Re} \left( T\Psi \underline{L}\overline{\Psi}\right),   \\
F^Y_{\underline{L}} &= \frac{1}{2} \xi |\mathring{\slashed{\nabla}}^{[s]}  \Psi |^2 + \frac{1}{2} \left( \xi \left[s^2 -\frac{6Mr}{r^2+a^2} \frac{r^2-a^2}{r^2+a^2} - \frac{7a^2\Delta}{(r^2+a^2)^2}  \right] \right) |\Psi|^2  \nonumber \\
& \ \ \ \ + \frac{a^2}{r^2+a^2}  \xi  |\Phi\Psi|^2 - \frac{1}{2}a\xi \textrm{Re} \left( \Phi\Psi L\overline{\Psi} \right) -\frac{1}{2} \xi a^2 \sin^2\theta |T\Psi|^2 + \frac{1}{2}  \xi a^2 \sin^2\theta \textrm{Re} \left( T\Psi \underline{L}\overline{\Psi}\right),
\\
 \label{redshiftIY}
I^{Y} &=-\frac{1}{2} \left(\frac{\xi}{w}\right)^\prime |\underline{L}\Psi|^2  +\frac{1}{2}\xi^\prime |\mathring{\slashed{\nabla}} \Psi |^2 + \frac{1}{2} \left( \xi \left[s^2 -\frac{6Mr}{r^2+a^2} \frac{r^2-a^2}{r^2+a^2} - \frac{7a^2\Delta}{(r^2+a^2)^2}  \right] \right)^\prime |\Psi|^2 \nonumber \\
& \ \ \  + \frac{2r a \xi}{r^2+a^2}  \textrm{Re} \left( \Phi \Psi \underline{L}\Psi\right)  - \left[ \underline{L} \left( \frac{a^2}{r^2+a^2}  \xi \right)\right] |\Phi\Psi|^2- \frac{1}{2} a \xi^\prime \textrm{Re} \left( \Phi \Psi \left(\underline{L} + L\right) \overline{\Psi}\right)- \frac{1}{2} a\xi \textrm{Re} \left( \Phi \Psi \left[L, \underline{L} \right] \overline{\Psi}\right)
\nonumber \\
& \ \ \ -\frac{1}{2} \xi^\prime a^2 \sin^2\theta | T\Psi|^2 -2sa \cos \theta \xi \textrm{Im} \left(T\Psi \underline{L} \overline{\Psi} \right) .
\end{align}

We apply the identity (\ref{redshiftphysid}) to the equation satisfied by $\Psi^{[+2]}$. In particular, $ \mathfrak{G}^{[s]}=0$ because $\alpha^{[+2]}$ satisfies the \emph{homogeneous} Teukolsky equation.
Upon integration over $\widetilde{\mathcal{R}}\left(\tau_1,\tau_2\right)$ 
(recalling Remark \ref{rem:convert}) we obtain (\ref{basdegmor}) and (\ref{ewbnd}) 
after making the following observations:
\begin{itemize}
\item The first term in $F^Y_L$ and the first term in $F^Y_{\underline{L}}$ are manifestly non-negative and produce precisely the desired improvement in the $\underline{L}$ derivative and the missing angular derivative in the horizon term respectively. All other terms appearing as boundary terms can now be controlled using Cauchy--Schwarz and (\ref{azimuthal}) by the energies without the overbar (sometimes borrowing an $\epsilon$ from the just obtained good $\underline{L}$-derivative term and the good angular term respectively is required).
\item Examining (\ref{redshiftIY}), the term $\frac{1}{2} \xi \frac{w^\prime}{w^2} | \underline{L} \Psi|^2$ is manifestly positive and produces precisely the desired improvement of the $|\underline{L} \Psi|^2$ in the spacetime energy without the overbar. All other terms can be controlled by the spacetime energy without the overbar, sometimes borrowing an $\epsilon$ from the improved $|\underline{L} \Psi|^2$ term.
\item The error term 
\[
\int_{\mathcal{M}\left(\tau_1,\tau_2\right)} \Big|\mathcal{J}^{[+2]} \frac{1}{w} \xi \underline{L}\overline{\Psi}\Big| \frac{1}{\rho^2} \frac{r^2+a^2}{\Delta}
\]
is controlled using Cauchy's inequality with $\epsilon$ and the energies $\mathbb{I}_0 \left[\uppsi^{[+2]}\right]  \left(\tau_1,\tau_2\right)
+\mathbb{I}_0 \left[\upalpha^{[+2]}\right]  \left(\tau_1,\tau_2\right)$.
\end{itemize}

Finally, the estimate (\ref{basnondegmor}) follows from its un-overbarred version by adding the just established (\ref{basdegmor}).

\section{The $r^p$-weighted hierarchy and the main decay result}
\label{rphierarchysec}
To complete the proof of Theorem~\ref{finalstatetheor},
it remains to obtain statements 3.~and 4.~concerning the $r^p$-weighted
hierarchy and polynomial decay. 
The required statement is  contained in {\bf Theorem~\ref{prop:maindecayprop}} below.

\subsection{Statement of the decay theorem}

\begin{theorem} \label{prop:maindecayprop}
Let $\upalpha^{[\pm2]}$, $\Psi^{[\pm2]}$ and
$\uppsi^{[\pm2]}$ be as in Theorem~\ref{finalstatetheor}. Then  the following holds
 for any $\tau > \tau_0=0$.

For $\boxed{s=+2}$ we have
\begin{align}
\overline{\mathbb{E}}_{\widetilde\Sigma_{\tau},\eta} \left[\Psi^{[+2]}\right] \left(\tau\right)   + {\mathbb{E}}_{\widetilde\Sigma_{\tau},\eta} \left[ \uppsi^{[+2]}\right] \left(\tau\right) + \mathbb{E}_{\widetilde\Sigma_{\tau},\eta} \left[ \upalpha^{[+2]}\right] \left(\tau\right) \lesssim \frac{\mathbb{D}_{2, 2} \left[\Psi^{[+2]}, \uppsi^{[+2]}, \upalpha^{[+2]}\right] \left(\tau_0\right)}{\tau^{2-\eta}} \,  \label{step1}
\end{align}
for the initial data energy
\begin{align}
\mathbb{D}_{2, 2} \left[\Psi^{[+2]}, \uppsi^{[+2]}, \upalpha^{[+2]}\right] \left(\tau_0\right) = &\sum_{k=0}^1 \left(\overline{\mathbb{E}}_{\widetilde\Sigma_{\tau},2} \left[T^k\Psi^{[+2]}\right] \left(\tau_0\right) +   \mathbb{E}_{\widetilde\Sigma_{\tau},2} \left[T^k \uppsi^{[+2]}\right] \left(\tau_0\right) + \mathbb{E}_{\widetilde\Sigma_{\tau},2} \left[T^k\upalpha^{[+2]}\right] \left(\tau_0\right)   \right)
\nonumber \\  
 &+ \overline{\mathbb{E}}_{\widetilde\Sigma_{\tau},\eta} \left[T^{2}\Psi^{[+2]}\right] \left(\tau_0\right)  +\mathbb{E}_{\widetilde\Sigma_{\tau},\eta} \left[T^{2}\uppsi^{[+2]}\right] \left(\tau_0\right)  
+ \mathbb{E}_{\widetilde\Sigma_{\tau},\eta} \left[T^{2}\upalpha^{[+2]}\right] \left(\tau_0\right) \nonumber \, .
\end{align}
For $\boxed{s=-2}$ we have
\begin{align} 
\overline{\mathbb{E}}_{\widetilde\Sigma_{\tau},\eta} \left[\Psi^{[-2]}\right] \left(\tau\right)  +   {\mathbb{E}}_{\widetilde\Sigma_{\tau}} \left[ \uppsi^{[-2]}\right] \left(\tau\right) +\mathbb{E}_{\widetilde\Sigma_{\tau}} \left[ \upalpha^{[-2]}\right] \left(\tau\right)  \lesssim \frac{\mathbb{D}_{2, 2} \left[\Psi^{[-2]}, \uppsi^{[-2]}, \upalpha^{[-2]}\right] \left(\tau_0\right)}{\tau^{2-\eta}} \,  \label{step1n}
\end{align}
and
\begin{align}
\overline{\mathbb{E}}_{\widetilde\Sigma_{\tau},\eta} \left[\Psi^{[-2]}\right] \left(\tau\right)  +   \overline{\mathbb{E}}_{\widetilde\Sigma_{\tau}} \left[ \uppsi^{[-2]}\right] \left(\tau\right) +\overline{\mathbb{E}}_{\widetilde\Sigma_{\tau}} \left[ \upalpha^{[-2]}\right] \left(\tau\right)  \lesssim \frac{\overline{\mathbb{D}}_{2, 2} \left[\Psi^{[-2]}, \uppsi^{[-2]}, \upalpha^{[-2]}\right] \left(\tau_0\right)}{\tau^{2-\eta}} \,  \label{step3n}
\end{align}
for the initial data energy
\begin{align}
\mathbb{D}_{2, 2} \left[\Psi^{[-2]}, \uppsi^{[-2]}, \upalpha^{[-2]}\right] \left(\tau_0\right) = &\sum_{k=0}^1 \left(\overline{\mathbb{E}}_{\widetilde\Sigma_{\tau},2} \left[T^k\Psi^{[-2]}\right] \left(\tau_0\right) +    \mathbb{E}_{\widetilde\Sigma_{\tau}} \left[T^k \uppsi^{[-2]}\right] \left(\tau_0\right) + \mathbb{E}_{\widetilde\Sigma_{\tau}} \left[T^k\upalpha^{[-2]}\right] \left(\tau_0\right)   \right)
\nonumber \\  
 &+ \overline{\mathbb{E}}_{\widetilde\Sigma_{\tau},\eta} \left[T^{2}\Psi^{[-2]}\right] \left(\tau_0\right)  +\mathbb{E}_{\widetilde\Sigma_{\tau}} \left[T^{2}\uppsi^{[-2]}\right] \left(\tau_0\right)  
+ \mathbb{E}_{\widetilde\Sigma_{\tau}} \left[T^{2}\upalpha^{[-2]}\right] \left(\tau_0\right) \nonumber \, ,
\end{align}
and with $\overline{\mathbb{D}}_{2, 2} \left[\Psi^{[-2]}, \uppsi^{[-2]}, \upalpha^{[-2]}\right] \left(\tau_0\right)$ defined by putting an overbar on all energies appearing in \\ $\mathbb{D}_{2, 2} \left[\Psi^{[-2]}, \uppsi^{[-2]}, \upalpha^{[-2]}\right] \left(\tau_0\right)$.
\end{theorem}

\subsection{Proof of Theorem \ref{prop:maindecayprop} for $s=+2$}

The $s=+2$ case of Theorem \ref{prop:maindecayprop} will be proven in Section \ref{sec:finishit} by combining basic estimates from the $r^p$ hierarchy associated with the inhomogeneous wave equation satisfied by $\Psi^{[+2]}$ (derived in Section \ref{sec:rph}) and basic transport estimates for $\uppsi^{[+2]}$ and $\upalpha^{[+2]}$ (derived in Section \ref{sec:rptrans}).

\subsubsection{The weighted $r^p$ hierarchy for $\Psi^{[+2]}$ in physical space} \label{sec:rph}

\begin{proposition} \label{prop:nondegweighted}
Under the assumptions of Theorem~\ref{prop:maindecayprop} we have
for any $\tau_2 > \tau_1\geq 0$ 
and for $p=2$, $p=1$ and $p=\eta$ the estimate
\begin{align}
& \ \ \  \overline{\mathbb{E}}_{\widetilde\Sigma_{\tau},p} \left[\Psi^{[+2]}\right] \left(\tau_2\right)  + \overline{\mathbb{I}}^{\rm deg}_{p} \left[\Psi^{[+2]}\right] \left(\tau_1,\tau_2\right) + \mathbb{E}_{\mathcal{I}^+,p} \left[\Psi^{[+2]}\right] \left(\tau_1,\tau_2\right) \nonumber \\
& \lesssim  \overline{\mathbb{E}}_{\widetilde\Sigma_{\tau},p} \left[\Psi^{[+2]}\right] \left(\tau_1\right)  +\mathbb{E}_{\widetilde\Sigma_{\tau},\eta} \left[\uppsi^{[+2]}\right] \left(\tau_1\right)  
+ \mathbb{E}_{\widetilde\Sigma_{\tau},\eta} \left[\upalpha^{[+2]}\right] \left(\tau_1\right) \, \nonumber  
\end{align}
and the non-degenerate estimate
\begin{align} \label{estph}
& \ \  \ \overline{\mathbb{E}}_{\widetilde\Sigma_{\tau},p} \left[\Psi^{[+2]}\right] \left(\tau_2\right)  + \overline{\mathbb{I}}_{p} \left[\Psi^{[+2]}\right] \left(\tau_1,\tau_2\right) + \mathbb{E}_{\mathcal{I}^+,p} \left[\Psi^{[+2]}\right] \left(\tau_1,\tau_2\right) \nonumber \\
& \lesssim  \overline{\mathbb{E}}_{\widetilde\Sigma_{\tau},p} \left[\Psi^{[+2]}\right] \left(\tau_1\right)  +\mathbb{E}_{\widetilde\Sigma_{\tau},\eta} \left[\uppsi^{[+2]}\right] \left(\tau_1\right)  
+  \mathbb{E}_{\widetilde\Sigma_{\tau},\eta} \left[\upalpha^{[+2]}\right] \left(\tau_1\right) \nonumber \\ 
& + \overline{\mathbb{E}}_{\widetilde\Sigma_{\tau},\eta} \left[T\Psi^{[+2]}\right] \left(\tau_1\right)  +\mathbb{E}_{\widetilde\Sigma_{\tau},\eta} \left[T\uppsi^{[+2]}\right] \left(\tau_1\right)  
+ \mathbb{E}_{\widetilde\Sigma_{\tau},\eta} \left[T\upalpha^{[+2]}\right] \left(\tau_1\right) \nonumber . 
\end{align}
\end{proposition}

\begin{proof}
Given $\alpha^{[+2]}$ we 
apply the multiplier identity (\ref{rpphysspaceid}) to $\Psi^{[+2]}$. To the identity that is being produced after integration over $\widetilde{\mathcal{R}}(\tau_1,\tau_2)$, we can add a large constant $B$ (depending only on $M$) times the basic estimate (\ref{basdegmor}) such that the following holds: For the boundary term we have for all $p \in \left[\eta,2\right]$
\begin{align}  
\int_{\widetilde{\mathcal{R}}(\tau_1,\tau_2)} \left( L \big\{ F^{r^p}_L\big\} + \underline{L} \big\{ F^{r^p}_{\underline{L}} \big\}\right)
 \frac{1}{\rho^2} \frac{r^2+a^2}{\Delta} dVol + B \cdot \overline{\mathbb{E}}_{\widetilde{\Sigma}_\tau,\eta} \left[\Psi^{[+2]}\right] (\tau_2) \nonumber \\
\gtrsim   b \cdot \overline{\mathbb{E}}_{\widetilde{\Sigma}_\tau,p} \left[\Psi^{[+2]}\right] (\tau_2) - B \cdot \overline{\mathbb{E}}_{\widetilde{\Sigma}_\tau,p} \left[\Psi^{[+2]}\right] (\tau_1) 
 + b\ \mathbb{E}_{\mathcal{I}^+,p} \left[\Psi^{[+2]}\right] \left(\tau_1,\tau_2\right) \, .
\end{align}
For the spacetime term we have
\begin{equation} \label{ugc1}
\int_{\widetilde{\mathcal{R}}(\tau_1,\tau_2)} \left( I^{r^{p}}\right)
 \frac{1}{\rho^2} \frac{r^2+a^2}{\Delta} dVol + B \cdot \overline{\mathbb{I}}^{deg}_{\eta} \left[\Psi^{[+2]}\right] (\tau_1,\tau_2)
\geq b \cdot \overline{\mathbb{I}}^{\rm deg}_{p} \left[\Psi^{[+2]}\right] (\tau_1,\tau_2) \, 
\end{equation}
for $p \in \left[\eta,2\right)$ and for $p=2$
\begin{equation} \label{ugc1b}
\sum_{p=2-\eta, p=2} \int_{\widetilde{\mathcal{R}}(\tau_1,\tau_2)} \left( I^{r^{p}}\right)
 \frac{1}{\rho^2} \frac{r^2+a^2}{\Delta} dVol + B \cdot \overline{\mathbb{I}}^{deg}_{\eta} \left[\Psi^{[+2]}\right] (\tau_1,\tau_2)
\geq b \cdot \overline{\mathbb{I}}^{\rm deg}_{2} \left[\Psi^{[+2]}\right] (\tau_1,\tau_2) \, ,
\end{equation}
the latter case being special because for $p=2$ we lose control of the angular derivatives in (\ref{ipb}). For the error term (which in view of $\xi$ being supported for large $r$ is supported for large $r$) we have, for any $\lambda>0$,
\begin{align} 
\int_{\widetilde{\mathcal{R}}(\tau_1,\tau_2)} \Big| \mathcal{J}^{[+2]}| |\xi| |\beta_4| |r^p L \Psi^{[+2]}|  \frac{r^2+a^2}{\Delta \rho^2} dVol  \lesssim  \int_{\widetilde{\mathcal{R}}(\tau_1,\tau_2)} dVol   \frac{r^2+a^2}{\Delta \rho^2} \left( \lambda r^{p-1} |L\Psi^{[+2]}|^2 + \frac{r^{p+1}}{\lambda}  \Big| \mathcal{J}^{[+2]}|^2 \right) \nonumber \\
\lesssim \lambda  \overline{\mathbb{I}}^{deg}_{p} \left[\Psi^{[+2]}\right] (\tau_1,\tau_2) + \frac{a^2}{\lambda}\left( \mathbb{I}_\eta \left[\uppsi^{[+2]}\right]  \left(\tau_1,\tau_2\right)
+\mathbb{I}_\eta \left[\upalpha^{[+2]}\right]  \left(\tau_1,\tau_2\right)\right) \, . \nonumber
\end{align}
Note that there is no $\mathfrak{G}^{[+2]}$ error term as $F^{[+2]}=0$ and hence $\mathfrak{G}^{[+2]}=0$.
Combining the above estimates yields the first estimate of the Proposition after using the basic estimate (\ref{basdegmor}) yields and choosing $\lambda$ sufficiently small (depending only on $M$). The second estimate follows immediately be combining the first one with the non-degenerate (\ref{basnondegmor}).
\end{proof}

\subsubsection{Physical space weighted transport for $\uppsi^{[+2]}$ and $P^{[+2]}$} \label{sec:rptrans}
We now turn to deriving weighted Morawetz and boundedness estimates for $\uppsi^{[+2]}$ and $\upalpha^{[+2]}$ from the transport equations they satisfy. Combining (\ref{weightedtransportright}) with the basic estimate (\ref{basdegmor}) we immediately obtain

\begin{proposition} \label{prop:weightedtransport}
Under the assumptions of~Theorem~\ref{prop:maindecayprop} we have
for any $\tau_2 > \tau_1\geq 0$ and for $p \in \{\eta,1,2\}$ the estimate
\begin{align} 
\mathbb{E}_{\widetilde\Sigma_{\tau},p} \left[\upalpha^{[+2]}\right] \left(\tau_2\right) + \mathbb{I}_p \left[\upalpha^{[+2]}\right]  \left(\tau_1,\tau_2\right) \lesssim \mathbb{I}_p \left[\uppsi^{[+2]}\right]  \left(\tau_1,\tau_2\right) + \mathbb{E}_{\widetilde\Sigma_{\tau},p} \left[\upalpha^{[+2]}\right] \left(\tau_1\right)
\end{align}
and the estimate
\begin{align} 
\mathbb{E}_{\widetilde\Sigma_{\tau},p} \left[\uppsi^{[+2]}\right] \left(\tau_2\right) + \mathbb{I}_p \left[\uppsi^{[+2]}\right]  \left(\tau_1,\tau_2\right) &\lesssim \mathbb{I}^{\rm deg}_p \left[\Psi^{[+2]}\right]  \left(\tau_1,\tau_2\right) + \mathbb{E}_{\widetilde\Sigma_{\tau},p} \left[\uppsi^{[+2]}\right] \left(\tau_1\right) \nonumber \\
& \ \ + \mathbb{E}_{\widetilde\Sigma_{\tau},\eta} \left[\Psi^{[+2]}\right] \left(\tau_1\right) 
+ \mathbb{E}_{\widetilde\Sigma_{\tau},\eta} \left[\upalpha^{[+2]}\right] \left(\tau_1\right) \
\end{align}
\end{proposition}

\subsubsection{Completing the proof of Theorem \ref{prop:maindecayprop}}  \label{sec:finishit}

Combining the estimate of Proposition \ref{prop:nondegweighted} with that of Proposition \ref{prop:weightedtransport} we deduce for $p \in \{\eta, 1, 2\}$ (first for $K=0$ and then by trivial commutation with the Killing field $T$ for any $K \in \mathbb{N}$) the estimate

\begin{align}
&\sum_{k=0}^K \left( \mathbb{E}_{\widetilde\Sigma_{\tau},p} \left[T^k\upalpha^{[+2]}\right] \left(\tau_2\right) + \mathbb{E}_{\widetilde\Sigma_{\tau},p} \left[T^k \uppsi^{[+2]}\right] \left(\tau_2\right) + \overline{\mathbb{E}}_{\widetilde\Sigma_{\tau},p} \left[T^k\Psi^{[+2]}\right] \left(\tau_2\right)   \right)
\nonumber \\
 +&\sum_{k=0}^K \left( \mathbb{I}_p \left[T^k\upalpha^{[+2]}\right]  \left(\tau_1,\tau_2\right)+\mathbb{I}_p \left[T^k\uppsi^{[+2]}\right]  \left(\tau_1,\tau_2\right) 
 + \overline{\mathbb{I}}^{\rm deg}_{p} \left[T^k\Psi^{[+2]}\right] \left(\tau_1,\tau_2\right)  \right)
 \nonumber \\
 \lesssim &\sum_{k=0}^K \left( \mathbb{E}_{\widetilde\Sigma_{\tau},p} \left[T^k\upalpha^{[+2]}\right] \left(\tau_1\right) + \mathbb{E}_{\widetilde\Sigma_{\tau},p} \left[T^k \uppsi^{[+2]}\right] \left(\tau_1\right) + \overline{\mathbb{E}}_{\widetilde\Sigma_{\tau},p} \left[T^k\Psi^{[+2]}\right] \left(\tau_1\right)   \right)
 \label{dyadcon}
\end{align}
and also
\begin{align}
&\sum_{k=0}^K \left( \mathbb{E}_{\widetilde\Sigma_{\tau},p} \left[T^k\upalpha^{[+2]}\right] \left(\tau_2\right) + \mathbb{E}_{\widetilde\Sigma_{\tau},p} \left[T^k \uppsi^{[+2]}\right] \left(\tau_2\right) + \overline{\mathbb{E}}_{\widetilde\Sigma_{\tau},p} \left[T^k\Psi^{[+2]}\right] \left(\tau_2\right)   \right)
\nonumber \\
 +&\sum_{k=0}^K \left( \mathbb{I}_p \left[T^k\upalpha^{[+2]}\right]  \left(\tau_1,\tau_2\right)+\mathbb{I}_p \left[T^k\uppsi^{[+2]}\right]  \left(\tau_1,\tau_2\right) 
 + \overline{\mathbb{I}}_{p} \left[T^k\Psi^{[+2]}\right] \left(\tau_1,\tau_2\right)  \right)
 \nonumber \\
 \lesssim &\sum_{k=0}^K \left( \mathbb{E}_{\widetilde\Sigma_{\tau},p} \left[T^k\upalpha^{[+2]}\right] \left(\tau_1\right) + \mathbb{E}_{\widetilde\Sigma_{\tau},p} \left[T^k \uppsi^{[+2]}\right] \left(\tau_1\right) + \overline{\mathbb{E}}_{\widetilde\Sigma_{\tau},p} \left[T^k\Psi^{[+2]}\right] \left(\tau_1\right)   \right)
\nonumber \\  
 &    
+ \mathbb{E}_{\widetilde\Sigma_{\tau},\eta} \left[T^{K+1}\upalpha^{[+2]}\right] \left(\tau_1\right) +\mathbb{E}_{\widetilde\Sigma_{\tau},\eta} \left[T^{K+1}\uppsi^{[+2]}\right] \left(\tau_1\right) + \overline{\mathbb{E}}_{\widetilde\Sigma_{\tau},\eta} \left[T^{K+1}\Psi^{[+2]}\right] \left(\tau_1\right) \, . \label{dyad}
\end{align}
Let us denote the right hand side of the second estimate on the initial data slice $\widetilde\Sigma_0$ (i.e.~for $\tau_1=\tau_0$) by $\mathbb{D}_{K+1, p} \left[\Psi^{[+2]}, \uppsi^{[+2]}, \upalpha^{[+2]}\right] \left(\tau_0\right)$. 

Applying (\ref{dyad}) for $K=1$ and $p=2$ implies (after using a standard argument involving dyadic sequences) along a dyadic sequence $\tau_n \sim 2^{n}\tau_0$ the estimate
\begin{align}
\sum_{k=0}^1 \left( \mathbb{E}_{\widetilde\Sigma_{\tau},1} \left[T^k \upalpha^{[+2]}\right] \left(\tau_n\right) + \mathbb{E}_{\widetilde\Sigma_{\tau},1} \left[T^k \uppsi^{[+2]}\right] \left(\tau_n\right) + \overline{\mathbb{E}}_{\widetilde\Sigma_{\tau},1} \left[T^k\Psi^{[+2]}\right] \left(\tau_n\right) \right)  \lesssim \frac{\mathbb{D}_{2, 2} \left[\Psi^{[+2]}, \uppsi^{[+2]}, \upalpha^{[+2]}\right] \left(\tau_0\right)}{\tau_n} \, . \nonumber
\end{align}
Using the above and applying (\ref{dyadcon}) for $p=1$, $K=1$ between the time $\tau_1=\tau_n$ and any $\tau_2 \in \left(\tau_n, \tau_{n+1}\right]$ yields the previous estimate for any $\tau$, not only the members of the dyadic sequence.
Turning back to (\ref{dyad}) now with $K=0$ we use the previous estimate and a similar dyadic argument to produce along a dyadic sequence the estimate
\begin{align}
\mathbb{E}_{\widetilde\Sigma_{\tau},\eta} \left[ \upalpha^{[+2]}\right] \left(\tau_n\right) + \mathbb{E}_{\widetilde\Sigma_{\tau},\eta} \left[ \uppsi^{[+2]}\right] \left(\tau_n\right) + \overline{\mathbb{E}}_{\widetilde\Sigma_{\tau},\eta} \left[\Psi^{[+2]}\right] \left(\tau_n\right)  \lesssim \frac{\mathbb{D}_{2, 2} \left[\Psi^{[+2]}, \uppsi^{[+2]}, \upalpha^{[+2]}\right] \left(\tau_0\right)}{\tau_n^{2-\eta}} \, . \nonumber 
\end{align}
Using the above and applying (\ref{dyadcon}) with $p=\eta$ and $K=0$ now yields the estimate (\ref{step1}) of Theorem \ref{prop:maindecayprop}.

\subsection{Proof of Theorem \ref{prop:maindecayprop} for $s=-2$}

The $s=-2$ case of Theorem \ref{prop:maindecayprop} will be proven in Section \ref{sec:finishit2} by combining basic estimates from the $r^p$ hierarchy associated with the inhomogeneous wave equation satisfied by $\Psi^{[-2]}$ (derived in Section \ref{sec:rph2}) and basic transport estimates for $\uppsi^{[-2]}$ and $\upalpha^{[-2]}$ (derived in Section \ref{sec:rptrans2}).

\subsubsection{The weighted $r^p$ hierarchy for $\Psi^{[-2]}$ in physical space} \label{sec:rph2}

\begin{proposition} \label{prop:nondegweighted2}
Under the assumptions of Theorem \ref{prop:maindecayprop} we have
for any $\tau_2 > \tau_1\geq 0$ 
and for $p=2$, $p=1$ and $p=\eta$ the estimate
\begin{align} \label{estph1a}
& \ \ \overline{\mathbb{E}}_{\widetilde\Sigma_{\tau},p} \left[\Psi^{[-2]}\right] \left(\tau_2\right)  + \overline{\mathbb{I}}^{\rm deg}_{p} \left[\Psi^{[-2]}\right] \left(\tau_1,\tau_2\right) + \mathbb{E}_{\mathcal{I}^+,p} \left[\Psi^{[-2]}\right] \left(\tau_1,\tau_2\right) \nonumber \\ 
 \lesssim & \ \ \overline{\mathbb{E}}_{\widetilde\Sigma_{\tau},p} \left[\Psi^{[-2]}\right] \left(\tau_1\right)  +\mathbb{E}_{\widetilde\Sigma_{\tau}} \left[\uppsi^{[-2]}\right] \left(\tau_1\right)  
+ \mathbb{E}_{\widetilde\Sigma_{\tau}} \left[\upalpha^{[-2]}\right] \left(\tau_1\right)  \nonumber \\
& \ \ + |a|\left(\mathbb{E}_{\mathcal{I}^+} \left[\uppsi^{[-2]}\right] \left(\tau_1,\tau_2\right) +\mathbb{E}_{\mathcal{I}^+} \left[\upalpha^{[-2]}\right] \left(\tau_1,\tau_2\right)\right)
\end{align}
and the non-degenerate estimate
\begin{align} \label{estph2}
& \ \ \overline{\mathbb{E}}_{\widetilde\Sigma_{\tau},p} \left[\Psi^{[-2]}\right] \left(\tau_2\right)  + \overline{\mathbb{I}}_{p} \left[\Psi^{[-2]}\right] \left(\tau_1,\tau_2\right) + \mathbb{E}_{\mathcal{I}^+,p} \left[\Psi^{[-2]}\right] \left(\tau_1,\tau_2\right) \nonumber \\ 
 \lesssim & \ \ \overline{\mathbb{E}}_{\widetilde\Sigma_{\tau},p} \left[\Psi^{[-2]}\right] \left(\tau_1\right)  +\mathbb{E}_{\widetilde\Sigma_{\tau}} \left[\uppsi^{[-2]}\right] \left(\tau_1\right)  
+ \mathbb{E}_{\widetilde\Sigma_{\tau}} \left[\upalpha^{[-2]}\right] \left(\tau_1\right) \nonumber \\ & + \overline{\mathbb{E}}_{\widetilde\Sigma_{\tau},\eta} \left[T\Psi^{[-2]}\right] \left(\tau_1\right)  +\mathbb{E}_{\widetilde\Sigma_{\tau}} \left[T\uppsi^{[-2]}\right] \left(\tau_1\right)  
+ \mathbb{E}_{\widetilde\Sigma_{\tau}} \left[T\upalpha^{[-2]}\right] \left(\tau_1\right) 
\nonumber \\
& \ \ + a\left(\mathbb{E}_{\mathcal{I}^+} \left[\uppsi^{[-2]}\right] \left(\tau_1,\tau_2\right) +\mathbb{E}_{\mathcal{I}^+} \left[\upalpha^{[-2]}\right] \left(\tau_1,\tau_2\right)\right) .
\end{align}
\end{proposition}

\begin{proof}
The proof is exactly as in Proposition \ref{prop:nondegweighted} except that we need to inspect carefully the error term $\mathcal{J}^{[-2]}$. (This is of course because the Regge--Wheeler operators are almost identical for $s=\pm 2$.) Checking the $r$-weights in the application of the 
Cauchy--Schwarz inequality, the analogous computation
\begin{align} 
\int_{\widetilde{\mathcal{R}}(\tau_1,\tau_2)} \Big| \mathcal{J}^{[-2]}| |\xi| |\beta_4| |r^p L \Psi^{[-2]}|  \frac{r^2+a^2}{\Delta \rho^2} dVol  \lesssim  \int_{\widetilde{\mathcal{R}}(\tau_1,\tau_2)} dVol   \frac{r^2+a^2}{\Delta \rho^2} \left( \lambda r^{p-1} |L\Psi^{[-2]}|^2 + \frac{r^{p+1}}{\lambda}  \Big| \mathcal{J}^{[-2]}|^2 \right) \nonumber \\
\lesssim \lambda  \overline{\mathbb{I}}^{deg}_{p} \left[\Psi^{[-2]}\right] (\tau_1,\tau_2) + \frac{a^2}{\lambda}\left( \mathbb{I} \left[\uppsi^{[-2]}\right]  \left(\tau_1,\tau_2\right)
+\mathbb{I} \left[\upalpha^{[-2]}\right]  \left(\tau_1,\tau_2\right)\right) \, \nonumber
\end{align}
is seen to be valid only for $p\in \left[\eta,2-\eta\right]$. For $p=2$ we need to integrate by parts. Note that the two worst (the others being controlled by the above estimate for $\lambda$ depending only on $M$) contributions from the error $\mathcal{J}^{[-2]}\beta_4 \xi r^p L \overline{\Psi}$ can be written (omitting taking real parts for the moment)
\begin{align}
 ar^{-2} \Phi \left(\sqrt{\Delta} \psi^{[-2]}\right) r^2 L \overline{\Psi^{[-2]}}  = L \left( a\Phi \left(\sqrt{\Delta} \psi^{[-2]}\right) \overline{\Psi^{[-2]}}\right) + a\Delta \left(r^2+a^2\right)^{-2} \left(\Phi \Psi^{[-2]}\right) \overline{\Psi^{[-2]}}  \, ,
\end{align} and
\begin{align}
 a^2 r^{-2} \left(r^2+a^2\right)^{-3/2} \upalpha^{[-2]} \ r^2 L \overline{\Psi^{[-2]}}  = L \left( a^2\left(r^2+a^2\right)^{-3/2} \upalpha^{[-2]}\overline{\Psi^{[-2]}}\right) - \frac{a^2\Delta}{ \left(r^2+a^2\right)^{2}} \left(\sqrt{\Delta}\psi^{[-2]} \right) \overline{\Psi^{[-2]}} \nonumber \, ,
\end{align}
where we have used the relations (\ref{auxrel2}) and (\ref{auxrel1}). Now upon taking real parts and integration, the second term in each line can be controlled by the basic Cauchy--Schwarz inequality, the integration by parts having gained a power in $r$. The first term in each line is a boundary term and controlled
by the terms appearing on the right hand side of the estimate (\ref{estph1a}), where for the boundary term on null infinity we borrow from the term $\mathbb{E}_{\mathcal{I}^+,p} \left[\Psi^{[-2]}\right] \left(\tau_1,\tau_2\right)$ appearing on the left hand side.
\end{proof}

\subsubsection{Physical space weighted transport for ${\uppsi}^{[-2]}$ and $\Psi^{[-2]}$}
\label{sec:rptrans2}

\begin{proposition} \label{prop:weightedtransport2}
Under the assumptions of Theorem~\ref{prop:maindecayprop} we have
for any $\tau_2 > \tau_1\geq 0$ the estimate
\begin{align} \label{setx2}
\mathbb{E}_{\widetilde\Sigma_{\tau}} \left[\upalpha^{[-2]}\right] \left(\tau_2\right) + \mathbb{I} \left[\upalpha^{[-2]}\right]  \left(\tau_1,\tau_2\right) + \mathbb{E}_{\mathcal{I}^+} \left[\upalpha^{[-2]}\right] \left(\tau_1,\tau_2\right) \lesssim \mathbb{I} \left[\uppsi^{[-2]}\right]  \left(\tau_1,\tau_2\right) + \mathbb{E}_{\widetilde\Sigma_{\tau}} \left[\upalpha^{[-2]}\right] \left(\tau_1\right)
\end{align}
and the estimate
\begin{align} \label{sety2}
\mathbb{E}_{\widetilde\Sigma_{\tau}} \left[\uppsi^{[-2]}\right] \left(\tau_2\right) + \mathbb{I} \left[\uppsi^{[-2]}\right]  \left(\tau_1,\tau_2\right)+\mathbb{E}_{\mathcal{I}^+} \left[\uppsi^{[-2]}\right] \left(\tau_1,\tau_2\right) \nonumber \\ 
\lesssim \mathbb{I}_{\eta}^{\rm deg} \left[\Psi^{[-2]}\right]  \left(\tau_1,\tau_2\right) + \mathbb{E}_{\widetilde\Sigma_{\tau}} \left[\uppsi^{[-2]}\right] \left(\tau_1\right) 
+ \mathbb{E}_{\widetilde\Sigma_{\tau},\eta} \left[\Psi^{[-2]}\right] \left(\tau_1\right) 
+ \mathbb{E}_{\widetilde\Sigma_{\tau}} \left[\upalpha^{[-2]}\right] \left(\tau_1\right) \, .
\end{align}
The same estimates hold with an overbar on all terms.
\end{proposition}

\begin{proof}
Multiply (\ref{klo}) by a cut-off function $\xi$ which is equal to $1$ for $r\geq 9M$ and equal to zero for $r\leq 8M$. Bringing $\xi$ inside the first bracket produces an error-term supported in $\left[8M,9M\right]$, which (upon integration) is for any $n$ controlled by the basic estimate (\ref{basdegmor2}). Upon integration of the resulting identity we deduce (\ref{sety2}) for $\Gamma$ being the identity in the energies appearing. Now we observe that the same estimate holds for the $T$ and $\Phi$ commuted equations (note that (\ref{Pbart}) commutes trivially with the Killing fields $T$ and $\Phi$ so the estimate (\ref{idp2}) trivially holds for the commuted variables). 

The estimate (\ref{setx2}) is proven completely analogously except that here no cut-off is required in view of the non-degenerate norm of $\uppsi^{[-2]}$ appearing on the right hand side: One first applies (\ref{klp}) and the same estimate for the $T$ and $\Phi$ commuted variables. 

To obtain the estimates with an overbar 
one first commutes (\ref{Pbart}) and (\ref{psibart}) with $ \underline{L}$ and notes that the analogue of (\ref{klo}) and (\ref{klp}) can now be applied with the error from the commutator $\left[L, \underline{L}\right] \sim \frac{a}{r^3} \Phi $ being controlled by the previous step.
Secondly, one commutes (\ref{Pbart}) and (\ref{psibart})  with the vectorfield $\frac{r^2+a^2}{\Delta} \underline{L}$ which extends regularly to the horizon and observes that the additional commutator term leads to a good sign (near the horizon) in the estimates (\ref{klo}) and (\ref{klp}).
\end{proof}

\subsubsection{Completing the proof of Theorem \ref{prop:maindecayprop} for $s=-2$} \label{sec:finishit2}
Combining the estimate of Proposition \ref{prop:nondegweighted2} with that of Proposition \ref{prop:weightedtransport2} we deduce for $p \in \{\eta, 1, 2\}$ (first for $K=0$ and then by trivial commutation with the Killing field $T$ for any $K \in \mathbb{N}$) the estimate

\begin{align}
&\sum_{k=0}^K \left( \mathbb{E}_{\widetilde\Sigma_{\tau}} \left[T^k\upalpha^{[-2]}\right] \left(\tau_2\right) + \mathbb{E}_{\widetilde\Sigma_{\tau}} \left[T^k \uppsi^{[-2]}\right] \left(\tau_2\right) + \overline{\mathbb{E}}_{\widetilde\Sigma_{\tau},p} \left[T^k\Psi^{[-2]}\right] \left(\tau_2\right)   \right)
\nonumber \\
 +&\sum_{k=0}^K \left( \mathbb{I} \left[T^k\upalpha^{[-2]}\right]  \left(\tau_1,\tau_2\right)+\mathbb{I} \left[T^k\uppsi^{[-2]}\right]  \left(\tau_1,\tau_2\right) 
 + \overline{\mathbb{I}}^{\rm deg}_{p} \left[T^k\Psi^{[-2]}\right] \left(\tau_1,\tau_2\right)  \right)
 \nonumber \\
 \lesssim &\sum_{k=0}^K \left( \mathbb{E}_{\widetilde\Sigma_{\tau}} \left[T^k\upalpha^{[-2]}\right] \left(\tau_1\right) + \mathbb{E}_{\widetilde\Sigma_{\tau}} \left[T^k \uppsi^{[-2]}\right] \left(\tau_1\right) + \overline{\mathbb{E}}_{\widetilde\Sigma_{\tau},p} \left[T^k\Psi^{[-2]}\right] \left(\tau_1\right)   \right)
 \label{dyadcon2}
\end{align}
and also
\begin{align}
&\sum_{k=0}^K \left( \mathbb{E}_{\widetilde\Sigma_{\tau}} \left[T^k\upalpha^{[-2]}\right] \left(\tau_2\right) + \mathbb{E}_{\widetilde\Sigma_{\tau}} \left[T^k \uppsi^{[-2]}\right] \left(\tau_2\right) + \overline{\mathbb{E}}_{\widetilde\Sigma_{\tau},p} \left[T^k\Psi^{[-2]}\right] \left(\tau_2\right)   \right)
\nonumber \\
 +&\sum_{k=0}^K \left( \mathbb{I} \left[T^k\upalpha^{[-2]}\right]  \left(\tau_1,\tau_2\right)+\mathbb{I} \left[T^k\uppsi^{[-2]}\right]  \left(\tau_1,\tau_2\right) 
 + \overline{\mathbb{I}}_{p} \left[T^k\Psi^{[-2]}\right] \left(\tau_1,\tau_2\right)  \right)
 \nonumber \\
 \lesssim &\sum_{k=0}^K \left( \mathbb{E}_{\widetilde\Sigma_{\tau}} \left[T^k\upalpha^{[-2]}\right] \left(\tau_1\right) + \mathbb{E}_{\widetilde\Sigma_{\tau}} \left[T^k \uppsi^{[-2]}\right] \left(\tau_1\right) + \overline{\mathbb{E}}_{\widetilde\Sigma_{\tau},p} \left[T^k\Psi^{[-2]}\right] \left(\tau_1\right)   \right) \nonumber \\
 & \ +  \mathbb{E}_{\widetilde\Sigma_{\tau}} \left[T^{K+1}\upalpha^{[-2]}\right] \left(\tau_1\right) + \mathbb{E}_{\widetilde\Sigma_{\tau}} \left[T^{K+1} \uppsi^{[-2]}\right] \left(\tau_1\right) + \overline{\mathbb{E}}_{\widetilde\Sigma_{\tau},\eta} \left[T^{K+1}\Psi^{[-2]}\right] \left(\tau_1\right)  \, .
 \label{dyad2}
\end{align}
Let us denote the right hand side of the second estimate on the initial data slice $\widetilde\Sigma_0$ (i.e.~for $\tau_1=\tau_0$) by $\mathbb{D}_{K+1, p} \left[\Psi^{[-2]}, \uppsi^{[-2]}, \upalpha^{[-2]}\right] \left(\tau_0\right)$. 

Applying the estimate (\ref{dyad2}) for $K=1$ and $p=2$ implies (after using a standard argument involving dyadic sequences) along a dyadic sequence $\tau_n \sim 2^{n}\tau_0$ the estimate
\begin{align}
\sum_{k=0}^1 \left( \mathbb{E}_{\widetilde\Sigma_{\tau}} \left[T^k \upalpha^{[-2]}\right] \left(\tau_n\right) + \mathbb{E}_{\widetilde\Sigma_{\tau}} \left[T^k \uppsi^{[-2]}\right] \left(\tau_n\right) + \overline{\mathbb{E}}_{\widetilde\Sigma_{\tau},1} \left[T^k\Psi^{[-2]}\right] \left(\tau_n\right) \right)  \lesssim \frac{\mathbb{D}_{2, 2} \left[\Psi^{[-2]}, \uppsi^{[-2]}, \upalpha^{[-2]}\right] \left(\tau_0\right)}{\tau_n} \, . \nonumber
\end{align}
Using the above and applying (\ref{dyadcon2}) for $p=1$, $K=1$ between the time $\tau_1=\tau_n$ and any $\tau_2 \in \left(\tau_n, \tau_{n+1}\right]$ yields the previous estimate for any $\tau$, not only the members of the dyadic sequence.
Turning back to (\ref{dyad2}) now with $K=0$ we use the previous estimate and a similar dyadic argument to produce along a dyadic sequence the estimate
\begin{align}
\mathbb{E}_{\widetilde\Sigma_{\tau}} \left[ \upalpha^{[-2]}\right] \left(\tau_n\right) + \mathbb{E}_{\widetilde\Sigma_{\tau}} \left[ \uppsi^{[-2]}\right] \left(\tau_n\right) + \overline{\mathbb{E}}_{\widetilde\Sigma_{\tau},\eta} \left[\Psi^{[-2]}\right] \left(\tau_n\right)  \lesssim \frac{\mathbb{D}_{2, 2} \left[\Psi^{[-2]}, \uppsi^{[-2]}, \upalpha^{[-2]}\right] \left(\tau_0\right)}{\tau_n^{2-\eta}} \, . \nonumber 
\end{align}
Using the above and applying (\ref{dyadcon2}) with $p=\eta$ and $K=0$ now yields (\ref{step1n}) of Theorem \ref{prop:maindecayprop}. To obtain the second estimate, one simply repeats the above proof using that the estimate of Proposition \ref{prop:weightedtransport2} also holds for the energies with an overbar.

\appendix

\section{Derivation of the equation for $\Psi^{[s]}$}
\label{sec:Psiderivation}

\subsection{The Teukolsky equation}
The Teukolsky equation has appeared in various forms throughout the paper. The most concise is perhaps the mode decomposed form (\ref{basic}), which we use as a starting point in our derivations. Of course all computations in this appendix could be carried out also in physical space. It is only for the purpose of cleaner notation that we use the mode decomposed version for our derivations.

Recall from (\ref{wdefinit}) the definition
\[
w = \frac{\Delta}{\left(r^2+a^2\right)^2} \, .
\]
Using the relation
\[
\frac{w^{\prime \prime}}{w} - 2 \frac{\left(w^\prime\right)^2}{w^2}-V_0^{[+2]} -2w= \frac{w^{\prime \prime}}{w} - 2 \frac{\left(w^\prime\right)^2}{w^2}-V_0^{[-2]} +2w = -3w \frac{a^4+a^2r^2-2Mr^3}{(r^2+a^2)^2} +2w \, ,
\]
we can rewrite the separated Teukolsky equation (\ref{basic}) for spin $s=+2$ as
\begin{align} \label{teudnp}
-L\underline{L}  \left(u^{[+2]}w\right) = & -2\frac{w^\prime}{w} \underline{L} \left(u^{[+2]}w\right) + w \left(\Lambda_{m \ell}^{[+2]} + 2\right) \left(u^{[+2]}w\right) - 3 w\left(+2\right) \frac{r}{r^2+a^2} iam\left(u^{[+2]}w\right) \nonumber \\
&+ \left(u^{[+2]}w\right) \left(-3w \frac{a^4+a^2r^2-2Mr^3}{(r^2+a^2)^2} +2w\right) - \frac{\Delta}{\rho^2} F^{[+2]}\, .
\end{align}
Similarly, we can write the separated Teukolsky equation for spin $s=-2$ as
\begin{align} \label{teudnm}
-\underline{L}L  \left(u^{[-2]}w\right) = & 2\frac{w^\prime}{w} L  \left(u^{[-2]}w\right) + w \left(\Lambda_{m \ell}^{[-2]} - 2\right)  \left(u^{[-2]}w\right) - 3 w\left(-2\right) \frac{r}{r^2+a^2} iam \left(u^{[-2]}w\right) \nonumber \\
&+  \left(u^{[-2]}w\right)\left(-3w \frac{a^4+a^2r^2-2Mr^3}{(r^2+a^2)^2} +2w\right)- \frac{\Delta}{\rho^2}F^{[-2]} \, .
\end{align}
Introducing
\[
Q^{[\pm 2]} = \mp 6 \frac{r}{r^2+a^2} iam -3 \frac{a^4+a^2r^2-2Mr^3}{(r^2+a^2)^2} +2
\]
as well as recalling the definitions
\begin{align} \label{ghu}
\underline{L} \left(u^{[+2]} \cdot w\right) = - 2w \sqrt{\Delta}\psi^{[+2]} \ \ \ , \ \ \ \underline{L} \left(\sqrt{\Delta} \psi^{[+2]}\right) = w \Psi^{[+2]} \, ,
\end{align}
we can write (\ref{teudnp}) as
\begin{align} \label{teubest}
2 L \left(\sqrt{\Delta}\psi^{[+2]}\right) - 2\sqrt{\Delta}\psi^{[+2]} \frac{w^\prime}{w}  =  \left(\Lambda_{m \ell}^{[+2]} + 2\right) \left(u^{[+2]}w\right) + Q^{[+2]} \left(u^{[+2]}w\right) - \frac{\Delta}{w\rho^2} F^{[+2]} 
 \, .
\end{align}
Similarly, recalling the definitions 
\begin{align} \label{ghu2}
{L} \left(u^{[-2]} \cdot w\right) =  2w \sqrt{\Delta}\psi^{[-2]} \ \ \ , \ \ \ {L} \left(\sqrt{\Delta} \psi^{[-2]}\right) = -w \Psi^{[-2]} \, ,
\end{align}
we can write (\ref{teudnm}) as 
\begin{align} 
-2 \underline{L} \left(\sqrt{\Delta}{\psi}^{[-2]}\right) - 2\sqrt{\Delta}{\psi}^{[-2]} \frac{w^\prime}{w}  =  \left(\Lambda_{m \ell}^{[-2]} - 2\right) \left({u}^{[-2]}w\right) +Q^{[-2]} \left(u^{[-2]}w\right) - \frac{\Delta^3}{w\rho^2} F^{[-2]} \, . \nonumber
\end{align} 
\begin{remark}
Recall that ${u^{[-2]}w} \sim \Delta^2$ and ${\psi}^{[-2]} \sim \sqrt{\Delta}$ near the horizon and rewriting this slightly one sees that the left hand side is also $\mathcal{O}\left(\Delta^2\right)$:
\[
-2 \underline{L} \left(\sqrt{\Delta}{\psi}^{[-2]}\right) - 2\sqrt{\Delta}{\psi}^{[-2]} \frac{w^\prime}{w}  = +2 (r^2+a^2) w\underline{L} \left(\frac{\sqrt{\Delta}{\psi}^{[-2]}}{w(r^2+a^2)}\right) - 2\sqrt{\Delta}{\psi}^{[-2]}\frac{2r \Delta}{(r^2+a^2)^2} \, .
\]
\end{remark}

\subsection{Derivation of the $\Psi^{[+2]}$ equation (separated version)}
  
Observing the commutation relation
\[
\left[\underline{L},L\right] = \frac{4ra}{r^2+a^2} w \cdot i m \, ,
\]
we obtain after applying $\underline{L}$ to the Teukolsky equation in the form (\ref{teubest}) recalling again (\ref{ghu})
\begin{align} 
2 L \left(w \Psi^{[+2]}\right) + \frac{8ra}{r^2+a^2} w \cdot im\left(\sqrt{\Delta} \psi^{[+2]}\right) - 2w^\prime \Psi^{[+2]} + 2\left(\sqrt{\Delta} \psi^{[+2]}\right) \left(\frac{w^\prime}{w}\right)^\prime \nonumber \\
= -2w \left(\Lambda_{m \ell}^{[+2]} + 2\right) \left(\sqrt{\Delta} \psi^{[+2]}\right) -2w Q^{[+2]} \left(\sqrt{\Delta}\psi^{[+2]}\right) - \left(\partial_{r^*} Q^{[+2]}\right) \left(u^{[+2]}w\right) -\underline{L} \left(\frac{\Delta}{w\rho^2} F^{[+2]}\right)\, ,
\end{align}
which we immediately simplify to 
\begin{align}  \label{elliptic1}
 L \left( \Psi^{[+2]}\right) + \frac{4ra}{r^2+a^2}  \cdot im\left(\sqrt{\Delta} \psi^{[+2]}\right)  + \left(\sqrt{\Delta} \psi^{[+2]}\right) \frac{1}{w} \left(\frac{w^\prime}{w}\right)^\prime \nonumber \\
= - \left(\Lambda_{m \ell}^{[+2]} + 2\right) \left(\sqrt{\Delta} \psi^{[+2]}\right) - Q^{[+2]} \left(\sqrt{\Delta}\psi^{[+2]}\right) - \frac{1}{2w}\left(Q^{[+2]}\right)^\prime \left(u^{[+2]}w\right) - \frac{1}{2w}\underline{L} \left(\frac{\Delta}{w\rho^2} F^{[+2]}\right)\, .
\end{align}
We now apply another $\underline{L}$ derivative, which produces 
\begin{align}  \label{elm}
 L \underline{L} \left( \Psi^{[+2]}\right) &+w \frac{8ra}{r^2+a^2}  \cdot im \Psi^{[+2]} + w Q^{[+2]} \Psi^{[+2]}
 + \left(\frac{w^\prime}{w}\right)^\prime \Psi^{[+2]}  +w \left(\Lambda_{m \ell}^{[+2]} + 2\right) \Psi^{[+2]}
 \nonumber \\
= &\left[ 4aim \left(\frac{r}{r^2+a^2}\right)^\prime + \left(\frac{1}{w} \left(\frac{w^\prime}{w}\right)^\prime \right)^\prime + 2\left(Q^{[+2]}\right)^\prime \right] \left(\sqrt{\Delta}\psi^{[+2]}\right) \nonumber \\
 &+ \left[ \left(\frac{1}{2w} \left(Q^{[+2]}\right)^\prime \right)^\prime \right] \left(u^{[+2]}w\right) -\frac{1}{2}\underline{L} \left( \frac{1}{w} \underline{L} \left(\frac{\Delta}{w\rho^2} F^{[+2]}\right)\right)\, . 
\end{align}
Using elementary algebra we simplify the terms in the first line of (\ref{elm}) to
\begin{align}
 L \underline{L} \left( \Psi^{[+2]}\right) &+w \frac{2ra}{r^2+a^2}  \cdot im \Psi^{[+2]} + w \left(\Lambda_{m \ell}^{[+2]} +6 - \frac{6M}{r} \frac{r^2-a^2}{r^2+a^2} -7a^2 w\right) \Psi^{[+2]} \, ,
\end{align}
which can be written more succinctly as
\begin{align}
\frac{1}{2} \left( L \underline{L} + \underline{L} L \right) \left( \Psi^{[+2]}\right) + w \left(\Lambda_{m \ell}^{[+2]} +6 - \frac{6M}{r} \frac{r^2-a^2}{r^2+a^2} -7a^2 w\right) \Psi^{[+2]} \, .
\end{align}
The terms in the second line simplify to
\begin{align}
w \left[ -8aim \cdot \frac{-r^2+a^2}{r^2+a^2} + 20a^2 \frac{r^3-3Mr^2+ra^2+Ma^2}{\left(r^2+a^2\right)^2}\right]\sqrt{\Delta}\psi^{[+2]} \, .
\end{align}
Finally, for the third line of (\ref{elm}) excluding the inhomogeneous term we obtain
\begin{align}
\left[ +12a^3 w \frac{r}{r^2+a^2} im -3a^2 w  \frac{r^4 -a^4+10Mr^3-6Ma^2r}{(r^2+a^2)^2} \right] \left(u^{[+2]} w\right) \, .
\end{align}

\subsection{Derivation of the $\Psi^{[-2]}$ equation (separated version)}
  
Observing the commutation relation
\[
\left[\underline{L},L\right] = \frac{4ra}{r^2+a^2} w \cdot i m \, ,
\]
we obtain after applying ${L}$ to the Teukolsky equation in the form 
\begin{align} 
-2 \underline{L} \left(\sqrt{\Delta}{\psi}^{[-2]}\right) - 2\sqrt{\Delta}{\psi}^{[-2]} \frac{w^\prime}{w}  =  \left(\Lambda_{m \ell}^{[-2]} - 2\right) \left({u}^{[-2]}w\right) +Q^{[-2]} \left(u^{[-2]}w\right) - \frac{\Delta^3}{w\rho^2} F^{[-2]} \, \nonumber
\end{align}
recalling again (\ref{ghu2})
\begin{align} 
+2 \underline{L} \left(w \Psi^{[-2]}\right) + \frac{8ra}{r^2+a^2} w \cdot im\left(\sqrt{\Delta} \psi^{[-2]}\right) + 2w^\prime \Psi^{[-2]} - 2\left(\sqrt{\Delta} \psi^{[-2]}\right) \left(\frac{w^\prime}{w}\right)^\prime \nonumber \\
= 2w \left(\Lambda_{m \ell}^{[-2]} - 2\right) \left(\sqrt{\Delta} \psi^{[-2]}\right) +2w Q^{[-2]} \left(\sqrt{\Delta}\psi^{[-2]}\right) + \left(\partial_{r^*} Q^{[-2]}\right) \left(u^{[-2]}w\right)- L \left(\frac{\Delta^3}{w\rho^2} F^{[-2]}\right) \, ,\nonumber
\end{align}
which we immediately simplify to 
\begin{align} \label{elliptic2}
 \underline{L} \left( \Psi^{[-2]}\right) + \frac{4ra}{r^2+a^2}  \cdot im\left(\sqrt{\Delta} \psi^{[-2]}\right)  - \left(\sqrt{\Delta} \psi^{[-2]}\right) \frac{1}{w} \left(\frac{w^\prime}{w}\right)^\prime \nonumber \\
=  \left(\Lambda_{m \ell}^{[-2]} - 2\right) \left(\sqrt{\Delta} \psi^{[-2]}\right) + Q^{[-2]} \left(\sqrt{\Delta}\psi^{[-2]}\right) +\frac{1}{2w}\left(Q^{[-2]}\right)^\prime \left(u^{[-2]}w\right) - \frac{1}{2w}L \left(\frac{\Delta^3}{w\rho^2} F^{[-2]}\right)\, .
\end{align}
We now apply another ${L}$ derivative, which produces 
\begin{align}  \label{elm2}
 L \underline{L} \left( \Psi^{[-2]}\right) &-w \frac{4ra}{r^2+a^2}  \cdot im \Psi^{[-2]} + w Q^{[-2]} \Psi^{[-2]}
 + \left(\frac{w^\prime}{w}\right)^\prime \Psi^{[-2]}  +w \left(\Lambda_{m \ell}^{[-2]} - 2\right) \Psi^{[-2]}
 \nonumber \\
= &\left[ -4aim \left(\frac{r}{r^2+a^2}\right)^\prime + \left(\frac{1}{w} \left(\frac{w^\prime}{w}\right)^\prime \right)^\prime + 2\left(Q^{[+2]}\right)^\prime \right] \left(\sqrt{\Delta}\psi^{[-2]}\right) \nonumber \\
 &+ \left[ \left(\frac{1}{2w} \left(Q^{[+2]}\right)^\prime \right)^\prime \right] \left(u^{[-2]}w\right)- \frac{1}{2} L \left( \frac{1}{w}L \left(\frac{\Delta^3}{w\rho^2} F^{[-2]}\right)\right) \, .
\end{align}
Using elementary algebra we simplify the terms in the first line of (\ref{elm2}) to
\begin{align}
 L \underline{L} \left( \Psi^{[-2]}\right) &+w \frac{2ra}{r^2+a^2}  \cdot im \Psi^{[-2]} + w \left(\Lambda_{m \ell}^{[-2]} +2 - \frac{6M}{r} \frac{r^2-a^2}{r^2+a^2} -7a^2 w\right) \Psi^{[-2]} \, ,
\end{align}
which can be written more succinctly as
\begin{align}
\frac{1}{2} \left( L \underline{L} + \underline{L} L \right) \left( \Psi^{[-2]}\right) + w \left(\Lambda_{m \ell}^{[-2]} +2 - \frac{6M}{r} \frac{r^2-a^2}{r^2+a^2} -7a^2 w\right) \Psi^{[-2]} \, .
\end{align}
The terms in the second line simplify to
\begin{align}
w \left[ +8aim \cdot \frac{-r^2+a^2}{r^2+a^2} + 20a^2 \frac{r^3-3Mr^2+ra^2+Ma^2}{\left(r^2+a^2\right)^2}\right]\sqrt{\Delta}\psi^{[-2]} \, .
\end{align}
Finally, for the third line of (\ref{elm}) excluding the inhomogeneous term we obtain
\begin{align}
\left[ -12a^3 w \frac{r}{r^2+a^2} im -3a^2 w  \frac{r^4 -a^4+10Mr^3-6Ma^2r}{(r^2+a^2)^2} \right] \left(u^{[-2]} w\right) \, .
\end{align}

\subsection{Summary: The $\Psi^{[\pm 2]}$ equation (separated version)} \label{appendix:RW}
The separated equation for $\Psi^{[+2]}$ and $\Psi^{[-2]}$ can be written as 
\begin{align} \label{dnf}
-\frac{1}{2} \left( L \underline{L}+ \underline{L} L\right) \Psi^{[s]}  &- \frac{\Delta}{\left(r^2+a^2\right)^2} \left(\lambda_{m \ell}^{[s]} -2am\omega + a^2 \omega^2 +s^2 +s \right)  \Psi^{[s]} \nonumber \\ &
+ \frac{\Delta}{\left(r^2+a^2\right)^2} \frac{6Mr}{r^2+a^2} \frac{r^2-a^2}{r^2+a^2}\Psi^{[s]}  
 +7a^2 \frac{\Delta^2}{(r^2+a^2)^4} \Psi^{[s]}  = \frac{\Delta}{\rho^2}\mathcal{J}^{[s]} + \mathfrak{G}^{[s]} \, ,
\end{align}
where the right hand is side given by
\begin{align}
\mathcal{J}^{[s]} = & \frac{\rho^2}{\left(r^2+a^2\right)^2} \left[ -4s\frac{r^2-a^2}{r^2+a^2} aim -20a^2 \frac{r^3-3Mr^2+ra^2+Ma^2}{\left(r^2+a^2\right)^2}\right]\left( \sqrt{\Delta}\psi^{[s]} \right) \nonumber \\
& +a^2 \frac{\rho^2}{\left(r^2+a^2\right)^2} \left[ -6s \frac{r}{r^2+a^2} aim
+3  \left( \frac{r^4 -a^4+10Mr^3-6Ma^2r}{(r^2+a^2)^2}\right) \right] \left(u^{[s]}w\right) \, ,
\end{align}
\begin{align}
\mathfrak{G}^{[+2]}= \frac{1}{2} \underline{L} \left( \frac{1}{w}\underline{L} \left(\frac{\Delta}{w\rho^2} F^{[+2]}\right)\right) \ \ \ , \ \ \ \mathfrak{G}^{[-2]}= \frac{1}{2} L \left( \frac{1}{w}L \left(\frac{\Delta^3}{w\rho^2} F^{[-2]}\right)\right) \, .
\end{align}
Finally, note that we can write (\ref{dnf}) also as
\[
\left({\Psi}^{[s]}\right)^{\prime \prime} + \left(\omega^2 - \mathcal{V}^{[+2]} \right){\Psi}^{[s]} = \frac{\Delta}{\rho^2}\mathcal{J}^{[s]} +\mathfrak{G}^{[s]}\, 
\]
for the potential
\begin{align} \label{Vdefappendix}
\mathcal{V}^{[s]}&=\frac{ \Delta \left(\lambda^{[s]}_{m \ell} +a^2\omega^2+s^2+s\right) + 4Mram\omega -a^2 m^2 }{\left(r^2+a^2\right)^2} -\frac{\Delta}{(r^2+a^2)^2}\frac{6Mr (r^2-a^2)}{(r^2+a^2)^2} -7 a^2 \frac{\Delta^2}{(r^2+a^2)^4} \nonumber \\
&= \mathcal{V}^{[s]}_0 + \mathcal{V}^{[s]}_1 + \mathcal{V}^{[s]}_2 \, .
\end{align}

\subsection{The $\Psi^{[\pm2]}$ equation (physical space)}
We can now translate the equation in the previous subsection to physical space to obtain the result of Proposition \ref{follfundprop}. We have
\begin{align}
-\frac{1}{2} \left( L \underline{L} + \underline{L} L\right) \Psi^{[s]}  &- \frac{\Delta}{\left(r^2+a^2\right)^2} \Big\{ \left(\mathring{\slashed\triangle}^{[s]}_m \left(0\right)+s^2 +s\right) \Psi^{[s]} -\frac{6Mr}{r^2+a^2} \frac{r^2-a^2}{r^2+a^2} \Psi^{[s]} -7a^2 \frac{\Delta}{(r^2+a^2)^2} \Psi^{[s]}\Big\}       \nonumber \\
& + \frac{\Delta}{\left(r^2+a^2\right)^2}\left(2 a T \Phi (\Psi^{[s]}) +a^2 \sin^2\theta TT(\Psi^{[s]}) - 2i s a\cos \theta T (\Psi^{[s]} ) \right) = \mathcal{J}^{[s]} + \mathfrak{G}^{[s]},
\end{align}
where
\begin{align}
\mathcal{J}^{[+2]} = & \frac{\Delta}{\left(r^2+a^2\right)^2} \left[ \frac{-8r^2+8a^2}{r^2+a^2} a\Phi -20a^2 \frac{r^3-3Mr^2+ra^2+Ma^2}{\left(r^2+a^2\right)^2}\right]\left( \sqrt{\Delta}\uppsi^{[+2]} \right) \nonumber \\
& +a^2 \frac{\Delta}{\left(r^2+a^2\right)^2} \left[ -12 \frac{r}{r^2+a^2}  a\Phi
+3  \left( \frac{r^4 -a^4+10Mr^3-6Ma^2r}{(r^2+a^2)^2}\right) \right] \left(\upalpha^{[+2]}\Delta^2 \left(r^2+a^2\right)^{-\frac{3}{2}}\right) \nonumber
\end{align}
and
\begin{align}
\mathcal{J}^{[-2]} = 
& \frac{\Delta}{\left(r^2+a^2\right)^2} \left[ \frac{8r^2-8a^2}{r^2+a^2} a \Phi  -20a^2 \frac{r^3-3Mr^2+ra^2+Ma^2}{\left(r^2+a^2\right)^2} \right]\left( \sqrt{\Delta}\psi^{[-2]} \right) \nonumber \\
& +a^2 \frac{\Delta}{\left(r^2+a^2\right)^2} \left[  +12 \frac{r}{r^2+a^2} a \Phi
+3  \left( \frac{r^4 -a^4+10Mr^3-6Ma^2r}{(r^2+a^2)^2}\right) \right] \left(\upalpha^{[-2]} \left(r^a+a^2\right)^{-\frac{3}{2}}\right)  \, .
\end{align}

 \newpage

\section{Auxilliary calculations for physical space multipliers}
\label{otherappendix}

\def\Psib{\overline\Psi}
We first recall the relations
\[
L + \underline{L} = 2T +2\frac{a}{r^2+a^2} \Phi \, \ \ , \ \ L-\underline{L} = 2 \partial_{r^*} \, .
\]
We will consider the identities generated by the following four multipliers (the smooth radial cut-offs $\chi$, $\xi$ and the smooth radial functions $f$, $h$, $y$ are chosen appropriately in the body of the paper)
\begin{enumerate}
\item The $T$-energy: $T\Psib$
\item The Lagragian multiplier: $h \Psib$
\item The $\Phi$-multiplier: $\omega_+\chi\Phi \Psib$  ($\chi$ a radial cut-off)
\item The $y$-multiplier: $f \left(L-\underline{L}\right)\Psib$
\item The redshift multiplier: $\frac{1}{w}\xi \underline{L} \Psib$ ($\xi$ a radial cut-off near the horizon)
\item The $r^p$ weighed multiplier: $r^p \beta_k \xi L\overline{\Psi}$ with $\beta_k=1+k\frac{M}{r}$ ($\xi$ a radial cut-off near infinity)
\end{enumerate}
each acting on the second order terms in the equation (\ref{RWtypeinthebulk}), namely (recall $w=\frac{\Delta}{\left(r^2+a^2\right)^2}$)
\begin{enumerate}[I.]
\item $\frac{1}{2} \left( L \underline{L} + \underline{L} L \right)\Psi$
\item $w\mathring{\slashed\triangle}^{[s]}_m \left(0\right) \Psi = \frac{1}{2} w\left(\eth \overline{\eth}+\overline{\eth} \eth \right)\Psi$
\item $ w 2 a T \Phi \Psi$
\item $ w a^2 \sin^2\theta TT\Psi$
\end{enumerate}
The point is that the $0^{th}$ order terms in (\ref{RWtypeinthebulk}) are easy to handle while for the (only) first order term in (\ref{RWtypeinthebulk}), $2is w a \cos \theta T\Psi$, we observe that for $X$ any real vectorfield commuting with $T$ we have
\begin{align}
2is w a \cos \theta T\Psi X \overline{\Psi} = & T \left(2is w a \cos \theta \Psi X \overline{\Psi}\right) - X \left( 2is w a \cos \theta \Psi T \overline{\Psi}  \right) + X \left(2isw a \cos \theta\right) \Psi T \overline{\Psi} \nonumber \\
& +2is w a \cos \theta X\Psi T\overline{\Psi} \, ,
\end{align}
and hence
\begin{align} \label{convert}
\textrm{Re} \left(2is wa \cos \theta T\Psi X\overline{\Psi} \right) &=  T \left(iswa \cos \theta  \Psi X \overline{\Psi} \right)- X \left( is w a \cos \theta \Psi T \overline{\Psi}  \right)+ X \left(isw a \cos \theta\right) \Psi T \overline{\Psi} \,  \nonumber \\
 &=  -T \left(swa \cos \theta \textrm{Im} \Psi X \overline{\Psi} \right)+ X \left( s w a \cos \theta \textrm{Im} \Psi T \overline{\Psi}  \right)- X \left(sw a \cos \theta\right) \textrm{Im} \Psi T \overline{\Psi}.
\end{align}
In particular for $X=T$ the right hand side is zero while for $X=\omega_+ \chi \Phi $ only the first two terms survive (and only the first after integration in $\phi$).

\newpage

\subsection{The $T$-multiplier: $T\Psib$} \label{sec:compT}
\subsubsection{Part I: $\frac{1}{2} \left( L \underline{L} + \underline{L} L \right)\Psi$}
\begin{align}
\frac{1}{2}{\rm Re} \left( L \underline{L} + \underline{L} L \right)\Psi\, T \Psib &= \frac{1}{4} {\rm Re}\left(L+\underline{L}\right) \left(L + \underline{L}\right) \Psi T \Psib -  \frac{1}{4} {\rm Re}\left(L-\underline{L}\right) \left(L - \underline{L}\right) \Psi T \Psib \nonumber \\
&=  \frac{1}{4} {\rm Re}\left( L + \underline{L}\right) \Big\{ \left(L + \underline{L}\right) \Psi T \Psib \Big\} - \frac{1}{8} T \Big\{ |\left(L+\underline{L}\right)\Psi|^2 \Big\} \nonumber \\
& \ \ -  \frac{1}{4} {\rm Re}\left( L - \underline{L}\right) \Big\{ \left(L - \underline{L}\right) \Psi T \Psib \Big\} + \frac{1}{8} T \Big\{ |\left(L-\underline{L}\right)\Psi|^2 \Big\}
\end{align}
which we write as
\begin{align}
\frac{1}{2} \textrm{Re} \left( L \underline{L} + \underline{L} L \right) \Psi T \overline{\Psi} &= \frac{1}{16} \left(L+\underline{L}\right) \Big\{  |\left(L+\underline{L}\right)\Psi|^2+ |\left(L-\underline{L}\right)\Psi|^2 -\frac{4a}{r^2+a^2}{\rm Re} \Phi\Psib \left(L+\underline{L}\right) \Psi \Big\} \nonumber \\
& \ \ +\frac{1}{8} \Phi \Big\{ \frac{a}{r^2+a^2} \left(  |\left(L+\underline{L}\right)\Psi|^2-|\left(L-\underline{L}\right)\Psi|^2\right)  \Big\} \nonumber \\
&\ \ -  \frac{1}{4} \left( L - \underline{L}\right) {\rm Re} \Big\{ \left(L - \underline{L}\right) \Psi T \Psib \Big\} 
\end{align}

\subsubsection{Part II: $w\left(\mathring{\slashed\triangle}^{[s]}_m \left(0\right) +s\right)\Psi$ (after integration over $\int \sin \theta d\theta d\phi$, see \eqref{laplacef})}

\begin{align}
 w{\rm Re}\left(\mathring{\slashed\triangle}^{[s]}_m \left(0\right) +s\right)\Psi T \Psib = +\frac{1}{2} T \Big\{ w |\mathring{\slashed{\nabla}}^{[s]} \Psi |^2\Big\}
\end{align}

\subsubsection{Part III: $w 2 a T \Phi \Psi$}
\begin{align}
 w 2 a {\rm Re} T \Phi \Psi T \Psib = \Phi \Big\{ aw |T\Psi|^2\Big\}
\end{align}

\subsubsection{Part IV: $w a^2 \sin^2\theta TT\Psi$}

\begin{align}
w a^2 \sin^2\theta {\rm Re} TT\Psi T \Psib = \frac{1}{2} T \Big\{ wa^2 \sin^2\theta |T\Psi|^2 \Big\}
\end{align}

\newpage

\subsection{The Lagrangian term: $h \Psib$} \label{sec:compL}
\subsubsection{Part I: $\frac{1}{2} \left( L \underline{L} + \underline{L} L \right)\Psi$}
\begin{align}
\frac{1}{2} \left( L \underline{L} + \underline{L} L \right){\rm Re} \Psi h \Psib &= \frac{1}{4} \left(L+\underline{L}\right) \left(L + \underline{L}\right) {\rm Re}\Psi h \Psib -  \frac{1}{4} \left(L-\underline{L}\right) \left(L - \underline{L}\right) {\rm Re}\Psi h \Psib \nonumber \\
&= \frac{1}{4} \left(L+\underline{L} \right) {\rm Re}\Big\{  \left(L + \underline{L}\right) \Psi h \Psib  \Big\} - \frac{1}{4} \left(L-\underline{L} \right) {\rm Re}\Big\{  \left(L - \underline{L}\right) \Psi h \Psib  \Big\} \nonumber \\
& \ \ \ +\frac{1}{4} h \Big[  |(L-\underline{L})\Psi|^2 -  |(L+\underline{L})\Psi|^2 \Big] + \frac{1}{4}h^\prime \left(L-\underline{L}\right) |\Psi|^2 \nonumber \\
&= \frac{1}{4} \left(L+\underline{L} \right) {\rm Re}\Big\{  \left(L + \underline{L}\right) \Psi h \Psi  \Big\} - \frac{1}{4} \left(L-\underline{L} \right) \Big\{  \left(L - \underline{L}\right) {\rm Re}\Psi h \Psi - h^\prime |\Psi|^2  \Big\} \nonumber \\
& \ \ \ +\frac{1}{4} h \Big[  |(L-\underline{L})\Psi|^2 -  |(L+\underline{L})\Psi|^2 \Big] - \frac{1}{2}h^{\prime \prime}  |\Psi|^2
\end{align}

\subsubsection{Part II: $w\left(\mathring{\slashed\triangle}^{[s]}_m \left(0\right) +s\right)\Psi $ (after integration over $\int \sin \theta d\theta d\phi$)}

\begin{align}
w\left(\mathring{\slashed\triangle}^{[s]}_m \left(0\right) +s\right)\Psi  h \Psib = + h w  |\mathring{\slashed{\nabla}}^{[s]} \Psi |^2
\end{align}

\subsubsection{Part III: $ w 2 a T \Phi \Psi$}
\begin{align}
 w 2 a {\rm Re} (T \Phi \Psi h \Psib) = \Phi {\rm Re}\left(w 2 a T \Psi h \Psib \right) - w2 a h {\rm Re}\{(T  \Psi)  (\Phi \Psib )\}
\end{align}

\subsubsection{Part IV: $wa^2 \sin^2\theta TT\Psi$}

\begin{align}
w a^2 \sin^2\theta{\rm Re}\{ TT\Psi h \Psib\} = T {\rm Re}\left(w a^2 \sin^2\theta T\Psi h\Psib \right) - w a^2 \sin^2\theta h |T\Psi|^2
\end{align}

\newpage

\subsection{The $\Phi$ multipier: $\omega_+\chi \Phi \Psib$} \label{sec:compphi}
\subsubsection{Part I: $\frac{1}{2} \left( L \underline{L} + \underline{L} L \right)\Psi$}

\begin{align}
\frac{1}{2} \left( L \underline{L} + \underline{L} L \right){\rm Re}\{\Psi\upomega_+\chi \Phi \Psib\}   &= \frac{1}{4} \left(L+\underline{L}\right) \left(L + \underline{L}\right) {\rm Re}\{\Psi \upomega_+\chi \Phi \Psib\} -  \frac{1}{4} \left(L-\underline{L}\right) \left(L - \underline{L}\right) {\rm Re}\{\Psi \upomega_+\chi \Phi \Psib\}\nonumber \\
&= \frac{1}{4} \left(L+\underline{L}\right) {\rm Re}\Big\{  \left(L + \underline{L}\right) \Psi \upomega_+\chi \Phi \Psib \Big\} - \frac{1}{8} \Phi \Big\{ \upomega_+ \chi\left|\left(L + \underline{L}\right) \Psi \right|^2 \Big\} \nonumber \\
& \ \ -\frac{1}{4} \left(L-\underline{L}\right) {\rm Re}\Big\{  \left(L - \underline{L}\right) \Psi \upomega_+\chi \Phi \overline{\Psi} \Big\} + \frac{1}{8} \Phi \Big\{ \upomega_+ \chi\left|\left(L - \underline{L}\right) \Psi \right|^2 \Big\} \nonumber \\
& \ \ +\frac{1}{2} \upomega_+ \chi^\prime  {\rm Re}\left((L-\underline{L})\Psi\right)(\Phi \overline{\Psi})
\end{align}

\subsubsection{Part II: $w\left(\mathring{\slashed\triangle}^{[s]}_m \left(0\right) + s\right) \Psi$ (after integration over $\int \sin \theta d\theta d\phi$)}

\begin{align}
w{\rm Re}\left(\mathring{\slashed\triangle}^{[s]}_m \left(0\right) + s\right) \Psi \upomega_+ \chi \Phi \Psib= +\frac{1}{2} \Phi \left(\upomega_+ \chi   |\mathring{\slashed{\nabla}}^s \Psi |^2 \right)
\end{align}

\subsubsection{Part III: $ w2 a T \Phi \Psi$}
\begin{align}
w2a{\rm Re}\{T\Phi \Psi \upomega_+\chi \Phi \Psib\} = T \Big\{ wa \upomega_+\chi |\Phi \Psi|^2 \Big\}
\end{align}

\subsubsection{Part IV: $w a^2 \sin^2\theta TT\Psi$}
\begin{align}
w a^2 \sin^2\theta {\rm Re}\{TT\Psi \upomega_+\chi \Phi \Psib\}  = T {\rm Re}\Big \{ w a^2 \sin^2\theta \upomega_+\chi (T\Psi)  \Phi \overline{\Psi} \Big\} - \frac{1}{2} \Phi \Big\{w a^2 \sin^2\theta \upomega_+\chi |T\Psi|^2 \Big\} 
\end{align}

\newpage 

\subsection{The $y$-multiplier: $y (L-$\underline{$L$}$)\Psib$} \label{sec:compy}
\subsubsection{Part I: $\frac{1}{2} \left( L \underline{L} + \underline{L} L \right)\Psi$}

\begin{align}
\frac{1}{2} {\rm Re}\{\left( L \underline{L} + \underline{L} L \right)\Psi \left(y (L-\underline{L})\Psib\right)\}  &= \frac{1}{4} {\rm Re}\{\left(L+\underline{L}\right) \left(L + \underline{L}\right) \Psi \left(y (L-\underline{L})\Psib\right)\}-  \frac{1}{4} {\rm Re}\{\left(L-\underline{L}\right) \left(L - \underline{L}\right) \Psi \left(y (L-\underline{L})\Psib\right)\} \nonumber \\
&=\frac{1}{4} {\rm Re}\left(L+ \underline{L}\right) \Big\{  y \left(L + \underline{L}\right) \Psi \left( (L-\underline{L})\Psi\right) \Big\}-\frac{1}{8} \left(L- \underline{L}\right) \Big\{  y  \left|(L-\underline{L})\Psi\right|^2 \Big\} \nonumber \\
&-\frac{1}{8}\left(L-\underline{L} \right) \Big\{ y  \left|(L+\underline{L})\Psi\right|^2 \Big\} + \frac{1}{4} y {\rm Re}\{\left[L-\underline{L},L+\underline{L}\right] \Psi \left(L+\underline{L}\right)\Psib\}
\nonumber \\
&+\frac{1}{4} y^\prime \left[ \left|(L+\underline{L})\Psi\right|^2 + \left|(L-\underline{L})\Psi\right|^2 \right]
\end{align}
Using the commutator identity
\[
\left[L-\underline{L},L+\underline{L}\right] = 2 \left[L,\underline{L} \right] = 4 \partial_{r^*} \left(\frac{a}{r^2+a^2}\right) \Phi = - \frac{8r a}{\left(r^2+a^2\right)^2} \frac{\Delta}{r^2+a^2} \Phi
\]
we conclude
\begin{align}
\frac{1}{2} \left( L \underline{L} + \underline{L} L \right)\Psi \left(y (L-\underline{L})\Psi\right)  
&=\frac{1}{4} {\rm Re}\left(L+ \underline{L}\right) \Big\{  y \left(L + \underline{L}\right) \Psi \left( (L-\underline{L})\Psi\right) \Big\} \nonumber \\
&-\frac{1}{8}\left(L-\underline{L} \right) \Big\{ y \left|(L+\underline{L})\Psi\right|^2 +y  \left|(L-\underline{L})\Psi\right|^2\Big\} 
\nonumber \\
&+\frac{1}{4} y^\prime \left[ \left|(L+\underline{L})\Psi\right|^2 + \left|(L-\underline{L})\Psi\right|^2 \right]-2y \frac{r a}{\left(r^2+a^2\right)^2} \frac{\Delta}{r^2+a^2} {\rm Re}\{\Phi \Psi \left(L+\underline{L}\right)\Psib\} \nonumber
\end{align}

\subsubsection{Part II: $w\left(\mathring{\slashed\triangle}^{[s]}_m \left(0\right) +s\right) \Psi$}

\begin{align}
{\rm Re}\int \sin \theta d\theta d\phi \,w\left(\mathring{\slashed\triangle}^{[s]}_m \left(0\right) +s\right) \Psi \left(y (L-\underline{L})\Psib\right) &= +\frac{1}{2} \left(L-\underline{L}\right)\Big\{\int \sin \theta d\theta d\phi \,  w y |\mathring{\slashed{\nabla}}^{[s]} \Psi |^2 \Big\} \nonumber\\ &- \frac{1}{2} \left[ \left(L-\underline{L}\right) \left(w y \right)\right] \int \sin \theta d\theta d\phi \,|\mathring{\slashed{\nabla}}^{[s]} \Psi |^2  
\end{align}

\subsubsection{Part III: $ w 2 a T \Phi \Psi$}

\begin{align}
w 2 a {\rm Re}\{T \Phi \Psi \left(y (L-\underline{L})\Psib\right)\} &=w  a y {\rm Re}\{\left(L+\underline{L} -2\frac{a}{r^2+a^2}\Phi \right) \Phi\Psi \left( (L-\underline{L})\Psib\right)\} \nonumber \\
&=-\Phi {\rm Re}\Big\{\frac{2a^2}{r^2+a^2}  w y  (\Phi \Psi) \left( (L-\underline{L})\Psib\right) \Big\} +   (L-\underline{L}) \Big \{ \frac{a^2}{r^2+a^2}  w y  |\Phi\Psi|^2  \Big\} \nonumber \\
&\ \ \ \ - \left[ (L-\underline{L}) \left( \frac{a^2}{r^2+a^2}  w y \right)\right] |\Phi\Psi|^2
 +\frac{1}{2} \Phi \Big\{ way |L\Psi|^2\Big\} - \frac{1}{2} \Phi \Big\{ way |\underline{L}\Psi|^2\Big\} \nonumber \\
& \ \ \ \ -L {\rm Re}\Big\{ way \Phi \Psi \underline{L} \Psib \Big\} + \underline{L} {\rm Re}\Big\{ way \Phi \Psi {L} \Psib\Big\} -\frac{4ra^2}{(r^2+a^2)^2} \frac{\Delta}{r^2+a^2} w y |\Phi \Psi|^2
\nonumber \\
& \ \ \ \ + a \left[ L \left(wy\right) \right]{\rm Re}\{ (\Phi\Psi)(\underline{L} \Psib)\} - a \left[ \underline{L} \left(wy\right) \right] {\rm Re}\{(\Phi\Psi)({L} \Psib)\}
\end{align}

\subsubsection{Part IV: $ wa^2 \sin^2\theta TT\Psi$}

\begin{align}
w a^2 \sin^2\theta {\rm Re}\{TT\Psi \left(y (L-\underline{L})\Psib\right)\} =& T {\rm Re}\Big\{ w a^2 \sin^2\theta T\Psi \left(y (L-\underline{L})\Psib \right) \Big\} - \frac{1}{2} \left(L-\underline{L}\right) \Big\{w a^2 \sin^2\theta y |T\Psi|^2 \Big\} \nonumber \\
&+\left(w y\right)^\prime a^2 \sin^2\theta | T\Psi|^2
\end{align}

\subsection{The redshift multiplier: $\frac{1}{w}\xi$\underline{$L$}$\overline{\Psi}$} \label{sec:comprs}
\subsubsection{Part I: $\frac{1}{2} \left( L \underline{L} + \underline{L} L \right)\Psi$}

\begin{align}
\frac{1}{2} \textrm{Re} \left( L \underline{L} + \underline{L} L \right)\Psi \left(\frac{1}{w} \xi \underline{L}\Psi\right)  &= L \underline{L} \Psi \left(\frac{1}{w} \xi \underline{L}\Psi\right) - \frac{1}{2} \textrm{Re} \left( \left[L,\underline{L}\right] \Psi \left(\frac{1}{w} \xi \underline{L}\Psi\right) \right)\nonumber \\
&=\frac{1}{2} L  \left(\frac{1}{w} \xi |\underline{L}\Psi|^2\right) -\frac{1}{2} \left(\frac{\xi}{w}\right)^\prime |\underline{L}\Psi|^2 + \frac{2r a \xi}{r^2+a^2}  \textrm{Re} \left( \Phi \Psi \underline{L}\Psi\right) \, .
\end{align}

\subsubsection{Part II: $w\left(\mathring{\slashed\triangle}^{[s]}_m \left(0\right) +s\right) \Psi$}

\begin{align}
\textrm{Re} \left( w\left(\mathring{\slashed\triangle}^{[s]}_m \left(0\right) +s\right) \Psi \left(\frac{1}{w} \xi \underline{L}\overline{\Psi}\right) \right)= +\frac{1}{2} \underline{L}\Big\{  \xi |\mathring{\slashed{\nabla}} \Psi |^2 \Big\} +\frac{1}{2}\xi^\prime |\mathring{\slashed{\nabla}} \Psi |^2  
\end{align}

% 94e070c29e

\subsubsection{Part III: $ w 2 a T \Phi \Psi$}

\begin{align} \label{imuz}
w 2 a T \Phi \Psi \left(\frac{1}{w} \xi \underline{L}\overline{\Psi}\right) &=  a \xi \left(L+\underline{L} -2\frac{a}{r^2+a^2}\Phi \right) \Phi\Psi \underline{L}\overline{\Psi} \nonumber \\
&=-\Phi \Big\{ \frac{2a^2}{r^2+a^2}  \xi  (\Phi \Psi) \underline{L}\overline{\Psi}\Big\} +   \underline{L} \Big\{ \frac{a^2}{r^2+a^2}  \xi  |\Phi\Psi|^2  \Big\} \nonumber \\
&\ \ \ \ - \left[ \underline{L} \left( \frac{a^2}{r^2+a^2}  \xi \right)\right] |\Phi\Psi|^2
 +\frac{1}{2} \Phi \Big\{ a\xi |\underline{L}\Psi|^2\Big\} + a \xi L  \Phi\Psi \underline{L}\overline{\Psi} \, .
\end{align}
In view of
\begin{align}
a\xi L\Phi \Psi \underline{L} \overline{\Psi} &= L \left(a \xi \Phi \Psi \underline{L} \overline{\Psi} \right) -  a \xi^\prime \Phi \Psi \underline{L} \overline{\Psi} - a\xi \Phi \Psi \left[L, \underline{L} \right] \overline{\Psi} - a\xi \Phi \Psi \underline{L} L\overline{\Psi} \nonumber \\
&=L \left(a \xi \Phi \Psi \underline{L} \overline{\Psi} \right) -  a \xi^\prime \Phi \Psi \underline{L} \overline{\Psi} - a\xi \Phi \Psi \left[L, \underline{L} \right] \overline{\Psi} - \underline{L} \left( a \xi \Phi \Psi L \overline{\Psi}  \right)  - a\xi^\prime \Phi \Psi L\overline{\Psi} + a\xi \Phi \underline{L}\Psi L \overline{\Psi} 
\end{align}
and hence
\begin{align}
2 \textrm{Re} \left( a\xi L\Phi \Psi \underline{L} \overline{\Psi}\right) = & L \left(a \xi \textrm{Re} \left(\Phi \Psi \underline{L} \overline{\Psi}\right) \right) -  a \xi^\prime \textrm{Re} \left( \Phi \Psi \left(\underline{L} + L\right) \overline{\Psi}\right)- a\xi \textrm{Re} \left( \Phi \Psi \left[L, \underline{L} \right] \overline{\Psi}\right) \nonumber \\
& \ \ - \underline{L} \left( a \xi \textrm{Re} \left( \Phi \Psi L\overline{\Psi}\right)  \right) +  \Phi \left(a \xi \textrm{Re} \left( \underline{L}\Psi L \overline{\Psi} \right)\right) \nonumber
\end{align}
we conclude from (\ref{imuz})
\begin{align}
\textrm{Re} \left( w 2 a T \Phi \Psi \left(\frac{1}{w} \xi \underline{L}\overline{\Psi}\right) \right)
&=\Phi \Big\{-\frac{2a^2}{r^2+a^2}  \xi \textrm{Re} \left( \Phi \Psi \underline{L}\overline{\Psi} \right)+ \frac{1}{2} a \xi \textrm{Re} \left( \underline{L}\Psi L \overline{\Psi} \right)+\frac{1}{2} a\xi |\underline{L}\Psi|^2\Big\} 
 \\ &\ \ \ \
+   \underline{L} \Big \{ \frac{a^2}{r^2+a^2}  \xi  |\Phi\Psi|^2 - \frac{1}{2}a\xi \textrm{Re} \left( \Phi\Psi L\overline{\Psi} \right) \Big\} +L \left(\frac{1}{2} a \xi \textrm{Re} \left( \Phi \Psi \underline{L} \overline{\Psi} \right) \right) \nonumber \\
&\ \ \ \ - \left[ \underline{L} \left( \frac{a^2}{r^2+a^2}  \xi \right)\right] |\Phi\Psi|^2- \frac{1}{2} a \xi^\prime \textrm{Re} \left( \Phi \Psi \left(\underline{L} + L\right) \overline{\Psi}\right)- \frac{1}{2} a\xi \textrm{Re} \left( \Phi \Psi \left[L, \underline{L} \right] \overline{\Psi}\right)
 \, . \nonumber
\end{align}

\subsubsection{Part IV: $ wa^2 \sin^2\theta TT\Psi$}

\begin{align}
\textrm{Re} \left( w a^2 \sin^2\theta TT\Psi\left(\frac{1}{w} \xi \underline{L}\overline{\Psi}\right)\right)=& T \Big\{ \xi a^2 \sin^2\theta \textrm{Re} \left( T\Psi \underline{L}\overline{\Psi}\right) \Big\} - \frac{1}{2} \underline{L} \Big\{\xi a^2 \sin^2\theta |T\Psi|^2 \Big\} \nonumber \\
&-\frac{1}{2} \xi^\prime a^2 \sin^2\theta | T\Psi|^2 \, .
\end{align}

\subsection{The $r^p$ multiplier: $r^p \beta_k \xi L\overline{\Psi}$ with $\beta_k=1+k\frac{M}{r}$} \label{sec:comprp}
\subsubsection{Part I: $\frac{1}{2} \left( L \underline{L} + \underline{L} L \right)\Psi$}

\begin{align}
\textrm{Re} \left( \frac{1}{2} \left( L \underline{L} + \underline{L} L \right)\Psi \left(r^p\beta_k\xi L\overline{\Psi}\right)\right)  &= \textrm{Re} \left( \underline{L} L\Psi \left(r^p\beta_k \xi L\overline{\Psi}\right)\right) + \textrm{Re} \left( \frac{1}{2} \left[L,\underline{L}\right] \Psi \left(r^p \beta_k\xi L\overline{\Psi}\right)\right) \nonumber \\
&=\frac{1}{2} \underline{L}  \left(\xi r^p \beta_k |{L}\Psi|^2\right) +\frac{1}{2} \left( \xi \left(pr^{p-1} + \mathcal{O}\left(r^{p-2}\right)\right) + \xi^\prime r^p \beta_k \right) |{L}\Psi|^2 \nonumber \\
& \ \ \ - \frac{2r a \xi r^p \beta_kw}{r^2+a^2}  \textrm{Re} \left( \Phi \Psi {L}\overline{\Psi}\right)
\end{align}

\subsubsection{Part II: $w\left(\mathring{\slashed\triangle}^{[s]}_m \left(0\right) +s\right) \Psi$}

\begin{align}
\textrm{Re}\left( w\left(\mathring{\slashed\triangle}^{[s]}_m \left(0\right) +s\right) \Psi \left(r^p\xi L\Psi\right)\right)  &= +\frac{1}{2} {L}\Big\{ w \xi r^p \beta_k |\mathring{\slashed{\nabla}} \Psi |^2 \Big\} \nonumber \\
& \ \ +\frac{1}{2}\left(\xi \left[\frac{\left(2-p\right)}{r^{3-p}} + \frac{(3-p)\left(k-2\right)M}{r^{4-p}} + \mathcal{O}\left(r^{p-5}\right) \right]+  \frac{\xi^\prime w}{r^{-p}}\right) |\mathring{\slashed{\nabla}} \Psi |^2  
\end{align}

\subsubsection{Part III: $ w 2 a T \Phi \Psi$}

\begin{align} \label{imu}
\textrm{Re} \left( w 2 a T \Phi \Psi \left(r^p \beta_k \xi {L}\overline{\Psi}\right) \right)&=  \textrm{Re} \left( a \xi w r^p \beta_k \left(L+\underline{L} -2\frac{a}{r^2+a^2}\Phi \right) \Phi\Psi {L}\overline{\Psi}\right) \\
&=-\Phi \Big\{\textrm{Re} \left( \frac{2a^2}{r^2+a^2}  \xi w r^p \beta_k  (\Phi \Psi) L\overline{\Psi} \right)\Big\} +  {L} \Big \{ \frac{a^2}{r^2+a^2}  \xi w r^p \beta_k |\Phi\Psi|^2  \Big\} \nonumber \\
&\ \ \ \ -  \left( \frac{a^2}{r^2+a^2}  \xi w r^p \beta_k \right)^\prime |\Phi\Psi|^2
 +\frac{1}{2} \Phi \Big\{ a\xi w r^p \beta_k |L\Psi|^2\Big\} + \textrm{Re} \left(a \xi w r^p \beta_k \underline{L}  \Phi\Psi {L}\overline{\Psi}\right) \nonumber \, .
\end{align}
In view of
\begin{align}
a\xi w r^p \beta_k \underline{L}\Phi \Psi {L} \overline{\Psi} &= \underline{L} \left(a \xi w r^p \beta_k \Phi \Psi {L} \overline{\Psi} \right) +  \left(a \xi w r^p \beta_k\right)^\prime \Phi \Psi {L} \overline{\Psi} + a\xi w r^p \beta_k \Phi \Psi \left[L, \underline{L} \right] \overline{\Psi} - a\xi w r^p \beta_k \Phi \Psi  L\underline{L}\overline{\Psi} \nonumber \\
&= \underline{L} \left(a \xi w r^p \beta_k \Phi \Psi {L} \overline{\Psi} \right) +  \left(a \xi w r^p \beta_k\right)^\prime \Phi \Psi {L} \overline{\Psi} + a\xi w r^p \beta_k \Phi \Psi \left[L, \underline{L} \right] \overline{\Psi} \nonumber \\
&  - {L} \left( a \xi w r^p \beta_k \Phi \Psi \underline{L}\overline{\Psi}  \right)  + \left( a\xi w r^p \beta_k\right)^\prime \Phi \Psi \underline{L}\overline{\Psi} + \Phi\left( a\xi w r^p \beta_k  L\Psi \underline{L} \overline{\Psi} \right) - a\xi w r^p \beta_k L \Psi \underline{L} \Phi \overline{\Psi}  \nonumber
\end{align}
we conclude from (\ref{imu})
\begin{align}
&\textrm{Re} \left( w 2 a T \Phi \Psi \left(r^p \beta_k \xi {L}\overline{\Psi}\right) \right)
=\Phi \Big\{-\textrm{Re} \left( \frac{2a^2 w r^p}{r^2+a^2}  \xi  \beta_k  (\Phi \Psi) L\overline{\Psi} \right) + \frac{1}{2} a\xi w r^p \beta_k |L\Psi|^2 + \frac{1}{2} \textrm{Re}\left( a\xi w r^p \beta_k  L\Psi \underline{L} \overline{\Psi}\right) \Big\} \nonumber \\
& \phantom{XXXXXX} +  {L} \Big \{ \frac{a^2 w r^p}{r^2+a^2}  \xi  \beta_k |\Phi\Psi|^2 -\frac{1}{2} \textrm{Re} \left(a \xi w r^p \beta_k \Phi \Psi \underline{L}\overline{\Psi}\right) \Big\} + \underline{L} \Big\{\frac{1}{2} \textrm{Re} \left(a \xi w r^p \beta_k \Phi \Psi {L} \overline{\Psi} \right) \Big\} \nonumber \\
& \phantom{XXXXXX} -  \left( \frac{a^2}{r^2+a^2}  \xi w r^p \beta_k \right)^\prime |\Phi\Psi|^2 + \frac{1}{2} \textrm{Re} \left( \left(a \xi w r^p \beta_k\right)^\prime \Phi \Psi \left( {L} + \underline{L}\right) \overline{\Psi}  \right) + \frac{1}{2} \textrm{Re} \left(a\xi w r^p \beta_k \Phi \Psi \left[L, \underline{L} \right] \overline{\Psi}\right)
  \nonumber \, .
\end{align}

\subsubsection{Part IV: $ wa^2 \sin^2\theta TT\Psi$}

\begin{align}
\textrm{Re} \left( w a^2 \sin^2\theta TT\Psi\left(r^p\beta_k \xi {L}\overline{\Psi}\right)\right)=& T \Big\{ w a^2 \sin^2\theta r^p \beta_k \xi \textrm{Re} \left( T\Psi {L}\overline{\Psi}\right) \Big\} - \frac{1}{2} {L} \Big\{w a^2 \sin^2\theta \xi  r^p\beta_k  |T\Psi|^2 \Big\} \nonumber \\
& \ \ +\frac{1}{2} \left(\xi \left(w r^p \beta_k\right)^\prime + \xi^\prime w r^p \beta_k \right) a^2 \sin^2\theta | T\Psi|^2 \, .
\end{align}

\bibliographystyle{DHRalpha}
\bibliography{Teukpaper}

\end{document}